\documentclass[a4wide,11pt,leqno]{amsart}
\usepackage{a4wide,color,array}
\usepackage{mathtools}
\usepackage{hyperref}
\usepackage{cite}
\hypersetup{linktocpage,colorlinks, citecolor={magenta}}
\usepackage{esint,amsmath,amssymb,amsthm}
\usepackage[english]{babel}
\usepackage[utf8]{inputenc}
\usepackage[cyr]{aeguill}
\usepackage{esint}
\usepackage{xargs}
\usepackage{enumitem}
\usepackage{lmodern}
\usepackage{bm}
\usepackage{mathrsfs}
\usepackage{accents}

\usepackage{tikz}
\usetikzlibrary{shapes.misc}
\usetikzlibrary{decorations.pathreplacing,calligraphy}
\usetikzlibrary{shapes,arrows}
\tikzset{cross/.style={cross out, draw=black, minimum size=2*(#1-\pgflinewidth), inner sep=0pt, outer sep=0pt},
cross/.default={1pt}}
\usepackage{subcaption}

\newcommand{\Z}{\mathbb{Z}}
\newcommand{\Q}{\mathbb{Q}}
\newcommand{\R}{\mathbb{R}}
\newtheorem{theo}{Theorem}
\newtheorem*{theo*}{Theorem}
\newtheorem{prop}{Proposition}[section]
\newtheorem{lem}[prop]{Lemma}
\newtheorem{coro}[prop]{Corollary}
\newtheorem{remark}[prop]{Remark}
\newtheorem{defi}[prop]{Definition}

\theoremstyle{plain}

\numberwithin{equation}{section}

\def\t0{\rightarrow 0} % Vers zÃ©ro
\newcommand{\f}{\frac}

\newcommand{\ep}{\varepsilon}

\newcommand{\hal}{\frac{1}{2}}

\newcommand{\supp}{\mathrm{supp }\,} % Support
\def\div{\mathrm{div} \, } % Divergence
\def\1{\mathbf{1}} % Fonction caractÃ©ristique

\def \mc{\mathcal}
\def \ep{\varepsilon}
\def\ux{X}

\def \ZNbeta{Z_{N,\beta}} % Fonction de partition

\def\({\left(}
\def\){\right)}

\def\yg{|y|^\gamma}
\def\P{\mathbb{P}} % ''Vraies'' mesures
\def \PNbeta{\P_{N, \beta}} % Mesure de Gibbs Ã  \beta
\def \PgN2{\mathbf{P}_{N,2}} % Processus Gibbsien Ã  N points et Ã  \beta = 2..
\def \HN{\mathcal{H}_N}

\def\Esp{\mathbb{E}} % EspÃ©rance
\def \E{\Esp}

\def \Ent{\mathrm{Ent}}   % Entropie relative classique
\def \F{\mathcal{F}} % Rate function pour le Sanov de la mesure de 

\def \K{\mathcal{K}}

\def\M{\mathsf{M}}

\def\Ani{\mathsf{A}}

\def \dist{d}
\def\l{\ell}

\def \dist{\mathrm{dist}}

\newcommand{\cm}[1]{{\color{blue}{*** #1 ***}}}
\def\muv{\meseq}

\def\I{\mathcal I}

\def\nab{\nabla}

\def\indic{\mathbf{1}}

\def \Var{\mathrm{Var}}

\def \XN{X_N}

\def\K{\mathsf{K}}
\def\F{\mathsf{F}}
\def\FN{\F_N}

\def\G{\mathsf{G}}

\def\muv{\mu_V}
\def\mut{\mu_t}

\def\rr{\mathsf{r}}
\def\rrh{\hat{\mathsf{r}}}
\def\mn{{\bar{\mathrm{n}}}}
\def\rrt{\tilde{\mathsf{r}}}

\def\Xint#1{\mathchoice
   {\XXint\displaystyle\textstyle{#1}}%
   {\XXint\textstyle\scriptstyle{#1}}%
   {\XXint\scriptstyle\scriptscriptstyle{#1}}%
   {\XXint\scriptscriptstyle\scriptscriptstyle{#1}}%
   \!\int}
\def\XXint#1#2#3{{\setbox0=\hbox{$#1{#2#3}{\int}$}
     \vcenter{\hbox{$#2#3$}}\kern-.5\wd0}}
\def\dashint{\Xint-}

\def \carr{\square} % CarrÃ©

\def \Old{\mathcal{O}}
\def \New{\mathcal{N}}

\def \Escr{E^{\rm{scr}}}

\def\P{\mathbb{P}} % ''Vraies'' mesures
\def \PNbeta{\P_{N, \beta}} % Mesure de Gibbs Ã  \beta
\def \PgN2{\mathbf{P}_{N,2}} % Processus Gibbsien Ã  N points et Ã  \beta = 2..
\def\g{\mathsf{g}}

\def \I{\mathcal{E}}

\def \C{\mathcal{C}}

\def\nab{\nabla}

\def\pa{{\partial}}

\def\ep{\varepsilon}

\def \fluct{\mathrm{fluct}}
\def\Fluct{\mathrm{Fluct}}

\def\hal{\frac{1}{2}}

\makeatletter
\def\namedlabel#1#2{\begingroup
    #2%
    \def\@currentlabel{#2}%
    \phantomsection\label{#1}\endgroup
}

\makeatother
\def \d{\mathsf{d}}
\def \s{\mathsf{s}}

\def \f{\mathsf{f}}

\def \p{\partial}

\def \cds{\mathsf{c}_{\d,\s}} %% Constante c,d,s.
\def \c{\cds}
\def\Rd{\R^\d} % Espace physique

\def \drd{\delta_{\Rd}}

\def\cd{\mathsf{c}_{\d}}

\def\Esp{\mathbb{E}} % Espâ?šÂ©rance
\def \E{\Esp}

\def \be{\begin{equation}}
\def \ee{\end{equation}}
\def \beq*{\begin{equation*}}
\def \eeq*{\end{equation*}}
\def \ba{\begin{eqnarray}}
\def \ea{\end{eqnarray}}
\def \ba*{\begin{eqnarray*}}
\def \ea*{\end{eqnarray*}}

\def\N{{n_\mathcal{O}}}

\def\mf{f_{\d,\s}}

\def\id{I}

\def\omc{{\overset{\circ}{\Omega}}}

\def\Vt{V_t}
\def\a{\alpha}

\theoremstyle{definition}

\newtheorem{rem}[prop]{Remark}

\def\veta{\vec{\eta}}

\renewcommand{\div}{\mathrm{div}\,}

\def\bulk{\hat \Sigma}
\def\Gc{\mathcal{G}}

\makeatletter
\let\@wraptoccontribs\wraptoccontribs
\makeatother

\author[Luke Peilen]{Luke Peilen}
\address[Luke Peilen]{Department of Mathematics, Temple University, 1805 N. Broad St., Philadelphia, PA 19122, United States}
\email{luke.peilen@temple.edu}

\author[Sylvia Serfaty]{Sylvia Serfaty}
\address[Sylvia Serfaty]{Courant Institute of Mathematical Sciences, New York University, 251 Mercer St., New York, NY 10012, United States\\ and Sorbonne Universit\'e,
 CNRS, Universit\'e de Paris,  Laboratoire Jacques-Louis Lions (LJLL), F-75005 Paris }
\email{serfaty@cims.nyu.edu}

\contrib[with an appendix by]{Xavier Ros-Oton}
\address{ICREA, Pg. Llu\'is Companys 23, 08010 Barcelona, Spain \& Universitat de Barcelona, Departament de Matem\`atiques i Inform\`atica, Gran Via de les Corts Catalanes 585, 08007 Barcelona, Spain \& Centre de Recerca Matem\`atica, Edifici C, Campus Bellaterra, 08193 Bellaterra, Spain}
\email{xros@icrea.cat}

\title{Local Laws and Fluctuations for Super-Coulombic Riesz Gases}
\begin{document}
\maketitle
%%%%%%%%%%%%%%%%%%
\begin{abstract}
We study the local statistical behavior of the super-Coulombic Riesz gas of particles in Euclidean space of arbitrary dimension, with inverse power distance repulsion integrable near $0$, and with a general confinement potential,  in a certain regime of inverse temperature. Using a bootstrap procedure, we prove local laws on the next order energy and control on fluctuations of linear statistics that are valid down to the microscopic lengthscale, and provide controls for instance, on the number of particles in a (mesoscopic or microscopic) box, and the existence of a limit point process up to subsequences.

As a consequence of the local laws, we derive an almost additivity of the free energy that allows us to exhibit for the first time a CLT for Riesz gases corresponding to small enough inverse powers, at small mesoscopic length scales, which can be interpreted as the convergence of the associated potential to a fractional Gaussian field.

Compared to the Coulomb interaction case, the main new  issues arise from the nonlocal aspect of the Riesz kernel. This manifests in  (i) a  novel technical difficulty in generalizing the transport approach of Lebl\'e and the second author to the Riesz gas which now requires analyzing a degenerate and singular elliptic PDE, (ii)  the fact that the transport map is not localized, which makes it more delicate to localize the estimates, (iii) the need for coupling the local laws and the fluctuations control inside the same bootstrap procedure. 
\end{abstract}

%%%%%%%%%%%%%%%Introduction
\section{Introduction}

We are interested in proving local laws and studying the fluctuations of super-Coulombic  Riesz gases. These are ensembles of point configurations $\XN=(x_1, \dots, x_N)$ with $x_i\in \R^\d$ whose law is given by 
\be \label{PN}
d\PNbeta(x_1, \dots, x_N)= \frac{1}{\ZNbeta} e^{-\beta N^{-\frac{\s}{\d}} \HN(x_1, \dots, x_N)} dx_1 \dots dx_N\ee
with 
\be \label{defHN}
 \HN(\XN)= \hal \sum_{i\neq j} \g(x_i-x_j) + N \sum_{i=1}^N V(x_i)\ee
 where 
 %\be \label{glog}
 %\g(x)=-\log |x| \quad \text{for} \  \s=0,\ee or 
\be \label{riesz}
\g(x)=\begin{cases} \frac{1}{\s} |x|^{-\s} &\quad  \text{for} \  \s \neq 0\\
  -\log |x| &\quad \text{for} \ \s=0\end{cases} \ee
with  the condition 
\be\label{intervalles} \d-2< \s<\d .\ee
The condition \eqref{intervalles} that we will use throughout implies that the case  $\s=0$, or log gas case,  is then only encountered in dimension $\d=1$. %or $\d=2$.
The case $\s=\d-2$    corresponds to the Coulomb case in any dimension, this is why the condition \eqref{intervalles} corresponds to a super-Coulombic Riesz gas. Note that as  $\s$ becomes larger than $\d$, the  kernel becomes nonintegrable near $0$ but integrable at infinity.  The regime $\s>\d$,  called the {\it hypersingular case} \cite{BHS}, corresponds to a short-range  interaction regime, and the behavior is quite different (see for instance \cite{hlss}), hence the restriction to $\s<\d$. The sub-Coulomb case $\s<\d-2$ on the other hand is  longer range due to the slower decay of the interaction, bringing in new difficulties that are outside the scope of this paper.

%Note that the case $\s>\d$, called the hypersingular case, is very different because the interaction decays faster, we refer to \cite{HLSS} \cm{add reference} and \cite{BHS19}, while the sub-Coulomb Riesz case of $\s \le \d-2$ is also quite different due to the slower decay of the interaction, and remains open.

The function $V$ is an external confining potential, on which we shall place assumptions later. 

The parameter $\beta>0$ is an inverse temperature, and we have chosen to multiply the energy by $N^{-\frac\s\d}$ because $N^{\frac\s\d}$ is the typical energy per particle.
Finally, the normalization factor
\be \label{defZN}
\ZNbeta:= \int_{(\R^\d)^N} e^{-\beta N^{-\frac\s\d} \HN(x_1, \dots, x_N)} dx_1\dots dx_N\ee is called the partition function.

The Coulomb case is particularly natural and physical, because $\g$ is the fundamental solution to the Laplacian:
\be \label{gcoul} -\Delta \g= \cd \delta_0\ee
where $\cd$ is a constant depending only on the dimension, and $\delta_0$ is the Dirac mass. In the Riesz case with $\s\in (\d-2,\d)$, and in that interval only, $\g$ is instead the fundamental solution to a {\it fractional Laplacian}
\be \label{eqgriesz} (-\Delta)^{\frac{\d-\s}{2}}\g= \cds \delta_0. \ee
In the following we denote
\be \label{defalpha}
\alpha=\frac{\d-\s}{2}.\ee
The constant $\cds$ is given by (see \cite{K17})
\begin{equation}\label{constant definition}
\cds= \frac{\pi^{\frac\d2}4^\alpha \Gamma(\alpha)}{\Gamma(\frac\d2-\alpha)}, \quad \alpha=\frac{\d-\s}{2}.
\end{equation}
%In particular, $\cds$ is nonzero and depends only on $\d$ and $\s$. 
% is another explicit constant depending only on $\d$ and $\s$, we refer to the introduction of \cite{S24} for the explicit formulae.

 We recall that the fractional Laplacian is a nonlocal operator (contrarily to the Laplacian associated to the Coulomb case). It can be seen as an integral operator defined by (its  various definitions can be found for instance in \cite{K17} and \cite{LPGe20})
\begin{equation}\label{def fraclap}
(-\Delta)^\alpha f(x)=\text{P.V.}\int (f(x)-f(x+y))\frac{c_{\d,\alpha}}{|y|^{\d+2\alpha}}dy, \quad c_{\d,\alpha}= \frac{4^\alpha \Gamma(\frac\d2+\alpha)}{\pi^{\frac\d2} |\Gamma(-\alpha)|}=\frac{\alpha 4^\alpha \Gamma(\frac\d2+\alpha)}{\pi^{\frac\d2}\Gamma(1-\alpha)}.
\end{equation}
%with normalization constant chosen to correspond to the Fourier multiplier
%\begin{equation}\label{def Ffraclap}
%\mathcal{F}\left((-\Delta)^\alpha f\right)(\xi)=|\xi|^{2\alpha}\mathcal{F}(f)(\xi).
%\end{equation}

We will also use the homogeneous Sobolev norm  defined by 
\be \label{homogsobo}
\|f\|_{\dot{H}^{-\alpha}}^2:= \iint \g(x-y) df(x) df(y)\ee
which is equivalent with the usual $H^s$ definition via Fourier transform 
\be \label{soboFT}
\|f\|_{\dot{H}^{-\alpha}}^2:=\frac{2}{c_{\d,\s}}\int |\xi|^{-2\alpha}|\hat f|^2(\xi) d\xi
\ee
as in \cite[Proposition 3.4]{DPV12}, where $c_{\d,\s}$ is the constant in \eqref{def fraclap}, since $\hat \g(\xi)= C_{\d,\s} |\xi|^{\s-\d}$ for some constant $C_{\d,\s}$ (cf. \cite[Proposition 2.14]{S24}).

The nonlocality of the operator creates much of the difficulties encountered in the Riesz case.

The Coulomb gas is an important model of statistical physics, in particular due to its connection to plasma physics, random matrix theory, quantum mechanics models,  and conformal field theory. We refer to the introduction of \cite{S24} for more detail. The Riesz gas is less understood but also physically interesting (in solid state physics, ferrofluids, elasticity), see \cite{mazars,bbdr,CDR,CDFR,torquato}, and has attracted quite a bit of attention in the recent physics literature, see for instance \cite{schehr1d,schehr1d2,schehr1d3}.

\subsection{The equilibrium measure}\label{sec22}
The first order or mean-field asymptotic behavior of the gas is well understood (see for instance \cite[Chap.~2]{S24}). From \cite{F35}, \cite{C58}, if the potential $V$ is lower semicontinuous, bounded below, finite on a set of positive capacity, and satisfies the growth condition
\begin{equation}
\lim_{|x|\rightarrow +\infty}\(V(x)+\g(x)\)=+\infty,
\end{equation}
then the continuous approximation of the Hamiltonian  given by 
\begin{equation}\label{Continuous energy}
\I(\mu):=\frac{1}{2}\iint_{\R^\d\times \R^\d} \g(x-y)\, d\mu(x)d\mu(y)+\int_{\R^\d} V(x)\, d\mu(x)
\end{equation}
is well-defined as long as $\s<\d$ (hence the restriction to that regime -- this is called the {\it potential case}) and
has a unique, compactly supported minimizer $\muv$ among the set of probability measures on $\R^{\d}$, called the {\it equilibrium measure} and characterized by the following Euler-Lagrange equation: there exists a constant $c_V$ such that 
\begin{equation}\label{Riesz Euler-Lagrange equation}
\begin{cases}
h^{\muv}+V=c_V & \text{quasi-everywhere on }\Sigma \\
h^{\muv}+V \geq c_V & \text{quasi-everywhere}
\end{cases}
\end{equation}
where $h^{\muv}=\g\ast \muv$ is the potential generated by $\muv$, and $\Sigma$ denotes the support of $\muv$.   In the following, we  denote the corresponding \textit{effective potential} by 
\begin{equation}\label{Riesz effective potential}
\zeta_V:=h^{\muv}+V-c_V.
\end{equation}

It is worth noting  that in the Coulomb case, since $-\Delta h^{\muv}= \cd \muv$ in view of \eqref{gcoul}, taking the Laplacian of the first relation in \eqref{Riesz Euler-Lagrange equation} we find
\be \label{densmuv}\muv=\frac{1}{\cd}\Delta V\quad \text{ in } \overset{\circ}\Sigma,\ee where $\overset{\circ}\Sigma$ denotes the interior of $\Sigma$. 
In contrast, such a manipulation is no longer possible in the Riesz nonlocal case, and  there is then no local or explicit expression for $\muv$ in terms of $V$. The interested reader can refer to the  articles  \cite{chafaisaff1,chafaisaff2} for  examples of Riesz equilibrium measures.

In the Riesz case $\s\in [\d-2,\d)$, this equilibrium measure problem can be rephrased in terms of an obstacle problem / a fractional obstacle problem as was observed for instance in \cite[Chapter 2]{S15};  this will be very useful for us as it will allow us to use results in that area on the behavior of $\muv$ and additional results proved in the appendix by  X. Ros-Oton.  Let us recall more precisely the correspondence (the interested reader can also refer to \cite[Section 2.4]{S24}, \cite{CDM16,AS22}: the {\it fractional obstacle problem} 
\begin{equation}\label{fractional obstacle}
\min\{(-\Delta)^\alpha h, h-\varphi\}=0
\end{equation}
with obstacle $\varphi=c_V-V$ and $\alpha=\frac{\d-\s}{2}$ has solution $h=h^{\muv}$. The much-studied classical obstacle problem corresponds to the (Coulomb) case $\alpha=1$. The fractional obstacle problem, studied for instance in \cite{CSS08,CDS17,ROS17}, is another free boundary problem. There are two sets in the solution to \eqref{fractional obstacle}: the set
 $\{h^{\muv}=\varphi\}$ where the solution touches the obstacle is known as the \textit{contact set} or coincidence set, and its boundary is the \textit{free boundary}.  In the set where $h^{\muv}>\varphi$, then $\muv=(-\Delta)^{\alpha} h^{\muv}$ must vanish. This corresponds to the complement of $\Sigma$  in \eqref{Riesz Euler-Lagrange equation}. 
 Note that in general the contact  set $\{\zeta_V=h^{\muv}-\varphi=0\}$ contains the  \textit{droplet} $\Sigma=\supp \muv$, but may be larger, although  generically it is not. 
    In the Coulomb case, imposing  $\Delta V>0$ in a neighborhood of $\Sigma$ ensures that the two sets coincide by taking the Laplacian of $\zeta_V=0$ (as in the computation for \eqref{densmuv}), but in the fractional case this computation does not work. Since we want to appeal to the regularity of the free boundary, we take as an assumption that the droplet and free boundary coincide,    this is accomplished by requiring that $\zeta_V>0$ on $\Sigma^c$, which is guaranteed by assumption \eqref{itemnondeg}.

In the fractional case (contrarily to the Coulomb case), the density of $\muv$ generically vanishes as one approaches $\partial \Sigma$ from the inside, near {\it regular points}. This is one of the main results of the appendix that we will need: if $x_0$ is a regular point of $\partial \Sigma$, then
\begin{equation}\label{decayeqmeasure}
\muv(x) \sim \dist(x,\partial \Sigma)^{1-\alpha} \quad \text{as }  x \to x_0, x\in \Sigma.
\end{equation}
Note that this behavior for instance matches the {\it semi-circle law} behavior of the equilibrium measure  in the one-dimensional log case (for which $\s=0$ and $\alpha=1/2$). 

The recent paper \cite{colombofigalli} also exploits the correspondence with the fractional obstacle problem to prove that this behavior near all boundary points is generic with respect to $V$  in dimension $\d\le 3$, for any $\s \in [\d-2,\d)$. 
For simplicity, we will thus assume that all boundary points are regular, which can be guaranteed by assumption \eqref{itemposLap} and \eqref{itemfbLip} below. 

Additionally, we will need some regularity on the quotient $\frac{\muv}{\dist(x,\partial \Sigma)^{1-\alpha}}$,  this is assumption~\eqref{itemeqreg}.

Finally, the ``lift-off'' rate from the obstacle, i.e.~the growth of $\zeta_V$ is known from the fractional obstacle problem literature: it is 
\be \label{liftoff}\zeta_V(x) \ge c(x) \dist(x, \Sigma)^{1+\alpha}\ee
with $c(x)>0$ near regular points. This is in assumption \eqref{itemnondeg}.

\subsection{Goal of the paper}
Proving local laws for the gas roughly corresponds to understanding how much the distribution of the points deviates from the mean-field distribution $\muv$. We do so via local controls on the {\it next-order electric energy} (called modulated energy in the dynamics context) $\F_N(\XN, \muv)$, first introduced in \cite{PS17} and defined via
\be \label{defFN} \F_N(\XN, \mu):= \hal \iint_{\triangle^c} \g(x-y) \, d\Big( \sum_{i=1}^N \delta_{x_i}-N\mu\Big)(x)\,  d\Big( \sum_{i=1}^N \delta_{x_i}-N\mu\Big)(y),\ee where $\triangle$ denotes the diagonal in $\R^\d\times \R^\d$.
This quantity, whose properties are described in \cite[Chap.~4]{S24},  is bounded below  by $-CN^{1+\frac\s\d}$ where $C$ depends only on $\|\mu\|_{L^\infty}$ and behaves effectively like the square of a distance between the empirical measure $\frac1N\sum_i \delta_{x_i}$ and the reference probability density $\mu$, more precisely like 
$$ N^2 \left\|\frac1N\sum_{i=1}^N \delta_{x_i}-\mu\right\|^2_{\dot{H}^{\frac{\s-\d}{2}}}$$
(where $\dot{H}$ denotes a homogeneous Sobolev norm as in \eqref{homogsobo}--\eqref{soboFT}), but here defined in a {\it renormalized} manner that allows to admit Diracs thanks to the removal of the (infinite) diagonal self-interaction terms.

As shown in prior works \cite{PS17,LS15}, see \cite[Chap.~4]{S24}, the electric energy $\F_N$ provides good controls on the difference 
$ \frac1N\sum_i \delta_{x_i}-\mu$,  such as rough bounds on linear statistics, bounds on minimal distances between points and bounds on  charge discrepancies (i.e. integrals of the difference over balls or cubes). These controls, based on the {\it electric formulation} of the energy $\F_N$, will be recalled in Section \ref{sec:remindersF}, particularly Proposition \ref{pro:controlfluct}.

In parallel with the local laws, we wish to address the related question of understanding the asymptotic behavior of fluctuations of linear statistics, of the form 
\be\label{def:fluct} \Fluct_{\muv}(\varphi):=\sum_{i=1}^N \varphi(x_i)-N \int \varphi d\muv,\ee
for regular enough test-functions $\varphi$.

This program has been in large part completed in the Coulomb case in prior works:  \cite{L17} for local laws in two dimensions (see \cite{bbny1} for related local laws), \cite{AS22} for local laws in general dimension, \cite{LS18,BBNY19} for fluctuations in  the two-dimensional case, \cite{S22} for fluctuations in  general dimension. We also refer the reader to\cite{S24} for a recap. These local laws are proved by a method of  boostrap on scales first introduced in \cite{L17}, while the fluctuations analysis is based on a transport method, introduced in \cite{LS18} and re-used in \cite{S22}.

The one-dimensional logarithmic case, which corresponds to the well-known situation of $\beta$-ensembles, and is closely related to random matrix theory, was intensively studied: fluctuations and questions similar to local laws were analyzed in \cite{J98,Sh13,Sh14,BG13,BG16,BEY12,BEY14,BMP22,BL18,BLS18}, closer to our analysis here is the paper \cite{P24} which really studies local laws for \eqref{defFN} and serves partly as a blueprint for the present paper. Finally, the one-dimensional Riesz gas with general $\s<1$ was extensively studied in \cite{boursier1,boursier2},  leveraging on the convexity  of the interaction in dimension 1, which permits a very different treatment based on the Helffer-Sj\"ostrand representation that yields stronger results.

Our goal here is to carry out the program of \cite{L17,LS18,AS22,S22,P24} for super-Coulombic Riesz gases of arbitrary dimension.
Note that in contrast with \cite{S22} and \cite{S24}, we will not be able to obtain estimates that are optimal in $\beta$ as $\beta$ gets small, because this would require using the ``thermal equilibrium measure'' (see \cite[Chap. 2]{S24}) instead of the equilibrium measure and, contrarily to the Coulomb case \cite{AS22},  its behavior in the nonlocal case is much less understood than the  standard equilibrium measure for which we have detailed knowledge (describe above) thanks to the fractional obstacle problem correspondance.

%independent and can serve to consider $\beta \to 0$ as $N\to \infty$.
 The paper  focuses in particular 
 on the  new difficulties brought by the nonlocality of the fractional Laplacian, which are the following:
 \begin{itemize}
 \item Generalizing the transport approach of \cite{LS18}  to the Riesz gas now requires analyzing a degenerate and singular elliptic PDE.
 This was not needed in \cite{P24} because in the one-dimensional log case, and in that case only, explicit formulas are available for the ``master operator inversion''  or equivalently for the transport map. 
 \item  That transport map ends up having nonlocalized tails, even if the test function $\varphi$ one considers is localized at a mesoscale.
 \item As in prior work, the local laws rely on a {\it  screening procedure}, which itself relies on the electric formulation of the energy together with a dimension-extension procedure to handle the nonlocality. The screening in extended dimension requires subtle adaptations, in particular in this probabilistic setting, and controlling its errors requires one to couple the local laws and the fluctuations control inside the same bootstrap procedure, a difficulty not present in the (local) Coulomb case.
 \end{itemize}

\subsection{Assumptions} \label{sec:def}
We will use the notation $|f|_{C^\sigma}$ for the  H\"older semi-norm of order $\sigma$ for any $\sigma \ge 0$ (not necessarily integer). 
For instance 
$|f|_{C^0}=\|f\|_{L^\infty}$, $|f|_{C^k}= \|D^k f\|_{L^\infty}$ and if $\sigma \in (k, k+1)$ for some $k$ integer, we let
$$ |f|_{C^\sigma(\Omega)}= \sup_{x\neq y \in \Omega} \frac{|D^k f (x)- D^k f (y)|}{|x-y|^{\sigma-k}}.$$
We emphasize that with this convention $f\in C^k$  does not mean that $f$ is $k$ times differentiable but rather that $D^{k-1} f$ is Lipschitz.
For $k\ge 1$ integer, we will then use the notation 
\be\|f\|_{C^k}= \|f\|_{L^\infty}+ \sum_{m=1}^k |f|_{C^m}.\ee

For  our main results, we assume the following, for some integer $k\geq 3$ that will be specified and for some $\epsilon>0$.

 %\cm{in the main regime of low temperature we are interested in $b=-\frac{\s}{\d}$ however it is also interesting to consider larger regimes of temperature, so I keep $b$ general for now}
\begin{enumerate}[label=\textbf{A\arabic{enumi}}]
\item (Nondegeneracy): $\zeta_V\ge c\,  \dist(x, \partial \Sigma)^{1+\alpha}$ with $c>0$  in some strict neighborhood of  $\Sigma$. \label{itemnondeg}
\item (Positive Laplacian): $\Delta V \geq 0$ on $\{c_V>V\}$, where $c_V$ is as in \eqref{Riesz Euler-Lagrange equation}. %\cm{refer to appendix where this assumption is used} 
\label{itemposLap}
\item (Lipschitz Free Boundary): The free boundary $\partial \Sigma$  is a Lipschitz graph. \label{itemfbLip}
%\item (Regularity of External Potential): $V \in C^{3,1}$ \cm{doesn't sound right} \label{itempotreg}
\item (Regularity of Equilibrium Measure):  
\begin{equation}\label{eq reg}
\muv(x)=s(x)\dist(x,\partial \Sigma)^{1-\alpha}
\end{equation}
for some function $s(x) \in C^{k+\epsilon}(\Sigma)$ that is bounded from below. %\cm{is $C^3$ sufficient? maybe $C^k$?}
\label{itemeqreg}
\item (Regularity and growth of $\nab V$): We assume $V \in C^{k+1}$ and there exist $r\ge0$,  $c>0$ and $C>0$  such that 
\begin{equation}|\nab V|\ge  c |x|^r, \quad  |\nab^{\otimes m} V|\ge C |x|^{r-m+1}\quad \text{for } |x|\ge C, \ 1\le m \le k.\end{equation}  
\label{itemgrowthV}
\item (One-cut) $\Sigma=\supp \muv$ is a connected set. \label{onecut}

\end{enumerate}
Note that these assumptions are genericity assumptions that avoid irregular cases. \eqref{itemposLap} allows us to apply Proposition \ref{quantitative-regular}, which shows that all points of the free boundary are regular. The assumption \eqref{eq reg} is for instance reasonable in view of the generic behavior \eqref{decayeqmeasure}, as shown in \cite{colombofigalli}. We also make the additional assumption \eqref{onecut} that $\Sigma$ is connected (what is commonly referred to the \textit{one-cut} regime in the one-dimensional case) to avoid some of the additional difficulties that the nonlocality of the interaction creates in the multicut regime (cf. \cite{BG16} for a treatment of the log-gas in the multicut regime). We refer to further discussion of these assumptions in Section \ref{subsec: prelimfluct}.

%The \textit{electrostatic potential} is given by 
%\begin{equation}\label{truepot}
%u:=\g\ast \(\sum \delta_{x_i}-N\muv\)
%\end{equation}

\subsection{Main results}
We prove the following two theorems in conjunction, by a bootstrap on sales. We state both at the traditional scale (although the local law will be proven at blown up scale).

\subsubsection{Local laws and consequences}

The first theorem establishes local laws {\it  down to the  microscale} $N^{-1/\d}$. They are  expressed in terms of a  version of the energy localized in a cube  $\carr_\ell$ of size $\ell$ (centered at an unspecified point), denoted 
$\tilde \F_N^{\carr_\ell}(\XN)$.
The precise definition of $\tilde \F_N$ will be given in  Section \ref{sec: fluct prelim}, it relies on the electric formulation of the energy  (as first introduced in \cite{PS17}, see \cite[Chap.~4, Chap.~7]{S24}), in terms of the electric potential $h_N$ 
associated to the couple $(\XN, \mu)$ via 
\be h_N[\XN,\mu] = \g* \( \sum_{i=1}^N \delta_{x_i}- N \mu\)\ee 
and extended to be a function of $\R^{\d+1}$. Briefly,
\begin{equation*}
\tilde \F_N^{\carr_\ell}(\XN,\mu)= \int_{\square_\ell\times [-\ell, \ell]} \yg|\nab h_{N, \rr}|^2 
\end{equation*}
where $\yg$ is a weight to be specified later, and $h_{N,\rr}$ is a suitable truncation of $h_N[\XN,\mu]$.
At this point what matters is to know that this quantity is positive, coercive, and controls charge discrepancies and fluctuations of linear statistics in the cube $\carr_\ell$.
The lengthscale $\ell$ will range from $\ell=1$, which corresponds to the macroscale, to $\ell\propto N^{-1/\d}$ which corresponds to the local scale or microscale.

%where $h_{N,\rr}$ is a renormalized version of 
%\be h_N[\XN,\mu] = \g* \( \sum_{i=1}^N \delta_{x_i}- N \mu\)\ee 
%discussed in Section \ref{sec: fluct prelim}, and $\square_\ell$ is a cube of sidelengths in $[\ell,2\ell]$. It is a version of \eqref{defFN} localized to a set $\Omega$, whose precise definition will be given in Section \ref{sec: fluct prelim}, relying on the so-called {\it electric reformulation} of this energy as first introduced in \cite{PS17} (see \cite[Chap. 4, Chap. 7]{S24}).
%We recall from \cite[Prop. 4.2.8]{S24} that there exists a constant $C_0$ depending only on $\|\muv\|_{L^\infty}$ such that 
%for any configuration $\XN\in (\R^\d)^N$, we have 
%\be \label{introC0} \F^{\Omega}(\XN,\muv)+ \frac{\#I_\Omega \log N}{2\d}\indic_{\s=0}+C_0\#\{\XN \cap \Omega\}\ge 0.\ee

The local laws are only  valid in the bulk, i.e.~well in the interior of $\Sigma$, the support  of $\mu_V$,  except if $\ell =1$ (or $\ell $ is larger than a fixed constant) in which case they  are easily seen to hold and are valid globally (we will not state this but refer to \eqref{macrolaw1}). More precisely, they are valid in $\bulk$ defined as 
\be\label{bulk}\bulk:= \{ x\in \Sigma, \dist (x, \pa \Sigma ) \ge \ep\}\ee
where $\ep>0$ is some small fixed number.

In the whole paper we will denote $\#I_\Omega$ the cardinality of $\{\XN\}\cap \Omega$, i.e.~the number of points in the set $\Omega$.
The notation $A\lesssim B$ means that $A/B$ is bounded by a constant independent of $\beta$, $N$, and other parameters of the problem, except possibly for $\d,\s, V, \ep$. The notation
$C_\beta $ and $\C_\beta$ will denote constants which are independent of $\beta $ when $\beta$ is larger than a positive constant, say $1$, and which may depend on $\beta$ for $\beta \le 1$. In the same way $A\lesssim_\beta B$ means $ A \le C_\beta B$ for such a $C_\beta$, and $O_\beta$ is also defined in the same way.
\begin{theo}[Local laws]\label{Local Law}
Assume \eqref{itemnondeg}--\eqref{onecut}.
There exists  constants $C>0, C_0>0, C_1>0, C_2>0$, depending only on $\d,\s, V,\ep$, and $\C_\beta >0$ depending on $\beta$ only if $\beta<1$, and a lengthscale $4 \le \rho_\beta\lesssim_\beta 1$, such that the following holds.
For any $\ell\ge \rho_\beta N^{-\frac1\d}$ and any cube $ \carr_\ell \subset \bulk$, there exists  a good event $\mathcal{G}_{\ell}$ satisfying
\begin{equation}
\PNbeta( \mathcal{G}_{\ell}^c)\leq C_1 e^{-C_2 \beta \ell^\d N}
\end{equation}
 such that, if $\XN \in \mathcal G_\ell$ 
 %and 
 % $\carr_\ell\subset \R^\d$ is any   cube of sidelength $\ell$ included in $\bulk$,  
 then 
\begin{equation}
\label{loiloc}\tilde \F_N^{\carr_\ell}(\XN, \muv) 
 \leq \C_\beta \ell^\d N^{1+\frac{\s}{\d}}
\end{equation}
and \be\label{controlnbreintro}
C_0 \#I_{\carr_\ell}\le \C_\beta \ell^\d N .\ee
%\cm{probability the log term should be removed}

\end{theo}
Here, in contrast with \cite{AS22}, the local laws are not obtained as exponential moments controls, but only with bounds on the event tails. This is due to the nonlocal nature of the problem and the need to couple local laws controls and fluctuations controls at each scale. 
As in the Coulomb case in \cite{AS22}, the local laws are valid down to a {\it temperature-dependent minimal scale} $\rho_\beta$, which depends on $\beta$ and is expected to blow up as $\beta \to 0$.
Because we use a description in terms of the equilibrium measure  rather than in terms of the more precise thermal equilibrium measure as in \cite{S24}, our $\beta$-dependence in the estimates is sharp for $\beta \ge 1$  but less good when $\beta \to 0$, and also our estimate of the minimal scale is not optimal.

Thanks to the coercivity of $\tilde \F_N^{\carr_\ell}(\XN,\mu)$ and the controls  provided in  Proposition \ref{pro:controlfluct}, imported from  
\cite[Chap.~4]{S24}, we deduce the following corollaries.

\begin{coro}Assume \eqref{itemnondeg}--\eqref{onecut}. There exist $C_1, C_2, C>0$ depending only on $\d,\s, V,\ep$ such that the following holds.
\begin{enumerate}
\item(Discrepancy control)
 Let  $B_\ell$ be a ball of radius $\ell\ge \rho_\beta N^{-1/\d}$,  and if $\ell<1$ assume moreover that $B_\ell \subset \bulk$.
 Letting $D(B_\ell):=\int_{B_\ell} (\sum_{i=1}^N \delta_{x_i}-N\mu_V)$, we have either $|D(B_\ell)|\le C  N^{1-\frac1\d} \ell^{\d-1} $ or
\begin{equation}\label{controldiscintro}
\frac{D(B_\ell)^2}{\ell^\s} \left|\min \left(1, \frac{D(B_\ell)}{\ell^\d}\right)\right|\lesssim_\beta \ell^\d N^{1+\frac\s\d}
\end{equation}
 %and in particular 
%\begin{equation}\label{controlnbreintro}
%\#I_{B_\ell} \le C (1+\beta^{-1}) N\ell^\d,
%\ee
except with probability $\le C_1 e^{-C_2\beta N\ell^\d}$.
\item(Fluctuations control) Assume $\carr_\ell$ is a cube of sidelength $\ell\ge \rho_\beta N^{-1/\d}$ included in $\bulk$, and  $\varphi$ is a function such that $\carr_\ell$ contains a $2N^{-1/\d}$-neighborhood of the support of $\varphi$. For any $\eta \ge N^{-1/\d}$ and $\sigma \in [0,1]$, we have 
\begin{multline}\label{controlfluctrough}
\left|\int_{\R^\d} \varphi \( \sum_{i=1}^N \delta_{x_i}multline N d\mu\) \right| \lesssim_\beta \( \eta^{\gamma-1}\|\varphi\|_{L^2(\Omega)}^2 + \eta^{\gamma+1}\|\nab \varphi\|_{L^2(\Omega)}^2\)^\hal \( \ell^\d N^{1+\frac\s\d}\)^\hal \\
+\ell^\d N^{1+\frac{\s-\sigma}{\d}}  |\varphi|_{C^\sigma}
,\end{multline}
except with probability $\le C_1 e^{-C_2\beta N\ell^\d}$.
\item(Minimal distance control)
For a configuration $\XN$, let 
\be \label{defri}\rr_i:= \frac14 \min (\min_{j\neq i} |x_i-x_j|, N^{-1/\d}).\ee
Assume $\carr_\ell$ is a cube of sidelength $\ell$ included in $\bulk$. We have 
\begin{align}
& \sum_{x_i \in \carr_{\ell}} \g(\rr_i) \lesssim_\beta \ell^\d N^{1+\frac\s\d}&\text{if} \ \s\neq 0
\\
 &\sum_{x_i \in \carr_{\ell}} \g(40 \rr_i N^{-\frac1\d}) \lesssim_\beta \ell^\d N^{1+\frac\s\d}&\text{if} \ \s= 0,
\end{align}
except with probability $\le C_1 e^{-C_2\beta N\ell^\d}$.

\end{enumerate}
\end{coro}

%It was proven in \cite[Sec. 4.5]{S24} (drawing on \cite{PS17,RS16,LS15}), that this local energy is a good ``coercive" quantity that provides controls on fluctuations of linear statistics,  number of points that fall in $\Omega$, and interparticle distances, see Proposition \ref{pro:controlfluct}, imported from  
%\cite{S24}.
The control of linear statistics fluctuations of \eqref{controlfluctrough}  is not optimal, we will provide a better one below under stronger regularity assumptions on $\varphi$.

Since \eqref{controlnbreintro} provides an $N$-independent control on the number of points in a ball of size $  \rho_\beta N^{-1/\d}<\ell \lesssim N^{-1/\d}$, blowing up the configuration by the factor $N^{1/\d}$,  taking large enough microscopic balls and using a Borel-Cantelli type argument, we obtain the existence of a limiting point process up to extraction (we will not provide details as they are almost identical  to \cite[proof of Corollary 1.1]{AS22}).

\begin{coro} \label{corolimitpp} Assume \eqref{itemnondeg}--\eqref{onecut}. Let $x \in \bulk$ and for every $i$, let $x_i'= N^{1/\d} x_i$. For fixed $\beta$,  as $N \to \infty$, the law of the point configuration
$\{x_1'-x,\dots, x_N'-x\}$ converges after extraction of a subsequence, to a  limiting point process with simple points and finite first and second correlation functions.

\end{coro}

We thus provide the first evidence of the existence of a limiting point process, that one may call a Riesz-$\beta$ point process, but only after subsequences. Note that the only situation outside of $\s=0$ in which the existence of such a point process is known is the  one-dimensional Riesz gas, thanks to the work \cite{boursier2}.  For $\s=0$, it is known in one dimension and is the sine-$\beta$ point process, or in two dimensions in the special determinantal case of $\beta=2$, with the Ginibre point process.

We note that with the ingredients developed for Theorem \ref{Local Law} we could prove a local version of the Large Deviations Principle on empirical fields obtained in \cite{LS15} characterizing the local point process averaged at a scale $\gg N^{-1/\d}$ as the minimum of a rate function. In the particular case of energy minimizers, i.e.~$\beta=\infty$, taking advantage of the $\beta$-dependence in the estimates, we could in particular derive precise equidistribution of the energy and number of points as in \cite{PRN} but without the extra decay assumption needed there.
 The details, both in the case of general $\beta$ and $\beta=\infty$ would be completely similar to the proof in \cite{AS21}, so we omit these results.
 
%\cm{Include here free energy expansion with a rate}
\subsubsection{Fluctuations of linear statistics}

Our second theorem concerns fluctuations of linear statistics, as defined in \eqref{def:fluct}.
We will 
assume that $\varphi\in C^k_c(\R^\d) $, that $ \rho_\beta N^{-1/\d}<\ell<1$ and  $\supp \varphi \subset \carr_\ell \subset \carr_{2\ell}\subset \bulk$  for some cube $\carr_\ell$ of sidelength $\ell$. The case $\ell =1$ corresponds to the macroscopic case, the case $\ell \ll 1$ to the mesoscopic case. We believe that in the macroscopic case $\ell=1$, our results hold without the interior condition $\supp \varphi \subset \bulk$, however the  PDE analysis of the transport map equation is a little trickier and we chose not to pursue this generality here.

Because of the nonlocal nature of the problem, the results always involve $\alpha$-{\it harmonic extensions} of the test-functions. These are defined as follows: if $\varphi$ is a function in $\R^\d$, we let 
$\varphi^\Sigma$ denote its $\alpha$-harmonic extension to $\Sigma^c$, that is the solution to 
\begin{equation}\label{aharmext}
\begin{cases}
\varphi^{\Sigma}=\varphi & \text{in }\Sigma \\
(-\Delta)^\alpha \varphi^\Sigma=0 & \text{in } \Sigma^c.
\end{cases}
\end{equation}
This extension is possible for nice enough functions, and enjoys some regularity, that will be important to us. In particular it is shown in Lemma \ref{lem1}   that 
%\cm{is it really what is proved in the appendix? please check} 
\be(-\Delta )^\alpha \varphi^\Sigma =w(x)\dist(x, \pa \Sigma)^{-\alpha} \quad \text{as} \ x\to \pa \Sigma,  x \in \Sigma\ee
  for some bounded $w(x)$  
 and we will need to assume some regularity of  $w$  on $\pa \Sigma$. In particular, we will assume that 
 \be\label{assumpw}
 w \in C^k(U), \qquad \|w\|_{C^k} \lesssim \left\|\frac{\varphi^\Sigma -\varphi}{\dist(x,\partial \Sigma)^{\alpha}}\right\|_{C^k(U \setminus \Sigma)}.
 \ee
 \eqref{assumpw} holds for $k=0,1$ as a consequence of Lemma \ref{def aharmext} and Lemma \ref{lem1} in the Appendix ($k=0$) and \cite[Lemma 3.2]{FR2}. It is expected to hold for all $k$ once $\varphi \in C^{k+\alpha+\epsilon}$, although we were not able to locate this result in the literature.

We develop a Riesz transport method (counterpart of that of  \cite{LS18}, described in \cite[Chap 9,10]{S24}, in the Coulomb case), to obtain the following control of fluctuations of linear statistics, as defined in \eqref{def:fluct}. Here $\carr_\ell(z)$ denotes the cube of size $\ell$ centered at $z$.
\begin{theo}[Fluctuations control]\label{FirstFluct}
Assume \eqref{itemnondeg}--\eqref{onecut} with $k=5$ and \eqref{assumpw} with $w\in C^{5+\epsilon}$. 

Assume  $\varphi\in C^5_c(\R^\d) $, that $\ell \ge \rho_\beta N^{-1/\d}$, and  $\supp \varphi \subset \carr_\ell (z)\subset\carr_{2\ell}(z)\subset \bulk$, with the estimates 
\be\label{estxiintro}
\forall \sigma \le 5, \quad |\varphi|_{C^\sigma}\le \M \ell^{-\sigma}
%, \qquad \|(-\Delta)^{\alpha}\varphi\|_{L^\infty} \le \M \ell^{-2\alpha}
\ee
for some constant $\M>0$.
%Assume that  that $\varphi$ satisfies  \eqref{assumpw} for some $w(x)\in C^\sigma(U)$, $\sigma>0.$ 
%Finally assume
%\eqref{assconncomp}.

%Assume $\varphi= \varphi_0(\frac{\cdot-z}{\ell}) $ for some $\varphi_0\in C^5_c(\R^\d) $ and that either $\ell=1$ (macroscopic case), or that $ \rho_\beta<\ell<1$ and  $\supp \varphi_0 \subset \carr_\ell \subset \bulk$  for some cube $\carr_\ell$ of sidelength $\ell$ (mesoscopic bulk case). 

%Assume that $\Sigma $ has a unique connected component, or if not, that 
%\be \label{assconncomp}\int_{\Sigma_i} (-\Delta )^\alpha \varphi^\Sigma=0\ \text{on each connected component }\Sigma_i \text{ of } \Sigma.\ee
 Then, there exist $C_1, C_2>0$ depending only on $\d, \s,\ep, V, \M$ such that for $N$ sufficiently large, there is an event $\mathcal{G}_\ell$ with 
\begin{equation}
\PNbeta(\mathcal{G}_\ell^c)\leq C_1e^{-C_2 \beta N\ell^\d }
\end{equation}
 such that  if $\tau (\ell N^{\frac1\d})^{\s-\d}$ is smaller than a constant depending only on $\d,\s,\ep, V, \M$, \footnote{Since $\s<\d$, this is satisfied for $N$ large enough for instance  as soon as $\ell \gg N^{-1/\d}$.} we have 
%\begin{multline} 
%\left|\log\Esp_{\PNbeta}\left[\exp\( (-\beta t N^{1-\frac{\s}{\d}}\Fluct_{\muv}(\varphi)\indic_{\mathcal{G}_\ell}\)\right]\right|= \frac{tN}{\cds}\int_\Sigma(-\Delta)^\alpha \varphi^\Sigma(\log \muv)+\frac{\beta  N^{2-\frac\s\d} t^2}{2} 
%\frac{c_{\d,\frac{\d-\s}{2}}}{2\cds}\left\|\varphi^\Sigma\right\|_{\dot{H}^\frac{\d-\s}{2}}^2
%\\
%+ O_{\|\varphi_0\|_{C^5}}\(\max(1,\beta)|t|\ell^\s N+\beta t^2  N^{2-\frac{\s}{\d}}  \|\varphi_0\|_{C^5}^2\ell^\s\(\ell N^{\frac{1}{\d}}\)^{\frac{(\s-\d)(\s-(\d-2))}{2\s+4}}+t^2N\ell^{2\s-\d}+\beta |t|^3N^{2-\frac{\s}{\d}}\ell^{2\s-\d}\).
%\end{multline}
%Taking $t=-\frac{\tau}{\beta}N^{-1+\frac{\s}{\d}}$ we find 

\be
\left|\log\Esp_{\PNbeta}\left[\exp\( \tau \frac{\beta}{1+\beta} \Fluct_{\muv}(\varphi)\indic_{\mathcal{G}_\ell}\)\right]\right|\lesssim_\beta \((|\tau|+|\tau|^2 )(\ell N^{\frac{1}{\d}})^\s\).
\ee
\end{theo}

Note that if we did not make the assumption that $\Sigma$ is connected but instead the union of $n$ connected components, we would need to make the assumption
\be \label{assconncomp}\int_{\Sigma_i} (-\Delta )^\alpha \varphi^\Sigma=0\ \text{on each connected component }\Sigma_i \text{ of } \Sigma,\ee which suffices to build the needed transport map.
Without the condition \eqref{assconncomp}, the result may be false due to the possibility of integer fluctuations in the multi-component case (see \cite{BG16} for the treatment of this situation in the one-dimensional log case). This will also be reflected in the fact that without this condition, 
 the equation for the transport map may not be solvable. Since we couple the proofs of Theorem \ref{FirstFluct} and \ref{Local Law}, and in particular make use of Theorem \ref{FirstFluct} for functions that may not solve \eqref{assconncomp}, we restrict our attention throughout the paper to the case where $\supp(\muv)=\Sigma$ is connected.

In the Coulomb case, the transport method combined with a local free energy expansion with a good enough rate allows us in principle to derive a full Central Limit Theorem for fluctuations of linear statistics, as described in \cite{S22} and \cite[Chap 9,10]{S24}. 
In dimension two (and one as well), any rate is sufficient to conclude, as seen in \cite{LS18}. In dimension three and higher, the currently available rate, which corresponds to a surface error, is not enough to conclude. The situation in the Riesz case is similar: the method in principle allows to get convergence for any $\s$ and $\d$ such that $\s\in (\d-2,\d)$, except that the rate we are able to obtain is only sufficient when $\s$ is small enough and $\ell $ is small enough.

Let us first discuss the question of free energy expansion for the Riesz gas.  In \cite{LS15}, the following  expansion was  shown:
\begin{equation}\label{freeZexp}
 \begin{cases}
\log Z_{N,\beta}= - \beta N^{2-\frac\s\d}\mathcal{E}(\mu_V)+ \(\frac\beta{2\d} N\log N\)\indic_{\s=0} +\log \K_{N,\beta}(\mu_V,\zeta_V)\\
\log \K_{N,\beta}(\mu_V,\zeta_V)=- N \Ent(\mu_V) + N\mathcal{Z} (\beta, \mu_V) +o(N).
\end{cases}
\end{equation}
Here $\Ent (\mu) = \int_{\R^\d} \mu \log \mu$, $\mathcal E$ is as in \eqref{Continuous energy} and 
\begin{equation}\label{defcalZ}
\mathcal{Z}(\beta, \mu):= -\beta \int_{\R^\d} \mu^{1+\frac\s\d}(x) \mf(\beta \mu^{\frac\s\d}(x))dx + \frac{\beta}{2\d}\Ent(\mu)\indic_{\s=0}.
\end{equation}
The quantity $\mf$ is a function of the inverse temperature (note that here an effective temperature $\beta \mu^{\s/\d}$ appears), and corresponds to the pressure at that temperature, meaning the large volume limit of the free energy per unit volume for a unit charge density.  The expression \eqref{defcalZ} then simply corresponds to  scaling that limit in terms of the effective density, in other words, to a local density approximation.

In \cite{LS15}, the function $\mf$  was characterized via the minimum of a rate function governing a large deviations principle on point processes. Here we will characterize it directly as the pressure, by showing in Lemma \ref{lem: const dens} the existence of the large volume limit of the free energy (per unit volume) of a Riesz gas confined to a box, {\it with a rate}, corresponding to surface additivity errors. This will be done by an almost additivity property of the free energy obtained as a consequence of the local laws, as done in \cite{AS22} in the Coulomb case.
The rate obtained in that limit allows in principle to improve the $o(N)$ error in \eqref{freeZexp}. However, there are additional surface errors on $\partial \Sigma$ which damage this improvement. But in view of obtaining a CLT via the transport method, what matters is not the rate in the free energy expansion, but rather the rate in the {\it relative free energy expansion}
$$\log \K_{N,\beta}(\mu_t)-\log \K_{N,\beta}(\mu_0)$$
where $\mu_t$ is the {\it transported measure},  and $\mu_0=\mu_V$ the original one, and for that we can obtain an improved rate.

% which, if $\mu_t $ and $\mu_V$ coincide except in a small set, leads to much smaller errors. 
%While

To go further, we will thus  assume a relative expansion of next-order partition functions (see Sections \ref{sec: ll prelim} and \ref{sec: partition} for a formal definition)
\begin{multline} \label{2eway0 intro}
\log \K_{N,\beta}^{\mathcal{G}_\ell}( \mu_t,\zeta\circ \Phi_t^{-1})- \log \K_{N,\beta}(\mu_0,\zeta) +N(\Ent(\mu_t)-\Ent(\mu_0) ) \\= N\( \mathcal{Z} (\beta, \mu_t)-\mathcal{Z}(\beta, \mu_0) \) + O((\beta+1) N\ell^\d \mathcal R_t)\end{multline}
where $\mathcal R_t$ is the error rate and $\K_{N,\beta}^{\mathcal{G}_\ell}$ corresponds to the free energy restricted to the event $\mathcal{G}_\ell$ where local laws hold as in Theorem \ref{Local Law}.

 We will also need an additional assumption on $\mf$: we 
assume that the function $y \mapsto \mf(y) $ is $p$ times differentiable and satisfies 
\be \label{425} \forall n \le p,\ \text{for} \   y \in \left[\hal\beta \min_{\carr_\ell} \mu^{\s/\d}, 2\beta \max_{\carr_\ell}\mu^{\s/\d}\right],  \ \text{we have}\   |(y\mf(y))^{(n)} |\lesssim_\beta |y|^{1-n}  .\ee

%\cm{$f_\d$ has not been defined!}
%\be \label{assumf''}\beta  \text{ is such that } \ \sup_p  \sup_{x\in \carr_\ell} |\mf^{(p)} (\beta \muv(x)^{\frac{\s}{\d}}) |<\infty.\ee
This can be interpreted as a {\it no-phase transition} assumption at the effective temperature  $\beta \muv(x)^{\frac{\s}{\d}}$. It is reasonable to expect that the CLT result may fail if there is a phase transition. We formulated an assumption that is uniform as $\beta \ge 1$ in order to be able to deduce a zero temperature result by taking a uniform limit as $\beta \to \infty$, but one may dispense with the uniformity in the assumption is one is not interested in that.

Our CLT result has two parts: the first is conditional and asserts that if $\mathcal R_t$ can be shown to be small enough (in terms of the scale $N^{1/\d}\ell$), then the result holds. The second part asserts that the rate can  indeed be obtained  small enough when $\s>0$ is small enough and $\ell$ small enough, which effectively restricts the result to dimensions $1$ and $2$ (because of the constraint $\s>\d-2$). This is the same limitation as encountered in the Coulomb case in \cite{S22}. However, we expect that our error rate  which saturates at surface errors is not sharp, and thus that the result holds much beyond this setting.

%Taking $t=\frac{-\tau}{\sqrt{2\beta}}N^{\frac{\s}{2\d}-1}\ell^{-\frac{\s}{2}}$, we should get the following.
\begin{theo}[Central Limit Theorem]\label{CLT}Assume \eqref{itemnondeg}--\eqref{onecut} and \eqref{assumpw} for some $k\ge 5$.  
Assume $\varphi= \varphi_0(\frac{\cdot-z}{\ell}) $ for some $\varphi_0\in C^k_c(\R^\d) $ and that $ \rho_\beta N^{-1/\d}\le \ell\le 1$ and  $\supp \varphi \subset \carr_\ell(z)\subset \carr_{2\ep}(z) \subset \Sigma$  for some cube $\carr_\ell(z)$ of sidelength $\ell$.   If $\s \neq 0$ assume that \eqref{425} holds.
 
Assume  
  \eqref{2eway0 intro} with a rate $\mathcal R_t$ such that 
\begin{equation}\label{goodrate}
\(\ell N^{\frac{1}{\d}}\)^{\frac{\s}{2}}\(\max_{|s|\leq \ell^{\d-\s}}\mathcal{R}_s\)^{1-\frac{1}{p}}\rightarrow 0
\end{equation}
as $N \rightarrow \infty$ for some $ p \ge 2$ such that $k \ge 2p+3$.  Then,
\begin{equation}
\sqrt{2\beta} \frac{\Fluct_{\muv}(\varphi)}{\(N^{\frac{1}{\d}}\ell\)^{\frac{\s}{2}}}- \(\ell N^{-\frac{1}{\d}}\)^{-\frac{\s}{2}} \mathrm{Mean} (\varphi)
\end{equation}
converges in distribution to a centered Gaussian with variance 
\be \mathrm{Var}(\varphi)=\frac{c_{\d,\frac{\d-\s}{2}}}{2\cds}\begin{cases}
\|\varphi_0^\Sigma\|_{\dot{H}^{\frac{\d-\s}{2}}}^2& \text{if }  \ell =1\\
\|\varphi_0\|_{\dot{H}^{\frac{\d-\s}{2}}}^2& \text{if } \ell\to 0 \ \text{as} \, N\to \infty,\end{cases}
\ee
where $\cds$ and $c_{\d,\alpha}$ are as in \eqref{constant definition} and \eqref{def fraclap}, and
\begin{align}
\mathrm{Mean}(\varphi)=\begin {cases}  \displaystyle\frac{\sqrt2}{\sqrt{\beta}}\frac{1}{\cds}\(- 1+\frac{\beta}{2\d}\indic_{\s=0}\) \int_\Sigma(-\Delta)^\alpha \varphi^\Sigma(\log \muv ) &\\
 - \displaystyle\frac{\sqrt{2\beta}}{\cds}\int_{\Sigma} (-\Delta)^{\alpha}\varphi^\Sigma\(  (1+\frac\s\d) \mf(\beta\muv^{\frac\s\d})\muv^{\frac\s\d} + \frac\s\d \mf'(\beta \muv^{\frac\s\d}) \muv^{2\frac\s\d} \)\indic_{\s\ge 0}& \text{if} \ \ell=1\\
 0 & \text{if}\  \ell \to 0\end{cases}
 \end{align}
  Moreover, 
\be \label{secondstate} \sqrt{2} \frac{\Fluct_{\muv}(\varphi)}{\(N^{\frac{1}{\d}}\ell\)^{\frac{\s}{2}}}- \frac1{\sqrt{\beta}}\(\ell N^{-\frac{1}{\d}}\)^{-\frac{\s}{2}} \mathrm{Mean} (\varphi)\ee converges in distribution to $0$ as $N \to \infty$, uniformly as $\beta \to \infty$.

The result holds without the additional assumption \eqref{goodrate}    for $\d=1,2$ and $\s \leq \s_0$, where 
%\cm{additional condition on $\ell$ small}
\begin{equation*}
0<\s_0\approx \begin{cases}
0.03973& \text{in }\d=1,\\
0.06059 & \text{in }\d=2.
\end{cases}
\end{equation*}
%\cm{numbers need to be changed}
as long as
\begin{equation*}
\ell \ll N^{-\frac{\s}{\d(\s+2)}} .\end{equation*}
\end{theo}
This theorem provides us with the precise (expected) order of the fluctuations
after one removes  a deterministic shift (the factor containing the mean) which is {\it divergent} as soon as  $\s>0$.

Notice that the additional assumption on $\ell$ small holds automatically in the case $\s \leq 0$ if $\ell=o(1)$, which allows to  match the result of \cite{boursier1} in the one-dimensional case, which was the only prior results in the Riesz case (however much less regularity of the test function was needed in \cite{boursier1}).
We could also have a result for larger $\s>0$ but at the expense of being restricted to smaller $\ell$'s.

The second statement \eqref{secondstate} in the theorem is meant to allow to take   $\beta\to \infty $,  which implies, under the same suitable assumptions that 
for minimizers of the energy $\HN$, a convergence result for $\frac{\Fluct_{\muv}(\varphi) }{(N^{\frac1\d}\ell)^{\s/2}}$ after substracting off the appropriate shift.

Let us discuss the comparison of this result with the Coulomb case, which is solved only in dimension 1 (see \cite{S24}) and more interestingly, in dimension 2 in \cite{LS18,S22} (refer also to \cite{S24,S22}). In the Coulomb case, when $\varphi$ is supported in $\Sigma$, then its harmonic extension (in that case $\alpha=1$)  outside $\Sigma $ is simply itself and the mean and variance expressions involve only $\varphi$ itself.
However, when the support of $\varphi$ overlaps $\Sigma^c$, then, as first shown in \cite{AHM15} in the particular (determinantal) setting of $\beta=2$, the expressions involve $\varphi^\Sigma$ instead of $\varphi$.
In the Riesz case, even when $\varphi$ is supported in $\Sigma$, it does not coincide with its $\alpha$-harmonic extension due to the nonlocality of the fractional Laplacian, thus $\varphi^\Sigma$ is always involved, and this counts among the difficulties in handling the Riesz case. 

In the (two-dimensional) Coulomb case,  the variance expression was $\int |\nab \varphi|^2$ (or more generally $\int |\nab \varphi^\Sigma|^2$), leading to interpreting the limit of $\Delta^{-1}(\sum_{i=1}^N \delta_{x_i}-N \muv)$ as the \textit{Gaussian free field}. Here in the Riesz case the variance  is a homogeneous fractional Sobolev norm of $\varphi^\Sigma$. Modulo the $\alpha$-harmonic extension, this structure is characteristic of a \textit{fractional Gaussian field} (for a review of the notion, we refer to \cite{LSSW}).

This leads to the following interpretation:
\begin{coro}
If convergence holds, the quantity $(-\Delta)^{-\alpha} ( \sum_{i=1}^N \delta_{x_i}-N \muv) $ converges to a Gaussian field of Hurst parameter $-\s/2$ in the mesoscopic case $\ell \to 0$,  or  a variant\footnote{due to the harmonic extension} of it in the macroscopic case.
\end{coro}

To our knowledge, this is a new and natural  occurence of such fractional Gaussian fields.

\subsection{The transport method}
Let us next outline the implementation of the transport method, emphasizing the differences with the local situation.

The starting point is the rewriting of the Laplace transform (or moment generation function) of fluctuations as a ratio of partition functions, as done in this context  for instance since \cite{J98}: a straightforward computation shows that 
$\Fluct_{\muv}(\varphi)$ being  defined  in \eqref{def:fluct}, we have 
\begin{equation}\label{expansion1}
\Esp_{\PNbeta}\left[\exp\(-\beta t N^{1-\frac{\s}{\d}} \Fluct_{\muv}(\varphi)\)\right]=e^{t\beta N^{2-\frac{\s}{\d}}\int \varphi\, d\muv}\frac{\ZNbeta(V_t)}{\ZNbeta(V)},
\end{equation}
where $V_t:=V+t\varphi$ and 
%$t=-\frac{u}{\beta N^{1-s/d}}=-\frac{u}{\theta}$, and
$\ZNbeta(V)$ denotes the partition function of the Riesz gas of potential $V$, as in \eqref{defZN}.
The interesting regime will be the regime of small $t$, thus we can think of the Laplace transform of fluctuations as a ratio of two partition functions, that of a Riesz gas with perturbed potential $V_t$ to that of  the original Riesz gas.
It then appears that understanding the log Laplace transform of the fluctuations amounts to obtaining a fine expansion of the free energy $\log \ZNbeta(V)$ in terms of $V$. Since we cannot get a precise enough expansion for this, we will bypass this by combining a cruder expansion of the free energy,  with a second way of estimating the ratio of partition functions. That second way is via a transport, which is a good  choice of  change of variables. Computations reveal that a good choice is one that transports the original equilibrium measure $\muv$ into the equilibrium measure $\mu_{V_t}$ associated to the perturbed external potential. When working in the regime of small $t$, we do not need an exact transport map, it suffices to find a map which is a perturbation of identity and transports at leading order $\muv$ to $\mu_{V_t}$, i.e.~a map $\psi: \R^\d\to \R^\d$ such that 
\be \label{tr1o}
(I+t\psi)\# \muv= \mu_{V_t}+o(t),\ee
which, by linearization of the Monge-Amp\`ere equation for instance, amounts to solving 
\be \label{tr1o2}\div (\psi \muv)= \frac{\p}{\p t}\mu_{V_t}.\ee

In the bulk Coulomb case where the perturbation $\varphi$ is supported in the support of $\muv$ (call this the interior case), it is easy to check that the difference of the equilibrium measures is supported in $\Sigma$ and is then explicit in view of \eqref{densmuv} : it is 
\be \mu_{V_t}-\muv= \frac{t}{\cd} \Delta \varphi.\ee
This makes solving \eqref{tr1o2} immediate : it suffices to take $\psi= \frac{\nab \varphi}{\cd \muv}$.

Still in the Coulomb case, when $\varphi$ is not supported in $\Sigma$, this is no longer possible, as the support of the equilibrium measure itself varies, as studied in \cite{SS18}, and the difference in equilibrium measures involves the harmonic extension of $\varphi$. The resolution of \eqref{tr1o2} becomes much more delicate. In the present Riesz case, we always encounter that difficulty because the support of the equilibrium measure always changes, even in the interior  case.  To carry the method through, we need not only to find a solution $\psi$, but to ensure it is sufficiently regular, and have sharp estimates for its derivatives in terms of those of $\varphi$, in particular as $\varphi$ get supported on small scales. This is the core of the new analysis performed in Section \ref{sec: tport}, which is  of analysis and PDE nature.

Note that the one-dimensional log case is also a nonlocal case. In that setting, finding $\psi$ solving \eqref{tr1o2} turns out to be the same as the inversion of a so-called ``master operator'' (see for instance \cite{BG13,BG16,BL18} and references therein). This is an integral operator in one dimension, for which an explicit inversion formula is  available from classical analysis, see \cite{M92}, allowing to read off  the needed estimates.  This is the route that was taken in \cite{P24}. In higher dimension, or for $\s\neq 0$, these formulae are not available, and the inversion of the master operator should rather be seen as the solution of the appropriate PDE  
\eqref{Rtransport} which is a rather nonstandard PDE, for which we need to get sharp estimates.

An important difficulty  in the Riesz case is that, contrarily to the interior Coulomb case (where we could take $\psi=\frac{\nab \varphi}{\cd\muv}$),  the transport map $\psi$ is no longer   supported in the support of $\varphi$. When we study  fluctuations for test functions that live on a mesoscale, this prevents us from localizing the estimates to the support of $\varphi$. 
Instead, $\psi$ has decaying tails, and we need to estimate their decay speed away from $\supp \varphi$, and to estimate the
 contributions of the tails of the transport to the errors, which need to remain small in terms of the chosen mesoscale.

Once the transport has been found, a change of variables reduces the analysis of the ratio of partition functions to one crucial term corresponding to the expectation of quantities of the form 
\be \iint_{(\R^\d\times \R^\d) \backslash \triangle} 
(\psi(x)-\psi(y))^{\otimes n} : \nab^{\otimes n}  \g (x-y) d\Big( \sum_{i=1}^N \delta_{x_i}- N\mu\Big)(x)  \Big( \sum_{i=1}^N \delta_{x_i}- N\mu\Big)(y),\ee which correspond to the $n$-th variation of the energy $\F_N$ along the transport $(I+t\psi)$, as described for instance in \cite[Chap.~6]{S24}, and have a commutator structure. Such terms are the analogues of  the so-called ``loop equation terms'' in the case $\s=0$.
To control them, we will make crucial use of the sharp and localized commutator estimates recently proven in \cite{RS22} and recalled in Proposition \ref{pro:commutator}. These ensure that for all configurations,  these terms are all controlled by a constant (which depends on norms of derivatives of $\psi$) times (the localized version of) $\F_N + C N^{1+\frac\s\d}$.  The control of this quantity at first order ($n=1$), inserted into \eqref{expansion11} after the change of variables already allows to obtain the fluctuation bounds of Theorem \ref{FirstFluct}. To obtain Theorem \ref{CLT}, we need the commutator estimates up to order $n=p$, and the combination with another free energy expansion with a good enough rate.

\subsection{Rest  of the  proof and plan of the paper}
In addition to the transport problem, dealing with the nonlocal Riesz case involves other difficulties, already encountered in \cite{P24}. First, as introduced in \cite{PS17}, we use the Caffarelli-Silvestre extension to represent the fractional Laplacian as a {\it local} divergence-form operator with singular weight. This makes the electric formulation possible and  is described in Section \ref{sec: fluct prelim}. 
The next ingredient is the screening procedure which is performed in the extended space $\R^{\d+1}$, and is the crucial tool  to show the local laws, almost additivity and perform the bootstrap procedure. The need to control the energy on good ``heights'' in the extended dimension in order to perform the screening  requires to couple the bootstrap with the control of fluctuations.

The results are coupled and proved as follows:
\begin{itemize}
\item We  
 assume local laws down to scale $2\ell$.
 We show that this implies the control of fluctuations for test functions varying on scales $\ge 2\ell$ stated in Theorem \ref{FirstFluct}.
 \item
Still assuming local laws hold down to scale $2\ell$, we use the fluctuations control to get a control on the ``electric field'' at heights comparable to $\ell$. This allows to perform the screening, which then allows to show that local laws hold down to scale $\ell$.
\item One may then iterate down the scales to obtain the local laws at all scales $\ell\ge \rho_\beta N^{-1/\d}$, hence Theorem \ref{Local Law}.
\item This also implies that the fluctuations control of Theorem \ref{FirstFluct} holds at all scales $\ell\ge \rho_\beta N^{-1/\d}$.
\end{itemize}
Note that  this coupling of the two proofs was not needed in the Coulomb case where no control of the electric field in the extended dimension is needed for the screening.

As mentioned above, the proof of the CLT relies on the free expansion with a rate of \eqref{2eway0 intro}.
This in turn is done as in \cite{AS21,S22} by first obtaining a rate of convergence for the free energy in a cube with  uniform background  equilibrium measure, obtained by almost additivity of the free energy, then  using the almost additivity again to split the domain into regions where  $\muv$ is close to constant, and using the transport method to estimate the difference with the free energy with uniform background. This is done in Section \ref{sec: partition}.

The first part of the paper is devoted to proving the fluctuations bound stated in Theorem~\ref{FirstFluct} at scale $2\ell$, assuming the local law Theorem~\ref{Local Law} at scales $\ge 2\ell$. It starts with   Section \ref{sec: tport}, devoted to solving the transport map problem and showing good estimates for it. Section \ref{sec: fluct prelim} reviews the electric formulation and extends the transport method to the Riesz setting, and finally Section \ref{sec: fluct} concludes with the proof of  the fluctuations bound stated in Theorem~\ref{FirstFluct} at scale $2\ell$, assuming the local law Theorem~\ref{Local Law} at scales $\ge 2\ell$. Much of the argument for Theorems~\ref{FirstFluct} and \ref{CLT} can be understood without a technical understanding of the local laws bootstrap, and so we present the proof of these theorems first, for readers who are interested primarily in fluctuations of linear statistics. 
The second part of the paper, in particular Section~\ref{sec: mb},  is then devoted to proving that the local law, Theorem \ref{Local Law}, holds at scale $\ell$ if it holds at scale $\ge 2\ell$. In Section~\ref{sec: partition}, we obtain the almost additivity of the free energy as a consequence of the local law, and prove the free energy expansion with a rate via domain partitioning plus transport.
Section~\ref{sec: flap} gathers some independent estimates on fractional Laplacians that are needed, in particular in the construction of the transport map. Section~\ref{sec: pfscn} presents the proof of the screening procedure, adapted from \cite{PS17}. Finally the paper concludes with the appendix by Xavier Ros-Oton providing the fine behavior of solutions to the fractional obstacle problem.

\medskip

{\bf Acknowledgements:}
This work  was supported by NSF grant DMS-2247846  and by the Simons Foundation through the Simons Investigator program. The authors would also like to thank Susanna Terracini and Giorgio Tortone for helpful conversations regarding degenerate elliptic equations.

\section{Solving the transport problem }\label{sec: tport}

\subsection{Preliminary: reexpressing the Laplace transform}
Let us start with the rewriting of the Laplace transform of linear statistics, which will make the correct choice of transport appear.
%\subsection{First expansion}
We start by a generic change of variables $\Phi_t$ and will make a specific choice below.

\begin{lem}[Reexpressing the Laplace transform] \label{lem2.1}
%We have 
Let $\varphi$ be a compactly supported measurable test function, and let $\Fluct_{\muv}(\varphi)$ be as in \eqref{def:fluct}. Set $\Vt:= V+t\varphi$.  Then, for any $t\in \R$ and any event $\mathcal{G}\subset\R^{\d N}$, we have
\begin{equation}\label{expansion11}
\Esp_{\PNbeta}\left[\exp\(-\beta t N^{1-\frac{\s}{\d}} \Fluct_{\muv}(\varphi)\)\indic_{\mathcal{G}}\right]=\exp\( t\beta N^{2-\frac{\s}{\d}}\int_{\R^\d} \varphi\,d\muv\)\frac{\ZNbeta^{\Vt, \mathcal{G}}}{\ZNbeta},
\end{equation}
with 
%$t=-\frac{u}{\beta N^{1-s/d}}=-\frac{u}{\theta}$, and
\begin{equation*}
\ZNbeta^{\Vt,\mathcal{G}}:=\int_{\mathcal{G} \subset (\R^\d)^N}\exp \left(-\beta N^{-\frac{\s}{\d}}\left(\sum_{i \ne j}\frac{1}{2}\g(x_i-x_j)+N\sum_{i=1}^N V_t(x_i)\right)\right)~d\XN.
\end{equation*} 
Furthermore,  we can expand
\begin{equation}\label{splitt}
\Esp_{\PNbeta}\left[\exp\(- \beta t N^{1-\frac\s\d} \Fluct_{\muv}(\varphi)\)\indic_{\mathcal{G}}\right]=e^{T_0}\Esp_{\PNbeta}\left[e^{T_1+T_2}\indic_{\mathcal{G}}\right],
\end{equation}
where
\begin{align}
\label{t0}
T_0:= &-\beta N^{2-\frac{\s}{\d}}\biggl(\frac{1}{2}\iint_{\R^\d\times \R^\d} (\g(\Phi_t(x)-\Phi_t(y))-\g(x-y))\, d\muv(x)d\muv(y)+\int (V_t\circ \Phi_t-V)\, d\muv\biggr) \\
\notag&+N \int_{\R^\d} (\log \det D\Phi_t)\, d\muv+\beta tN^{2-\frac\s\d} \int_{\R^\d} \varphi\,d\muv,\end{align}
\be
\label{t1} T_1: =-\beta N^{1-\frac{\s}{\d}}  \int_{\R^\d}\(\int_{ \R^\d}\g(\Phi_t(x)-\Phi_t(y))-\g(x-y))\, d\muv(y)+(V_t \circ \Phi_t -V)(x)\)d\fluct_{\muv}(x),\ee 
\begin{align}
\label{t2} T_2:= &-\frac{\beta N^{-\frac{\s}{\d}}}{2}\iint_{\triangle^c}(\g(\Phi_t(x)-\Phi_t(y))-\g(x-y))\,d\fluct_{\muv}(y)d\fluct_{\muv}(x)\\ \notag &
+\int_{\R^\d}\log \det D\Phi_t\,d\fluct_{\muv}(x)\\
\notag  =&-\beta N^{-\frac\s\d}\(\F_N(\Phi_t(\XN),\Phi_t\#\muv)-\F_N(\XN, \muv)\)+\Fluct_{\muv}( \log \det D\Phi_t) 
\end{align}
where we have let 
\begin{equation}
\label{defpfluct}
\fluct_{\mu}:= \sum_{i=1}^N \delta_{x_i}- N \mu.\end{equation}
\end{lem}

\begin{proof}
The relation \eqref{expansion11} is immediate from spelling out the definition of $\Fluct_{\muv}(\varphi)$.
We  then perform the  change of variables $y_i=\Phi_t(x_i)$ in the integral defining $\ZNbeta^{V_t,\mathcal{G}}$  to obtain 
\begin{align*}
& \Esp_{\PNbeta}\left[\exp\(- \beta t N^{1-\frac\s\d} \Fluct_{\muv}(\varphi)\)\indic_{\mathcal{G}}\right] \exp\(-\beta  t N^{2-\frac\s\d} \int_{\R^\d}\varphi\, d\muv\)\\ &=\frac{1}{\ZNbeta}
 \int_{\mathcal{G}} \exp \biggl(-\beta N^{-\frac\s\d} \Big(\frac{1}{2}\sum_{i \ne j}\g(\Phi_t(x_i)-\Phi_t(x_j))+N\sum_{i=1}^NV_t(\Phi_t(x_i))\Big)+\sum_{i=1}^N \log \det D\Phi_t(x_i)\biggr)d\XN \\
&= \Esp_{\PNbeta} \biggl[\exp \biggl(-\beta N^{-\frac\s\d}\biggl( \frac{1}{2}\sum_{i \ne j}(\g(\Phi_t(x_i)-\Phi_t(x_j))-\g(x_i-x_j))+ N \sum_{i=1}^N (V_t\circ \Phi_t-V)(x_i)\biggr)\\\ &\qquad \qquad+\sum_{i=1}^N \log \det D\Phi_t(x_i)\biggr)\indic_{\mathcal{G}}\biggr].
\end{align*}
Expanding $\sum_{i=1}^N \delta_{x_i}$ as $N\muv+ \fluct_{\muv}$ and collecting terms, we find the result. The relation \eqref{t2} follows  by the  definition \eqref{defFN}.

%
%
%We may write \begin{multline}
%\Esp_{\PNbeta}\(e^{-\beta t N \sum_i \varphi(x_i)}\)= \frac{1}{\ZNbeta} \int \exp\(-\beta \Big(\hal\sum_{i\neq j} \g(x_i-x_j) +N\sum_{i=1}^N V_t(x_i) \Big) \)dx_1 \dots dx_N = \frac{\ZNbeta^{V_t}}{\ZNbeta^V} \end{multline}
%and making the change of variables $y_i= \Phi_t(x_i) $ where $\Phi_t= \id + t\psi$, we obtain
%\begin{multline}
%\Esp\(e^{-\beta t N \sum_i \varphi(x_i)}\)
%\\= \frac{1}{\ZNbeta} \int \exp\(-\beta \Big(\hal\sum_{i\neq j} \g(\Phi_t(x_i)-\Phi_t(x_j)) +N \sum_{i=1}^N V_t(\Phi_t(x_i))\Big)+ \sum_i \log \det D\Phi_t(x_i)\)   dx_1 \dots dx_N\\
%= \Esp_{\PNbeta}\Big[ 
%\exp\Big( -\beta \Big( \hal \sum_{i\neq j} \g(\Phi_t(x_i)-\Phi_t(x_j))- \g(x_i-x_j)  +N \sum_{i=1}^N (V_t\circ\Phi_t- V)(x_i) \Big)\\ + \sum_{i=1}^N \log \det D\Phi_t(x_i)\Big)   \Big)\Big]\end{multline}
% We then write the exponent as
% \begin{multline}
% -\beta \Big(\iint_{\triangle^c}\hal \( \g(\Phi_t(x)-\Phi_t(y))- \g(x-y )  \) d(\sum_{i=1}^N \delta_{x_i} )(x) d(\sum_{i=1}^N\delta_{x_i} ) (y) \\+N \int (V_t\circ\Phi_t- V)(x)d( \sum_{i=1}^N \delta_{x_i} ) (x) \Big)+ \int\log \det D\Phi_t(x ) d( \sum_{i=1}^N \delta_{x_i} ) (x)  
% \end{multline} and expand $\sum_i\delta_{x_i}$ as $N\muv+ \Fluct_{\muv}$
%   to find \eqref{splitt}.
\end{proof}
%\cm{insert here preliminary lemma, similar to [BLS] that says that for any  strict neighborhood $U$ of $\Sigma$,  that almost all the particles fall in $U$. \be \label{confinement} \PNbeta( \XN \in U^N) \ge 1- e^{-cN^{1-\frac\s\d}}\ee Insert this into the good event. This way we can reduce to the event where $\fluct$ is supported in $U$}

\subsection{Choice of transport and estimates}\label{subsec: transport}
The choice of transport is made so that $\Phi_t$ is a perturbation of identity of the form $\Phi_t=\id + t\psi$, chosen so that, at leading order in $t\to 0$, the $T_1$ term 
in the expansion above vanishes, leaving only the constant term $T_0$ to compute, and the subleading order term $T_2$ to analyze.

Plugging in $\Phi_t= \id + t\psi$ and linearizing the $T_1$ term in $t$, we find that $\psi$ must be chosen so that 
\be \int_{\R^\d}\(\int_{\R^\d}  \nab \g(x-y)\cdot (\psi(x)-\psi(y)) d\muv(y)+ \varphi(x)+ \nab V (x)\cdot \psi(x) \)d\fluct_{\muv}(x).\ee
We will choose $\psi$ such that the term in factor of $\fluct_{\muv}$ vanishes identically. The reader can recognize the condition $\Xi(\psi)=\varphi$ for a certain ``master operator'' $\Xi$ that needs to be inverted (as in \cite{BG13} and references therein for the one-dimensional log gas).  
%Since we reduce to a good event where $\XN \in U^N$ (see  \eqref{confinement}), it suffices for the term to vanish in $U$, the neighborhood of $\Sigma$. \cm{remove}

This can be rewritten as 
\be \label{2.7}(\nab (\g * \muv) + \nab V) \cdot \psi- \g* (\div (\psi \muv))+\varphi=0 \quad \text{in} \ \R^\d.\ee
Using  the notation $h^f$ for the potential $\g*f$ generated by $f$, and recalling the definition of $\zeta_V$ in \eqref{Riesz effective potential}, this becomes
\be\label{lowtemp}
\psi \cdot \nab \zeta_V- h^{\div(\psi \muv) }+\varphi=0\quad \text{in} \ \R^\d.\ee
This is the generalization to arbitrary dimension and interaction potential of the master operator inversion.

Let us now explain how to  solve \eqref{lowtemp}. Since $\muv$ vanishes outside $\Sigma$, the function $h^{\div(\psi \muv)}$ is $\alpha$-harmonic there, hence the only possibility to solve \eqref{lowtemp} is to make $\psi \cdot \nab \zeta_V+ \varphi$ be $\alpha$-harmonic outside of $\Sigma$. Since $\zeta_V$ vanishes in $\Sigma$, that function equals $\varphi$ in $\Sigma$. Thus, the only possibility to solve it continuously is to have $\psi \cdot \nab \zeta_V=\varphi^\Sigma-\varphi$.
In other words we need to find $\psi$ such that, in $\R^\d$, there holds
\be \label{eqars} 
\begin{cases}
h^{\div (\psi\muv)} =\varphi^\Sigma \\
\psi \cdot \nab \zeta_V=\varphi^\Sigma-\varphi 
\end{cases}\ee
or equivalently
\be \label{eqars2} 
\begin{cases}
\div (\psi\muv) =\frac1\cds(-\Delta)^\alpha \varphi^\Sigma \\
\psi \cdot \nab \zeta_V=\varphi^\Sigma-\varphi.
\end{cases}\ee
It is sufficient to solve \eqref{eqars2} in order to obtain \eqref{eqars} since convolution with $\g$ uniquely inverts $(-\Delta)^\alpha$ for $L^1$ functions, cf \cite[Theorem 2.4]{K19}. By Lemma \ref{def aharmext}, $\varphi^\Sigma \in \dot{H}^\alpha$ which embeds into $L^1$ by \cite[Theorem 6.5]{DPV12}. $\div(\psi \muv)$ is continuous in $\Sigma$, compactly supported and has an integrable singularity $(\sim (\dist(,x\partial \Sigma))^{-\alpha})$ at $\partial \Sigma$, so it is also in $L^1$.

Since $\alpha<1$, by \eqref{eq reg} $\muv$ vanishes at the boundary of $\Sigma$ and, if $\psi$ is regular enough,  $\psi \muv$ is continuous across $\pa \Sigma$ and thus 
\be\label{divpsimu0}\div (\psi \muv)= \div(\psi \muv) \indic_\Sigma,\ee
which means that the first equation needs to be solved in each connected component of $\Sigma$ only. While we only prove Theorems \ref{Local Law}-\ref{CLT} in the case where $\Sigma$ is connected, the results we state in this section remain true in the case where $\Sigma$ is a disjoint union of finitely many connected components. In that case, the solvability in each connected component is ensured by the condition \eqref{assconncomp}. 

Care is needed for the boundary condition because of the blow up of $(-\Delta )^\alpha \varphi^\Sigma$ as $w \, \dist (\cdot, \pa \Sigma)^{-\alpha}$ on the boundary as in \eqref{assumpw}.
Because $\psi \muv$ behaves like $\psi s \dist(\cdot, \pa \Sigma)^{1-\alpha}$ as in \eqref{eq reg}, whose normal derivative along the boundary blows up like $\psi s(1-\alpha) \dist(\cdot, \pa \Sigma)^{-\alpha}$, equating the two blowing-up rates leads to imposing $\psi \cdot \vec{n} s (1-\alpha) = w$, thus we are led to solving
\begin{equation}\label{Rtransport0}
\begin{cases}
\div(\psi \muv)=\frac{1}{\c}(-\Delta)^\a \varphi^{\Sigma} & \text{in }\Sigma\\
\psi \cdot \vec n =\frac{w}{(1-\alpha)s} & \text{on }\partial \Sigma \\
\psi \cdot \nab \zeta_V=\varphi^\Sigma-\varphi & \text{in }   \Sigma^c.
\end{cases}\ee
We will next see that this is solvable, and that the last two relations are compatible at $\pa \Sigma$. %thanks to the following 
%\begin{lem}\label{def aharmext}
%Let $U$ be a neighborhood of $\Sigma$.
%Let  $0<\alpha<1$ and let $\varphi\in C^{\sigma}(\R^\d)$ be compactly supported. Then, there exists a unique function $\varphi^{\Sigma} \in H^\alpha \cap C^{\sigma}(\R^\d)$ solving \eqref{aharmext}). For any $x_0 \in \partial \Sigma$, if $\varphi \in C^{2\sigma+\ep}(\R^\d)$ for some $\ep>0$ we also have the boundary estimate
%\begin{equation}\label{aharm reg}
%\left\|\frac{\varphi^{\Sigma}-\varphi}{\dist(x,\Sigma)^\alpha}\right\|_{C^{\sigma}(\overline{U \setminus \Sigma})} \lesssim \|(-\Delta)^\alpha \varphi\|_{L^\infty(U \setminus \Sigma)}
%\end{equation}
%for all $\sigma \in (0,\alpha)$ and $r$ sufficiently small.
%\end{lem}
%We could not find a succinct statement of this result for solutions of fractional elliptic problems in unbounded domains in the literature, although it is certainly known. For completeness we offer its proof in Appendix \ref{sec: flap}. Interior and boundary regularity for solutions of equations like \eqref{def aharmext}) are well-studied (see ) and indeed the optimal regularity is understood to be $C^\alpha$, see \cite{RoS14}, \cite[Theorem 1.1]{RoW24} and references therein. The regularity of $\frac{\varphi^{\Sigma}-\varphi}{\dist(x,\Sigma)^\alpha}$ then becomes an interesting question, which we rely crucially on here.
%Thanks to this relation and \eqref{liftoff}, it appears that the last two relations in \eqref{Rtransport0} are compatible.

Before proving the solvability with estimates, let us mention a word about the Coulomb case in which $\alpha=1$. In that case $\muv$ is generically discontinuous across the boundary of $\Sigma$ and $\div (\psi \muv)$ has a singular part on the boundary. The equation to solve is then 
\begin{equation}\label{Rtransportcoul}
\begin{cases}
\div(\psi \muv)=- \frac{1}{\cd}\Delta \varphi & \text{in }\Sigma\\
\psi \cdot \vec n =\frac{1}{\cd \muv} [\nab \varphi^\Sigma]\cdot \vec{n}  & \text{on }\partial \Sigma \\
\psi \cdot \nab \zeta_V=\varphi^\Sigma-\varphi & \text{in} \  \Sigma^c.
\end{cases}\ee
We refer to \cite[Lemma 3.4]{LS18} where this problem is formally derived and solved in any dimension.

\begin{prop}[Good transport map and estimates]\label{transport}Let $U$ be an open neighborhood of $\Sigma$. Suppose that $\muv$ satisfies \eqref{itemeqreg} for $s(x)\in C^{k+\epsilon}(\Sigma)$ for some $\epsilon \in (0,1)$, and that $\partial \Sigma$ is $C^{k+1+\epsilon}$.
Assume $V \in C^{k+1}(\R^{\d+1})$, $\varphi\in C^{k+1}_c(\R^\d) $,  that  $\ell\le 1$ and  $\supp \varphi \subset \carr_\ell (z)\subset\carr_{2\ell}(z)\subset \bulk$, with the estimates 
\be\label{estxi}
\forall \sigma \le k+1, \quad \|\varphi\|_{C^\sigma}\le \M \ell^{-\sigma}\ee
for some constant $\M>0$.
 Assume that   $\varphi$ satisfies  \eqref{assumpw} for some $w(x)\in C^{k+\epsilon}(U)$. 
 %\cm{which $\sigma$ needed? for $\sigma = k$?} 
Finally, assume \eqref{assconncomp}.

Then, there exists a map $\psi$ defined in $\R^\d$, continuous in $\R^\d$, and a map $\psi^\perp$  defined in $\R^\d\backslash \Sigma$, vanishing in $\R^\d\backslash U$ and continous in $\R^\d\backslash \Sigma$, perpendicular to $\nab \zeta_V$, and  such that 
\begin{equation}\label{Rtransport}
\begin{cases}
\div(\psi \muv)=\frac{1}{\c}(-\Delta)^\a \varphi^{\Sigma} & \text{in }\Sigma\\
\psi \cdot \vec n =\frac{w}{(1-\alpha)s} & \text{on }\partial \Sigma \\
\psi=(\varphi^{\Sigma}-\varphi)\frac{\nabla \zeta_V}{|\nabla \zeta_V|^2}+\psi^\perp & \text{in }\Sigma^c,
\end{cases}
\end{equation}  and $\psi$ satisfies  \eqref{lowtemp} in $\R^\d$. Moreover, %if $\sigma $ satisfies $\sigma+\alpha, \sigma-\alpha \notin \mathbb{N}$ then 
$\psi \in C^{k}(\Sigma) \cap C^{k}(\Sigma^c) \cap C(U)$. 
%\cm{again, which $\sigma$?}
% and satisfies
%\cm{I think this can be removed from the main result -- doesn't seem used anywhere
%\begin{equation}\label{transport reg}
%\|\psi\|_{C^{\sigma}(\Sigma)}+\|\psi\|_{C^{\sigma}(U\setminus \Sigma)}  \lesssim \left\|\frac{(-\Delta)^\a \varphi^{\Sigma}}{\dist(x,\partial \Sigma)^{-\alpha}}\right\|_{C^{\sigma}(\Sigma)}+\|(-\Delta)^\alpha \varphi\|_{C^{\sigma-\alpha}(U\setminus \Sigma)}.
%\end{equation}}

%Suppose moreover that $\partial \Sigma$ is $C^{k, 1^-}$, \cm{check} then, 
Furthermore, for any $m\leq k-3$, we have for any $x \notin \partial \Sigma$, 
\begin{equation}\label{tscale}
\left|\nab^{\otimes m} \psi(x)\right|\lesssim \M \begin{cases}
\frac{1}{\ell^{\d-\s+{m-1}}} & \text{if }x \in \carr_{2\ell}(z) \\
\frac{\ell^\d}{|x-z|^{2\d-\s+{m-1}}} & \text{if }x \in U\backslash \carr_{2\ell}(z)\\
\frac{\ell^{\d}}{|x-z|^{\s+2+m}}& \text{if } x\in U^c.\end{cases}
\end{equation}
%\cm{see if in macroscopic case this is the optimal estimate}
\end{prop}
 Notice that in the one-dimensional log case $\s=0$, this scaling agrees with that of  the equivalent formulation of that transport in \cite[Lemma 4.10]{P24} obtained there by exact formulas for the inversion of the master operator. We have been careful to assume that $s,w \in C^{k+\epsilon}$ and $\partial \Sigma$ is $C^{k+1+\epsilon}$ since some of the fractional and degenerate regularity results we seek to apply, namely \cite[Theorem 1.1]{TTV22} and \cite[Theorem 1.3]{ARo20}, do not necessarily hold for integer H\"older classes.

Before moving on to the proof of Proposition \ref{transport}, let us recall more about the solutions to the fractional obstacle problem.

\subsection{Preliminaries on the fractional obstacle problem}\label{subsec: prelimfluct}

%The first assumption deals with classifying the free boundary. 
Theorem \ref{FB-thm} in the appendix recalls the  result  from the work of \cite{CSS08}, \cite{CDS17} and \cite{ROS17}, which classifies the free boundary points as regular or degenerate (or singular) points.

%\begin{theostar}
%Let $\psi \in C^{1+2\alpha+\delta}$ for some $\delta>0$, and suppose $u$ solves \eqref{fractional obstacle}. Then, for every free boundary point $x_0$, either 
 In our notation, regular points are those for which 
 \begin{equation}\label{regular point}
\zeta_V(x)=c(x_0)\dist(x,\partial \Sigma)^{1+\alpha}(x)+O(|x-x_0|^{1+\alpha+a})
\end{equation}
for some $c(x_0)>0$ and $a>0$ with $\alpha+a<1$.
%, or 
%\begin{equation}\label{degenerate point}
%(u-\psi)(x)=O(|x-x_0|^{1+\alpha+\beta}).
%\end{equation}

%Points in \eqref{regular point} are called \textbf{regular} points, and points in \eqref{degenerate point} are called \textbf{degenerate}, or \textbf{singular}.
%\end{theostar}
%In dimensions $\d=2,3$, the authors in \cite{FR21} and \cite{FT22} established that free boundary points are generically regular for $\alpha=\frac{1}{2}$; similar results are expected to hold in $\d=2$ for all fractional powers and additional regularity on $V$. 
If we make an additional assumption on the obstacle, then we can guarantee that all points are regular. Namely, this holds if the free boundary is Lipschitz, $\psi \in C^{2+\gamma}$ for some $\gamma>\alpha$ and if $\Delta \psi<0$ on $\{\psi=c_V-V>0\}$; this is Proposition \ref{quantitative-regular} in the appendix. As we will see below, we actually need a slightly stronger condition on the free boundary; however, as long as \eqref{itemfbLip} and \eqref{itemposLap} hold, we can conclude that all points of the free boundary are regular. With the additional regularity of $V$, X. Ros Oton proves in Proposition \ref{prop1} in the appendix a quantified decay rate of $\muv$ near the free boundary: as $x \rightarrow x_0\in \partial \Sigma_V$ with $x_0$ a regular free boundary point then we have in our notation
\begin{equation}\label{decay of equilibrium measure}
\muv(x) \sim \dist(x,\partial \Sigma_V)^{1-\alpha},
\end{equation}
justifying the assumption \eqref{eq reg}.

Let $U \supset \Sigma$ be an open set such that \eqref{regular point} holds in $U \setminus \Sigma_V$. In order for us to define the transport map with $C^\sigma$ regularity for a general $0<\sigma<1$, we will need the free boundary to have $C^{1+\sigma}$ regularity. As a result of \cite[Theorem 1.1]{JN17}, we have $C^{2+\gamma}$ regularity at regular points for some $\gamma>0$ dependent on $\sigma$ %\cm{i don't understand this $\gamma$ depending on $\sigma$, what is this $\sigma$?}  
(and thus, in particular, $C^{1+\sigma}$) once $V \in C^{4}$; hence, for standard existence of a general H\"older continuous transport, we will need to assume $V\in C^4$. % \cm{same}
%\cm{this is a place where reg of $V$ is needed} 
Higher regularity can be obtained with higher regularity assumptions on $V$ and in fact if $V$ is smooth, the free boundary will be smooth as well. This result was established independently as well in \cite[Theorem 1.1]{KRS19}. The regularity assumptions \eqref{itemeqreg}-\eqref{itemgrowthV} are the additional higher regularity assumptions that we will need to recover the $C^k$ regularity estimates on $\psi$ given in Proposition \ref{transport}.
%Importantly, we note that with this assumption, \eqref{itemfbLip} is automatically strengthened to $C^{1,\sigma}$ regularity. \cm{same question}

\subsection{Proof of Proposition \ref{transport}}
We now give the proof of our main analytic result, on the existence of a good transport map with estimates on its derivatives. We will use some auxiliary results on fractional harmonic extensions proved in Section \ref{sec: flap}.

First, we observe that under our assumptions $\muv$ is an $A_2$-Muckenhoupt weight on $\Sigma$. Since $\muv$ is bounded below away from $\partial \Sigma$, we only need to verify this for small balls near the boundary. We use the assumption that all points of the free boundary are regular, and the quantitative form \eqref{itemeqreg} that we assume.
 Thus, for $\epsilon>0$ small we find a uniform bound 
\begin{equation*}
\dashint_{B_\epsilon}\muv\dashint_{B_\epsilon}\muv^{-1} \lesssim \dashint_{B_\epsilon}|x_1|^{1-\alpha}\dashint_{B_\epsilon}|x_1|^{\alpha-1}\lesssim \epsilon^{1-\alpha}\epsilon^{\alpha-1}\lesssim 1,
\end{equation*}
establishing the required uniform bound, by definition an $A_2$ weight. We now seek to establish the existence of weak energy solutions in $\Sigma$ of 
\be \begin{cases}
-\div (\muv\nab u)= \frac1\c(-\Delta)^\alpha \varphi^\Sigma& \text{in} \ \Sigma\\
\nabla u \cdot \vec{n}= \frac{w}{(1-\alpha)s}&\text{on} \ \Sigma\end{cases}
\ee
in the Hilbert space $H_{\muv}^1$. This is the Hilbert space (see \cite[Theorem 1]{GU09}) of functions $u$ such that 
\begin{equation*}
\int_{\Sigma}u^2 \muv+\int_{\Sigma}|\nabla u|^2\muv<\infty
\end{equation*}
equipped with inner product 
\begin{equation}\label{inner product}
\langle u, v \rangle=\int_{\Sigma}uv\muv+\int_{\Sigma}\nabla u \cdot \nabla v \muv.
\end{equation}
These weak solutions are ones that satisfy
\begin{equation}\label{weak solution}
\int_{\Sigma} \nabla u \cdot \nabla \phi \muv=\frac{1}{\cds}\int_{\Sigma} (-\Delta)^\alpha\varphi^{\Sigma} \phi
\end{equation}
for all $\phi \in H_{\muv}^1$.

\subsection*{Step 1: Existence of weak energy solutions}
Let us first restrict to the subspace $H$ of all functions $u$ with $\int u \muv=0$, which is a closed subspace since $\int u \muv$ is a bounded linear functional on $H_{\muv}^1$ as
\begin{equation*}
\left|\int u\muv \right|\leq \int |u|\sqrt{\muv}\sqrt{\muv} \leq \left(\int u^2 \muv\right)^{1/2}\left( \int \muv \right)^{1/2}\leq \|u\|_{H^1_{\muv}}.
\end{equation*}
Notice that on $H$, the inner product 
\begin{equation}\label{new inner product}
\langle u, v \rangle_1:=\int \nabla u \cdot \nabla v \muv
\end{equation}
induces a norm equivalent to the one defined in \eqref{inner product} because of the weighted Poincar\'e inequality for $A_2$-weights in \cite[Theorem 1.5]{FKS82}. Letting $\overline{u}=\int u \muv=\fint u \muv$ 
\begin{equation*}
\langle u, u\rangle_1 \leq \langle u, u \rangle=\langle u, u \rangle_1+\int u^2 \muv=\langle u, u \rangle_1+\int (u-\overline{u})^2\muv\leq C\langle u, u \rangle_1
\end{equation*}
since $\overline{u}=0$ in $H$.  We will conclude as usual via Riesz representation theorem once we establish that $\phi \mapsto \int (-\Delta)^\alpha \varphi^\Sigma \phi$ is a bounded linear functional on $H$; this, however, presents some difficulties because $(-\Delta)^\alpha \varphi^\Sigma$ blows up like $\dist(x,\partial \Sigma)^{-\alpha}$ as $x \rightarrow \partial \Sigma$. To get around this, we follow the approach of \cite[Theorem 2.10]{TTV22} and show that $(-\Delta)^\alpha \varphi^\Sigma$ is at least locally the divergence of a regular function.

\subsubsection*{Straightening the Boundary} 
Let $k$ be fixed as in the assumptions of Proposition \ref{transport}. %\cm{so $\sigma$ is $k$? or $\sigma \le k$?}
Choose $x_0 \in \partial \Sigma$, and without loss of generality set $x_0=0$. Locally then, in some open neighborhood $\mathcal{O}$ we can write $x_n=\phi(x_1,\dots,x_{n-1})$ with $\phi\in C^{1+k}(\mathcal{O}\cap \Sigma)$. 
Define the diffeomorphism
\begin{equation}\label{straigtening}
\Phi(x_1,\dots,x_{n-1},x_n)=(x_1,\dots,x_{n-1},x_n+\phi(x))=(\overline{x_1},\dots,\overline{x_n})
\end{equation}
which maps $B_R^+= B_R\cap \{x_n>0\} $ to $ \mathcal{O} \cap \{\overline{x_n}>0\}$ for some $R>0$. We can assume $\Phi(0)=\Phi^{-1}(0)=0$, and observe also that $\Phi(B_R \cap \{x_n =0\})\subset \mathcal{O}\cap \partial \Sigma$. The Jacobian of this map is then 
\begin{equation}\label{Jacobian}
(J_\Phi(x))_{ij}=\begin{cases} 
1 & \text{if }i=j \\
\frac{\partial \phi}{\partial x_i} & \text{if }i=n, \: 1 \leq j \leq n-1 \\
0 & \text{otherwise}
\end{cases}
\end{equation}
and satisfies $|\det(J_\Phi)|\equiv 1$. One can readily compute (albeit with some tedious multivariable calculus) that $-\div(\muv \nabla u)=\frac{1}{\cds}(-\Delta)^\alpha \varphi^\Sigma$ in $\mathcal{O} \cap \Sigma$ if and only if 
\begin{equation}\label{chart equation}
-\div(\overline{\muv}\overline{A}\nabla \overline{u})=\overline{f}
\end{equation}
in $B_R \cap \{x_n>0\}$, with $\overline{\muv}=\muv \circ \Phi$, $\overline{u}=u \circ \Phi$, $\overline{A}=J_{\Phi}^{-1}(J_{\Phi}^{-1})^t$ and $\overline{f}=\frac{1}{\cds}(-\Delta)^\alpha \varphi^\Sigma \circ \Phi$. Now, it follows (cf \cite[Lemma 2.4]{TTV22}, as in the proof of \cite[Corollary 2.9]{TTV22}) from our regularity assumptions on the boundary that 
\begin{equation*}
\frac{\dist(x,\partial \Sigma)\circ \Phi}{x_n} \in C^{k}(B_R^+) \hspace{3mm} \text{and}\hspace{3mm}\frac{\dist(x,\partial \Sigma)\circ \Phi}{x_n} \geq c>0
\end{equation*}
for some constant $c$. In particular, in view of \eqref{eq reg}, shrinking $R$ if necessary, we have
\begin{equation}\label{easy weight}
-\div(x_n^{1-\alpha}\tilde{A}\nabla \overline{u})=\overline{f}
\end{equation}
with $\tilde{A}=(s \circ \Phi)\frac{\dist(\Phi(x),\partial \Sigma)^{1-\alpha}}{x_n^{1-\alpha}}\overline{A} \in C^{k}(B_R^+)$. Additionally, using \eqref{assumpw}, we can write
\begin{equation*}
\overline{f}=(w \circ \Phi)\dist(\Phi(x),\partial \Sigma)^{1-\alpha}:=\overline{w}x_n^{-\alpha},
\end{equation*}
with $\overline{w}\in C^{k}(B_R^+)$ and bounded from below. Then, if we define 
\begin{equation}\label{vec field}
F=\( 0, 0, \dots, 0, \frac{1}{x_n^{1-\alpha}}\int_0^{x_n}\overline{w}(x_1,\dots,x_{n-1},t)t^{-\alpha}~dt\)
\end{equation}
we have $F \in C^{k}(B_R^+)$ with $\|F\|_{C^{k}(B_R^+)} \lesssim \|\overline{w}\|_{C^k(B_R^+)} \lesssim \|w\|_{C^{k}(B_R^+)}$ as in \cite[Lemma 2.4]{TTV22} and the arguments therein. Furthermore, we can rewrite \eqref{easy weight} as
\begin{equation*}
-\div(x_n^{1-\alpha}\tilde{A}\nabla \overline{u})=f=\div(x_n^{1-\alpha}F).
\end{equation*}
Reverting to our weight $\overline{\muv}$ and defining $\tilde{F}=\frac{x_n^{1-\alpha}}{\overline{\muv}}F \in C^{k}(B_R^+)$, we find 
\begin{equation}\label{div in chart}
-\div(\overline{\muv}\overline{A}\nabla \overline{u})=\div(\overline{\muv}\tilde{F})
\end{equation}
in $B_R^+$. Finally, if we define $G=J_\Phi(\tilde{F}\circ \Phi^{-1})$, we have $\|G\|_{C^k(\mathcal O)} \lesssim \|w\|_{C^k}$ and  $\tilde{F}=J_\Phi^{-1}(G \circ \Phi)$ and one can readily compute that 
\begin{equation*}
-\div(\overline{\muv}\overline{A}\nabla \overline{u})=\div(\overline{\muv}J_\Phi^{-1}(G \circ \Phi))
\end{equation*}
in $B_R^+$ if and only if 
\begin{equation}\label{div form}
-\div(\muv \nabla u)=\div(\muv G)
\end{equation}
in $\mathcal{O} \cap \Sigma$.

\subsubsection*{Bounded Linear Functional}
We use the divergence form above to show that $\int (-\Delta)^\alpha\varphi^\Sigma \phi$ is a bounded linear functional on $H$. Let $\{\mathcal{O}_i\}$ denote a finite collection of charts covering a small neighborhood $E$ of $\partial \Sigma$, and let $\{\psi_i\}$ denote a partition of unity on $E$ subject to this collection of charts. Using the previous substep, we can write
\begin{equation*}
\frac{1}{\cds}(-\Delta)^\alpha \varphi^\Sigma=\sum_i\psi_i\div(\muv G_i)
\end{equation*}
for some $G_i$ with $\|G_i\|_{C^{k}(\mathcal{O}_i)} \lesssim \|w\|_{C^{k}}$, and so 
\begin{align*}
\int_E \frac{1}{\cds} \phi (-\Delta)^\alpha \varphi^\Sigma =\sum_i \int_{\mathcal{O}_i} \frac{1}{\cds}\phi_i (-\Delta)^\alpha \varphi^\Sigma \psi_i &=\sum_i \int_{\mathcal{O}_i}\phi_i \div(\muv G_i) \psi_i  \\
&=-\sum_i \int_{\mathcal{O}_i}\muv G_i \cdot \nabla (\psi_i \phi) \\
&=-\sum_i \int_{\mathcal{O}_i}\sqrt{\muv}G_i \cdot \left(\sqrt{\muv}\psi_i \nabla \phi+\sqrt{\muv}\phi\nabla \psi_i  \right).
\end{align*}
By Cauchy-Schwarz and the boundedness of $\psi_i$, we find 
\begin{equation}\label{csb}
\left|\int_E \frac{1}{\cds}\phi(-\Delta)^\alpha \varphi^\Sigma \right|\lesssim \sum_i \( \int \muv |G_i|^2 \)^{1/2}\(\int_{\mathcal{O}_i}\muv |\nabla \phi|^2+\muv \phi^2 \)^{1/2} \lesssim \|\phi\|_H.
\end{equation}
We can similarly bound using Cauchy-Schwarz on $\Sigma \setminus E$ since $(-\Delta)^\alpha \varphi^\Sigma$ and $\muv$ are bounded above and below there to conclude that $\phi \mapsto \int \phi (-\Delta)^\alpha\varphi^\Sigma $ is a bounded linear functional on $H$.

\subsubsection*{Existence}
The Riesz Representation theorem guarantees a unique $u \in H$ such that \eqref{weak solution} holds for all $\phi \in H$. To extend \eqref{weak solution} to all test functions $\phi \in H_{\muv}^1$, we need
\begin{equation*}
\int_{ \Sigma}(-\Delta)^\alpha \varphi^{\Sigma}=0
\end{equation*}
for compatibility.
% If $\Sigma$ has several connected components, this is guaranteed by assumption \eqref{assconncomp}. If not, 
This follows from the fact that $(-\Delta)^\alpha$ is a mean-zero operator and thus 
\begin{equation*}
0=\int (-\Delta)^\alpha \varphi^{\Sigma}=\int_{\Sigma} (-\Delta)^\alpha \varphi^{\Sigma},
\end{equation*}
as desired.

\subsection*{Step 2: Boundary Condition}
Next, we show that these weak energy solutions have the desired boundary condition. From \cite[Theorem 1.1]{TTV22}, in the chart $\mathcal{O}_i$ the boundary condition satisfies 
\begin{equation*}
(\nabla u+G_i)\cdot \vec{n}=0,
\end{equation*}
where $\vec{n}$ is the outward pointing unit normal. Using the relation 
\begin{equation*}
\frac{1}{\cds}(-\Delta)^\alpha \varphi^\Sigma=w\, \dist(x,\partial \Sigma)^{-\alpha}=\div(\muv G_i)= \div(s \, \dist(x, \pa \Sigma)^{1-\alpha} G_i) 
\end{equation*}
in view of  \eqref{eq reg}, we find 
\begin{multline*}
w\, \dist(x,\partial \Sigma)^{-\alpha}=\dist(x,\partial \Sigma)^{1-\alpha}G_i \nabla s +s(1-\alpha)\dist(x,\partial \Sigma)^{-\alpha}\nabla \dist(x,\partial \Sigma)\cdot G_i \\
+s\, \dist(x,\partial \Sigma)^{1-\alpha}\div G_i,
\end{multline*}
which can be rewritten as 
\begin{equation*}
s(1-\alpha)\nabla \dist(x,\partial \Sigma) \cdot G_i=w-\dist(x,\partial \Sigma) \nabla s G_i-s\,\dist(x,\partial \Sigma)\div G_i.
\end{equation*}
Taking $x \rightarrow \partial \Sigma$ and using the regularity of $s, w$ and $G_i$ coupled with $\nabla \dist(x,\partial \Sigma) \rightarrow -\vec{n}$, we find 
\begin{equation*}
G_i \cdot \vec{n}=-\frac{w}{s(1-\alpha)}
\end{equation*}
and thus 
\begin{equation*}
\nabla u \cdot \vec{n}=\frac{w}{s(1-\alpha)}.
\end{equation*}
Setting $\psi=\nabla u$, we find a solution to \eqref{Rtransport} in $\Sigma$.

\subsection*{Step 3: Schauder Estimate up to the Boundary}
First, let us examine the behavior of $\psi=\nabla u$ in the interior of $\Sigma$. Writing $E=\cup \mathcal{O}_i$ as above for the union of charts covering $\partial \Sigma$, observe that $\muv$ is bounded below in $\Sigma \setminus E$. So, we can appeal to standard  elliptic regularity to control $u$ inside $E$, thus finding 
\begin{equation}\label{interior estimate}
\|u\|_{C^{2+\sigma}(A)}\lesssim \|w\|_{C^\sigma(A)}
\end{equation}
for any $A \subset E$, recalling that for integer $\sigma$ we define $C^{\sigma}$ as the H\"older space $C^{\sigma-1,1}$.  Applying \eqref{interior estimate} to $A=\carr_{2\ell}$ (we omit the $z$ in the notation) and using  \eqref{ptwise fraclap}, we obtain for any $ 1\le m\le k$, 
\begin{equation*}
\|\psi\|_{C^m(\carr_{2\ell})}\lesssim \ell^{\d+2}\|(-\Delta)^\alpha \varphi\|_{L^\infty}+\ell^{2+\s-\d}\|\varphi\|_{C^{1+m}} \lesssim 
\M (\ell^{2+\s} +\ell^{\s-\d-m+1})
\end{equation*}
where we have also used \eqref{estxi} and Lemma \ref{lemestxi}. This proves \eqref{tscale} in $\carr_{2\ell}$.
  
In $\hat \Sigma \backslash \carr_{2\ell}$ we can obtain the decay as well by applying elliptic estimates in dyadic annuli centered around $\carr_{2\ell}$ of radii $2^k\ell$ until we hit the edge of the bulk. More precisely, choose $k_\ast$ such that $2^{k_\ast}\ell$ is at macroscopic scale and $\carr_{2^{k_\ast} \ell}\subset \bulk$. Let $A_k=\carr_{2^\ell}\setminus \carr_{2^{k-2}\ell}$. We can use \eqref{interior estimate} and \eqref{ptwise fraclap} to obtain that 
\begin{equation*}
\left\|\psi\right\|_{C^m(A_k)}\lesssim \left\|(-\Delta)^\alpha \varphi^\Sigma \right\|_{C^{m-1}(A_k)}  \lesssim \ell^{\d+2}\|(-\Delta)^\alpha\varphi\|_{L^\infty}+\frac{\ell^\d\|\varphi\|_{L^\infty}}{\(2^{k-2}\ell\)^{2\d-\s+m-1}}.
\end{equation*}
The first term is dominant by Lemma \ref{lemestxi}, so 
\begin{equation*}
\left\|\psi\right\|_{C^m(A_k)}\lesssim \frac{\M\ell^\d}{\(2^{k-2}\ell\)^{2\d-\s+m-1}}
\end{equation*}
by \eqref{estxi}. Adjusting constants and summing over $A_k$ yields
%\cm{need to check again}
%Notice that we chose the largest possible exponent in the denominator, because at these mesoscopic scales the denominator is $\lesssim 1$, so maximizing the exponent maximizes the fraction. 
%Adjusting constants in each $A_k$ and applying \eqref{estxi} and Lemma \ref{lemestxi}, this tells us that  \cm{how do we sum the $\ell^{\d+2}\|(-\Delta)^\alpha\varphi\|_{L^\infty}$ over $k$?}
\begin{equation*}
|\nab^{\otimes m}\psi(x)|\lesssim \frac{\M\ell^\d}{|x-z|^{2\d-\s+m-1}}
\end{equation*}
which is \eqref{tscale} in $\bulk \backslash \carr_{2\ell}$.

Near the boundary of $\Sigma$, we need to be more careful due to the decay of $\muv$ and the blowup of $(-\Delta)^\alpha \xi^\Sigma$. Applying \cite[Theorem 1.1]{TTV22}, we obtain the control
\begin{equation}\label{prelim est}
\|u\|_{C^{1+\sigma}(\mathcal{O}_i \cap \Sigma) }\lesssim \|u\|_{L^2(\mathcal{O}_i \cap \Sigma)}+\|G_i\|_{C^{\sigma}(\mathcal{O}_i \cap \Sigma)}
\end{equation}
for any $\sigma \notin \mathbb{N}$. We can rewrite this estimate using the equation \eqref{Rtransport} and our divergence form for the right hand side. With $E=\cup \mathcal{O}_i$ as above and recalling that $\int |\nabla u|^2\muv$ is an equivalent norm on $H$, we write 
\begin{equation}\label{initial u L2}
\|u\|_{H}^2=\int_\Sigma | \nabla u |^2 \muv=\int_{\Sigma}\frac{1}{\cds}u (-\Delta)^\alpha \varphi^\Sigma =\int_{E}\frac{1}{\cds}(-\Delta)^\alpha \varphi^\Sigma u+\int_{\Sigma \setminus E}\frac{1}{\cds}(-\Delta)^\alpha \varphi^\Sigma u.
\end{equation}
The integral on $\Sigma \setminus E$ can be controlled via Cauchy-Schwarz, writing
\begin{align*}
\int_{\Sigma \setminus E}u (-\Delta)^\alpha \varphi^{\Sigma}  &\lesssim \left(\int_{\Sigma \setminus E}\frac{\left|(-\Delta)^\alpha \varphi^{\Sigma}\right|^2}{\muv}\right)^{1/2}\left(\int_{\Sigma \setminus E}u^2 \muv \right)^{1/2}
\end{align*}
and so since $\int \muv^{-1}<\infty$,
\begin{equation*}
\left|\int_{\Sigma \setminus E}\frac{1}{\cds}(-\Delta)^\alpha \varphi^\Sigma u\right| \lesssim \left\|(-\Delta)^\alpha \varphi^{\Sigma} \right\|_{L^\infty(\Sigma \setminus E)}\|u\|_H \lesssim \|w\|_{L^\infty(\Sigma)}\|u\|_H.
\end{equation*}
For the integral in $E$, we argue as in \eqref{csb}  that 
\begin{align*}
\left|\int_E \frac{1}{\cds}u (-\Delta)^\alpha \varphi^\Sigma \right|& \lesssim \sum_i \( \int \muv G_i \)^{1/2}\(\int_{\mathcal{O}_i}\muv |\nabla u|^2+\muv |u|^2 \)^{1/2}\\
& \lesssim \sum_i \|G_i\|_{L^\infty}\|u \|_H \lesssim \|w\|_{L^\infty(\Sigma)}\|u\|_H.
\end{align*}
Inserting these into \eqref{initial u L2} we find 
\begin{equation*}
\|u\|_{H}^2 \lesssim \|w\|_{L^\infty}\|u\|_H
\end{equation*}
and from \eqref{prelim est} we conclude that
\begin{equation*}
\|u\|_{C^{1+\sigma}(E)} \lesssim \|w\|_{C^{\sigma}(E)}.
\end{equation*}
for $\sigma \notin \mathbb{N}$. Recalling the definition $\psi=\nab u$ in $\Sigma$,  we can rephrase the global estimate as 
\begin{equation}\label{Schauder estimate}
\|\psi\|_{C^{\sigma}(\Sigma)} \lesssim \|w\|_{C^{\sigma}(\Sigma)}=\left\|\frac{(-\Delta)^\alpha\varphi^\Sigma}{\dist(x,\partial \Sigma)^{-\alpha}}\right\|_{C^{\sigma}(\Sigma)}.
\end{equation}
In view of H\"older interpolation of\eqref{ptwise fraclap}, we have (applying \eqref{Schauder estimate} for $\sigma=k+\epsilon$) 
%\cm{should be $k$ instead of $m$?}
\begin{equation*}
\|\psi\|_{C^{k+\epsilon}(\Sigma\backslash \bulk)} \lesssim \ell^{\d} \|\varphi\|_{L^\infty}+\ell^{\d+2}\|(-\Delta)^\alpha \varphi\|_{L^\infty}
\end{equation*}
which yields, in view of Lemma \ref{lemestxi},
\be\label{Schauder estimate2}
\|\psi\|_{C^{k}(\Sigma\backslash \bulk)} \lesssim \M \ell^\d
\ee 
using \eqref{estxi} and Lemma \ref{lemestxi}, which is \eqref{tscale} in $\Sigma \setminus \bulk$. Note that we needed to be careful to apply \eqref{Schauder estimate} for $\sigma \notin \mathbb{N}$ as \cite[Theorem 1.1]{TTV22} is stated only for noninteger $\sigma$.
%\cm{no, there is a discrepancy of $1$}

Finally, the above estimates were only for $\nab^{\otimes m} \psi$ with $m \geq 1$. Extending the above estimates on $\nab^{\otimes m} \psi$ for $m \geq 1$ to an $L^\infty$ $(m=0)$ bound on $\psi$ then follows from using that $|\psi|=O(\ell^d)$ at $\partial \Sigma$ and integrating the derivative estimate.

\subsection*{Step 4: Continuity Across the Boundary}
We now establish continuity of $\psi$ across the boundary. 
\subsubsection*{Rewriting the external transport at the boundary}
Let $x_0 \in \partial \Sigma$, and let $c_1(x_0)$ be the unique constant such that 
\begin{equation}\label{c1}
\nabla \zeta_V \cdot \vec n=c_1(x_0)\dist(x,\partial \Sigma)^\alpha+O\(|x-x_0|^{\alpha+a}\)
\end{equation}
where $a>0$ is such that $\alpha+a<1$. Thus, in $U \setminus \Sigma$, as $x$ approaches $x_0$ from the outside,
\begin{equation*}
\lim_{x \rightarrow x_0}\psi \cdot \vec n(x)=\lim_{x \rightarrow x_0}\frac{\varphi^\Sigma-\varphi}{c_1(x_0)\dist(x,\partial \Sigma)^\alpha}
\end{equation*}
Furthermore, since $\varphi \in C^2$ we know that $(-\Delta)^\alpha \varphi \in L^\infty$; then, regularity for fractional elliptic problems (cf. \cite[Theorem 1.2]{RoS14}) gives us a unique constant $c_2(x_0)$ near $x_0$ such that 
\begin{equation}\label{c2}
\varphi^\Sigma(x)-\varphi(x)=c_2(x_0)\dist(x,\partial \Sigma)^\alpha+o( \dist (x, \pa \Sigma)^\alpha) \quad \text{as } \ x\to x_0
\end{equation}
and so 
\begin{equation}\label{ext bdy}
\lim_{x\rightarrow x_0}\psi \cdot \vec n=\frac{c_2(x_0)}{c_1(x_0)}
\end{equation}
as $x\rightarrow x_0$ from $U \setminus \Sigma$.

%We know (see \eqref{regular point})) that at a regular point $x_0 \in \partial \Sigma$, we can write
%\begin{equation}\label{zeta expansion}
%\zeta_V(x)=c_1(x_0)\dist(x,\partial \Sigma)^{1+\alpha}+ O\(|x-x_0|^{1+\alpha+\beta}\)
%\end{equation}
%for some unique constant $c_1(x_0)$ and $\beta>0$ such that $\alpha+\beta<1$. In particular, for $x \rightarrow x_0$ we have
%\begin{equation}\label{grad zeta}
%\nabla \zeta_V(x)=(1+\alpha)c_1(x_0)\dist(x,\partial \Sigma)^{\alpha}\nabla \dist(x,\partial \Sigma)+O\(|x-x_0|^{\alpha+\beta}\),
%\end{equation}
%%\cm{I would like to be more explicit about the higher order terms, in particular it's not clear how to go from the first estimate to the second} 
%and so we can rewrite \eqref{Rtransport}) on $\Sigma^c$ as 
%\begin{equation*}
%\psi \cdot \nabla \dist(x,\partial \Sigma)=\frac{\varphi^\Sigma-\varphi}{(1+\alpha)c_1(x_0)\dist(x,\partial \Sigma)^\alpha}+\text{ higher order terms.}
%\end{equation*}
%Since $\nabla  \dist(x,\partial \Sigma) \rightarrow \vec{ n}$ as $x \rightarrow \partial \Sigma$, we can write 
%\begin{equation}\label{ext bdy}
%- \psi \cdot \vec{ n}(x_0)=\lim_{x \rightarrow x_0}\frac{\varphi^\Sigma(x)-\varphi(x)}{(1+\alpha)c_1(x_0)\dist(x,\partial \Sigma)^\alpha}.
%\end{equation}
\subsubsection*{Agreement of Boundary Conditions}
We now claim that the Neumann boundary condition in \eqref{Rtransport} and \eqref{ext bdy} agree. First, recalling that for $x_0$ on the boundary $$w(x_0)=\lim_{x \rightarrow x_0}\frac{1}{\c} (-\Delta)^\alpha \varphi^\Sigma(x) \dist(x,\partial \Sigma)^\alpha$$ and $s(x_0)=\lim_{x \rightarrow x_0}\frac{\muv(x)}{\dist(x,\partial \Sigma)^{1-\alpha}}$, we can rewrite the Neumann condition in \eqref{Rtransport} as 
\begin{equation}\label{int bdy}
\psi \cdot \vec{n}(x_0)=\lim_{x \rightarrow x_0} \frac{(-\Delta)^\alpha \varphi^\Sigma(x) \dist(x, \partial \Sigma)}{(1-\alpha)\cds \muv(x)}.
\end{equation}
%Next, since $\varphi \in C^2$ we know that $(-\Delta)^\alpha \varphi \in L^\infty$; then, regularity for fractional elliptic problems (cf. \cite[Theorem 1.2]{RoS14}) gives us a unique constant $c_2(x_0)$ near $x_0$ such that 
%\begin{equation}\label{ext diff}
%\varphi^\Sigma(x)-\varphi(x)=c_2(x_0)\dist(x,\partial \Sigma)^\alpha+o( \dist (x, \pa \Sigma)^\alpha) \quad \text{as } \ x\to x_0.
%\end{equation}
%This allows us to compute the limit in \eqref{ext bdy}) as 
%\begin{equation}\label{ext limit}
%-\psi \cdot \vec{ n}(x_0)=\lim_{x \rightarrow x_0}\frac{c_2(x_0)\dist(x,\partial \Sigma)^\alpha}{(1+\alpha)c_1(x_0)\dist(x,\partial \Sigma)^\alpha}=\frac{c_2(x_0)}{(1+\alpha)c_1(x_0)}.
%\end{equation}
%Let us compare this to the interior boundary condition. 
Lemma \ref{lem1}  yields
\begin{equation}\label{RHS exp}
(-\Delta)^\alpha \varphi^\Sigma=(-\Delta)^\alpha \varphi+\overline{c}_\alpha c_2(x_0)\dist(x,\partial \Sigma)^{-\alpha}+o(\dist(x, \partial \Sigma)^{-\alpha}) 
\end{equation}
where 
\begin{equation}\label{c alpha}
\overline{c}_\alpha=-\frac{\Gamma(1+\alpha)}{\Gamma(1-\alpha)}.
\end{equation}
Proposition \ref{prop1} also gives
\begin{equation}\label{eq exp}
(1-\alpha)\cds \muv(x)=\overline{c}_\alpha c_1(x_0)\dist(x,\partial \Sigma)^{1-\alpha}+o(\dist(x, \partial \Sigma)^{1-\alpha}) \quad \text{as } \ x\to x_0\end{equation}
where $c_1(x_0)$ is as in \eqref{c1}.
%where $c_3(x_0)$ is the unique constant such that 
%\begin{equation}\label{eqn from RO}
%\nabla \zeta_V \cdot \vec{n}=c_3(x_0)\dist(x,\partial \Sigma)^{\alpha}.
%\end{equation}
%Comparing \eqref{eqn from RO}) and \eqref{grad zeta}), we see that $c_3(x_0)=(1+\alpha)c_1(x_0)$. Combining this with \eqref{RHS exp}) and \eqref{eq exp}), we can compute \eqref{int bdy}) as 
We then compute 
\begin{align}
\notag  \psi \cdot \vec{ n}(x_0)&=\lim_{x \rightarrow x_0} \frac{(-\Delta)^\alpha \varphi^\Sigma(x) \dist(x, \partial \Sigma)}{(1-\alpha)\cds \muv(x)}\\
\notag &=\lim_{x \rightarrow x_0} \frac{\((-\Delta)^\alpha \varphi+\overline{c}_\alpha c_2(x_0)\dist(x,\partial \Sigma)^{-\alpha}\)\dist(x, \partial \Sigma)}{\overline{c}_\alpha c_1(x_0)\dist(x,\partial \Sigma)^{1-\alpha}}\\
\notag &=\lim_{x \rightarrow x_0}\frac{(-\Delta)^\alpha \varphi}{\overline{c}_\alpha c_1(x_0)}\dist(x,\partial \Sigma)^{\alpha}+\lim_{x \rightarrow x_0} \frac{\overline{c}_\alpha c_2(x_0)\dist(x,\partial \Sigma)^{1-\alpha}}{\overline{c}_\alpha c_1(x_0)\dist(x,\partial \Sigma)^{1-\alpha}} \\
 \label{int limit} &=\frac{c_2(x_0)}{c_1(x_0)},
\end{align} 
which agrees with \eqref{ext bdy}. Thus, $\psi$ is continuous in its normal component across the boundary. We may then build $\psi$ continuous exactly in the same way as in the proof of \cite[Lemma 3.4]{LS18} by compensating with an appropriate tangential vector field $\psi^\perp$: consider the trace $\psi -(\psi\cdot \vec{n}) \vec{n}$ on $\pa \Sigma$ from the inside, and  extend it   to a regular vector field vanishing outside $U$,   then subtract off the projection of that vector field onto $\nabla \zeta_V$ to obtain a vector field  $\psi^\perp$ which remains perpendicular to $\nabla \zeta_V$ and vanishing in $U^c$. 
In view of \eqref{Schauder estimate2}, and the $C^{k+1}$ regularity of $\partial \Sigma$, we can obtain $\psi^\perp$ such that, for $m\le k$, 
\be \label{contrpsiperp}
\|\nab^{\otimes m} \psi^\perp\|_{L^\infty(\Sigma^c)} \lesssim \|\varphi\|_{L^\infty}\ell^\d.\ee

The vector field $\psi$ is now defined in $\R^\d$.

\subsection*{Step 4: Exterior Schauder Estimate} %Using the expansions (\ref{grad zeta}) and (\ref{ext diff}), the regularity of $\psi$ on $\Sigma^c$ is controlled by that of $\frac{\varphi^\Sigma-\varphi}{\dist(x,\partial \Sigma_V)^\alpha}$. The $\gamma$-H\"older regularity for all $\gamma \in (0,\alpha)$ then follows from (\ref{aharm reg}). 
We apply the boundary Harnack inequality for $\alpha$-harmonic functions from \cite[Theorem 1.3]{ARo20}. Up to a rotation, we may as well assume at a point $x_0 \in \partial \Sigma$ that $\vec n=\vec{e}_n$, the unit vector in the $x_n$ direction. The regularity of $\psi^\perp$ is provided in \eqref{contrpsiperp}, so it is sufficiently to consider 
\begin{equation*}
(\varphi^\Sigma-\varphi)\frac{\nabla \zeta_V}{|\nabla \zeta_V|^2}.
\end{equation*}
Notice that outside of $\Sigma$, we have $(-\Delta)^\alpha \zeta_V=\muv+(-\Delta)^\alpha(V-c_V)=(-\Delta)^\alpha(V-c_V)$. Since $\zeta_V\equiv 0$ in $\Sigma$, we then have
\begin{equation*}
\begin{cases}
(-\Delta)^\alpha (\partial_i \zeta_V)=(-\Delta)^\a (V-c_V) & \text{in }B_r(x_0) \cap \Sigma^c \\
\partial_i \zeta_V=0 & \text{in }B_r(x_0) \cap \Sigma
\end{cases}
\end{equation*}
for $x_0 \in \partial \Sigma$ and $r>0$ sufficiently small. We also have 
\begin{equation*}
\begin{cases}
(-\Delta)^\alpha (\varphi^\Sigma-\varphi)=-(-\Delta)^\a \varphi& \text{in }B_r(x_0) \cap \Sigma^c \\
\varphi^\Sigma-\varphi=0 & \text{in }B_r(x_0) \cap \Sigma.
\end{cases}
\end{equation*}
Notice also that it follows from \eqref{regular point} (see also \cite{ROS17}) that $\partial_n \zeta_V \gtrsim \dist(x,\partial \Sigma)^\alpha$. Now, it follows from \cite[Theorem 1.3]{ARo20} that 
\begin{equation*}
\frac{\varphi^\Sigma-\varphi}{\partial_n \zeta_V}, \hspace{3mm} \frac{\partial_i \zeta_V}{\partial_n \zeta_V} \in C^{\sigma}
\end{equation*}
for any $\sigma\notin \mathbb{N}$ such that $\sigma+\alpha, \sigma-\alpha \notin \mathbb{N}$. As in the proof of \cite[Theorem 1.1]{ARo20}, we write
\begin{equation}\label{coordinate est}
\frac{(\varphi^\Sigma-\varphi)\partial_i\zeta_V}{|\nabla \zeta_V|^2}=\frac{(\varphi^\Sigma-\varphi)}{\partial_n \zeta_V}\frac{\partial_i \zeta_V \partial_n \zeta_V}{|\nabla \zeta_V|^2}.
\end{equation}
%\cm{$\sigma_i$ bad notation here. What is it good for anyway?}
Now, 
\begin{equation*}
\frac{\partial_i \zeta_V \partial_n \zeta_V}{|\nabla \zeta_V|^2}=\frac{\frac{\partial_i \zeta_V}{\partial_n\zeta_V}}{1+\sum_{i=1}^{n-1}\left(\frac{\partial_i \zeta_V}{\partial_n \zeta_V}\right)^2} \in C^{\sigma}(B_r(x_0))
\end{equation*}
by the boundary Harnack inequality, and 
\begin{align*}
 \left\|\frac{(\varphi^\Sigma-\varphi)}{\partial_n \zeta_V}\right\|_{C^{\sigma}(B_r(x_0))}& \lesssim \|(-\Delta)^\alpha \varphi\|_{C^{\sigma-\alpha}(B_r(x_0))}+\|\varphi^\Sigma-\varphi\|_{L^\infty(\R^n)} \\
 &\lesssim \|(-\Delta)^\alpha \varphi\|_{C^{\sigma-\alpha}(U \setminus \Sigma)}.
 \end{align*}
where we have used \cite[Theorem 1.2]{RoS16} to control the $L^\infty$ term. Inserting these into (\ref{coordinate est}) and using \cite[(4.7)]{GT01}, we find 
\begin{equation*}
\left\|(\varphi^\Sigma-\varphi)\frac{\nabla \zeta_V}{|\nabla \zeta_V|^2}\right\|_{C^{k+\epsilon}(B_r(x_0))} \lesssim  \|(-\Delta)^\alpha \varphi\|_{C^{k+\epsilon-\alpha}(U \setminus \Sigma)},
\end{equation*}
which, using \eqref{fraclap der} and \eqref{estxi} coupled with H\"older interpolation completes the proof. We have been careful to apply \cite[Theorem 1.3]{ARo20} for $\sigma=k+\epsilon$, since that result does not hold for integer $k$.

 To control the decay of $\psi$ in $U^c$, we note that by definition \eqref{Riesz effective potential} we have 
 $\nab \zeta_V=\nab h^{\mu_V}+\nab V$. Since $\mu_V$ has compact support, $\nab h^{\mu_V}$ and its derivatives decay like those of $\g$, i.e.~faster than $|x|^{-\s-1}$. Since $\s>\d-2\ge -1$, this means that $\nab h^{\mu_V}$ and all its derivatives tend to $0$ at infinity, thus $\nab \zeta_V\sim \nab V$ at infinity. It follows from the definition of $\psi$ in \eqref{Rtransport} (and the fact that $\psi^\perp$ is compactly supported) that 
 $$\nab^{\otimes m} \psi \sim \nab^{\otimes m} \( \varphi^\Sigma\frac{\nab V}{|\nab V|^2}\) \quad \text{as} \ |x|\to \infty$$ 
 
 Using \eqref{decayxisigma-scale}, \eqref{estxi}, Lemma \ref{lemestxi}, $\s >\d-2$ and \eqref{itemgrowthV} we have that 
  \begin{equation*}
 |\psi|\lesssim \frac{\left\|(-\Delta)^\alpha \varphi\right\|_{L^\infty}\ell^{\d+2}+\|\varphi\|_{L^\infty}\ell^{\d}}{|x-z|^{\s+2}}\lesssim \frac{\M\ell^{\d}}{|x-z|^{\s+2}},
 \end{equation*}
 which is \eqref{tscale}. For the derivatives, we have using \eqref{decayxisigma-scale}, \eqref{itemgrowthV}, \eqref{estxi} and the Faa-di Bruno formula that 
\begin{equation*}
|\nab^{\otimes m} \psi|\lesssim \sum_{j=0}^m \left\|D^j \varphi^\Sigma\right\| \left\|D^{m-j}\(\frac{1}{|\nabla V|}\)\right\| \lesssim \frac{\left\|(-\Delta)^\alpha \varphi\right\|_{L^\infty}\ell^{\d+2}+\|\varphi\|_{L^\infty}\ell^{\d}}{|x-z|^{\s+2+m}} \lesssim \frac{\M \ell^{\d}}{|x-z|^{\s+2+m}}\end{equation*}
for $|x|$ large enough.
This completes the proof of \eqref{tscale} and of Proposition \ref{transport}.

\smallskip

We will often need the following set of consequences.
\begin{lem}\label{computationratio}
Assume \eqref{estxi} hold for $k= 3$.
For $\psi$ as above, and any $n\ge 1$ integer, for  $U$ as above, we have
\begin{equation}\label{cpsi}
\int_U |\psi|^n \lesssim \M^n\ell^{(1-n)\d+n\s+n},
\end{equation}
\begin{equation}\label{cdpsi}
\int_U|D\psi|^n \lesssim  \M^n \ell^{(1-n) \d + n\s},
\end{equation}
\begin{equation}\label{cratio1}
\int_{\Sigma}\frac{|\psi(x)-\psi(y)|^n}{|x-y|^{\s+n}}\, dy\lesssim
\M^n\begin{cases}  \ell^{n\d}\max( |x-z|,\ell)^{\d(1-2n)+\s(n-1)}& \text{if} \ x\in U\\
\ell^{n\d} |x-z|^{-n(\s+2)}& \text{if} \ x\in U^c\end{cases}
\end{equation} where $z$ is the center of $\carr_{\ell}$, and 
\be\label{cratio}
\int_{\Sigma^2}\frac{|\psi(x)-\psi(y)|^n}{|x-y|^{\s+n}}\, dxdy \lesssim \M^n\ell^{(2-n)\d+(n-1)\s}.
\ee
\end{lem}
\begin{proof}
From the estimates \eqref{tscale}, we have
\be\int_U|\psi|^n\lesssim
\M^n\( \frac{\ell^\d}{\ell^{n(\d-\s-1)}}+ \ell^{n\d} \int_\ell^1\frac{r^{\d-1}}{r^{n(2\d-\s-1)}}\, dr \) \lesssim \M^n\ell^{\d(1-n) +n(\s+1)}.\ee 
Next, we write 
\begin{equation}\label{cdpsi2}
\int_U|D\psi|^n \lesssim \M^n\(\frac{\ell^\d}{\ell^{n(\d-\s)}}+\ell^{n\d}\int_\ell^1\frac{r^{\d-1}}{r^{n (2\d-\s)}}\, dr \)\lesssim \M^n\ell^{(1-n) \d + n\s}.\ee

For \eqref{cratio1}, we first consider $x\in \carr_{2\ell}$.
Then, using $\frac{|\psi(x)-\psi(y)|}{|x-y|}\lesssim \M\ell^{\s-\d}$ from \eqref{tscale}, we find
\begin{align*}
\int_{U}\frac{|\psi(x)-\psi(y)|^n}{|x-y|^{\s+n}}\, dy &\lesssim \int_{y\in U, |y-x|\leq \ell}\frac{|\psi(x)-\psi(y)|^n}{|x-y|^{\s+n}}\, d+\int_{y\in U, |y-x|>\ell}\frac{|\psi(x)|^n+|\psi(y)|^n}{|x-y|^{\s+n}}\, dy \\
&\lesssim \frac{\M^n}{\ell^{n(\d-\s)}}\int_0^\ell\frac{1}{r^\s}r^{\d-1}\, dr+\M^n\ell^{-n(\d-\s-1)}\int_{\ell \leq |u|\leq C}\frac{1}{|u|^{\s+n}}\, du\\
& \lesssim   \M^n  \ell^{(1-n)(\d-\s)} .\end{align*}

Next, for $x\in U \backslash \carr_{2\ell}$, we obtain similarly using \eqref{tscale}, $z$ being the center of $\carr_{2\ell}$,
\begin{align*}
&\int_{U}\frac{|\psi(x)-\psi(y)|^n}{|x-y|^{\s+n}}\, dy \lesssim \int_{y\in U, |y-x|\leq \hal |x-z|}\frac{|\psi(x)-\psi(y)|^n}{|x-y|^{\s+n}}\, dy+\int_{y\in U, |y-x|>\hal |x-z|}\frac{|\psi(x)|^n}{|x-y|^{\s+n}}\, dy \\
& \qquad \qquad + \int_{y\in U, |y-x|>\hal |x-z|}\frac{|\psi(y)|^n}{|x-y|^{\s+n}}\, dy\\
&\lesssim \frac{\M^n\ell^{n\d}}{|x-z|^{n(2\d-\s)}}\int_0^{\hal |x-z|}\frac{1}{r^\s}r^{\d-1}\, dr+\M^n\ell^{n\d}|x-z|^{-n(2\d-\s-1)}\int_{\hal |x-z| \leq |u|\leq C}\frac{1}{|u|^{\s+n}}\, du\\
& \qquad \qquad + \M^n\int_{\hal |x-z| \leq |y-x|\leq C}\frac{1}{|y-x|^{\s+n}}\frac{\ell^{n\d}}{\max(|y-z|,\ell)^{n(2\d-\s-1)}}\, dy
\\
& \lesssim   \M^n  \ell^{n\d} |x-z|^{\d(1-2n)+\s(n-1)} .\end{align*}
Finally, for $x \in U^c$,  we obtain similarly
\begin{align*}
&\int_{\Sigma}\frac{|\psi(x)-\psi(y)|^n}{|x-y|^{\s+n}}\, dy \lesssim \int_{y\in \Sigma, |y-x|\leq \hal |x-z|}\frac{|\psi(x)-\psi(y)|^n}{|x-y|^{\s+n}}\, dy+\int_{y\in U, |y-x|>\hal |x-z|}\frac{|\psi(x)|^n}{|x-y|^{\s+n}}\, dy \\
& \qquad \qquad + \int_{y\in \Sigma, |y-x|>\ep}\frac{|\psi(y)|^n}{|x-y|^{\s+n}}\, dy\\
&\lesssim  \M^n  \ell^{n\d} |x-z|^{\d(1-2n)+\s(n-1)}+  \M^n\frac{\ell^{n\d}}{|x-z|^{n(\s+2)}}.\end{align*}
This proves \eqref{cratio1}.

For \eqref{cratio}, we integrate \eqref{cratio1} over $U$ to find 
\begin{align*}
\int_{\Sigma^2}\frac{|\psi(x)-\psi(y)|^n}{|x-y|^{\s+n}}\, dxdy \lesssim 
 \M^n \( \ell^\d \ell^{ (\d-\s)(1-n)}+   \ell^{n\d}\ell^{\d+\d(1-2n) + \s(n-1)}\),\end{align*}
 which yields the result.
\end{proof}

\section{Splitting, the electric formulation, and transport calculus}\label{sec: fluct prelim}

In order to complete the proof of Theorems \ref{FirstFluct}-\ref{CLT}, we need to recall the main ingredients of our electric-formulation based analysis, originating in \cite{PS17}. Most of what follows is based on material that can be found in \cite{S24}.

\subsection{The next-order energy and partition functions}\label{sec:remindersF}

As discussed in the introduction, under our assumptions on $V$ the sequence of empirical measures $\frac{1}{N}\sum \delta_{x_i}$ converges in a large deviations sense \cite[Theorem 3.3]{S24} to the equilibrium measure $\muv$. Splitting off the main term of the interaction leads to a definition of the next-order energy.

\begin{lem}[Splitting Formula]\label{Riesz Splitting Formula}
Given any configuration $\XN \in (\R^\d)^N$, the energy being as in \eqref{defHN}, we have
\begin{align}\label{splitting}
\HN(\XN)=N^2\I(\muv)+N\sum_{i=1}^N\zeta_V(x_i)+\FN(\XN,\muv)
\end{align}
where $\I$ is defined in \eqref{Continuous energy} and  the next-order energy $\FN(\XN,\muv)$ is defined  in \eqref{defFN}.
% by 
%\begin{equation*}
%\FN(\XN,\muv)=\frac{1}{2}\iint_{\triangle^c}\g(x-y) d\left(\sum_{i=1}^N \delta_{x_i}-N\muv\right)(x)d \left(\sum_{i=1}^N \delta_{x_i}-N\muv\right)(y).
%\end{equation*}
\end{lem}
We refer to \cite[Lemma 5.1]{S24} for the proof.

Associated to this next-order energy is the \textit{next-order partition function}
\begin{equation}
\label{defK fluct}
 \K_{N,\beta}(\mu,\zeta):=  \int_{(\R^{\d })^N} 
\exp\left(- \beta N^{-\frac{\s}{\d}}\(\F_N(X_N,\mu) +N \sum_{i=1}^N \zeta(x_i) \) \right)
\, dX_N.\end{equation}
The Gibbs measure can then be rewritten as 
\begin{equation}\label{splitgibbs}
d\PNbeta(\XN)= \frac{1}{\K_{N,\beta}(\muv,\zeta_V)}\exp\(- \beta N^{-\frac\s\d} \(\F_N(\XN, \muv)+ N\sum_{i=1}^N \zeta_V(x_i)\) \) dX_N.\end{equation}

We will also  need in the proofs of Theorem \ref{FirstFluct} and \ref{CLT} a next-order partition function restricted to a given event $\mathcal{G}$, which will be denoted as  
\begin{equation}\label{defK fluct event}
 \K_{N,\beta}^{\mathcal{G}}(\mu,\zeta):=  \int_{(\R^{\d })^N} 
\exp\left(- \beta N^{-\frac{\s}{\d}}\(\F_N(X_N,\mu,U) +N \sum_{i=1}^N \zeta(x_i) \) \right)\indic_{\mathcal{G}}
\, dX_N.\end{equation}

In \cite[Corollary 5.23]{S24} an exponential moment control of the energy in the form 
\be \label{expmomentcontrol}
\left|\log \Esp_{\PNbeta}\( \exp\(\frac{\beta}{2}N^{-\frac\s\d} \( \F_N(\XN, \muv)+\(\frac{N}{2\d}\log N\) \indic_{\s=0}+N \sum_{i=1}^N \zeta_V(x_i) \)\)\) \right|\le C(1+\beta) N,
\ee where $C>0$ depends on $\s,\d, \zeta_V$ and $\|\muv\|_{L^\infty}$, 
are shown from upper and lower bounds on $ \K_{N,\beta}(\mu,\zeta)$.

They can be seen as a ``local law'' at the macroscale.
In particular it follows that, except with probability $e^{-C\beta N}$, we have 
\be \label{macrolaw1} \F_N(\XN, \muv)+\(\frac{N}{2\d}\log N\) \indic_{\s=0}+N \sum_{i=1}^N \zeta_V(x_i) \lesssim_\beta N^{1+\frac\s\d}\ee
where $C>0$ depends on $\s,\d, \zeta_V$ and $\|\muv\|_{L^\infty}$. We will thus be able to intersect all our good events with this large probability event.

\subsubsection{Electric formulation}

We will use the so-called \textit{electrostatic} approach to studying $\F_N$, introduced for the general Riesz gas in \cite{PS17}. A key tool in this approach is the Caffarelli-Silvestre \cite{CS07} extension procedure for reinterpreting fractional Laplace operators, as introduced in the Riesz context in \cite{PS17}. The connection is based on the observation that (up to a constant) $\g$ is the solution kernel for the fractional Laplace operator $(-\Delta)^\alpha$, where $\alpha=\frac{\d-\s}{2}$. Considering the extended function $\g(x,y)$ on $\R^{\d+1}$ (where $x \in \R^\d$ and $y \in \R$) one can observe  that $\g$ is a fundamental solution of a degenerate elliptic equation, i.e.~up to a constant solves
\begin{equation}\label{fundamental solution of extension}
\begin{cases}
\div(|y|^\gamma \nabla u)=0 & \text{in }\R^\d \times (\R \setminus \{0\}) \\
-\lim_{|y|\downarrow 0}|y|^{\gamma}\partial_y u(\cdot,y)=\delta_0 & \text{on }\R^\d \times \{0\}
\end{cases}
\end{equation}
with $\gamma$ satisfying
\begin{equation}\label{definition of gamma}
\d-1+\gamma=\s.
\end{equation}
In particular, if we let $\mu$ denote a measure on $\R^\d$, then the extension of the potential $h^\mu=\g \ast \mu$ to $\R^{\d+1}$ given by 
\begin{equation}\label{extended potential}
h^\mu(x,y)=\int_{\R^{\d+1}}\frac{1}{|(x,y)-(x',y')|^\s}\, d(\mu\delta)_{\R^\d}(x',y')= \int_{\R^{\d}}\frac{1}{|(x,y)-(x',0)|^\s}\, d \mu(x')
\end{equation}
solves 
\begin{equation}\label{definition of degenerate elliptic PDE}
-\div(|y|^\gamma\nabla h^\mu)=\cds\mu\delta_{\R^\d}\quad \text{in} \ \R^{\d+1},
\end{equation}
where we denote by $\delta_{\R^\d}$ the uniform measure on $\R^\d\times \{0\}$ characterized by the fact that for any continuous function $\varphi $ in $\R^{\d+1}$, $\int_{\R^{\d+1}} \varphi \, \drd=\int_{\R^\d}\varphi(x,0)dx$.

To formalize the electrostatic rewriting of the next-order energy, we will need a smearing procedure that regularizes the electrostatic potential generated by point charges.  That procedure is only needed in the case $\s\ge 0$ where $\g$ is singular at the origin. While much of what we define can be found in \cite[Sec.~4.1.3]{S24}, we restate much of it here for readability. If we define
\begin{equation}\label{extended truncation}
\f_\eta(x)=(\g(x)-\g(\eta))_+
\end{equation}
either for $x \in \R^\d$ or by extension for  $x\in \R^{\d+1}$, this is a function supported in $B(0,\eta)$ and  satisfying 
\begin{equation}\label{PDE truncation}
-\div(|y|^\gamma \nabla {\f}_\eta)=\cds\left(\delta_0-\delta_0^{(\eta)}\right)
\end{equation}
where we define
$\delta_{x_0}^{(\eta)}=-\frac1\cds \div ( \yg \nab \g_\eta(x-x_0))$. It is a weighted  measure of mass $1$ supported on $\partial B(0,\eta)$ satisfying 
\begin{equation*}
\int \varphi \delta_0^{(\eta)}=\frac{1}{\cds}\int_{\partial B(0,\eta)}\varphi(x,y)|y|^\gamma \g'(\eta)
\end{equation*}
for smooth $\varphi$.  
If $u$ solves
\begin{equation}\label{true Riesz potential}
-\div(|y|^\gamma \nabla u)=\cds\left(\sum_{i=1}^N \delta_{(x_i,0)}-\mu\delta_{\R^\d}\right),
\end{equation}
for any truncation vector $\vec{\eta}=(\eta_1,\dots,\eta_N)$, we define 
\begin{equation}\label{truncated potential}
u_{\vec{\eta}}(X):=u-\sum_{i=1}^N {\f}_{\eta_i}(x-(x_i,0))
\end{equation}
and observe that
\begin{equation}\label{PDE truncated potential}
-\div(|y|^\gamma \nabla u_{\vec{\eta}})=\cds \left(\sum_{i=1}^N \delta_{(x_i,0)}^{(\eta_i)}-\mu\delta_{\R^\d}\right).
\end{equation} 
We will also commonly identify $x_i\in \R^\d$ with $(x_i, 0)\in \R^{\d+1}$.

For our computations, we will often make use of the identification
\begin{equation}\label{extended truncation is Rd truncation}
\int_{U \times [-h,h]}{\f}_\eta(x,y)\, d\mu\delta_{\R^d}(x,y)=\int_U \f_\eta(x)\, d\mu(x).
\end{equation}
%for $h>\eta$, where we have defined 
%\begin{equation}\label{eq:truncation}
%\f_{\eta}(x)=(\g(x)-\g(\eta))_+.
%\end{equation}
We will frequently use that 
\be\label{eq:intf}
\int_{\R^{\d}} |\nab^{\otimes n} \f_{\eta}|\le C \eta^{\d-\s-n}.
\ee
A truncation parameter  that we will often choose is the minimal distance \eqref{defri}.
%\begin{equation}\label{minimal distance}
%\rr_i:=\frac{1}{4}\min \left(\min_{j \ne i} |x_i-x_j|, N^{-1/\d}\right).
%\end{equation}
With this truncation defined, we can now integrate by parts. The following exact formula follows from an examination of the proof of \cite[Proposition 1.6]{PS17} using the decay of $h_{N,\vec{\eta}}$ at infinity; see \cite[Lemma 4.10]{S24}.
\begin{lem}[Riesz electric formulation  of the next order energy]\label{riesz electric}
Let $\XN$ be a configuration of points in $\R^\d$, $\mu$ a probability density on $\R^\d$ such that 
\begin{equation}\label{electric assumption}
\iint |\g(x-y)|~d|\mu|(x)d|\mu|(y)<+\infty
\end{equation}
and let $\vec{\eta}$ denote a truncation vector with $\eta_i \leq \rr_i$ for all $1 \leq i \leq N$.  Let 
\be\label{defhN}
h_N[\XN,\mu] := \g* \( \sum_{i=1}^N \delta_{x_i}- N \mu\)\ee (which will most often be abbreviated as $h_N$).
Then, with the notation \eqref{truncated potential}, we have 
\begin{equation}\label{spelloutFN}
\F(\XN,\mu)=\frac{1}{2\cds}\left(\int_{\R^{\d+1}}|y|^\gamma|\nabla h_{N,\vec{\eta}}|^2-\cds \sum_{i=1}^N\g(\eta_i)\right)-N\sum_{i=1}^N \int_{\R^\d}\f_\eta(x-x_i)\, d\mu(x).
\end{equation}
\end{lem}
We will very often use the truncated potential at distances $\rr_i$, then simply denoted $h_{N,\rr}$.

We can then define the local energy as announced in the introduction.
\begin{defi}[Local energy]
If $\carr_\ell$ is some cube of sidelength $\ell$ included in $\R^\d$, we let 
\begin{equation}
\tilde \F_N^{\carr_\ell}(\XN,\mu)= \int_{\square_\ell\times [-\ell, \ell]} \yg|\nab h_{N, \rr}|^2 
\end{equation}
where $h_N$ is defined as in \eqref{defhN}.
\end{defi}
The following provides   minimal distance controls. 
%\begin{equation}
%\lambda= (N\|\mu\|_{L^\infty})^{-\frac1\d}.\end{equation} 
\begin{prop}\label{prop:MElb}
Assume $\s\in [(\d-2)_+, \d)$.
Let $\mu \in L^1(\R^\d)\cap L^\infty(\R^\d)$ with $\int_{\R^\d}\mu=1$, satisfying \eqref{electric assumption}, and let  $\ux_N \in (\R^\d)^N$ be a pairwise distinct configuration.

Let $\Omega \subset \R^\d$.
For any $\vec{\eta}$ satisfying $ \frac{1}{2} \rr_i \le \eta_i \le \rr_i$ for every $1\leq i\leq N$, it holds that 
\begin{multline}\label{eq:13}
 \frac{1}{\cd}\int_{\Omega \times [-N^{-1/\d}, N^{-1/\d}]}  \yg|\nab h_{N,\vec{\eta}}|^2 
 %+\left(\frac{\#I_\Omega  \log N }{\d}\right) \indic_{\s=0} +C \#I_\Omega  N^{\frac{\s}{\d}}  
\ge 
\begin{cases} 
\displaystyle \frac1{ C}\sum_{i: x_i\in \Omega, \dist(x_i, \pa \Omega) \ge \frac14 N^{-1/\d}}\g(\eta_i) & \text{if} \ \s\neq 0\\
\displaystyle \sum_{i: x_i\in  \Omega , \dist(x_i, \pa \Omega) \ge \frac14N^{-1/\d}}\g(40 \eta_i N^{1/\d} ) & \text{if} \ \s=0,\end{cases}
\end{multline}
where $C>0$ depends only on $\d$, $\s$ and $\|\mu\|_{L^\infty}$.

If $\Omega=\R^\d$, then  we also have that 
\begin{equation}\label{eq:pr1}
2 \left(\F_N(\ux_N, \mu)+ \frac{ N  \log N  }{{2} \d}  
\indic_{\s=0}\right) + C N^{1+\frac{\s}{\d}} \ge \begin{cases} 
\displaystyle \frac1{ C}\sum_{i=1}^N\g(\eta_i) & \text{if} \ \s\neq 0\\
\displaystyle \sum_{i=1}^N\g(\eta_i N^{1/\d} ) & \text{if} \ \s=0,\end{cases}
\end{equation}
and
\begin{equation}\label{eq:pr2}
\int_{\R^{\d+1}} |y|^\gamma |\nabla h_{N,\vec{\eta}}|^2 \leq C\(\F_N(\ux_N, \mu) +\frac{ N \log N }{2 \d} \indic_{\s=0}\) + C  N^{1+\frac{\s}{\d}} ,
\end{equation}
where $C>0$ depends only on $\d$, $\s$ and $\|\mu\|_{L^\infty}$.

\end{prop}
Note that it follows from \eqref{eq:pr2} (if $\s< 0$ there is nothing to prove) that 
\be \label{energylowerbound}
\F_N(\ux_N, \mu)+ \frac{ N  \log N  }{{2} \d}  \indic_{\s=0}\ge - C N^{1+\frac\s\d},
\ee
where $C>0$ depends only on $\d,\s$ and $\|\mu\|_{L^\infty}$, which proves that $\F_N$ is uniformly bounded below.

%We first want to define \textit{local energies} associated to a potential $w$ solving
%\begin{equation}\label{Riesz viable potential}
%-\div(|y|^\gamma \nabla w)=\cds \left(\sum_{i=1}^N \delta_{(x_i,0)}-\mu\delta_{\R^\d}\right),
%\end{equation}
%in $\Omega \times [-h,h]:=\Omega_h$. In the proof of the local laws we will take $h \sim \ell$, but when we reexamine fluctuations to prove a Central Limit Theorem we will want to optimize over different values of $h$.

%We will need a minimal distance that respects $\Omega$, which we will again appeal to below.
%\begin{equation}\label{defrrc4}
%\rr_i := \frac{1}{4}\begin{cases} 
%\min \( \min_{x_j\in \Omega, j\neq i} |x_i-x_j|, 1\) & \text{if} \ \dist(x_i, \partial \Omega) \ge \hal  \\ 1  & \text{if } \dist(x_i, \pa \Omega) \le \frac14\\
%t \min \( \min_{x_j\in \Omega, j\neq i} |x_i-x_j|, 1\)
%+(1-t) & \text{if } \dist(x_i, \pa  \Omega)= \frac{1+t}{4} 
%, t\in [0,1].\end{cases}
%\end{equation}
%Observe that if $x_i\in\Omega'\subset \Omega$, $\rr_i(\Omega') \geq \rr_i(\Omega)$. With this minimal distance we define the \textit{true local energy} $\F^{\Omega_h}$ by
%\begin{equation}\label{Riesz true local energy}
%\F^{\Omega_h}(\XN,\mu)=\frac{1}{2\cds}\left(\int_{\Omega_h}|y|^\gamma|\nabla h_{N,\rr}|^2-\cds \sum_{x_i \in \Omega} \g(\rr_i)\right)-\sum_{x_i \in \Omega} \int_{\R^\d}\f_{\rr_i}(x-x_i)~d\mu(x),
%\end{equation}
%with $u$ the electrostatic potential solving \eqref{true Riesz potential}) with $\mu=\muv'$. 
We note that \eqref{eq:pr2} implies that the macroscopic law \eqref{macrolaw1} can be expressed as a control of the form 
\be\label{macrolaw2}
\int_{\R^{\d+1}}\yg |\nab h_{N,\rr}|^2 + N \sum_{i=1}^N \zeta_V(x_i) \lesssim_\beta N^{1+\frac\s\d}\ee
except with probability $\le e^{-C \beta N}$, where $C>0$ depends only on $\d,\s, \zeta_V$ and $\|\muv\|_{L^\infty}$.

The following is obtained  by combining  Proposition 4.28, Lemma 4.20, Lemma 4.25 and Lemma 4.26 in \cite{S24} (taking $N^{-1/\d}$ for the value of $\lambda$ there since we do not track the $\|\mu\|_{L^\infty}$ dependence).

\begin{prop}[Control of fluctuations, discrepancies and minimal distances]\label{pro:controlfluct}
 Assume $\varphi$ is a function such that $\Omega \subset \R^\d$ contains a $ 2N^{-1/\d}$-neighborhood of its support.
For any configuration $\XN(\in \R^\d)^N$, let  $I_\Omega$ denote $\{i, x_i \in \Omega\}$ and $\#I_\Omega$ its cardinality. 
For any $\eta \ge N^{-1/\d}$, we have
  \begin{multline}\label{rieszfluct2} 
\left|\int_{\R^\d} \varphi\( \sum_{i=1}^N \delta_{x_i}- N d\mu\) \right|  
\\
\le C \( \eta^{\gamma-1} \|\varphi\|_{L^2(\Omega)}^2 + \eta^{\gamma+1}\|\nab \varphi\|_{L^2(\Omega)}^2 \)^\hal  
  \(\int_{ \Omega\times [-2\eta,2\eta]} \yg |\nab h_{N,\rr}|^2\)^{\hal} + C  \# I_\Omega |\varphi|_{C^1} N^{-\frac{1}{\d}} ,
  \end{multline}
where $C>0$ depends only on $\d$ and $\s$ and $\|\muv\|_{L^\infty}$.

If $B_R $ is some  ball  of radius $R>2 N^{-1/\d}$, letting $D(B_R)= \int_{B_R}d\( \sum_{i=1}^N \delta_{x_i}-N\muv\)$, we have either 
$|D(B_R)|\le CN^{1-\frac1\d}R^{\d-1}$ or 
\be \label{discest}
\frac{D(B_R)^2}{R^\s}\left| \min \( 1, \frac{D(B_R)}{R^\d}\) \right|\le C \int_{B_{2R}\times[-2R,2R] }\yg|\nab h_{N,\rr}|^2 [\XN, \mu]
\ee 
where $C>0$ depends on $\d,\s$ and $\|\mu\|_{L^\infty}$.

If $\Omega$ is a general set of finite perimeter and $\Omega_\delta$ its $\delta$-neighborhood, if  $D(\Omega) \ge 0$, for any  $\|\mu\|_{L^\infty}^{-1/\d} <\delta \le N^{1/\d}$,
  \be   \label{disc2buriesz}\Big( D(\Omega)- \|\mu\|_{L^\infty} |\Omega_\delta\backslash \Omega| \Big)_+^2\le C
  \( \frac{|\Omega_\delta|}{|\p \Omega_\delta|}\)^\gamma \frac{|\Omega_\delta|}{\delta} 
  \int_{ (\Omega_\delta \backslash \Omega) \times  -\frac{|\Omega_\delta|}{|\partial \Omega_\delta|}-\delta ,\frac{|\Omega_\delta|}{|\partial \Omega_\delta|}+\delta]}\yg |\nab h_{N,\rr}|^2 \ee
If  $D(\Omega) \le 0$, for any $-N^{1/\d} \le \delta < -\|\mu\|_{L^\infty}^{-1/\d} $,
  \be \label{disc1buriesz}
  \Big( D(\Omega)+\|\mu\|_{L^\infty} |\Omega\backslash \Omega_\delta|\Big)_-^2\le  C
  \( \frac{|\Omega|}{|\p \Omega|}\)^\gamma \frac{|\Omega|}{\delta} 
   \int_{ ( \Omega\backslash \Omega_\delta)\times [ -\frac{|\Omega|}{|\partial \Omega|}-|\delta| ,\frac{|\Omega|}{|\partial \Omega|}+|\delta|]}\yg |\nab h_{N, \rr}|^2   .\ee
   \end{prop}
%\cm{insert or not the control on $\sum\g(\rr_i)$?}
Thus, thanks to the control \eqref{loiloc}, which provides a control on $\#I_\Omega$, we easily deduce controls on the quantities in \eqref{rieszfluct2} and \eqref{discest} except with small probability.

We will also need the following simple control on fluctuations.
\begin{lem}\label{RoughFluct}
Assume that  $\varphi \in C^1(\R^\d)$ is  compactly supported in some cube $\carr_r\subset \bulk.$  Let $\XN$ be a configuration such that the local laws of Theorem \ref{Local Law} hold on $\carr_{2r} \supset \carr_r$. Then,
\begin{equation}\label{Linf}
\left|\int_{\carr_r}\varphi(x)\, d\fluct_{\muv}(x)\right|\lesssim_\beta  Nr^\d\|\varphi\|_{L^\infty(\carr_r)}.
\end{equation}
%\cm{insert $\beta$ dependence}
%and 
%\cm{maybe suppress}
%\begin{equation}\label{Lone}
%\left|\int_{Q_r} \varphi(x)~d\fluct_{\muv}(x)\right|\lesssim_{\mathfrak{c}}N\int_{Q_r}|\varphi(x)|\, dx
%\end{equation}
%for all $\XN \in \mathcal{G}_{r}$, where $\mathcal{G}_{r}$ is an event on which the macroscopic energy control $\FN(\XN,\muv) \lesssim N^{1+\frac{\s}{\d}}$ holds and 
%\begin{equation*}
%\PNbeta(\mathcal{G}_{r}^c)\leq C_1e^{-C_2\beta r^\d N}
%\end{equation*}
%for some constants $C_1$ and $C_2$ dependent only on $\muv$.
%\cm{but for the first thing we also need the mesoscale control}
\end{lem}
\begin{proof}
We may bound 
\begin{equation*}
\left|\int_{\carr_r}\varphi(x)\, d\fluct_{\muv}(x)\right|\le \|\varphi\|_{L^\infty(\carr_r)} (\# I_{\carr_R} + N |\carr_r| \|\muv\|_{L^\infty} ) .
\end{equation*}
The local laws on $\carr_{2r}$  provide the desired bound.

%For the second item, we split the bulk into roughly $\frac{1}{\mathfrak{c}^\d}$ cubes of side length $\mathfrak{c}$. Using the macroscopic energy control coupled with the discrepancy control as mentioned in the previous line, we can control the number of points in each cube by $\lesssim N \mathfrak{c}^\d$. This allows us to use a Riemann sum approximation of the integral and obtain \eqref{Lone}. \cm{here I don't understand and I'm not sure it's true. There should be an error involving the $C^1$ norm of $\varphi$ or a modulus of continuity of $\varphi$} The probability of the event on which this holds is a result of Lemma \ref{macro nrg}, which is \cite[Corollary 5.23]{S24} and gives $\FN \lesssim N^{1+\frac{\s}{\d}}$ on an event of probability $\lesssim e^{-C\beta N}$ in the constant $\beta$ regime.
\end{proof}

\subsection{Transport calculus and commutator estimates}
We may now recast the change of variables made in the proof of Lemma \ref{lem2.1} at this next-order level as follows. This is a ``post-splitting'' version of Lemma \ref{lem2.1}. 
%We also insert an event $\mathcal G$.

\begin{lem} Let $\Phi_t=\id + t\psi$ with $\psi$ as in Proposition \ref{transport}, and $\mu_t=\Phi_t\#\muv$. Let $\mathcal{G}$ be an event.
We have
\begin{align}\label{09}
 &\log\frac{\K_{N,\beta}^{\mathcal{G}}(\mu_t,\zeta_V\circ \Phi_t^{-1})}{\K_{N,\beta}(\muv,\zeta_V)} +N (\Ent(\mu_t)- \Ent(\muv))\\
 \notag
 &=\log \Esp_{\PNbeta} \(  \exp\(-\beta N^{-\frac\s\d}\(\F_N(\Phi_t(\XN),\Phi_t\#\muv)-\F_N(\XN, \muv)\)+\Fluct_{\muv}( \log \det D\Phi_t)  \)\indic_{\mathcal G}\)\\   \notag
 &=\log\Esp_{\PNbeta}(e^{T_2}\indic_{\mathcal G} ),\end{align}
 where $T_2$ is as in \eqref{t2}.
\end{lem}
\begin{proof} By the change of variables $y_i=\Phi_t(x_i)$ we have
 \begin{align*}
& \frac{\K_{N,\beta}^{\mathcal{G}}(\mu_t,\zeta\circ \Phi_t^{-1})}{\K_{N,\beta}(\muv,\zeta_V)} 
\\ &=
  \frac{1}{ \K_{N,\beta}(\muv, \zeta_V)} \int_{(\R^\d)^N } \exp\( -\beta N^{-\frac\s\d} \F_N(\Phi_t(\XN), \Phi_t \# \muv) + N\sum_{i=1}^N \zeta_V(x_i) \) \prod_{i=1}^N\det D\Phi_t (x_i)\indic_{\mathcal G}~ dX_N\\
&  = \Esp_{\PNbeta} \( \exp\(- \beta N^{-\frac\s\d} \( \F_N(\Phi_t(\XN), \Phi_t \# \muv) -\F_N(\XN, \muv) \) 
   +\sum_{i=1}^N \log \det D\Phi_t (x_i)\)\)\indic_{\mathcal G}~ dX_N,\end{align*}
   where we used \eqref{splitgibbs}.
Since $\mu_t=\Phi_t\#\mu_0$ we have  $\det D\Phi_t= \frac{\mu_0}{\mu_t\circ \Phi_t}$, and thus
\be\int \log \det D\Phi_t d\mu_0= \int \log \mu_0 d\mu_0- \int \log \mu_t(\Phi_t(x)) d\mu_t= \Ent(\mu_0)-\Ent(\mu_t). \ee
The result follows by  \eqref{t2}.
\end{proof}

Wishing to linearize these expressions as $t\to 0$, we are led to considering successive derivatives of $\F_N$ along a general transport. More precisely, we denote for any $\psi$, 
\be\label{defA1} \Ani_1(\XN , \mu, \psi):= \hal \int_{\triangle^c} \nab \g(x-y)\cdot (\psi(x)-\psi(y)) d\( \sum_{i=1}^N \delta_{x_i}- N\mu\)(x)\( \sum_{i=1}^N \delta_{x_i}- N\mu\)(y)
\ee
and more generally
\be\label{defAn} \Ani_n(\XN, \mu, \psi):=\hal  \int_{\triangle^c} \nab^{\otimes n} \g(x-y): (\psi(x)-\psi(y))^{\otimes n}  d\( \sum_{i=1}^N \delta_{x_i}- N\mu\)(x)\( \sum_{i=1}^N \delta_{x_i}- N\mu\)(y).
.\ee
It is easy to check \cite[Lemma 4.1]{S22}
that, if $\Phi_t = \id + t\psi$, we have 
\be \label{dtF}
\frac{d^n}{dt^n} \F_N(\Phi_t(\XN), \Phi_t\#\mu)= \Ani_n (\Phi_t(\XN), \Phi_t\#\mu, \psi\circ \Phi_t^{-1}) . \ee

To control such terms, we will need the  following recent sharp and localized  commutator estimates of all order from \cite{RS22}. \footnote{They are slightly restated, because up to allowing constants that depend on $\|\mu\|_{L^\infty}$ we can work with $\lambda= N^{-1/\d}$ there.}

\begin{prop}[\cite{RS22}]\label{pro:commutator}Let $\mu \in L^1(\R^\d)\cap L^\infty(\R^\d)$ with $\int_{\R^\d}\mu\,  dx=1$ satisfying \eqref{electric assumption}.  There exists a constant $C>0$ depending only $\d$, $\s$  and $\|\mu\|_{L^\infty}$ such that the following holds.  Let $\psi$ be  a Lipschitz vector field $\psi:\R^\d\rightarrow\R^\d$  and $\Omega$ be a closed set  containing a $2N^{-1/\d}$-neighborhood of $\supp \psi$.  For any pairwise distinct configuration $\XN \in (\R^\d)^N$, it holds that
\begin{multline}\label{main1}
\left|\int_{(\R^\d)^2\setminus\triangle}(\psi(x)-\psi(y))\cdot\nabla\g(x-y)d\Big(\sum_{i=1}^N\delta_{x_i} -N \mu\Big)^{\otimes 2}(x,y)\right| \\
\leq C\|\nabla \psi\|_{L^\infty}\Bigg( \int_{\Omega\times [-\ell, \ell]} \yg|\nab h_{N,\rr}|^2 + {C\# I_\Omega N^{\frac{\s}{\d}} }\Bigg).
\end{multline}
 Suppose in addition  that $\nabla^{\otimes(n-1)}\psi$ is Lipschitz and that  $\Omega$ contains a $(5\ell+N^{-1/\d})$-neighborhood of $\supp \psi$, where $\ell$ satisfies $\ell>2  N^{-1/\d}$. For any $n \ge  2$, we have 
\begin{align} \label{mainn}
& \left|\int_{(\R^\d)^2\setminus\triangle}\nabla^{\otimes n}\g(x-y) : (\psi(x)-\psi(y))^{\otimes n} d\Big(\sum_{i=1}^N\delta_{x_i} -N \mu\Big)^{\otimes 2}(x,y)\right| \\
\notag & \qquad \le 
C\sum_{p=0}^n (\ell\|\nab^{\otimes 2}\psi\|_{L^\infty})^p
 \sum_{\substack{1\leq c_1,\ldots,c_{n-p} \\ n-p\le c_1+\cdots+c_{n-p} \le 2n}} N^{\frac{1}\d ((n-p)-\sum_{k=1}^{p} c_{n-k} )} \|\nabla^{\otimes c_1} \psi\|_{L^\infty}\cdots\|\nabla^{\otimes c_{n-p}} \psi\|_{L^\infty}  \\ \notag
 &\qquad \qquad \qquad\qquad\qquad \qquad \qquad\qquad \times\Bigg( \int_{\Omega\times [-\ell,\ell]} \yg|\nab h_{N,\rr}|^2
+ C  \#I_\Omega N^{\frac\s\d}   \Bigg).\end{align}
\end{prop}

\subsection{Variation of energy along a transport}
In the macroscopic case $\ell=1$, taking $\Omega= \R^\d$, we can immediately control the energy of the transported configuration in terms of the initial configuration: let $$\Xi(t):= \F_N(\Phi_t(\XN), \Phi_t\# \mu_0) + \frac{N\log N}{2\d} \indic_{\s=0}+ C N^{1+\frac\s\d}, $$ for $C$ large enough; in view of \eqref{dtF}, 
applying the first order commutator estimate \eqref{main1} and combining it with \eqref{eq:pr2}, 
we find that 
\be \Xi'(t)\le C\|\nab \psi\|_{L^\infty} \Xi(t),\ee
and applying Gronwall's lemma yields
\be \Xi(t) \le \exp(C t\|\nab\psi\|_{L^\infty}) \Xi(0),\ee
which gives the desired control. 

The mesoscopic case $\ell <1$  is much more delicate, and presents additional difficulties compared to the Coulomb case, due to the nonlocalized nature of the transport. We address this difficulty by considering the transport on increasing dyadic scales and leveraging the decay of the transport away from $\supp \varphi$.

Let as above  $\carr_\ell$ be a cube of size $\ell$ such that  $\supp \varphi\subset \carr_\ell \subset \carr_{2\ell}\subset \bulk$ and assume without loss of generality, that it is centered at the origin. Let $k^*$ be such that 
\be\label{defkstar}
U \subset \carr_{2^{k*}\ell},\ee where $U$ is a neighborhood of $\Sigma$ as in Proposition~\ref{transport}.
For $k $  in $[0,k_*]$, let
\be\label{defAk} D_k= \carr_{2^k \ell}, \quad A_k= D_{k+4}\backslash D_{k-2}, \quad A_{k_*+1}= D_{k_*}^c.\ee
Finally, denote 
$$\rho=2^k\ell$$
where the dependence in $k$ is implicit.

First we prove a preliminary lemma, which is a weighted trace inequality.
\begin{lem}\label{trace}
Let $h$ be a function in $A_k\times [-\lambda\rho, \lambda \rho]$ such that $\int_{A_k\times [-\lambda\rho,\lambda \rho]} \yg |\nab h|^2<\infty$. 
Let $\bar h=\dashint_{A_k \times [-\lambda \rho,\lambda\rho]} h$. Then 
\be \label{traceth}
\int_{\pa(D_k\times [-\lambda \rho, \lambda\rho])} \yg |h-\bar h|^2\le C\rho \int_{A_k\times [-\lambda\rho, \lambda \rho]} \yg |\nab h|^2,\ee where $C>0$ depends only on $\s $, $\d$ and $\lambda>0$.
The same holds with $A_k$ replaced by $D_k$.
\end{lem}
\begin{proof}Let us denote by $
\tilde h_z=\dashint_{A_k\times \{z\}} h.$
Since $\bar h= \dashint_{A_k\times [-\lambda\rho,\lambda\rho]}= \tilde h_{z_0}$ for some $z_0\in [-\lambda \rho, \lambda \rho]$ by the mean value property, we may write, using Cauchy-Schwarz, assuming without loss of generality that $z_0\le z$, 
\begin{multline}\label{thhp}
|\tilde h_{z}-\bar h|^2\le \frac{C}{\rho^{2\d}} \( \int_{z_0}^{z}\partial_y\( \int_{A_k \times \{y\}} h (x,y) dx\) dy\)^2
\le \frac{C}{\rho^{2\d}} \(\int_{A_k \times [z,z_0] }\partial_y h\)^2\\
\le\frac{C}{\rho^{2\d} } \int_{A_k\times [-\lambda\rho,\lambda \rho]} \yg |\nab h|^2 \int_{A_k\times [-\lambda \rho,\lambda \rho]}  \frac{1}{\yg} \le C \rho^{1-\gamma-\d} \int_{A_k\times [-\lambda\rho,\lambda\rho]} \yg |\nab h|^2  , \end{multline} where $C>0$  depends  on $\d, \s$ and $\lambda$. 
It follows that, letting $\bar h_+=\dashint_{A_k\times [\lambda \rho/2, \lambda \rho]} h$, $|\bar h_+ -\bar h|^2$ is controlled in the same way.
Thus, using the standard trace theorem in $A_k\times  [\lambda \rho/2, \lambda \rho] $ and the triangle inequality, we deduce that 
\begin{multline}\int_{A_k\times \{\lambda \rho\}} \yg |h-\bar h|^2 \le
2\int_{ A_k\times \{\lambda\rho\}} \yg |h-\bar h_+|^2 +2\int_{A_k\times \{\lambda\rho\}} \yg |\bar h_+-\bar h|^2 
\\ \le
 C  \rho \int_{A_k \times [\lambda \rho/2,\lambda \rho]} \yg |\nab h|^2 + C\rho^{\gamma +\d+1-\gamma-\d}\int_{A_k\times [-\lambda \rho,\lambda\rho]} \yg |\nab h|^2\\
 \le C  \rho \int_{A_k \times [-\lambda \rho,\lambda \rho]} \yg |\nab h|^2 .\end{multline}
The same relation holds on $A_k\times \{-\lambda \rho\}$.
We next turn to the integral on  $\pa A_k \times [-\lambda \rho,\lambda \rho]$. By standard trace theorem, we have
$$\int_{\pa A_k \times \{y\}} \yg |h-\tilde h_y|^2 \le C\rho\int_{ A_k \times \{y\}}\yg |\nab h|^2,$$
and using \eqref{thhp} and the triangle inequality, we deduce 
$$\int_{\pa A_k \times [-\lambda \rho,\lambda \rho]} \yg |h-\bar h|^2 \le C\rho\int_{ A_k \times [-\lambda \rho,\lambda \rho]}\yg |\nab h|^2.$$
The result is proven for $A_k$. The proof is the same over $D_k$.
\end{proof}

\begin{lem}\label{corocontrenergyt2}
Assume $\mu$ is a bounded probability density satisfying \eqref{electric assumption}.    %Let  $\psi_t$, $t \in [0,1]$ be a Lipschitz vector field, and $\Phi_t$ solve \eqref{defflow}, and $\mu_t$ be as in \eqref{defmutflow}. 
Let   $\Omega$ be  $A_k$ for some $0\le k\le  k_*$. Let $\tilde \psi$ be a map supported in  $\Omega$ or in  $\Omega^c$ and  let $\tilde \Phi_t:= \id + t\tilde \psi$.  Assume that  $\partial \Omega$ is at distance $\ge \hal  \rho= 2^{k-1} \ell$ from $\supp \tilde \psi$, and $|\tau|\|\tilde \psi\|_{L^\infty} < \min (\frac14 \rho,\hal)$. 
 Then letting $\eta_i$ be the minimum over $t\in [0,\tau]$ of the $\rr_i$'s of the  $\tilde\Phi_t(\XN)$, we have
 \begin{multline} \label{controlenergyt2}  \forall t \in [0,\tau],  \quad 
\int_{\Omega\times [-\lambda \rho,\lambda \rho]}\yg|\nab h_{N,\vec{\eta}}[\tilde \Phi_t(\XN), \tilde \Phi_t\#\mu]|^2
\le C \exp\(Ct  (\rho^{-1}\|\tilde \psi\|_{L^\infty}+ |\tilde\psi|_{C^1})  \) \\ \times \(\int_{\Omega\times [-\lambda \rho,\lambda \rho]}\yg|\nab h_{N,\vec{\eta}}[\XN, \mu]|^2
+ \#I_\Omega N^{\frac\s\d}  +  N^{-\frac2\d}(\#I_\Omega )^2 \rho^{-\s-2}\) \end{multline}
where $C$ depends only on $\s, $  $\d$, $\lambda>0$  and $\|\mu\|_{L^\infty}$.

\end{lem}

\begin{proof} For shortcut, let us drop the tildes for the proof, and  denote $h_N^t = h_N[\tilde \Phi_t(\XN), \tilde \Phi_t\#\mu]$  as in \eqref{defhN}. 
Let us start with the case where $\psi$ is supported in $\Omega= A_k$.  We note that since $\supp \psi$ is at distance $\ge \hal\rho$ from $\partial \Omega$ and $|\tau|\|\psi\|_{L^\infty} <\frac14\rho$, $\Phi_t$ maps $\Omega$ to $\Omega$, and coincides with the identity in $\Omega^c$ and in the part  of $\Omega$ at distance $\le \frac14 \rho$ from $\partial \Omega$, for any $t\in [0,\tau]$. Denoting $\tilde \Omega=\{x\in \Omega, \dist(x, \partial \Omega) \ge \hal \rho\}$, we have that $\supp \psi\subset \tilde \Omega$ and $\supp\psi$ is even at distance $\ge \frac14 \rho\ge \frac14 \ell\ge \hal\rho_\beta N^{-1/\d}\ge 2N^{-1/\d}$ from $\partial \tilde \Omega$.

The commutator estimate provides us with an estimate on 
  $\frac{d}{dt} \F_N( \Phi_t(\XN), \mu_t)$, so we need to estimate the difference between that  and $\frac{d}{dt}\int_{\Omega\times [-\rho,\rho]}\yg |\nab h_{N,\rr}^t|^2$.

  Since $\Phi_t $ coincides with the identity map outside $\Omega$ for each $t \in [0,\tau]$, spelling out  the definition  \eqref{spelloutFN} (the equality case with $\eta_i=\rr_i$), we may write  that 
\begin{align}\notag  \F_N( \Phi_t(\XN), \mu_t)-
\F_N(\XN, \mu_0)
\notag &= 
\frac1{2\cds} \int_{\Omega\times[-\lambda \rho,\lambda\rho]} \yg|\nab h^t_{N,\veta}|^2 -\yg|\nab h^0_{N,\veta}|^2\\ \notag
& + \frac{1}{2\cds}\int_{\R^{\d+1}\setminus(\Omega\times [\lambda\rho,\lambda \rho])}\yg|\nab h_{N,\veta}^t|^2- \yg|\nab h_{N,\veta}^0|^2,\\ \label{alignemen}
 & - N\sum_{i\in I_\Omega} \int_{\R^\d}\f_{\eta_i}(x-\Phi_t(x_i)) d\mu_t(x) + N\sum_{i\in I_\Omega} \int_{\R^\d} \f_{\eta_i}(x-x_i) d\mu_0(x),
\end{align}
where we note that   $\rr_i$ for $\XN$  and $\eta_i$  are within a factor  in $[\hal, 2]$ of each other. 

Since  points at distance $\le \frac14\rho$ from $\pa \Omega$  and points in $\Omega^c$ are fixed points of $\Phi_t$, and since $\frac14\rho \ge \frac14 \ell\ge \frac14 \rho_\beta N^{-1,\d}$ and $\rr_i\le \frac 14 N^{-1/\d}$, since $\rho_\beta\ge 4$, the smeared charges $\delta_{x_i}^{(\eta_i)}$ coincide in $\Omega^c$, and thus 
the function  $h^t_{N,\veta}-h^0_{N,\veta}:=u^t$ solves $$-\div (\yg \nab u^t)=0 \quad \text{in} \ \R^{\d+1}\setminus(\Omega\times [-\lambda \rho, \lambda \rho]),$$ and decays at infinity, and its gradient as well. Hence, writing $\nab h^t_{N,\veta}= \nab h^0_{N,\veta}+\nabla u^t$ and  integrating by parts, we obtain
\begin{align}\notag
&\int_{\R^{\d+1}\setminus(\Omega\times [-\lambda \rho, \lambda \rho])} 
\yg|\nab h_{N,\veta}^t|^2- \yg|\nab h_{N,\veta}^0|^2=
\int_{\R^{\d+1}\setminus(\Omega\times [-\lambda \rho, \lambda \rho])}\yg |\nab u^t|^2 \\ 
\label{intomegc}
&+ 2\int_{\partial (\Omega\times [-\lambda \rho, \lambda \rho])} \yg ( h_{N,\veta}^0 -\bar h_{N,\veta}^0)\frac{\partial  u^t}{\partial n}\end{align}
where $\vec{n}$ is the inner pointing  unit normal to $\pa \Omega$ and $\bar h_{N,\veta}^0$ is the weighted average of $h_{N,\veta}^0$ on $\pa(\Omega \times [-\lambda \rho, \lambda \rho])$.  By Cauchy-Schwarz and Lemma \ref{trace} we may write  
\begin{multline} \label{inthntu}\left|\int_{\partial (\Omega\times [-\lambda \rho, \lambda \rho])} \yg( h_{N,\veta}^0 -\bar h_{N,\veta}^0)\frac{\partial  u^t}{\partial n}\right|
\le 
C\rho^{\hal} \|   \nab h_{N,\veta}^0\|_{L^2_{\yg}(\Omega\times [-\lambda \rho, \lambda \rho])} \|\nab u^t\|_{L^2_{\yg}(\pa (\Omega\times [-\lambda \rho, \lambda \rho]))}
.\end{multline}
We next estimate $\nab u^t$ in the complement of $\Omega \times [-\lambda \rho, \lambda \rho]$.
By definition, 
$$u^t(x)= \int_{\R^\d} \g(x-x') d\(\sum_{i=1}^N \delta_{\Phi_t(x_i)}^{(\eta_i)}-N\mu_t\)(x')- \int_{\R^\d} \g(x-x') d\(\sum_{i=1}^N \delta_{x_i}^{(\eta_i)}-N\mu_0\)(x').$$
We compute that 
\begin{align*}
 \partial_t u^t(x)&= \partial_t \(\sum_{i=1}^N\dashint_{\pa B(0, \eta_i)} \g(x-\Phi_t(x_i)-\cdot)-N\int_{\R^\d} \g(x-\Phi_t(y)) d\mu_0(y)\)\\
& = \sum_{i=1}^N\int_{\R^\d}\nab \g(x-y) \cdot \psi_t(\Phi_t(x_i))\delta_{\Phi_t(x_i)}^{(\eta_i)}(x')- N\int_{\R^\d}\nab \g(x-x') \cdot \psi_t(x') d\mu_t(x') 
\end{align*}
where we have denoted $\psi_t:=\frac{d\Phi_t}{dt}$. Hence 
\begin{align}\notag u^t(x)&=\int_0^t  \int_{\R^\d}\nab  \g(x-x')\cdot \psi_s(x') d\(\sum_{i=1}^N \delta_{\Phi_s(x_i)}^{(\eta_i)}-N\mu_s\)(x')\, ds \\ &+\int_0^t \sum_{i=1}^N\int_{\R^\d} \nab \g(x-x') \cdot \(\psi_s(\Phi_s(x_i))- \psi_s(x')\) \delta_{\Phi_t(x_i)}^{(\eta_i)}(x')\, ds.
\end{align}
and 
\begin{align}\notag \nab u^t(x)&=-\frac1{\cds}\int_0^t  \int_{\R^\d}\nab^{\otimes 2}  \g(x-x'): \psi_s(x') \div(\yg \nab h_N^s)(x')\, ds \\ \label{nabutformula}&+\int_0^t \sum_{i=1}^N\int_{\R^\d\times \R^1} \nab^{\otimes 2} \g(x-y) : \(\psi_s(\Phi_s(x_i))- \psi_s(x')\) \delta_{\Phi_t(x_i)}^{(\eta_i)}(x')\, ds.
\end{align}
The function $\chi_x(x'):= \nab^{\otimes 2}  \g(x-x'): \psi_s(x')$ is compactly supported where $\psi_s$ is, moreover, it can be checked to satisfy 
\be |\nab \chi_x(x')|\lesssim \(\frac{\|\psi_s\|_{L^\infty(\Omega)}}{\dist(x, \supp \psi_s)^{\s+3}}+ \frac{|\psi_s|_{C^1(\Omega)}}{\dist(x, \supp \psi_s)^{\s+2}}\).\ee
We extend the function $\chi_x$ to $\Omega \times [-\lambda \rho, \lambda \rho]$ by multiplying $\chi_x$ by a function $\varphi(y)$ supported in $[-\lambda \rho, \lambda \rho]$ with $\|\varphi\|_{L^\infty}\leq 1$ and $\left\|\frac{d}{dy}\varphi\right\|_{L^\infty}\lesssim \frac{1}{\rho}$. Then,  for $x \in \Omega^c$, 
\be \int_{\Omega\times [-\lambda \rho, \lambda \rho]} \yg|\nab \chi_x|^2  \lesssim (\|\psi_s\|_{L^\infty(\Omega)}+\rho |\psi_s|_{C^1(\Omega)})^2 \rho^{-\s-4}\ee
using  $\d+\gamma=\s+1$.
Using Green's formula and the Cauchy-Schwarz inequality to control the first term in \eqref{nabutformula}, and a rough bound for the second term, we are led to 
\begin{multline*} 
|\nab u^t(x)|^2 \lesssim\\
t \int_0^t   (\|\psi_s\|_{L^\infty}+\rho |\psi_s|_{C^1})^2 \rho^{-\s-4}  \int_{\Omega\times [-\lambda \rho, \lambda \rho]}\yg |\nab h^s_{N,\veta}|^2 ds + t \int_0^t |\psi_s|_{C^1}^2\( \sum_{i\in I_\Omega} \eta_i \rho^{-\s-2}\)^2 ds.
\end{multline*}
Since this is true for all $x\in (\Omega \times [-\lambda \rho, \lambda \rho])^c$, 
inserting  $\eta_i\le  N^{-1/\d}$, it  follows that 
\begin{multline*} \int_{\pa (\Omega\times [-\lambda \rho, \lambda \rho])} \yg|\nab u^t|^2 \lesssim \rho^{\d+\gamma}
\\   \times  t \int_0^t (\|\psi_s\|_{L^\infty}+\rho |\psi_s|_{C^1})^2\( \rho^{-\s-4}\int_{\Omega\times [-\lambda \rho, \lambda \rho]} \yg|\nab h^s_{N,\veta}|^2 +  N^{-\frac2\d}(\#I_\Omega )^2 \rho^{-2\s-6}\) ds\end{multline*}
Combining with \eqref{intomegc} and \eqref{inthntu}, and noting that $\int_{\R^{\d+1}\setminus(\Omega\times [-\lambda \rho, \lambda \rho])}\yg |\nab u^t|^2$ is an order $O(t^2)$ term,  it follows that 
\begin{align*}
&\left|\frac{d}{dt}_{|t=0} \int_{\R^{\d+1}\setminus(\Omega\times[-\lambda \rho, \lambda \rho])}\yg 
|\nab h_{N,\veta}^t|^2\right|
\lesssim \rho^{\hal} \|\nab h_{N,\vec{\eta}}^0\|_{L^2_{\yg}(\Omega\times [-\lambda \rho, \lambda \rho]))} \rho^{\frac{\s+1}{2}}   (\|\psi_0\|_{L^\infty}+\rho |\psi_0|_{C^1})\\
& \hspace{3cm}\times \( \rho^{-\frac{\s}2-2} \|\nab h^0_{N,\veta}\|_{L^2_{\yg}(\Omega\times [-\lambda \rho, \lambda \rho])} +  N^{-\frac1\d}(\#I_\Omega ) \rho^{-\s-3}\)
.\end{align*}
On the other hand, by definition  \eqref{extended truncation},  we have
\begin{align*}
&\left|\frac{d}{dt}_{|t=0} \sum_{i\in I_\Omega} \int_{\R^\d} \f_{\eta_i}(x-\Phi_t(x_i))d\mu_t(x)\right|
= \left|\sum_{i\in I_\Omega}  \int_{\R^\d} \nab\f_{\eta_i}(x-x_i)\cdot (\psi_0(x)-\psi_0(x_i)) d\mu_0(x)\right|
\\
&\le |\psi_0|_{C^1} \sum_{i\in I_\Omega} \int_{B(0,\eta_i)} |x|^{-\s}
\le C \#I_\Omega|\psi_0|_{C^1} N^{-\frac{\d-\s}{\d}}.
\end{align*}

Combining the above with \eqref{alignemen} and \eqref{intomegc}, and by definition \eqref{dtF} and the commutator  estimate in the form \eqref{main1}, after using Young's inequality, we are led to
\begin{align*}
&\left|\frac{d}{dt}_{|_{t=0} }\int_{\Omega\times [-\lambda \rho, \lambda \rho]}\yg|\nab h_{N,\vec{\eta}}^t|^2\right| \le C (\|\psi_0\|_{L^\infty}+\rho |\psi_0|_{C^1})\times
\\ &    \(\rho^{-1} \int_{\Omega\times [-\lambda \rho, \lambda \rho]}\yg|\nab h_{N,\vec{\eta}}^0|^2  +\rho^{-1} \int_{\tilde\Omega\times [-\lambda \rho, \lambda \rho]}\yg |\nab h_{N,\rr}^0|^2+  N^{-\frac2\d}(\#I_\Omega )^2 \rho^{-\s-3}+ \#I_\Omega \rho^{-1} N^{\frac\s\d}\)
 .\end{align*}
 We next wish to estimate $\int \yg |\nab h_{N,\rr}|^2$ in terms of $\int \yg |\nab h_{N,\vec{\eta}}|^2$.
 Using the definition \eqref{truncated potential}, if $\alpha_i $ and $\kappa_i$ are $\le \rr_i$, we have
 $$\nab h_{N,\vec{\alpha}}= \nab h_{N,\vec{\kappa}}+ \sum_{i=1}^N\nab ( \f_{\kappa_i}-\f_{\alpha_i})(x-x_i)$$
 hence, using that the $B(x_i,\rr_i)$ are disjoint, we find
 \begin{align}
 \notag \int_{ \Omega\times [-\lambda \rho, \lambda \rho]}\yg |\nab h_{N,\vec{\alpha}}|^2& \le 2 \int_{\Omega\times [-\lambda \rho, \lambda \rho]} \yg |\nab h_{N,\vec{\kappa}}|^2 +\sum_{i\in I_\Omega} \int_{\R^{\d+1}} \yg |\nab (\f_{\kappa_i}-\f_{\alpha_i})|^2 \\
 & \lesssim \int_{\Omega\times [-\lambda \rho, \lambda \rho]} \yg |\nab h_{N,\vec{\kappa}}|^2 + \sum_{i\in \Omega} (\kappa_i^{-\s}+\alpha_i^{-\s}).
 \end{align}
 Applying in $\tilde \Omega$  to $\alpha_i=\rr_i$ and $\kappa_i=\eta_i \ge \hal \rr_i$, we obtain 
  $$\int_{\tilde \Omega\times [-\lambda \rho, \lambda \rho]} \yg |\nab h_{N,\rr}^0|^2\lesssim  \int_{\tilde \Omega\times [-\lambda \rho, \lambda \rho]}\yg |\nab h_{N,\vec{\eta}}^0|^2+  \sum_{i\in I_{\tilde \Omega} }\rr_i^{-\s},$$
 and using \eqref{eq:13} and the definition of $\tilde \Omega$, we may absorb these terms into the others.

The same reasoning yields that the relation is true at any $t$, so we find 
\begin{align}\notag
 & \left|\frac{d}{dt}\int_{\Omega\times [-\lambda \rho, \lambda \rho]}\yg|\nab h_{N,\veta}^t|^2\right|\\ \label{dtOt}
 & 
 \le C (\rho^{-1}\|\psi_t\|_{L^\infty}+ |\psi_t|_{C^1}) \( \int_{\Omega\times [-\lambda \rho, \lambda \rho]}|\nab h_{N,\vec{\eta}}^t|^2+ \#I_\Omega N^{\frac\s\d}  +  N^{-\frac2\d}(\#I_\Omega )^2 \rho^{-\s-2}\).
 \end{align}
 Applying Gronwall's lemma,  we deduce the result \eqref{corocontrenergyt2} in the case $\psi$ supported in $\Omega$.
  
 Let us now turn to the case where the support of $\psi$ is in $\Omega^c$. With the same notation, we may write $h_{N,\vec{\eta}}^t-h_{N,\vec{\eta}}^0= u^t$ where $u^t$ solves $-\div (\yg \nab u^t)=0$ in  $\Omega \times \R$, and we may write 
 \begin{multline}\int_{\Omega\times [-\lambda \rho, \lambda \rho]} 
\yg|\nab h_{N,\veta}^t|^2- \yg|\nab h_{N,\veta}^0|^2\\=
\int_{\Omega\times [-\lambda \rho, \lambda \rho]}\yg |\nab u^t|^2 + 2\int_{\partial( \Omega\times [-\lambda \rho, \lambda \rho]} \yg ( h_{N,\veta}^0 -\bar h_{N,\veta}^0)\frac{\partial  u^t}{\partial n}\\
= O(t^2) + O \( \( \int_{\Omega\times [-\lambda \rho, \lambda \rho]} \yg|\nab h_{N,\vec{\eta}}|^2 \)^\hal \( \int_{\pa \Omega \times [-\lambda \rho, \lambda \rho]} \yg \left|\frac{\pa u^t}{\pa n}\right|^2\)^\hal\),
\end{multline}
where $\vec{n}$ is the outwards pointing  unit normal to $\pa \Omega$ and $\bar h_{N,\veta}^0$ is the weighted average of $h_{N,\veta}^0$ on $\pa\Omega$. The rest of the computation is identical to the first step and yields \eqref{dtOt}, and we finish the proof in the same way. We can then check that all the cases considered in the statement lemma have been treated.
\end{proof}

\begin{prop}\label{pf loiloc}
Let $\psi$ be as constructed in Proposition \ref{transport}, and let $\Phi_t=\id+t\psi$.
Assume that $\XN\in \mathcal G_\ell$, a set of configurations with $\PNbeta(\mathcal G_\ell^c) \le C_1 e^{-\beta C_2 \ell^\d N}$ such that the local laws \eqref{loiloc}--\eqref{controlnbreintro} hold in each $D_k$ of \eqref{defAk} with $k \geq 1$ (in $\bulk$) and that \eqref{macrolaw2} holds.  Then if $t$ is small enough that 
$|t| \ell^{\s-\d}\M $ is smaller than a constant depending only on  the constant in \eqref{tscale}, for any $k\in [1, k_+]$ integer and $D_k$ as in \eqref{defAk},  $k_*$ as in \eqref{defkstar}, we have 
\be \label{localizationGN}
\int_{(  D_{k+\frac32}\backslash D_{k+\hal} )\times [-2^k \ell,2^k\ell]} \yg |\nab h_N[\Phi_t(\XN), \Phi_t\#\muv]|^2 \le 
C(2^k \ell)^\d N^{1+\frac\s\d},\ee
and
\be\label{localizationGNext}
\int_{D_{3/2}\times [-\ell, \ell]} \yg |\nab h_N[\Phi_t(\XN), \Phi_t\#\muv]|^2 \le 
C \ell^\d N^{1+\frac\s\d},
\ee where $C>0$ depends only on $\s,\d,\|\muv\|_{L^\infty}$ and the prior constants.
\end{prop}

\begin{proof}
Let, for $k$ integer, 
\be \label{defBk}
B_k= D_{k+2}\backslash D_{k} \quad  \text{for}\,  k \le k_*, \quad B_{k_*+1}= D_{k_*+1}^c.\ee
Let  $\chi_k$ be a partition of unity associated to $\{B_k\}_{k=0}^{k_*+1}$, such that $\supp \chi_k \subset B_k$,  $\sum_{k=0}^{k_*+1}\chi_k=1$ and $|\nab \chi_k|\lesssim 2^{-k}\ell^{-1}$.

Let $k\in [0, k_*]$ be a given integer.  
We are going to decompose $\id + t\psi$ as a composition of transports $\id + t\psi_j$ such that each $\psi_j$ has support in $B_j$.
For that, we let $m_0, \dots, m_{k_*+1}$ be the sequence given by 
$$\begin{cases} 
m_0= k-1&\\
m_1=k&\\
m_2= k+1 & \\
m_p= p- 3 & \text{for }3\le  p \le k+1\\
m_p=p & \text{for } k+2 \le p \le k_*+1
\end{cases}
$$
i.e.~we take away the indices from $k-1$ to $k+1$ and put them at the beginning. If $k=0$, we do the same but moving only two indices to the beginning.

 We define inductively $\psi_j$ by 
\be \label{inductionpsik}
\id+ t\psi_j= \Big(\id +t (\sum_{p=0}^j \chi_{m_p})\psi\Big) \circ \Big( \id +t(\sum_{p=0}^{j-1}\chi_{m_p}) \psi\Big)^{-1}.\ee
In other words, we have decomposed the transport into successive localized transports $\id + t\psi_{j}$. One may check by induction that 
\be \label{inductionpsik2}\id + t\psi= \id +t \sum_{j=0}^{k_*+1} \chi_j \psi = (\id + t\psi_{k_*+1}) \circ \dots \circ(\id+t\psi_0).\ee
We decompose this as 
\be \id+ t\psi= T_{\mathrm{out}}\circ T_{\mathrm{inn}}\ee
where $T_{\mathrm{inn}}= (\id + t\psi_{2}) \circ (\id + t\psi_1) \circ(\id+t\psi_0)$, and 
$T_{\mathrm{out}}= \prod_{j=3}^{k_*+1} (\id + t\psi_j)$ (or the same with only two maps  set aside if appropriate).

If $t\ell^{\s-\d}\|\varphi_0\|_{C^3}$ is small enough, then by \eqref{tscale}, we can make $t\|\psi\|_{L^\infty}\lesssim t \|\varphi_0\|_{C^3}\frac{\ell}{\ell^{\d-\s}}<\frac14\ell$.
%where $C$ is the constant in \eqref{main1}.
One may then check that in view of \eqref{inductionpsik} and definition of the $\chi_j$, for $j \ge 3$, $\psi_j$ is supported in $B_{m_j}$,  and letting $\rho_j= 2^{m_j} \ell$, and in view of \eqref{tscale} we have 
\be | \psi_j|_{C^1} \lesssim \frac{\ell^\d}{\rho_j^{2\d-\s} },
\qquad \|\psi_j\|_{L^\infty} \lesssim \frac{\ell^\d}{\rho_j^{2\d-\s-1}}.
\ee  

Thus, the transport $T_{\mathrm{inn}}$  leaves $A_k$ invariant (where $A_k$ is as in \eqref{defAk}), is supported at distance $\ge \frac14 2^k\ell$ from its boundary, and coincides with identity outside. Thus we may apply Lemma \ref{corocontrenergyt2}  in that set with $\tilde \psi=T_{\mathrm{inn}}$ (we have that $|t|\| T_{\mathrm{inn}}\|_{L^\infty} \lesssim \frac{\ell^\d}{(2^k\ell)^{2\d-\s}}$ hence if $|t|\ell^{\s-\d} \|\varphi_0\|_{C^1}$ is small depending on the constant in \eqref{tscale}, the assumption $|t|\|T_{\mathrm{inn}}\|_{L^\infty} < \min (\frac14 (2^k \ell) , \frac14)$ is verified)  and obtain 
\begin{multline}\label{itertinn}
  \int_{A_k \times [-2^k\ell,2^k\ell]}
  |\nab h_{N,\rr}[T_{\mathrm{inn}}(\XN), T_{\mathrm{inn}}\#\muv]|^2 \\ \lesssim  \exp\( C t\frac{\ell^\d}{(2^k\ell)^{2\d-\s}}\) 
\int_{A_k\times [-2^k\ell,2^k\ell]}   |\nab h_{N,\rr}[\XN,\muv]|^2
+ \#I_{A_k} N^{\frac\s\d}  +  N^{-\frac2\d}(\#I_{A_k} )^2 (2^k\ell)^{-\s-2}.
\end{multline}

For $j\ge 3$, the transports $\id + t\psi_j$ have support in  $(D_{k+2}\backslash D_k)^c$, so we can say that their support is at distance $\ge 2^k \ell$ from $\partial (D_{k+\frac32}\backslash D_{k+\hal})$ (resp. $\partial D_{3/2}$ if $k=0$).  We may thus apply Lemma \ref{corocontrenergyt2} iteratively in $D_{k+\frac32}\backslash D_{k+\hal}$, resp. $D_{3/2}$ if $k=0$,  (which is a set of similar nature as in the assumption) with $\tilde \psi=\psi_{j}$ (again we verify that the assumption $|t|\|\psi_{j}\|_{L^\infty}<\min(\frac{2^k \ell}{4}, \hal)$ is satisfied if $|t|\ell^{\s-\d}\|\varphi_0\|_{C^1}$ is small enough), and find
\begin{multline}\label{itertout}
  \int_{(  D_{k+\frac32}\backslash D_{k+\hal} )\times [-2^k \ell,2^k\ell]}
  |\nab h_{N,\rr}[T_{\mathrm{out}}(T_{\mathrm{inn}}(\XN)), T_{\mathrm{out}}\#(T_{\mathrm{inn}}\#\muv)]|^2 \\ \lesssim  \exp\( C t \sum_{j=0}^{2^*+1}\frac{\ell^\d}{(2^j\ell)^{2\d-\s}}+\frac{\ell^\d}{ 2^k \ell (2^{j\ell})^{2\d-\s-1}} \) \\ \times
\Bigg( \int_{(  D_{k+\frac32}\backslash D_{k+\hal} )\times [-2^k \ell,2^k\ell]}   |\nab h_{N,\rr}[T_{\mathrm{inn}}(\XN),T_{\mathrm{inn}}(\XN)\#\muv]|^2 \\+  \#I_{D_{k+\frac32}\backslash D_{k+\hal}} N^{\frac\s\d}  +  N^{-\frac2\d}(\#I_{D_{k+\frac32}\backslash D_{k+\hal}} )^2 (2^k\ell)^{-\s-2}\Bigg)
\end{multline}
Combining \eqref{itertinn} and \eqref{itertout}, we are led to 
\begin{multline}
\int_{(  D_{k+\frac32}\backslash D_{k+\hal} )\times [-2^k \ell,2^k\ell]}
  |\nab h_{N,\rr}[(\id + t\psi)(\XN), (\id + t\psi)\# \muv]|^2 
\\ \lesssim e^{Ct \ell^{\s-\d}} \(\int_{(A_k\times [-2^k \ell,2^k\ell]}   |\nab h_{N,\rr}[\XN,\muv]|^2+
 \#I_{A_k} N^{\frac\s\d}  +  N^{-\frac2\d}(\#I_{A_k} )^2 (2^k \ell)^{-\s-2}\).
\end{multline}
The result then follows from the local laws for $\XN$ and \eqref{discest} to control $(\#I_{A_k})^2$. We note that since $\carr_\ell$ is at distance $ \ge \ep>0$ from $\pa \Sigma$, either $A_{k}$ is included in the set where the local laws hold, or $2^k \ell$ is larger than a constant depending on $\ep$, in which case we can use the macroscopic law (for $\ell=1$) valid up to the boundary.

\end{proof}

These local controls on the transported energy cover all dyadic annuli and allow us to control the $\Ani_n$ terms, except with small probability.

\begin{lem}\label{lemsubdi}Assume that $\XN\in \mathcal G_\ell$ with $\Gc_\ell$ as in Proposition \ref{pf loiloc}.
%, a set of configurations such that the local laws \eqref{loiloc}--\eqref{controlnbreintro} hold for each $D_k$ in \eqref{defAk} (in $\bulk$) and \eqref{macrolaw2} holds.
 Let $\psi$ be the transport constructed in Proposition \ref{transport} for $\varphi$ satisfying \eqref{estxi} at order $k=5$.  If $|t| \ell^{\s-\d}\M $ is smaller than a constant depending only on  the constant in \eqref{tscale}, 
we have 
\be \label{controlani}
|\Ani_1(\Phi_t(\XN),\Phi_t\# \muv, \psi\circ \Phi_t^{-1})|\lesssim_\beta \M N^{1+\frac{\s}{\d}}\ell^\s
\ee 
and more generally, if \eqref{estxi} holds at order $2n+3$,  we have
\be \label{controlani higher}
|\Ani_n(\Phi_t(\XN), \Phi_t\# \muv, \psi \circ \Phi_t^{-1})|\lesssim_\beta
\M^n N^{1+\frac{\s}{\d}}\ell^{n\s-(n-1)\d}.\ee
\end{lem}
\begin{proof} Let us first treat the easiest situation of $n=1$. Let $k_*$ be as above.
Let $\chi_k$ be the  partition of unity  relative to the  $A_k$'s this time defined as 
\be A_k= D_k\backslash D_{k-\frac32}, \ \text{for } k \le k_*+1, \quad A_{k*+2}= D_{k_*+\hal}^c\ee
with $|\nab^{\otimes m} \chi_k|\lesssim (2^k \ell)^{-m}$. The $A_k$'s are chosen to form  a covering of $\R^\d$, and each $A_k$ intersects only $A_{k-1}$ and $A_{k+1}$, with $A_{k-1} $ and $A_{k+1}$ at positive distance $2^{k/2}\ell$ from each other.
Let us decompose $\psi$ into $\sum_{k=0}^{k_*+1} (\chi_k \psi)$ and use the linearity of $\Ani_1$ with respect to $\psi$, the commutator estimate \eqref{main1}, the local laws \eqref{loiloc}, the law \eqref{macrolaw2} in $U^c$, the estimate on $\psi$ \eqref{tscale} and the result of the above proposition to obtain that for $\XN\in \mathcal G_\ell$, we have 
\begin{align*}  
&|\Ani_1(    \Phi_t(\XN),\Phi_t\# \muv, \psi)| \lesssim                                                                                                                                                                                                                                           \sum_{k=0}^{k_*+1} \|\nab (\chi_k\psi)\|_{L^\infty}\int_{B_k\times [-2^k\ell, 2^\ell]} \yg|\nab h_{N,\rr} (\Phi_t(\XN),\Phi_t( \muv))|^2 \\
&\lesssim_\beta\( \sum_{k=0}^{k_\ast}\|\nabla(\psi \chi_k)\|_{L^\infty}\(2^k\ell\)^\d N^{1+\frac{\s}{\d}} +\|\nab (\psi \chi_{k_*+1})\|_{L^\infty(U^c)} N^{1+\frac\s\d} \)\\
&\lesssim_\beta  \M\(\sum_{k=0}^{k_\ast}\frac{\(2^k\ell\)^\d N^{1+\frac{\s}{\d}}\ell^\d}{\(2^k\ell\)^{2\d-\s}} +\ell^\d N^{1+\frac\s\d}\)\\
&\lesssim_\beta \M\ell^{\s} N^{1+\frac{\s}{\d}}\( \sum_{k=0}^{k_\ast}\frac{1}{\(2^{\d-\s}\)^k} + 1\).
\end{align*}
The proof for larger $n$ relies similarly on a dyadic decomposition coupled with the higher order commutator estimates \eqref{mainn}, although the combinatorics are a bit more complicated due to the tensor products in the integrand. We first describe the approach for $n=2$ since it is instructive and easier to follow, before giving a detailed proof for all $n \geq 2$.

For notational ease, we will write $\XN$, $\mu$ and $\psi$ instead of $\Phi_t(\XN)$, $\Phi_t\#\muv$ and $\psi \circ \Phi_t^{-1}$; this makes the notation below more readable, and the local law applies as desired by Proposition~\ref{pf loiloc}.

\textbf{Description for $n=2$.}

As in the $n=1$ case, we decompose $\psi$ as $\sum_{k=0}^{k_*+1} (\chi_k\psi)$; denoting $\chi_k\psi$ by $\psi^k$ for notational ease, we find 
\begin{equation}\label{ani 2 exp}
\Ani_2(\XN,\mu,\psi)=\frac{1}{2}\sum_{k,m=0}^{k_\ast+1} \iint_{\triangle^c}\nabla^{\otimes 2}\g(x-y) :(\psi^k(x)-\psi^k(y))\otimes (\psi^m(x)-\psi^m(y))~d\fluct_N^{\otimes 2}(x,y).
\end{equation}
There are three kinds of terms in the summands:
\begin{enumerate}
\item $k=m$
\item $|m-k|=1$
\item $|m-k|\geq 2$.
\end{enumerate}
The $k=m$ case is simplest, as the term is itself a higher order commutator 
\begin{equation*}
\Ani_2(\XN,\mu,\psi)=\frac{1}{2}\iint_{\triangle^c}\nabla^{\otimes 2}\g(x-y) :(\psi^k(x)-\psi^k(y))^{\otimes 2}~d\fluct_N^{\otimes 2}(x,y).
\end{equation*}
which we will estimate using \eqref{mainn} and the local law \eqref{loiloc} at the scale $2^k\ell$ at which $\psi^k$ lives.

The case where $|m-k|=1$ is a bit challenging, because although it is not literally a commutator it is very close to one as $\psi^k$ and $\psi^m$ live at similar scales. The idea is to use polarization; letting
\begin{equation*}
\varphi(k,m):= \iint_{\triangle^c}\nabla^{\otimes 2}\g(x-y) :(\psi^k(x)-\psi^k(y))\otimes (\psi^m(x)-\psi^m(y))~d\fluct_N^{\otimes 2}(x,y)
\end{equation*}
we notice that $\varphi(k,m)$ is bilinear in $k$ and $m$, where the sum $k+m$ corresponds to $\psi^k+\psi^m$. Polarizing, we find
\begin{equation*}
\varphi(k,m)=\frac{\varphi(k+m,k+m)-\varphi(k,k)-\varphi(m,m)}{2}.
\end{equation*}
Each $\varphi(i,i)$ term is then a commutator, which we can control using \eqref{mainn} and the local law at the corresponding scale \eqref{loiloc} as in Case 1.

In the third case, $|m-k|\geq 2$, the analysis is quite different since the supports of $\psi^k$ and $\psi^m$ are then disjoint. While this case will be a bit computationally intensive, it is conceptually easier because the summand can then be seen as the double fluctuation of a regular function. Notice that the domain of integration for 
\begin{multline*}
\iint_{\triangle^c}\nabla^{\otimes 2}\g(x-y) :(\psi^k(x)-\psi^k(y))\otimes(\psi^m(x)-\psi^m(y))~d\fluct_N^{\otimes 2}(x,y) \\
=\sum_{i,j=1}^\d \iint_{\triangle^c}\partial_{ij}\g(x-y)(\psi_i^k(x)-\psi_i^k(y))(\psi_j^m(x)-\psi_j^m(y))~d\fluct_N^{\otimes 2}(x,y),
\end{multline*}
where we have denoted the $i$th coordinate of $\psi^k$ by $\psi_i^k$, is only on the support of $(\psi_i^k(x)-\psi_i^k(y))(\psi_j^m(x)-\psi_j^m(y))$. Both of these factors must be nonzero, and since the supports are disjoint then we need one of $x$ and $y$ to be in each. In particular, via the symmetry of $\g$,
\begin{multline*}
\sum_{i,j=1}^\d \iint_{\triangle^c}\partial_{ij}\g(x-y)(\psi_i^k(x)-\psi_i^k(y))(\psi_j^m(x)-\psi_j^m(y))~d\fluct_N^{\otimes 2}(x,y)\\
=2\sum_{i,j=1}^\d \iint_{\triangle^c}\partial_{ij}\g(x-y)\psi_i^k(x)\psi_j^m(y)~d\fluct_N^{\otimes 2}(x,y).
\end{multline*}
We will apply the energy estimate Proposition \ref{pro:controlfluct} twice to obtain the requisite control, once on each fluctuation. With this background for $n=2$, we now explain how to work through the argument for generic $n$.
\textbf{Detailed argument for generic $n \geq 2$.}

Decomposing $\psi=\sum \psi^k$ as above, we write
\begin{equation}\label{ani n exp}
\Ani_n(\XN,\mu,\psi)=\frac{1}{2}\sum_{\vec k \in [k_\ast+2]^n}\iint_{\triangle^c}\nabla ^{\otimes n}\g(x-y) : \bigotimes_{i=1}^n \(\psi^{k_i}(x)-\psi^{k_i}(y)\)~d\fluct_N^{\otimes 2}(x,y)
\end{equation}
where $[k]^n=\{0,1,\dots,k\}^n$. We have three kinds of terms in the summation.

\subsubsection*{Case 1: $\vec k=(k,k,\dots,k)$}
This corresponds to the first type discussed in the $n=2$ case. It is a true commutator and we compute using \eqref{mainn}, \eqref{loiloc} and \eqref{tscale} at scale $2^k\ell$
\begin{align*}
&\left|\iint_{\triangle^c}\nabla ^{\otimes n}\g(x-y) :\(\psi^{k}(x)-\psi^{k}(y)\)^{\otimes n}~d\fluct_N^{\otimes 2}(x,y)\right| \lesssim_\beta \(2^k\ell\)^\d N^{1+\frac{\s}{\d}}\sum_{p=0}^n \(2^k\ell\|\nab^{\otimes 2}\psi\|_{L^\infty}\)^p \\
&\hspace{3cm} \times
 \sum_{\substack{1\leq c_1,\ldots,c_{n-p} \\ n-p\le c_1+\cdots+c_{n-p} \le 2n}} N^{\frac{1}{\d}((n-p)-\sum_{k=1}^{n-p} c_{k}) } \|\nabla^{\otimes c_1} \psi\|_{L^\infty}\cdots\|\nabla^{\otimes c_{n-p}} \psi\|_{L^\infty} \\
 &\lesssim_\beta \(2^k\ell\)^\d N^{1+\frac{\s}{\d}}\sum_{p=0}^n \frac{\ell^{\d p}\M^p}{\(2^k\ell\)^{2\d-\s}}\sum_{\substack{1\leq c_1,\ldots,c_{n-p} \\ n-p\le c_1+\cdots+c_{n-p} \le 2n}}N^{\frac1\d((n-p)-\sum_{k=1}^{n-p} c_{k}) } \prod_{q=1}^{n-p}\frac{\ell^\d\M}{\(2^k\ell\)^{2\d-\s+c_q-1}} \\
 %&\lesssim \(2^k\ell\)^\d N^{1+\frac{\s}{\d}}\sum_{p=0}^n \frac{\ell^{\d p}\|\varphi_0\|_{C^5}^p}{\(2^k\ell\)^{\d+2\alpha}}\sum_{\substack{1\leq c_1,\ldots,c_{n-p} \\ n-p\le c_1+\cdots+c_{n-p} \le 2n}} \lambda^{-(n-p)+\sum_{k=1}^{n-p} c_{k} }\(2^k\ell\)^{(n-p)+\sum_{q=1}^{n-p}c_q}\frac{\ell^{\d(n-p)}\|\varphi_0\|_{C^{2n+3}}^{n-p}}{\(2^k\ell\)^{(\d+2\alpha)(n-p)}} \\
 &\lesssim_\beta \(2^k\ell\)^\d N^{1+\frac{\s}{\d}}\frac{\ell^{\d n}\M^n}{\(2^k\ell\)^{(2\d-\s)n}}=\ell^{\d(1-n) +\s n}N^{1+\frac{\s}{\d}}\M^n\(2^{n(\s-\d)+\d(1-n)}\)^k
\end{align*}
where we have used that the largest possible $c_q$ is $2n$ and that
\begin{equation*}
N^{\frac{1}{\d}( (n-p)-\sum_{k=1}^{n-p} c_{k}) }\(2^k\ell\)^{(n-p)+\sum_{q=1}^{n-p}c_q}\lesssim 1
\end{equation*}
since $N^{-1/\d} \leq 2^k \ell$ and $(n-p)-\sum_{k=1}^{n-p} c_{k} \geq 0$. Summing over $k$ yields 
\begin{multline*}
\sum_{\vec k=(k,k,\dots,k), k\in [k_*+1]}\left|\iint_{\triangle^c}\nabla ^{\otimes n}\g(x-y) : \bigotimes_{i=1}^n \(\psi^{k_i}(x)-\psi^{k_i}(y)\)~d\fluct_N^{\otimes 2}(x,y)\right| \\
\lesssim_\beta \ell^{(1-n) \d+\s n}N^{1+\frac{\s}{\d}}\M^n\end{multline*}
In the case $k=k_*+2$, we use the bound \eqref{tscale} and \eqref{macrolaw2} to obtain that this term is controlled by $\ell^{\d}N^{1+\frac\s\d}$, which can be incorporated into the previous estimate. %\cm{please check}

\subsubsection*{Case 2: $\bigcup_i\supp(\psi^{k_i})$ is connected}
This corresponds to the second type discussed in the $n=2$ case. Since the supports of $\psi^k$ and $\psi^m$ only overlap for $|m-k|\leq 1$, this is only possible if the index vector $\vec k$ takes the form (up to reordering)
\begin{equation*}
\vec k=(k, k+a_1, k+a_1+a_2, \dots, k+a_1+\dots + a_{n-1}), \qquad a_i\in \{0,1\}
\end{equation*}
Note that for each $k$, there are only a constant dependent on $n$ number of such vectors. For these indices, we make use of multilinear polarization. For notational ease, we again let
\begin{equation*}
\varphi(k_1,\dots,k_n):= \iint_{\triangle^c}\nabla ^{\otimes n}\g(x-y) : \bigotimes_{i=1}^n \(\psi^{k_i}(x)-\psi^{k_i}(y)\)~d\fluct_N^{\otimes 2}(x,y).
\end{equation*}
Notice that $\varphi(k_1,\dots,k_n)$ is symmetric and multilinear in $\vec k$, where the sum $k+m$ again corresponds to $\psi^k+\psi^m$. Using the polarization formula for symmetric, multilinear forms we have
\begin{equation*}
\varphi(k_1,\dots,k_n)=\frac{1}{n!}\sum{p=1}^n \sum_{1 \leq j_1 < \cdots < j_p \leq n}(-1)^{n-p}\varphi(k_{j_1}+\cdots+k_{j_p})
\end{equation*}
where we have simplified $\varphi(m):=\varphi(m,m,\dots,m)$. Now, for a vector of the form $\vec k=(k, k+a_1, k+a_1+a_2, \dots, k+a_1+\dots + a_{n-1})$ with $a_i$ as above, all of the $k_i$ live at the same scale (up to an $n$-dependent constant). Hence, we can apply \eqref{loiloc} and \eqref{tscale} at scale $2^k\ell$, up to adjusting constants, which coupled with \eqref{mainn} as in Case 1 yields
\begin{equation*}
|\varphi(k_1,\dots,k_n)|\lesssim_\beta \ell^{\d(1-n)+\s n }N^{1+\frac{\s}{\d}}\M^n\(2^{ n(\s-\d)+(1-n)\d}\)^k
\end{equation*}
In the case $k=k_*+2$, we again use the bound \eqref{tscale} and \eqref{macrolaw2} to obtain that this term is controlled by $\ell^{\d}N^{1+\frac\s\d}$, which can be incorporated into the previous estimate. %\cm{please check} 
Summing over $k$ again yields
\begin{multline*}
\sum_{\vec k \text{ is Case }2}\left|\iint_{\triangle^c}\nabla ^{\otimes n}\g(x-y) : \bigotimes_{i=1}^n \(\psi^{k_i}(x)-\psi^{k_i}(y)\)~d\fluct_N^{\otimes 2}(x,y)\right| \\
\lesssim_\beta \M^n \ell^{n\s-(n-1)\d}N^{1+\frac{\s}{\d}}
\end{multline*}
using that there are an $O(n)$ number of $\vec k$ associated to each $k$.
\subsubsection*{Case 3: $\bigcup_i\supp(\psi^{k_i})$ is not connected}
This corresponds to the third type discussed in the $n=2$ case. The first observation to make here is that we only need to consider the situation where $\bigcup_i\supp(\psi^{k_i})$ is a disjoint union of two connected sets. Indeed, suppose that it could be written as a disjoint union of three connected sets. Then, since the support of distinct $\psi^k$ and $\psi^m$ only overlap for $|m-k|\leq 1$, we can write the $\vec k$ (up to reordering) as 
\begin{multline*}
\vec k=(k, k+a_0,\dots, k+a_1+\cdots+a_{p-1}, m, m+b_0, \dots, m+b_1+\cdots+b_{q-1}, \\
r, r+c_0,\dots, r+c_1+\cdots+c_{n-p-q-1})
\end{multline*}
where all of the $a_i,b_i$ and $c_i$ belong to $\{0,1\}$ and the sets 
\begin{equation*}
A=\bigcup_{i=1}^{p}\supp(\psi^{k+a_1+\cdots+a_{i-1}}), \hspace{3mm}
B=\bigcup_{i=1}^{q}\supp(\psi^{m+b_1+\cdots+b_{i-1}}), \hspace{3mm}
C=\bigcup_{i=1}^{n-p-q}\supp(\psi^{r+c_1+\cdots +c_{i-1}}) 
\end{equation*}
are pairwise disjoint. We have set $a_0=b_0=c_0$ for notational ease. Now, let us examine the corresponding terms of $\Ani_n$, which look like 
\begin{multline*}
\sum_{i_1,\dots,i_n}\iint_{\triangle^c}\partial_{i_1,\dots,i_n}\g(x-y)\prod_{j=1}^p\(\psi_{i_j}^{k+a_0+\cdots+a_{i-1}}(x)-\psi_{i_j}^{k+a_0+\cdots+a_{i-1}}(y)\)\times \\
\prod_{j=1}^q\(\psi_{i_{p+j}}^{m+b_0+\cdots+b_{i-1}}(x)-\psi_{i_{p+j}}^{m+b_0+\cdots+b_{i-1}}(y)\) \times \\
 \prod_{j=1}^{n-p-1}\(\psi_{i_{p+q+j}}^{r+c_0+\cdots+c_{i-1}}(x)-\psi_{i_{p+q+j}}^{r+c_0+\cdots+c_{i-1}}(y)\)~d\fluct_N^{\otimes 2}(x,y)
\end{multline*}
In order for the integral to be nonzero, all three factors need to not vanish. So, at least one of $x$ or $y$ needs to belong to  $A$, $B$ and $C$. However, since $A$, $B$ and $C$ are pairwise disjoint this is impossible! The same argument works for more than three disjoint sets. 
 Hence, $\bigcup_i\supp(\psi^{k_i})$ is a disjoint union of two connected sets and $\vec k$ takes the form (up to reordering)
\begin{equation*}
\vec k=(k, k+a_1,\dots, k+a_1+\cdots+a_{p-1}, m, m+b_1, \dots, m+b_1+\cdots+b_{n-p-1})
\end{equation*}
where all of the $a_i,b_i$ and $c_i$ belong to $\{0,1\}$ and the sets $A$ and $B$ as above are disjoint. There are a constant bounded by a number dependent only on $n$ number of vectors associated to each $k$ and $m$. The corresponding terms of $\Ani_n$ take the form
\begin{multline*}
\iint_{\triangle^c}\partial_{i_1,\dots,i_n}\g(x-y)\prod_{j=1}^p\(\psi_{i_j}^{k+a_0+\cdots+a_{i-1}}(x)-\psi_{i_j}^{k+a_0+\cdots+a_{i-1}}(y)\) \times \\
\prod_{j=1}^{n-p}\(\psi_{i_{p+j}}^{m+b_0+\cdots+b_{i-1}}(x)-\psi_{i_{p+j}}^{m+b_0+\cdots+b_{i-1}}(y)\) ~d\fluct_N^{\otimes 2}(x,y)
\end{multline*}
and the integral is only nonzero if $x \in A$ and $y \in B$ (or vice versa). Without loss of generality, assume $x \in A$. Then, the terms above become
\begin{multline*}
\iint_{\triangle^c}\partial_{i_1,\dots,i_n}\g(x-y)\prod_{j=1}^p\psi_{i_j}^{k+a_0+\cdots+a_{i-1}}(x)\prod_{j=1}^{n-p}\psi_{i_{p+j}}^{m+b_0+\cdots+b_{i-1}}(y) ~d\fluct_N^{\otimes 2}(x,y)\\
:=\iint_{\triangle^c}\partial_{i_1,\dots,i_n}\g(x-y)\varphi_1(x)\varphi_2(y) ~d\fluct_N^{\otimes 2}(x,y)
\end{multline*}
with $\partial_{i_1,\dots,i_n}\g(x-y)\sim |x-y|^{-\s-n}$. We bound as described in the third type discussed in the $n=2$ case. We set 
\begin{equation*}
f(x)=\int \partial_{i_1,\dots,i_n}\g(x-y)\varphi_2(y)~d\fluct_N(y)
\end{equation*}
and first bound $f$ and its derivatives using Proposition \ref{pro:controlfluct} and \eqref{loiloc} and \eqref{tscale} at scale $2^m \ell$ (if $m=k_\ast+2$ then the same estimates hold with $2^{m}\ell\sim1$). Using $|x-y|\gtrsim 2^m \ell$, we have
\begin{multline*}
\| \partial_{i_1,\dots,i_n}\g(x-y)\varphi_2(y)\|_{L^2}^2 \lesssim_\beta \int_{\cup_{k=m}^{m+b_0+\cdots+b_{n-p-1}}A_k}\frac{1}{|x-y|^{2\s+2n}}\(\frac{\ell^{2\d}\M^2}{|y|^{4\d-2\s-2}}\)^{n-p}~dy \\
\lesssim_\beta \ell^{2\d(n-p)}\M^{2(n-p)}\int_{O\(\(2^{m+b_0+\cdots+b_{n-p-1}}\)^\d\)}\frac{r^{\d-1}}{r^{2\s+2n+(4\d-2\s-2)(n-p)}}~dr \\
\lesssim_\beta \frac{\ell^{2\d(n-p)}\M^{2(n-p)}}{\(2^m \ell\)^{2\s+2n+(4\d-2\s-2)(n-p)-\d}}.
\end{multline*}%\cm{replace the integration set by the $A_k$'s. aren't there many $b$'s to sum} 
An analogous computation yields 
\begin{align*}
\| \nabla\(\partial_{i_1,\dots,i_n}\g(x-y)\varphi_2(y)\)\|_{L^2}^2 &\lesssim_\beta \frac{\ell^{2\d(n-p)}\M^{2(n-p)}}{\(2^m \ell\)^{2\s+2n+(4\d-2\s-2)(n-p)-\d+2}} \\
\| \nabla^{\otimes 2}\(\partial_{i_1,\dots,i_n}\g(x-y)\varphi_2(y)\)\|_{L^2}^2 &\lesssim_\beta \frac{\ell^{2\d(n-p)}\M^{2(n-p)}}{\(2^m \ell\)^{2\s+2n+(4\d-2\s-2)(n-p)-\d+4}} 
\end{align*}
We first use Proposition \ref{pro:controlfluct} to get $L^\infty$ bounds on $h$ and its gradient in $\supp \psi_i^k$. As a simplification, notice that it is sufficient to control
\begin{equation*}
 \(\eta^{\gamma-1}\|\varphi\|_{L^2(\Omega)}^2+\eta^{\gamma+1}\|\nabla \varphi\|_{L^2(\Omega)}^2\)^{1/2}\(\(2^m\ell\)^\d N^{1+\frac{\s}{\d}}\)^{1/2}
\end{equation*} 
%\cm{which $\ep$ is chosen?}
for $\eta$ to be chosen below, with $\varphi(y)=\partial_{i_1,\dots,i_n}\g(x-y)\varphi_2(y)$ using \eqref{loiloc} at scale $2^m \ell$, since the error term is strictly smaller; one can bound $\#I_\Omega (N^{-\frac1\d})^{\d-\s}$ by the energy at scale $2^m \ell$. A computation shows that the terms above are balanced by choosing $\eta=2^m\ell$, and we conclude by Proposition \ref{pro:controlfluct} that
\begin{align*}
|f(x)|^2&\lesssim_\beta \frac{\ell^{2\d(n-p)}\M^{2(n-p)}\(2^m \ell\)^\d N^{1+\frac{\s}{\d}}}{\(2^m \ell\)^{2\s+2n+(4\d-2\s-2)(n-p)-\d+1-\gamma}}\\
|\nabla f(x)|^2&\lesssim_\beta \frac{\ell^{2\d(n-p)}\M^{2(n-p)}\(2^m \ell\)^\d N^{1+\frac{\s}{\d}}}{\(2^m \ell\)^{2\s+2n+(4\d-2\s-2)(n-p)-\d+3-\gamma}}.
\end{align*}
We now use these estimates to control
\begin{equation*}
\int \varphi_1(x)f(x)\, d\fluct_N(x)
\end{equation*}
using Proposition \ref{pro:controlfluct}. We find using \eqref{tscale} (again, if $k=k_\ast+2$ then the same estimates hold with $2^{k}\ell\sim1$)
\begin{multline*}
\|\varphi_1(x)f(x)\|_{L^2}^2\lesssim_\beta \frac{\ell^{2\d(n-p)}\M^{2(n-p)}N^{1+\frac{\s}{\d}}}{\(2^m \ell\)^{2\s+2n+(4\d-2\s-2)(n-p)-2\d+1-\gamma}}\int_{\cup_{m=k}^{k+a_0+\cdots+a_{p-1}}A_k}\(\frac{\ell^{2\d}\M^2}{|y|^{4\d-2\s-2}}\)^p~dy \\
\lesssim_\beta  \frac{\ell^{2\d n}\M^{2n}N^{1+\frac{\s}{\d}}}{\(2^m \ell\)^{2\s+2n+(4\d-2\s-2)(n-p)-2\d+1-\gamma}\(2^k\ell\)^{(4\d-2\s-2)p-\d}}.
\end{multline*} 
We also have 
\begin{multline*}
\|\nabla(\varphi_1(x)f(x))\|_{L^2}^2\lesssim_\beta \frac{\ell^{2\d(n-p)}\M^{2(n-p)}N^{1+\frac{\s}{\d}}}{\(2^m \ell\)^{2\s+2n+(4\d-2\s-2)(n-p)-2\d+3-\gamma}}\int_{\cup_{m=k}^{k+a_0+\cdots+a_{p-1}}A_k}\(\frac{\ell^{2\d}\M^2}{|y|^{4\d-2\s-2}}\)^p~dy \\
+\frac{\ell^{2\d(n-p)}\M^{2(n-p)}N^{1+\frac{\s}{\d}}}{\(2^m \ell\)^{2\s+2n+(4\d-2\s-2)(n-p)-2\d+1-\gamma}}\int_{\cup_{m=k}^{k+a_0+\cdots+a_{p-1}}A_k}\(\frac{\ell^{2\d}\M^2}{|y|^{4\d-2\s-2}}\)^{p-1}\(\frac{\ell^{2\d}\M^2}{|y|^{4\d-2\s}}\)~dy 
\end{multline*}
which yields the bound 
\begin{equation*}
\|\nabla(\varphi_1(x)f(x))\|_{L^2}^2\lesssim_\beta\frac{\ell^{2\d n}\M^{2n}N^{1+\frac{\s}{\d}}}{\(2^m \ell\)^{\s+2n+(4\d-2\s-2)(n-p)-\d}\(2^k\ell\)^{(4\d-2\s-2)p-\d}}\(\frac{1}{\(2^m\ell\)^2}+\frac{1}{\(2^k\ell\)^2}\).
\end{equation*}
using $\gamma+\d-1=\s$. Applying Proposition \ref{pro:controlfluct} with $\eta=\(\frac{1}{\(2^m\ell\)^2}+\frac{1}{\(2^k\ell\)^2}\)^{-1/2}$ and \eqref{loiloc} at scale $2^k\ell$ to find 
\begin{equation*}
\left|\int \varphi_1(x)f(x)\, d\fluct_N(x)\right| \lesssim_\beta\frac{\ell^{\d n}\M^{n}\(2^k\ell\)^{\frac{\d}{2}}N^{1+\frac{\s}{\d}}}{\(2^m \ell\)^{\frac{\s}{2}+n+(2\d-\s-1)(n-p)-\frac{\d}{2}}\(2^k\ell\)^{(2\d-\s-1)p-\frac{\d}{2}}}\(\frac{1}{\(2^m\ell\)^2}+\frac{1}{\(2^k\ell\)^2}\)^{\frac{1-\gamma}{4}}.
\end{equation*}
which, simplifying, yields
\begin{multline*}
\iint_{\triangle^c}\partial_{i_1,\dots,i_n}\g(x-y)\prod_{j=1}^p\psi_{i_j}^{k+a_0+\cdots+a_{i-1}}(x)\prod_{j=1}^{n-p}\psi_{i_{p+j}}^{m+b_0+\cdots+b_{i-1}}(y) ~d\fluct_N^{\otimes 2}(x,y)\\
\lesssim_\beta \frac{\M^n\ell^{n\s-(n-1)\d}N^{1+\frac{\s}{\d}}}{\(2^m\)^{\frac{\s}{2}+n+(2\d-\s-1)(n-p)-\frac{\d}{2}}+\(2^k\)^{(2\d-\s-1)p-\frac{\d}{2}}}\(\frac{1}{2^{2m}}+\frac{1}{2^{2k}}\)^{\frac{1-\gamma}{4}}.
\end{multline*}

Summing over all choices of $k$ and $m$, we find 
\begin{equation*}
\sum_{\vec k \text{ is Case }3}\left|\iint_{\triangle^c}\nabla ^{\otimes n}\g(x-y) : \bigotimes_{i=1}^n \(\psi^{k_i}(x)-\psi^{k_i}(y)\)~d\fluct_N^{\otimes 2}(x,y)\right| 
\lesssim_\beta \M^n\ell^{n\s-(n-1)\d}N^{1+\frac{\s}{\d}}.
\end{equation*}
%\cm{again check that case $k_*+2$ is ok}
Combining Cases $1-3$ yields the result.
\end{proof}

%\cm{what follows is false because the transport is not localized}
% Let
%\be \varphi(t):= \F_N(\Phi_t(\X_N),\Phi_t\# \mu) +\( \frac{N}{4} \log N\) \indic_{\s=0}+  C N^{1+\frac\s\d} \|\mu_t\|_{L^\infty}^{\frac\s\d}\ee

%As in \cite[Corollary 9.1.2]{S24}, combining \eqref{dtF} with $n=1$ and \eqref{main1} we have 
%$$\varphi'(t)\le C(1+t |\psi|_{C^1} \varphi(t)$$
%Thus, if $t |\psi|_{C^1} <\hal$, we have by Gronwall's lemma,
%\begin{equation}
%\label{bXi}
%\varphi(t)\le e^{Ct}\varphi(0).
%\ee

As a corollary, we obtain a first bound on the term $T_2$ defined in \eqref{t2}.

\begin{coro}[A first bound on $T_2$]\label{first T2 bound}
Under the same assumptions, 
supposing that \eqref{estxi} holds at order $4$,   if $t \ell^{\s-\d}\M$ is small enough, for every $N$ sufficiently large, we have
\begin{equation}\label{firstt2bound}
\left|\log \Esp_{\PNbeta}(e^{T_2}\indic_{\mathcal{G}_\ell})\right|\lesssim_\beta \beta \M|t|N\ell^\s
 \end{equation}
 where $\Gc_\ell$ is as in Proposition \ref{pf loiloc}.
\end{coro}
\begin{proof}
Let $\XN\in \mathcal{G}_\ell$.
From \eqref{dtF}, we have 
\begin{align*}
\F_N(\Phi_t(\XN), \Phi_t\#\mu_0)- \F_N(\XN, \mu_0)= \int_0^t \Ani_1(\Phi_s(\XN), \mu_s , \psi\circ\Phi_s^{-1})ds.\end{align*}
In view of \eqref{controlani}
 this is $\lesssim_\beta \M N^{1+\frac{\s}{\d}}\ell^\s$. Similarly, using \eqref{Linf} on dyadic scales can bound
\begin{equation*}
\left|\Fluct_{\muv}(\log \det D\Phi_t)\right|\lesssim_\beta |tN|\M\sum \(2^k\ell\)^\d\frac{\ell^\d}{\(2^k\ell\)^{2\d-\s}}
 \lesssim_\beta  |t|N\|\varphi_0\|_{C^4}\ell^\s
\end{equation*}
using the transport bounds \eqref{tscale}.  Combining all of the above  with \eqref{t2} yields the result.
\end{proof}

\section{Proof of Theorems \ref{FirstFluct} and  \ref{CLT}}\label{sec: fluct}
We now have all of the tools needed to examine fluctuations of linear statistics, and complete the proofs of Theorems \ref{FirstFluct} and \ref{CLT}. Let us first recall the expansion in Lemma \ref{lem2.1}, which allows us to expand the Laplace transform of $\Fluct_{\muv}(\varphi)$ on a good event $\mathcal{G}$ by
\begin{equation*}
\Esp_{\PNbeta}\left[\exp\(- \beta t N^{1-\frac\s\d} \Fluct_{\muv}(\varphi)\)\indic_{\mathcal{G}}\right]=e^{T_0}\Esp_{\PNbeta}\left[e^{T_1+T_2}\indic_{\mathcal{G}}\right].
\end{equation*}
We now examine each of these terms individually.

\subsection{Computation of $T_0$}\label{subsec: T0}
Using the explicit choice of $\psi$ in Proposition \ref{transport}, we can compute $T_0$, which is a completely deterministic computation. We start with the $\log \det D\Phi_t$ term.
\begin{lem} Let $\Phi_t=\id+ t\psi$ with $\psi$ as in Proposition \ref{transport}. Denoting $\mu_t= \Phi_t \# \muv$, 
%and letting 
%\be\label{defentropie} \Ent(\mu):=\int_{\R^\d} \mu \log \mu\ee
 we 
 have 
\be\label{ilogdet}
\int_{\R^\d} \log \det D\Phi_t \, d\muv=\Ent( \muv)- \Ent(\mut)=   t \, M(\varphi) +O\(t^2\M^2\ell^{2\s-\d}\)\ee
with 
\be\label{defmean}M(\varphi)=  \frac{1}{\c} \int_\Sigma (-\Delta)^{\a} \varphi^\Sigma (\log \muv).\ee   \end{lem}
%\cm{Check signs}
\begin{proof} Since $\Phi_t=\id + t\psi$, and 
$\mu_t=\Phi_t\#\muv$ we have  $\det D\Phi_t= \frac{\muv}{\mu_t\circ \Phi_t}$, and thus
\be\int \log \det D\Phi_t d\muv= \int \log \muv d\muv- \int \log \mu_t(\Phi_t(x)) d\muv= \Ent(\muv)-\Ent(\mu_t). \ee
We can also write explicitly
 \be
 \log \det D\Phi_t=t\, \div \psi+O\(t^2|D\psi|^2\)
\ee
and have
 \be
  \int_{\R^\d}(\div\psi)\muv= \int_{\R^\d} \div(\psi \muv)-\int_{\R^\d} \psi \muv  \cdot \nabla \log \muv= \int_{\R^\d} \div(\psi \muv) (1+ \log \muv).\ee
  Thus, using \eqref{divpsimu0} and \eqref{Rtransport}
   \begin{align*}
    \int_{\R^\d}(\div\psi)\muv & = \frac{1}{\cds}\int_{\Sigma} (-\Delta)^\alpha \varphi^\Sigma(1+ \log \muv)\\  
    &  = \frac{1}{\cds}\int_{\R^\d} (-\Delta)^\alpha \varphi^\Sigma(1+ \log \muv)\\
    & =\frac1{\cds} \int_\Sigma (-\Delta)^\alpha \varphi^\Sigma \log \muv,
 \end{align*} where we have used that $(-\Delta)^\a \varphi^\Sigma$ is mean zero.
 
 Integrating the error term and using \eqref{cdpsi}, we obtain the result.
\end{proof}
Let us turn to the remaining terms.
\begin{lem}Let $\Phi_t=\id+ t\psi$ with $\psi$ as in Proposition \ref{transport} with \eqref{estxi} at order  $ 4$. 
Let $T_0$ be as in \eqref{t0}. Then, we have 
%\be\label{pbt0}  |T_0|\le C \beta N^{2-\frac\s\d} t^2  \|\varphi_0\|_{C^4}^2 \ell^{2\s-2\d}  +C Nt  \|\varphi_0\|_{C^{4}}\ell^{\s-\d}, \ee  and
\be\label{bt0}
T_0=   \frac{\beta  N^{2-\frac\s\d} t^2}{2}  \Var(\varphi)
+ t N   M(\varphi) +O\(t^3 N^{2-\frac\s\d}\beta\M^3\ell^{2\s-\d} + t^2\M^2N\ell^{2\s-\d}\)   \ee
where
\be
\Var(\varphi)=\frac{c_{\d,\alpha}}{2\cds}\left\|\varphi^\Sigma\right\|_{\dot{H}^\frac{\d-\s}{2}}^2
\ee
with $M(\varphi)$ as in \eqref{defmean} and $c_{\d, \alpha}$ as in \eqref{def fraclap}.
\end{lem}
\begin{proof}
%For \eqref{pbt0}, we Taylor expand to first order. First, $|\log \det D\Phi_t|=O\(|t|\|\psi\|_{C^{1}}\).$  A careful computation yields 
%\begin{equation*}
%\g(\Phi_t(x)-\Phi_t(y))-\g(x-y)=t\nabla \g(x-y)\cdot (\psi(x)-\psi(y))+O\(t^2\|\psi\|_{C^{1}}^2\frac{1}{|x-y|^\s}\)
%\end{equation*}
%where the $O$ is uniform despite the singularity in the derivative of $\g$; this can be seen by factoring out $|x-y|$ so we Taylor expand instead about $\frac{x-y}{|x-y|}$, where $\g$ is bounded. Similarly
%\begin{equation*}
%\int (V_t \circ \Phi_t -V -t\varphi)\, d\muv=\int t \nabla V \cdot \psi\, d\muv+O\(t^2\|\psi\|_{L^\infty}\(\|\psi\|_{L^\infty}+\|\varphi\|_{C^{1}}\)\)
%\end{equation*}
%The order $t$ terms are (by symmetry)
%\begin{multline*}
%\frac{t}{2}\iint \nabla \g(x-y)\cdot(\psi(x)-\psi(y))\, d\muv(x)d\muv(y)+t\int \nabla V \cdot \psi(x)\, d\muv(x) \\
%=t\int \nabla \(h^{\muv}+V\right)\cdot \psi(x)\, d\muv(x)=0,
%\end{multline*}
%where we used \eqref{Riesz effective potential} and the fact that $\zeta_V$ vanishes in the support of $\muv$.
%Integrating $\frac{1}{|x-y|^\s}$ in $x$ and $y$ against $\muv$ yields  
%\be  |T_0|\lesssim  \beta N^{2-\frac\s\d} t^2  \|\psi\|_{C^{1}}^2 +C Nt  \|\psi\|_{C^{1}}. \ee Inserting \eqref{tscale}, we obtain \eqref{pbt0}.

The proof is based on a second order  Taylor expansion in $t$. 
First, a careful computation yields 
\begin{multline*}
\frac{1}{2}\left(\g(\Phi_t(x)-\Phi_t(y))-\g(x-y)\right)=\frac{t}{2}\nabla \g(x-y)\cdot (\psi(x)-\psi(y)) \\
+\frac{t^2}{4}\sum_{i,j}\partial_{i,j}\g(x-y)(\psi_i(x)-\psi_i(y))(\psi_j(x)-\psi_j(y))+O\(t^3\frac{|\psi(x)-\psi(y)|^3}{|x-y|^{3+\s}}\)
\end{multline*}
where the $O$ is uniform despite the singularity in the derivative of $\g$; this again can be obtained by factoring out $|x-y|$ so that we Taylor expand instead about $\frac{x-y}{|x-y|}$, where $\g$ is bounded. Similarly, 
\begin{equation*}
V(x+t\psi(x))-V(x)=t\nabla V \cdot \psi(x)+\frac{t^2}{2}\sum_{i,j}\partial_{i,j}V\psi_i(x)\psi_j(x)+O\(t^3|\psi(x)|^3\)
\end{equation*}
and 
\begin{equation*}
t\varphi(x+t\psi(x))-t\varphi(x)=t^2\nabla \varphi \cdot \psi(x)+O\(t^3|D^2\varphi||\psi|^2\).
\end{equation*} After integration, the first order terms  give by symmetry
\begin{multline*}
\frac{t}{2}\iint \nabla \g(x-y)\cdot(\psi(x)-\psi(y))\, d\muv(x)d\muv(y)+t\int \nabla V \cdot \psi(x)\, d\muv(x) \\
=t\int \nabla \(h^{\muv}+V\right)\cdot \psi(x)\, d\muv(x)=0,
\end{multline*}
where the vanishing is thanks to \eqref{Riesz effective potential} and the fact that $\zeta_V$ vanishes in the support of $\muv$. 
 Now, expanding $(\psi_i(x)-\psi_i(y))(\psi_j(x)-\psi_j(y))$ and using symmetry we find 
\begin{multline*}
\iint \frac{t^2}{4}\sum_{i,j}\partial_{i,j}\g(x-y)(\psi_i(x)-\psi_i(y))(\psi_j(x)-\psi_j(y))\, d\muv(x)d\muv(y) \\
=\frac{t^2}{2}\sum_{i,j}\iint \partial_{i,j}\g(x-y)\psi_i(x)\psi_j(x)\, d\muv(x)d\muv(y)-\frac{t^2}{2}\sum_{i,j}\iint \partial_{i,j}\g(x-y)\psi_i(x)\psi_j(y)\, d\muv(x)d\muv(y).
\end{multline*}
Notice that 
\begin{multline*}
\int \sum_{i,j}\partial_{ij}V \psi_i(x)\psi_j(x)~d\muv(x)+\sum_{i,j}\iint \partial_{i,j}\g(x-y)\psi_i(x)\psi_j(x)\, d\muv(x)d\muv(y) \\
=\int \partial_{i,j}(V+ \g*\muv)(x)\psi_i(x)\psi_j(x)\, d\muv(x)=0
\end{multline*} by the same argument as above using \eqref{Riesz effective potential}. 
Thus, the order $t^2$ terms that remain are just 
\begin{equation*}
t^2 \int  \nabla \varphi \cdot \psi\, d\muv-\frac{t^2}{2}\sum_{i,j}\iint \partial_{i,j}\g(x-y)\psi_i(x)\psi_j(y)\, d\muv(x)d\muv(y).
\end{equation*} 
For the first term, integrating by parts and using \eqref{divpsimu0} and \eqref{Rtransport} yields 
\begin{equation*}
t^2 \int_{\R^\d} \nabla \varphi \cdot \psi\, d\muv=-t^2\int_{\R^\d} \varphi \, \div(\psi \muv)=-\frac{t^2}{\cds}\int \varphi^\Sigma (-\Delta)^\alpha \varphi^\Sigma=-\frac{t^2}{2}\frac{c_{\d,\alpha}}{\cds}\left\|\varphi^\Sigma\right\|_{H^\alpha}^2
\end{equation*} 
since we have via fractional integration by parts and \cite[Proposition 3.6]{DPV12} that
\begin{equation*}
\int f (-\Delta)^\alpha f=\int| (-\Delta)^{\alpha/2}f|^2=
%\int |\xi|^{2\alpha}|\hat f|^2~d\xi=
\frac{c_{\d,\alpha}}{2}\|f\|_{\dot{H}^\alpha}^2
\end{equation*}
with $c_{\d,\alpha}$ the constant in \eqref{def fraclap}, in view of the definitions \eqref{def fraclap} and \eqref{soboFT}.
%\cm{it would be best not to have to use Fourier here -- following the definition \eqref{homogsobo} this is $\|(-\Delta)^\alpha \varphi^\Sigma \|_{\dot{H}^{-\alpha}}$. How can we transform that into what we want? argue fractional integration by parts
%}
Integrating the second term by parts in $x$ and $y$ and using \eqref{Rtransport} yields
\begin{align*}
&-\frac{t^2}{2}\sum_{i,j}\iint \partial_{i,j}\g(x-y)\psi_i(x)\psi_j(y)\, d\muv(x)d\muv(y)=-\frac{t^2}{2}\iint \g(x-y)\div(\psi\muv)(x)\div(\psi\muv)(y) \\
&=\frac{-t^2}{2\cds^2}\iint \g(x-y)(-\Delta)^\alpha \varphi^\Sigma(y)(-\Delta)^\alpha \varphi^\Sigma(x) \\
&=\frac{-t^2}{2\cds}\int \varphi^\Sigma(x)(-\Delta)^\alpha \varphi^\Sigma(x)=\frac{-t^2}{4}\frac{c_{\d,\alpha}}{\cds}\left\|\varphi^\Sigma\right\|_{H^\alpha}^2
.
\end{align*}Finally, we need to integrate the error terms. This is done using   \eqref{cpsi}, \eqref{cdpsi} and  \eqref{cratio},  to write \begin{equation*}
\int |D^2\varphi||\psi|^2 \lesssim \frac{1}{\ell^2}\ell^{-\d+2\s+2}\M^3,
\end{equation*}
and similarly for the other terms.
Combining  with \eqref{ilogdet} we have thus established \eqref{bt0}.
\end{proof}
%From this it is clear that we should choose 
%\be t= \frac{\tau \L_N^{-\frac{\s}{2} }  }{    \sqrt{\beta}N} .\ee
%Then $\beta t N= \sqrt{\beta} \L_N^{-\frac{\s}{2}}$.
%Notice that the error terms made are thus $o\(N^{\frac{\s}{\d}}\)$, uniformly in $\beta$ as long as $\beta N$ is bounded below.
 
 \subsection{Control of $T_1$}\label{subsec: T1}
 We next turn our attention to $T_1$ as in \eqref{t1}; the choice of transport $\psi$ made in Section \ref{sec: fluct} was so that $T_1$ would vanish at leading order in $t$. We verify that here, and estimate the next order in $t$.
\begin{lem}\label{T1 cancellation} 
Let $\Phi_t=\id+ t\psi$ with $\psi$ as in Proposition \ref{transport} with \eqref{estxi} at order  $ 5.$
If $t$ is small enough that $\|t\psi\|_{C^1}<\hal$ and $|t| \ell^{\d-\s}\M$ is smaller than a constant, then 
$$T_1=  \beta N^{1-\frac{\s}{\d}}t^2 \int_{\R^\d} u(x)  d\fluct_{\muv}(x)$$
where  $u$ satisfies
\be \label{pw control}
|u(x)|\lesssim \M^2\begin{cases} \ell^{2\d} \max(|x-z|,\ell)^{-3\d+\s}& \text{if} \ x\in U\\
\ell^{2\d} |x-z|^{-2(\s+2)}& \text{if} \ x\in U^c.\end{cases}\ee
and
\begin{equation}\label{der control}
\|Du\|_{L^\infty}\lesssim  \frac{\M^2}{\ell^{\d-\s+1}}.\end{equation}
%\cm{we could have gotten better for $\int |\nab u|^2$}
\end{lem}

 \begin{proof} 
  The approach is similar to \cite[Lemma 4.14, Lemma 5.5]{P24}.

We use Taylor's formula to write ($m$ denoting multi-indices of length $2$)
\begin{align*}
& \int_{\R^\d}( \g(\Phi_t(x)- \Phi_t(y)) - \g(x-y))d\muv(y) +   (V_t\circ \Phi_t- V) (x)) 
\\ & =   t\(  \int \nab \g(x-y)\cdot  (\psi(x)-\psi(y)) d\muv(y) +  (\nab V\cdot \psi + \varphi)(x)\)
\\ &+2t^2\int_{\R^\d}  \sum_{|m|=2} \frac{1}{m!}  (\psi(x)-\psi(y))^{m} \int_0^1(1-a) D^{ m} \g(x-y+at(\psi(x)-\psi(y)) \, da\, \muv(y) 
\\ &
+ 2t^2 \int \sum_{|m|=2} \frac{1}{m!} \psi(x)^{ m}  \int_0^1(1-a) D^{m} V(x+at\psi(x)) \, da+ t^2 \psi(x) \cdot \int   D\varphi(x+at\psi(x))\, da,\end{align*}
and denote $t^2 u(x)$ the sum of the last two lines.

As in \eqref{2.7}--\eqref{lowtemp}, 
we then may rewrite 
\begin{equation*}
  \int \nab \g(x-y)\cdot  (\psi(x)-\psi(y)) d\muv(y) +  \nab V\cdot \psi + \varphi
= \psi \cdot \nab \zeta_V - h^{\div (\psi \muv)} +\varphi\end{equation*}
which vanishes by our definition of $\psi$ (see  Proposition \ref{transport}). It follows that 
$$T_1=\beta N^{1-\frac\s\d}t^2\int_{\R^\d} u(x) d\fluct_{\muv}(x).$$ 
%\beta t  N^{1-\frac{s}{d}} \int f(x) d\Fluct_{\muv}(x)  + 

\noindent
{\bf Step 1: Pointwise control on $u$.}
Let \begin{equation}\label{grem}
u_1(x):= 2\sum_{|m|=2}\frac{1}{m!}\int(\psi(x)-\psi(y))^m\int_0^1(1-a)D^m \g(x-y+at(\psi(x)-\psi(y)))~da~d\muv(y).
\end{equation}
 A computation shows that
\begin{equation*}
D^{ij}\g(z)=\frac{-\s}{|z|^{\s+2}}\indic_{i=j}+\frac{\s(\s+2)z_iz_j}{|z|^{\s+4}} \quad \text{for} \ z\neq0
\end{equation*}
and so the $da$-integrand in \eqref{grem} can be reexpressed as 
\begin{multline*}
\int_0^1(1-a)\biggl( \frac{-\s}{|x-y+at(\psi(x)-\psi(y))|^{\s+2}} \\+\frac{\s(\s+2)(x-y+at(\psi(x)-\psi(y)))_i(x-y+at(\psi(x)-\psi(y)))_j}{|x-y+at(\psi(x)-\psi(y))|^{\s+4}}\biggr)~da.
\end{multline*}
Now, using $\|t\psi\|_{C^1}\leq \frac{1}{2}$ we can write $|x-y+at(\psi(x)-\psi(y))| \gtrsim |x-y|$ and control \eqref{grem}, using \eqref{cratio1},  by 
\begin{equation}\label{gexp control} |u_1(x)|\lesssim  \int \frac{|\psi(x)-\psi(y)|^2}{|x-y|^{\s+2}}\, d\muv(y)\lesssim   \M^2\begin{cases} \ell^{2\d} \max(|x-z|,\ell)^{-3\d+\s}& \text{if} \ x\in U\\
\ell^{2\d} |x|^{-2(\s+2)}& \text{if} \ x\in U^c.\end{cases}
\end{equation}
We can give a similar control, much more easily, for
\begin{equation}\label{Vexp}
u_2(x)=  2  \sum_{|m|=2} \frac{1}{m!} \psi(x)^m  \int_0^1(1-a) D^{m} V(x+at\psi(x)) ~ da+  \psi(x) \cdot \int_0^1  D\varphi(x+at\psi(x))~da.
 \end{equation}
 For $x \in \carr_{2\ell}$, using \eqref{tscale}, we can immediately bound
 \begin{equation*}
 |u_2(x)|\lesssim \M^2 \ell^{2(\s-\d+1)}+\frac{1}{\ell}\M^2\ell^{-\d+\s+1}\lesssim \M^2( \ell^{\s-\d + (\s-\d+2)}+  \ell^{\s-\d}) \lesssim  \M^2 \ell^{\s-\d}.
 \end{equation*}

 For $x \notin \carr_{2\ell}$, if $\|t\psi\|_{C^1}$ is small enough, $x+t\psi(x) \notin \supp \varphi$,  hence $ \int_0^1   D\varphi(x+at\psi(x))~da$ vanishes. Thus, we can immediately apply the decay of $\psi(x)$ to the remaining term to control
 \begin{equation*}
 |u_2(x)|\lesssim  |\psi(x)|^2 \end{equation*}
 finally giving  the decay bound 
 \be\label{Vexp control}
 |u_2(x)|
 \lesssim  \M^2\begin{cases} 
   \ell^{2\d} \max(|x-z|, \ell)^{-4\d+2\s+2} 
   & \text{if }x \in U\\
  \frac{\ell^{2\d}}{|x|^{2\s+ 4} }& \text{if} \ x\in U^c.
   \end{cases}
   \ee
   Taking the dominant terms in \eqref{gexp control} and \eqref{Vexp control}, and recalling that $u=u_1+u_2$, yields the decay  estimate in \eqref{pw control}.

\noindent {\bf Step 2: Derivative control on $u$.}
We will only need an $L^\infty$ control on $Du$.  A computation allows us to write
\begin{multline}\label{factored gexp}
\int (\g(\Phi_t(x)-\Phi_t(y))-\g(x-y))~d\muv(y)-t\int \nabla \g(x-y)\cdot(\psi(x)-\psi(y))~d\muv(y) \\
=\int \frac{1}{|x-y|^\s}\left(\g\(\frac{x-y}{|x-y|}+t\frac{\psi(x)-\psi(y)}{|x-y|}\)-1-t\nabla \g\(\frac{x-y}{|x-y|} \)\cdot \frac{\psi(x)-\psi(y)}{|x-y|}\right)~d\muv(y);
\end{multline}   
the utility of this computation is that the derivatives of $\g$ are uniformly bounded in a neighborhood of $\frac{x-y}{|x-y|}$. Differentiating with respect to any $x$ variable, we see that the derivative either falls on $\frac{1}{|x-y|^\s}$ or the function in parentheses.

Let us do the latter first. We can write the function in parentheses as 
\begin{equation}\label{grem factored}
2t^2\sum_{|m|=2}\frac{1}{m!}\(\frac{\psi(x)-\psi(y)}{|x-y|}\)^m\int_0^1(1-a)D^m \g\(\frac{x-y}{|x-y|}+at\(\frac{\psi(x)-\psi(y)}{|x-y|}\)\)~da.
\end{equation}
A computation, using that the derivatives of $\g$ are uniformly bounded near $\frac{x-y}{|x-y|}$ and a mean value bound, yields that for $|x-y|\le \ell$, 
\begin{multline*}
\left|D\left(2t^2\sum_{|m|=2}\frac{1}{m!}\(\frac{\psi(x)-\psi(y)}{|x-y|}\)^m\int_0^1(1-a)D^m \g\(\frac{x-y}{|x-y|}+at\(\frac{\psi(x)-\psi(y)}{|x-y|}\)\)~da\right)\right|\\
\lesssim t^2 \|\psi\|_{C^1}\|\psi\|_{C^2}+t^2\|\psi\|_{C^1}^2+t^3\|\psi\|_{C^1}^2\|\psi\|_{C^2} \lesssim t^2\|\psi\|_{C^1}\|\psi\|_{C^2},
\end{multline*}
where we have used $\|t\psi\|_{C^1}\lesssim 1$. For $|x-y|\geq \ell$, we do not apply a mean value control to quotients of the form $\frac{\psi(x)-\psi(y)}{x-y}$ (except for those coming from the chain rule on $D^m\g$) and instead bound
\begin{multline*}
\left|D\left(2t^2\sum_{|m|=2}\frac{1}{m!}\(\frac{\psi(x)-\psi(y)}{|x-y|}\)^m\int_0^1(1-a)D^m \g\(\frac{x-y}{|x-y|}+at\(\frac{\psi(x)-\psi(y)}{|x-y|}\)\)~da\right)\right| \\
\lesssim t^2\frac{\|\psi\|_{L^\infty}^2}{|x-y|^3}+t^2\frac{\|\psi\|_{L^\infty}^2}{|x-y|^2}+t^3\frac{\|\psi\|_{L^\infty}^2\|\psi\|_{C^1}}{|x-y|^3}\lesssim  t^2\frac{\|\psi\|_{L^\infty}^2}{|x-y|^3}
\end{multline*}
where we have again used $\|t\psi\|_{C^1}\lesssim 1$ and $|x-y|\lesssim 1$. An explicit analysis of this computation in one dimension can be found in \cite[Appendix C]{P24}. Integrating these bounds, we find for any $i$,
\begin{align*}&\left|
\int \frac{1}{|x-y|^\s}\partial_i\left(\g\(\frac{x-y}{|x-y|}+t\frac{\psi(x)-\psi(y)}{|x-y|}\)-1-t\nabla \g\(\frac{x-y}{|x-y|} \)\cdot \frac{\psi(x)-\psi(y)}{|x-y|}\right)~d\muv(y)\right| \\
&\lesssim \int_{|x-y|\leq \ell}\frac{t^2\|\psi\|_{C^1}\|\psi\|_{C^2}}{|x-y|^\s}~d\muv(y)+\int_{|x-y|\geq \ell}t^2\frac{\|\psi\|_{L^\infty}^2}{|x-y|^{3+\s}}\, d\muv(y) \\
&\lesssim \frac{t^2\M^2}{\ell^{ 2\d-2\s+1 }}\int_0^\ell \frac{1}{r^\s}r^{\d-1}~dr+\frac{t^2\M^2}{\ell^{2\d-2\s-2}}\int_\ell^\infty \frac{1}{r^{3+\s}}r^{\d-1}~dr \\
&\lesssim \frac{t^2\M^2}{\ell^{\d-\s+1}}.
\end{align*}
 Now, returning to \eqref{factored gexp}, we also need to deal with terms of the form 
   \begin{equation*}
   \int \partial_i\left(\frac{1}{|x-y|^\s}\right)\left(\g\(\frac{x-y}{|x-y|}+t\frac{\psi(x)-\psi(y)}{|x-y|}\)-1-t\nabla \g\(\frac{x-y}{|x-y|} \)\cdot \frac{\psi(x)-\psi(y)}{|x-y|}\right)\, d\muv(y)
   \end{equation*}
   where $\partial_i$ indicates differentiation with respect to the $i$th component of the $x$ variable.   This can be dealt with via integration by parts, using that the integrand is symmetric in $x$ and $y$ and that $\muv(y)$ vanishes on $\partial \Sigma$; in the second and third lines of the following environment, $\partial_i$ denotes differentiation in $y$:
   \begin{align*}
   &   \int \partial_i\left(\frac{1}{|x-y|^\s}\right)\left(\g\(\frac{x-y}{|x-y|}+t\frac{\psi(x)-\psi(y)}{|x-y|}\)-1-t\nabla \g\(\frac{x-y}{|x-y|} \)\cdot \frac{\psi(x)-\psi(y)}{|x-y|}\right)~d\muv(y) \\
   &=   \int\left(\frac{1}{|x-y|^\s}\right) \partial_i\left(\g\(\frac{x-y}{|x-y|}+t\frac{\psi(x)-\psi(y)}{|x-y|}\)-1-t\nabla \g\(\frac{x-y}{|x-y|} \)\cdot \frac{\psi(x)-\psi(y)}{|x-y|}\right)~d\muv(y) \\
   &+ \int\left(\frac{1}{|x-y|^\s}\right) \left(\g\(\frac{x-y}{|x-y|}+t\frac{\psi(x)-\psi(y)}{|x-y|}\)-1-t\nabla \g\(\frac{x-y}{|x-y|} \)\cdot \frac{\psi(x)-\psi(y)}{|x-y|}\right)\partial_i \muv(y)~dy.
   \end{align*}
   The second line is what we have just controlled. Furthermore, since $\muv(y)\sim s(y) \dist(y,\partial \Sigma)^{1-\alpha}$ by \eqref{eq reg}, $\partial_i \muv(y) \sim s(y) \dist(y, \partial \Sigma)^{-\alpha}$ is integrable at the boundary. Hence, we can repeat the proof of \eqref{gexp control} (which never used the explicit density $\muv$, just that it was integrable at $\partial \Sigma$) and retrieve the same estimate. In particular, 
   
\begin{multline}\label{gexp derivative control}
\left\|D\int (\g(\Phi_t(x)-\Phi_t(y))-\g(x-y))~d\muv(y)-t\int \nabla \g(x-y)\cdot(\psi(x)-\psi(y))~d\muv(y) \right\|_{L^\infty}\\
\lesssim \frac{t^2 \M^2}{\ell^{\d-\s+1}}.
\end{multline}

The remaining terms  are much easier to control. These terms are 
\begin{equation*}
V_t \circ \Phi_t(x)-V(x)-t\nabla V \cdot \psi(x)-\varphi(x).
\end{equation*}
Differentiating the expansion \eqref{Vexp} yields control
\begin{multline*}
\left\|D\left(V_t \circ \Phi_t(x)-V(x)-t\nabla V \cdot \psi(x)-\varphi(x)\right)\right\|_{L^\infty} \\
\lesssim t^2\|\psi\|_{L^\infty}\|D\psi\|_{L^\infty}+t^2\|\psi\|_{L^\infty}^2+t^3\|\psi\|_{L^\infty}^2\|D\psi\|_{L^\infty}+t^2\|D\psi\|_{L^\infty}\|\varphi\|_{C^1} \\+t^2\|\psi\|_{L^\infty}\|\varphi\|_{C^2}+t^3 \|\psi\|_{L^\infty}\|\varphi\|_{C^2} \|\psi\|_{C^1}
\end{multline*}
which controlling $\|t\psi\|_{C^1}\lesssim 1$ allows us to simplify, using  \eqref{tscale} and \eqref{estxi}, as 
\begin{multline*}
\left\|D\left(V_t \circ \Phi_t(x)-V(x)-t\nabla V \cdot \psi(x)-\varphi(x)\right)\right\|_{L^\infty} \\
\lesssim t^2\|\psi\|_{L^\infty}\|D\psi\|_{L^\infty}+t^2\|\psi\|_{L^\infty}^2+\frac{t^2}{\ell}\|D\psi\|_{L^\infty}\M +\frac{t^2}{\ell^2}\|\psi\|_{L^\infty}\M \\
\lesssim \frac{t^2\M^2}{\ell^{2\d-2\s-2}}+\frac{t^2\M^2}{\ell^{\d-\s+1}}
\lesssim \frac{t^2\M^2}{\ell^{\d-\s+1}},
\end{multline*}
where we have used that $2\d-2\s-1\le \d-\s+1$. Coupling this with \eqref{gexp derivative control} yields the $\|Du\|_{L^\infty}$ bound in \eqref{der control}.

 \end{proof}

 Notice that in the mesoscopic case, $\|t\psi\|_{C^1}$ is small for $\ell N^{\frac{1}{\d}}\rightarrow +\infty$, since by \eqref{tscale}
\begin{equation*}
t\|\psi\|_{C^1}\lesssim N^{-1+\frac{\s}{\d}}\ell^{-\s+\d}=\left(N^{\frac{1}{\d}}\ell \right)^{-(\d-\s)}\rightarrow0.
\end{equation*} 
Once we have these scaling estimates, we can estimate $T_1$. 
\begin{lem}\label{lem: t1}
Let $\Phi_t=\id+ t\psi$ with $\psi$ as in Proposition \ref{transport} with \eqref{estxi} satisfied  at order  $ 5.$
Assume that $\XN\in \mathcal G_\ell$, with $\Gc_\ell$ as in Proposition \ref{pf loiloc} and Lemma \ref{lemsubdi}, a set of configurations such that the local laws \eqref{loiloc}--\eqref{controlnbreintro} hold for each $D_k$ in \eqref{defAk} (in $\bulk$) and \eqref{macrolaw2} holds.
Then, for $N$ large enough,  we have \begin{equation}\label{t1 bound}
|T_1|\lesssim_\beta  \beta  t^2  N^{2-\frac{\s}{\d}}  \M^2\times \begin{cases}
\ell^\s\( \(\ell N^{\frac{1}{\d}}\)^{\frac{(\d-\s)(\s-2\d)}{2(3\d-\s)}} + \ell^{2\d-\s}\) & \text{if }\ell \lesssim N^{\frac{\s-\d}{\d(7\d-3\s)}} \\
 \ell^{2\d} \( (N^{\frac1\d}\ell)^{\frac{\d-\s}{2}} \ell^{3\d-\s}\)^{\frac{-2\s-4}{\frac{\d}{2}+2\s+4}} & \text{otherwise}.
\end{cases}\end{equation}
In either case, we can write the (suboptimal) bound
\begin{equation}\label{applied t1 bound}
|T_1|\lesssim_\beta t^2N^{2-\frac{\s}{\d}}\M^2\ell^\s (1+ (N^{\frac1\d}\ell)^{-\sigma})
\end{equation}
for some $\sigma>0$.
%\cm{resp. $\beta t^2 N^{2-\frac\s\d} \|\varphi_0\|_{C^5}^2 \ell^{2\d} \( (N^{\frac1\d}\ell)^{\frac{\d-\s}{2}} \ell^{3\d-\s}\)^{\frac{-2\s-4}{\frac{\d}{2}+2\s+4}}$} 
\end{lem}

% This shows that for our choice of $t \sim N^{-1+\frac{\s}{2\d}}\ell^{-\frac{\s}{2}}$, $T_1=o(1)$ on $\mathcal{G}_\ell$.
\begin{proof}% First, let us treat the case where $\ell\ge \ep$, with $\ep$ as in the definition of $\bulk$.

The idea is to use Proposition \ref{pro:controlfluct} on the  cube $\carr_R(z)$ with $R>2\ell$ to be determined below, and estimate more easily with Lemma \ref{RoughFluct} outside.
 Notice that, for  $u$ as in Lemma \ref{T1 cancellation},
\begin{equation*}
\|u\|_{L^2(\carr_R)}^2 \lesssim \frac{\M^4R^\d}{\ell^{2(\d-\s)}}
\end{equation*}
and 
\begin{equation*}
\|\nabla u\|_{L^2(\carr_R)}^2 \lesssim \frac{\M^4R^\d}{\ell^{2(\d-\s)}\ell^2}.
\end{equation*}
If $\ell<1$ and $R<\ep$ (as in the definition of $\bulk$) then the local laws holds in $\carr_R$,  if not we use the global law \eqref{macrolaw2}.
Using Proposition \ref{pro:controlfluct} (applied with $\eta=\ell$) coupled with the local / macroscopic  law depending on the case, we obtain
\begin{equation*}
\left|\int_{\carr_R} u(x)\, d\fluct_{\muv}(x)\right|^2 \lesssim_\beta\(\ell^{\gamma-1}\frac{\M^4R^\d}{\ell^{2(\d-\s)}}+\ell^{\gamma+1}\frac{\M^4R^\d}{\ell^{2(\d-\s)}\ell^2}
\)\min (1,R^\d) N^{1+\frac{\s}{\d}}
\end{equation*}
where we have also used \eqref{discest} (or its consequence \eqref{controlnbreintro}) to show that the additive  error term in \eqref{rieszfluct2} is strictly smaller.  We find 
\begin{multline}\label{int1err}
\left|\int_{\carr_R} u(x)\, d\fluct_{\muv}(x)\right|\lesssim_\beta \( \ell^{\gamma-1+2\s-2\d}\M^4\min (R^\d, R^{2\d})N^{1+\frac{\s}{\d}}\)^{\frac{1}{2}}\\
\lesssim_\beta \M^2\ell^{\frac{3}{2}(\s-\d)}\min (R^{\d/2}, R^\d) N^{\frac{1}{2}+\frac{\s}{2\d}}
\end{multline}
using $\d-1+\gamma=\s$. On the other hand, using Lemma \ref{T1 cancellation} coupled with the local laws on dyadic scales $\ge 2 \ell$, as long as $2^k \ell <\ep$, or \eqref{macrolaw2} otherwise,  and Lemma \ref{RoughFluct}, we obtain
\begin{multline*}
\left|\int_{\carr_R^c \cap U } u(x)\, d\fluct_{\muv}(x)\right|\lesssim_\beta \sum_{k \geq \log_2\frac{R}{\ell}} N\(\min(1, 2^k\ell)\)^\d\frac{\ell^{2\d} \M^2}{(2^k\ell)^{3\d-\s}}\\
 \lesssim_\beta N \ell^{\s}\M^2\sum_{k \ge \log_2 \frac{R}{\ell}}\min \(\ell^{-\d}( 2^{\s-3\d})^k , (2^{\s-2\d})^k \)\lesssim_\beta N \ell^{\s}\M^2\(\frac{\ell}{R}\)^{2\d-\s}\min (1, R^{-\d})\\
 \lesssim_\beta N \M^2\ell^{2\d}\min (R^{\s-2\d} , R^{\s-3\d}).
\end{multline*}
Note that if $R$ exceeds a large enough constant, then $\carr_R^c \cap U=\varnothing$, so the integral vanishes, and we can thus replace the result by 
\begin{equation}\label{ext1err2}
\left|\int_{\carr_R^c \cap U } u(x)\, d\fluct_{\muv}(x)\right|\lesssim_\beta  N \M^2\ell^{2\d} R^{\s-2\d}.
\end{equation}

%\cm{this seems incorrect now. estimate is 
%$$\sum N \(2^k\ell\)^\d\frac{\ell^\d \|\varphi_0\|_{C^4}^2}{\ell^{\d-\s}(2^k)^{3\d -\s}}$$}

Finally, using a very crude bound and \eqref{pw control}, we have
\begin{equation}\label{ext2err}
\left|\int_{Q_R^c \cap U^c } u(x)\, d\fluct_{\muv}(x)\right|\lesssim_\beta N \|u\|_{L^\infty(U^c)} \lesssim N 
\ell^{2\d} \min (1, R^{-2(\s+2)})\M^2.
\end{equation}

We then choose 
\begin{equation}\label{Rdef}
R:=\ell \(\ell N^{\frac{1}{\d}}\)^{\frac{\d-\s}{2(3\d-\s)}}
\end{equation}
if this is such that $R $ is smaller than a constant, which happens when $\ell \lesssim N^{\frac{\s-\d}{\d(7\d-3\s)}}$.  Then 
\eqref{int1err} and \eqref{ext1err2} balance and $R \gg \ell$. Indeed, \begin{equation*}
%\ell^{3(\s-\d)}R^{2\d}N^{1+\frac{\s}{\d}}=\ell^{3(\s-\d)}R^{2\d-2\s-4}N^{1+\frac{\s}{\d}}R^{2\s+4} \\
%=\ell^{3(\s-\d)}R^{2\d-2\s-4}N^{1+\frac{\s}{\d}}\ell^{\s+\d+4}N^{1-\frac{\s}{\d}}=N^2 \ell^{4\s-2\d+4}R^{2\d-2\s-4}
\ell^{\frac{3}{2}(\s-\d)}R^\d N^{\frac{1}{2}+\frac{\s}{2\d}}=N\ell^{2\d}R^{\s-2\d}\iff R^{3\d-\s}=\ell^{\frac{7}{2}\d-\frac{3}{2}\s}N^{\frac{1}{2}-\frac{\s}{2\d}}=\ell^{3\d-\s}\(\ell N^{\frac{1}{\d}}\)^{\frac{\d-\s}{2}}.
\end{equation*}
We arrive at the desired result in that case.

Otherwise, we optimize the sum of \eqref{int1err} and \eqref{ext2err} and take 
\begin{equation}\label{Rdef2}
R= \( (N^{\frac1\d}\ell)^{\frac{\d-\s}{2}} \ell^{3\d-\s}\) ^{\frac{1}{\frac\d2+2\s+4}};
\end{equation}
notice that this is $\gtrsim 1$ precisely when $\ell \gtrsim N^{\frac{\s-\d}{\d(7\d-3\s)}}$, which is the regime under consideration. Then, for $N$ large enough we have $\carr_R^c \cap U=\varnothing$.
%\cm{clearly when $\ell$ is small we should choose \eqref{Rdef}, and when it is close to 1 we should choose \eqref{Rdef2}, but we need to check that all values of $\ell$ are accounted for and the result is consistent.}
In that case we also obtain the result by substituting \eqref{Rdef2} into \eqref{int1err} and \eqref{ext2err}.

%\begin{equation}\label{Rsimp}
%R^{2\s+4}=\ell^{\s+\d+4}N^{1-\frac{\s}{\d}}=\ell^{2\s+4}\ell^{\d-\s}N^{1-\frac{\s}{\d}}=\ell^{2\s+4}\(\ell N^{\frac{1}{\d}}\)^{\d-\s}\gg \ell^{2\s+4}
%\end{equation}

\end{proof}
%\cm{maybe suppress}
%\begin{prop}\label{first T1 control}
%Suppose that $\varphi_0 \in C^4$ is a compactly supported test function, and let $\varphi$ denote the associated rescaled test function at scale $\ell$ at least microscopic with support included in $\Sigma$ and at distance $\ge \ep>0$ from $\partial \Sigma$. Then, for every $N$ sufficiently large and $\XN \in \mathcal{G}_{\mathsf{M}}$,
%\cm{probably it should be on $\mathcal{G}_\ell$}
%\begin{equation*}
%\left|\int \beta N^{1-\frac{\s}{\d}}t^2  u(x)\, d\fluct_{\muv}(x)\right|\lesssim \beta N^{2-\frac{\s}{\d}}t^2\ell^s\|\varphi_0\|_{C^4}^2
%\end{equation*}
%\end{prop}
%\begin{proof}
%The proof is a direct application of Lemma \ref{RoughFluct} and Proposition \ref{T1 cancellation}. We have
%\begin{equation*}
%\left|\int_{Q_{2\ell}}\beta N^{1-\frac{\s}{\d}}t^2  u(x)\, d\fluct_{\muv}(x)\right|\lesssim \beta N^{1-\frac{\s}{\d}}t^2\frac{\|\varphi_0\|_{C^4}^2}{\ell^{\d-\s}}\ell^\d N=\beta N^{2-\frac{\s}{\d}}t^2\ell^\s\|\varphi_0\|_{C^4}^2.
%\end{equation*}
%Using \eqref{Lone} on $J^c$ yields
%\begin{align*}
%\left|\int_{Q_{2\ell}^c} \beta N^{1-\frac{\s}{\d}}t^2  u(x)~d\fluct_{\muv}(x)\right|&\lesssim \beta N^{1-\frac{\s}{\d}}t^2\ell^{\d+2-4\alpha}\|\varphi_0\|_{C^3}^2\int_{Q_{2\ell}^c}\frac{1}{|x-z|^{\s+2}}~dx \\
%&\lesssim  \beta N^{1-\frac{\s}{\d}}t^2\ell^{\d+2-4\alpha}\|\varphi_0\|_{C^3}^2\ell^{\d-\s-2} \\
%&=\beta N^{2-\frac{\s}{\d}}t^2\ell^\s\|\varphi_0\|_{C^3}^2
%\end{align*}
%for $\s>\d-2$ (Riesz). \cm{what about Coulomb?}
%\end{proof}
%
%
%

\subsection{Proof of Theorem \ref{FirstFluct}.}
 As in the Coulomb case, once we have control of the $T_0$ term via explicit controls of the transport, of the $T_1$ term via bounds on fluctuations, and of the $T_2$ term  from the commutator estimate, we have a first bound on the fluctuations, which is Theorem~\ref{FirstFluct}. We take $\Gc_\ell$ as in Proposition \ref{pf loiloc} and Lemma \ref{lemsubdi}, a set of configurations such that the local laws \eqref{loiloc}--\eqref{controlnbreintro} hold for each $D_k$ in \eqref{defAk} (in $\bulk$) and \eqref{macrolaw2} holds, which we can assume satisfies $
 \PNbeta(\Gc_\ell^c)\leq  C_1e^{-C_2\beta \ell^\d N}$,
 up to adjusting the definitions of $C_1$ and $C_2$.

Combining  \eqref{splitt},  \eqref{bt0},  \eqref{applied t1 bound} and \eqref{firstt2bound}, we are led to 
 \begin{align*}
&\left|\log \Esp_{\PNbeta}\( e^{-\beta  t N^{1-\frac\s\d} \Fluct_{\muv}(\varphi)} \indic_{\mathcal G_\ell}\)-
 \frac{\beta  N^{2-\frac\s\d} t^2}{2}  \frac{c_{\d,\frac{\d-\s}{2}}}{2\cds}\left\|\varphi^\Sigma\right\|_{\dot{H}^\frac{\d-\s}{2}}^2
-t N   M(\varphi) \right|\\ &\lesssim_\beta\( t^3 N^{2-\frac\s\d}\beta\M^3\ell^{2\s-\d} + t^2 (1+\beta)\M^2N\ell^{2\s-\d}\)  
+\beta t^2 N^{2-\frac\s\d}\M^2 \ell^\s (1+ (N^{\frac1\d}\ell)^{-\sigma})
+\beta \M|t| N\ell^\s,\end{align*}
for some $\sigma>0$, with $M(\varphi)$ is as in \eqref{defmean}.
We then let  $t=-\frac{1}{1+\beta} N^{-1+\frac{\s}{\d}}\tau$, and note that the condition $t\ell^{\s-\d}\M$ small enough that was needed for our proofs amounts to $\tau (\ell N^{\frac1\d})^{\s-\d}\M$ small enough.  In view of the definition \eqref{defmean} and Lemma~\ref{variancemeso}, we obtain the result under this assumption.

\subsection{H\"older trick for the CLT}

We can now turn to the proof of the Central Limit Theorem for fluctuations.  We now assume that $\varphi = \varphi_0(\frac{\cdot-z}{\ell})$, which implies \eqref{estxi} with $\M= \|\varphi_0\|_{C^{k+1}}$. 
The goal of this section is to improve the estimate on $T_2$, assuming that we have an expansion of the relative free energy with a good enough rate. The method of proof, introduced in \cite{LS18} consists in comparing this  expansion of the difference of free energies with the relative expansion obtained by transport in \eqref{09}.

%We make use of energy estimates and the partition function expansion proved in the previous section to improve our error control here.
We next assume that we have an expansion of the form \eqref{2eway0 intro}, that is 
\begin{multline} \label{2eway0}
\log \K_{N,\beta}^{\mathcal{G}_\ell}( \mu_t,\zeta_V\circ \Phi_t^{-1})- \log \K_{N,\beta}(\mu_0,\zeta_V) +N(\Ent(\mu_t)-\Ent(\mu_0) ) \\= N\( \mathcal{Z} (\beta, \mu_t)-\mathcal{Z}(\beta, \mu_0) \) + O((1+\beta) N\ell^\d \mathcal R_t)\end{multline}Here,  $\mathcal{Z}$ is as in \eqref{defcalZ}, $\K^{\mathcal G}_{N,\beta}$ is as in \eqref{defK fluct event}, 
 $\mathcal R_t$ is the error rate, and $\mf$ is the pressure for the unit density system defined in Lemma \ref{lem: const dens}. Indeed, this is precisely the expansion that we will prove in Proposition \ref{Kcomp} leveraging the local laws of Theorem \ref{Local Law} down to microscopic scales. We will analyze when the rate $\mathcal{R}_t$ we obtain is sufficient at the end of this section. 

Let us record some information about the function $\mathcal{Z}$.

\begin{lem} 
Denote $\Phi_t=\id + t\psi$, $\mu_t:= \Phi_t\#\muv$ where $\psi$ is the transport map constructed in Proposition~\ref{transport},  with $\varphi_0 \in C^5$.  Assume that the function $y \mapsto \mf(y) $ is $p$ times differentiable and satisfies \eqref{425}. Letting then  $B_k(\beta, \mu, \psi)$ be  the $k$-th derivative at $t=0$ of  the function $\phi(t):=\mathcal{Z}(\beta , \mu_t)$, if $|t|\ell^{\s-\d}\|\varphi_0\|_{C^5}$ is small enough, we have 
\be\label{tayexp}
\mathcal{Z}(\beta,  \mu_t) - \mathcal{Z}(\beta, \mu_0)= \sum_{k=1}^{p-1}\frac{ t^k}{k!} B_k(\beta, \mu_0, \psi) + O\(t^p  \beta  \|\varphi_0\|_{C^4}^p\ell^{(1-p) \d+p\s}\).\ee
and 
\be \label{borneB}
|B_k(\beta, \mu_0, \psi)|\lesssim_\beta \beta \|\varphi_0\|_{C^4}^k\ell^{(1-k) \d+k\s}.\ee 
 We also have the explicit expression 
\begin{multline}\label{expliB1}
B_1(\beta, \muv, \psi)= 
-\frac{\beta}{\cds} (1+\frac\s\d) \int (-\Delta)^{\alpha}\varphi^\Sigma \mf(\beta\muv^{\frac\s\d})\muv^{\frac\s\d} -\frac{\beta}{\cds} \frac\s\d\int \mf'(\beta \muv^{\frac\s\d}) \muv^{2\frac\s\d}(-\Delta)^{\alpha}\varphi^\Sigma \\
+\frac{1}{\cds}\frac\beta{2\d} \indic_{\s=0}\int(-\Delta)^\alpha \varphi^\Sigma \log \muv.\end{multline}

\end{lem}
\begin{proof}
We  may write \begin{equation*}
\int \beta \mu_t^{1+\frac{\s}{\d}}\mf(\beta  \mu_t^{\frac{\s}{\d}})=
\int \beta  \mu_t^{\frac{\s}{\d}}\mf(\beta \mu_t^{\frac{\s}{\d}}) \Phi_t\# \mu_0
= \int \beta \(  \mu_t\circ \Phi_t\)^{\frac{\s}{\d}} \mf\(\beta ( \mu_t\circ \Phi_t)^{\frac{\s}{\d}}\) d\mu_0.
\end{equation*}
Next we recall that by definition of the push forward we have 
\be\label{explicitmu} \mu_t\circ \Phi_t= \frac{\muv}{\det(\id+ tD\psi)},\ee
 hence if $t|D\psi|<\hal$, which in view of \eqref{tscale} is implied by $ |t|\ell^{\s-\d} \|\varphi_0\|_{C^5}$ small enough, we may bound 
\be \left|\frac{d^j}{dt^j}  \mu_t \circ \Phi_t\right|\le |D\psi|^j .\ee
We also set $g(x)= \beta x^{\frac{\s}{\d}} \mf(\beta x^{\frac{\s}{\d}})$ 
and check that by the  assumption \eqref{425}, we have $|g^{(n)}(y)|\le C\beta $ for $n\le p$ when $y$ takes values $ \beta \muv(x)^{\frac\s\d}$.
Using the Faa di Bruno formula, we now have 
\begin{equation*}
\frac{d^k}{dt^k} g(  \mu_t \circ \Phi_t)\\ = \sum_{j_1+2j_2+\dots + kj_k=k} C_{\mathbf{j}} g^{(j_1+\dots +j_k)} (  \mu_t \circ \Phi_t)\prod_{l=1}^k \( \frac{d^l}{dt^l}  \mu_t \circ \Phi_t\)^{j_l} \end{equation*}
where $C_{\mathbf{j}}$ is some combinatorial factor, and inserting the above estimates we deduce that
$$\left|\frac{d^k}{dt^k} g(  \mu_t \circ \Phi_t)\right|\lesssim\beta  |D\psi|^{ k},$$ with a uniform constant over $\beta \ge 1$. In the same way 
$$\int \mu_t\log \mu_t=\int \log ( \mu_t\circ\Phi_t) d\muv$$
and the derivatives of $\log (\mu_t\circ \Phi_t) $ are bounded by $C   |D\psi|^{ k}.$
Integrating against  $ d\muv$ on the support of $\psi$, and using \eqref{cdpsi},  we deduce that 
\be \label{borneB0}  |\phi^{(k)}(t)| \le   C\beta   \int_\Sigma  |D\psi|^{ k}\lesssim  \beta \|\varphi_0\|_{C^5}^k\ell^{(1-k) \d+k\s}, \ee with a constant that is uniform for $\beta \ge 1$.
 The result \eqref{borneB} follows by Taylor expansion.
For \eqref{expliB1}, a direct calculation using that $\partial_t \mu_t|_{t=0}=-\div (\psi \mu_0)$ yields 
\begin{align*}
B_1(\beta, \muv, \psi)= &
\beta (1+\frac\s\d) \int \div (\psi \muv) \mf(\beta\muv^{\frac\s\d})\muv^{\frac\s\d} +\beta \frac\s\d\int \mf'(\beta \muv^{\frac\s\d}) \muv^{2\frac\s\d}\div (\psi \muv) \\
&-\frac\beta{2\d} \indic_{\s=0}\int \div(\psi\muv) \log \muv.\end{align*}
Inserting \eqref{divpsimu0} and \eqref{Rtransport}, we obtain the result.
\end{proof}

We can now obtain our main result on the expansion of partition functions relevant to $T_2$.

\begin{prop} Assume the same hypotheses as in the previous lemma.
Let $\mu_t= (\id+t\psi)\# \muv$ as in the previous lemma. Let $\mathcal{G}_\ell$ be as in Proposition \ref{pf loiloc} and Lemma \ref{lemsubdi}, a set of configurations such that the local laws \eqref{loiloc}--\eqref{controlnbreintro} hold for each $\carr_{2^k\ell}$ with $k \geq 1$ (in $\bulk$) and \eqref{macrolaw2} holds.
Assume we know that for each $r \le \ell^{\d-\s}$, we have
\be \label{relexp}\log \frac{\K_{N,\beta}^{\mathcal{G}_\ell}(\mu_r,\zeta_V\circ\Phi_r^{-1})}{\K_{N,\beta}(\mu)}+ N(\Ent(\mu_r)-\Ent(\mu_0)) =  N\( \mathcal{Z}(\beta,\mu_r)- \mathcal{Z}(\beta,\mu_0)\)+O(( \beta+1) N\ell^\d \mathcal R_r)\ee
with 
$\max_{|r|\le \ell^{\d-\s}} \mathcal R_r \le C$ and $\mathcal R_r$ continuous in $r$ for $|r|\le \ell^{\d-\s}$.
For any integer $p\ge 1$ such that $\varphi_0 \in C^{2p+3}$,
for every $t$ such that  $|t| \ell^{\d-\s} \|\varphi_0\|_{C^{2p+3}} \( \max_{|r|\le \ell^{\d-\s}} \mathcal R_r\)^{-1/p}$ is smaller than a small enough constant (depending only on $\d,\s, \muv,p$), $a$ being the largest number $\le \ell^{\d-\s}$ such that 
\be \label{defa}
a= \frac{c}{\|\varphi_0\|_{C^{2p+3}}} \( \max_{[-a,a]}\mathcal R_r \)^{\frac1p}   \ell^{\d-\s} ,
\ee for some $c>0$ small enough,
  we have 
\be
\label{hold0} \left|\log \Esp_{\PNbeta} \( \exp\Big( \sum_{k=1}^p \gamma_k t^k\Big)\indic_{\mathcal{G}_\ell} \) \right|\le C |t|
(\beta +1) N\ell^\s \|\varphi_0\|_{C^{2p+3}}\(\max_{|r|\le a}\mathcal R_r\)^{1-\frac1p }
\ee where $\gamma_k=\beta N^{-\frac{\s}{\d}} 
\Ani_k(\XN, \mu_0, \psi) - \frac{N}{k!}B_k(\beta, \mu_0, \psi)$.
Moreoever, 
\begin{multline}\label{pconc}
\left|\log \Esp_{\PNbeta}\( e^{T_2} \indic_{\mathcal G_\ell}\)-  N \sum_{k=1}^{p-1}\frac{t^k}{k!}B_k(\beta, \muv,\psi)\right|\\ \lesssim_\beta t (\beta +1) N\ell^{\s} \Big( \max_{|r|\le a}\mathcal R_r \Big)^{1-\frac1p}  +
t^p  \beta  \|\varphi_0\|_{C^{2p+3}}^p N\ell^{(1-p) \d+p\s},
\end{multline}
%\log \frac{\K_{N,\beta}^{\mathcal{G}_\ell}( \mu_t,\zeta\circ\Phi_t^{-1})}{\K_{N,\beta}(\mu_0,\zeta)}+N(\Ent(\mu_t)-\Ent(\mu_0))= \sum_{k=1}^{p-1}t^k  \frac{N}{k!}B_k(\beta, \mu_0,\psi )\\+ 
%O \( t (\beta +1) N\ell^{\s} \Big( \max_{|r|\le \ell^{\d-\s}}\mathcal R_r \Big)^{1-\frac1p}  \)+
%O\(t^p  \beta  \|\varphi_0\|_{C^{2p+3}}^p N\ell^{(1-p) \d+p\s}\).
 
 %\end{multline} 
 where the bound depends  on $p$ and the above bounds.
 
\end{prop}
\begin{proof}
%{\bf Step 1: Comparing two ways of expanding $\log 1$}.\\
Expanding the next-order energy to order $p$ via \eqref{dtF}, we have
\begin{multline}\label{diffFtransport2}
\F_N(\Phi_t(\XN), \Phi_t\#\mu_0)- \F_N(\XN, \mu_0)\\=\sum_{k=1}^{p-1} \frac{t^k}{k!}\Ani_k(\XN, \mu_0 , \psi) + \frac{1}{p!}\int_0^t(t-s)^{p-1}  \Ani_p(\Phi_s(\XN), \mu_s, \psi\circ \Phi_s^{-1}) ds.\end{multline}
In the same way 
\be \log \det D\Phi_t= \sum_{k=1}^{p-1} t^k c_k(\psi) + O(t^p |D\psi|^p),
\ee
for certain expressions $c_k(\psi)$.
Inserting into \eqref{09} and using \eqref{controlani higher} to bound $\Ani_p$, we obtain 
\begin{multline}\label{exp2forme}
e^{T_2}\indic_{\mathcal{G}_\ell}=
\exp\Bigg( -\beta N^{-\frac{\s}{\d}}  \Bigg(\sum_{k=1}^{p-1} t^k \Ani_k(\XN, \mu_0, \psi) + O_\beta\(|t|^p\|\varphi_0\|_{C^{2p+3}}^p\ell^{(1-p)\d+p\s}N^{1+\frac{\s}{\d}}\)\Bigg) \\
+ \sum_{k=1}^{p-1}t^k \Fluct_{\muv}(c_k(\psi)) + O\(  \Fluct_{\muv}\(t^p|D\psi|^p\) \) \Bigg)\indic_{\mathcal{G}_\ell}.
\end{multline}
Let us first bound $\Fluct_{\muv} (|D\psi|^p)$. For that we rewrite $\psi$ as $\sum_k (\chi_k\psi)$ where $\chi_k$ is a partition of unity relative to the $A_k$'s as in the proof of Lemma \ref{lemsubdi}. We use the rough bound  of Lemma \ref{RoughFluct} (and a rough bound by $N \|\psi\|_{L^\infty}$ in $A_{k_*+2}$), and obtain using \eqref{tscale} that 
\begin{align}
\Fluct_{\muv}\(|D\psi|^p\) \indic_{\mathcal G_\ell}&\lesssim  \sum_{k=0}^\infty \Fluct_{\muv}\(|D(\chi_k\psi)|^p\) \indic_{\mathcal G_\ell}\lesssim N \|\varphi_0\|_{C^5}^p\(
\sum_{k=0}^\infty\(\frac{\ell^\d}{(2^k\ell)^{2\d-\s}}\)^p  (2^k\ell)^\d+  \ell^\d\)\\
&
\lesssim \|\varphi_0\|_{C^5}^p
 N\( \ell^{\d(1-p)+\s p} +\ell^\d\)\lesssim \|\varphi_0\|_{C^5}^p
 N\ell^{\d(1-p)+\s p} .
\end{align}
Inserting this into \eqref{exp2forme} and inserting  \eqref{exp2forme}
into \eqref{09}, we obtain an expansion that we may equate with the  expansion \eqref{relexp} and \eqref{tayexp}.   Setting 
\be \label{defgk}
\gamma_k =  -\beta N^{-\frac{\s}{\d}} 
\Ani_k(\XN, \mu_0, \psi)+ \Fluct(c_k(\psi)) - \frac{N}{k!}B_k(\beta, \mu_0, \psi)\ee
this yields that  \begin{equation*}
\log \Esp_{\PNbeta} \( \exp\Big(  \sum_{k=1}^{p-1} t^k \gamma_k\Big)\indic_{\mathcal{G}_\ell}\)
=O_\beta\(|t|^p (1+ \beta) N \|\varphi_0\|_{C^{2p+3}}^p\ell^{(1-p) \d+p\s}\) +O\((1+\beta) N\ell^\d(\mathcal R_t+\mathcal R_0)\).
\end{equation*}
%Using Cauchy-Schwarz's inequality and the local laws \eqref{locallawint0} we then deduce
%\begin{equation*}
%\log \Esp_{\PNbeta} \( \exp\(  \sum_{k=1}^{p-1} t^k \gamma_k  \) \)
%= O\( t^p \beta\chi(\beta) N   \ell^{\d-2p}   \)+O\(\beta \chi(\beta) N\ell^\d(\mathcal R_0+\mathcal R_t)\).%\end{equation*}
We next wish to choose $a\le \ell^{\d-\s}$ such  that 
\begin{equation*}
a^p  \|\varphi_0\|_{C^{2p+3}}^p\ell^{(1-p)\d+p\s} 
\le C \ell^\d(\mathcal R_0+\mathcal R_a).
\end{equation*}
For that we choose 
\begin{equation*}
a=\sup\left\{ b\le \ell^{\d-\s},\  b \le \frac{c}{\|\varphi_0\|_{C^{2p+3}}} \( \max_{r\in[-b,b]}\mathcal R_r \)^{\frac1p}   \ell^{\d-\s} \right\}.
\end{equation*}
We note that $a <\ell^{\d-\s}$   if  $ \max_{r\in [-\ell^{\d-\s},\ell^{\d-\s}]}\mathcal R_r $ is bounded and $c>0$ is chosen small enough. Thus by continuity, we must have 
\be  a= \frac{c}{\|\varphi_0\|_{C^{2p+3}}} \( \max_{[-a,a]}\mathcal R_r \)^{\frac1p}   \ell^{\d-\s}. \ee

With this choice we then have 
\be\label{debut0}
\log \Esp_{\PNbeta} \( \exp\Big(  \sum_{k=1}^{p-1} a^k \gamma_k  \Big)\indic_{\mathcal{G}_\ell} \)= O_\beta\( (\beta+1) N\ell^\d \max_{|r|\le a}\mathcal R_r\).
\ee
We note by H\"older's inequality we have $\Esp(e^L) \Esp(e^{-L} )\ge 1$ thus we can transform \eqref{debut} into 
\be\label{debut}
\left|\log \Esp_{\PNbeta} \( \exp\Big( \pm \sum_{k=1}^{p-1} a^k \gamma_k  \Big)\indic_{\mathcal{G}_\ell} \)\right|= O_\beta\( (\beta+1) N\ell^\d \max_{|r|\le a}\mathcal R_r\).
\ee

Let now $X_1, \dots X_p$ be $p$ equally-spaced points in $[1/2,1]$, and let $P_i$'s be the Lagrange interpolation polynomials of degree $p-1$ associated to $X_1, \dots, X_p$, i.e.~such that $P_i(X_j)=\delta_{ij}$. We may expand each $P_i$ as $\sum_{n=0}^{p-1} c_{i, n} X^n$, where the coefficients $c_{i,n}$ depend only on $p$.
Expressing the polynomial $\sum_{n=1}^{p-1} \gamma_n a^n X^n$ along the $P_i$'s we obtain that 
$$\sum_{n=1}^{p-1} \gamma_n a^n X^n= \sum_{i=1}^p ( \sum_{k=1}^{p-1} \gamma_k a^k X_i^k) P_i= \sum_{n=0}^{p-1}\sum_{i=1}^p  ( \sum_{k=1}^{p-1} \gamma_k a^k X_i^k) c_{i,n} X^n.$$
Equating the coefficients, it follows that  for $1 \le n \le p-1$, 
\be\label{gammaaX}
\gamma_n a^n = \sum_{i=1}^{p} c_{i,n} \( \sum_{k=1}^p \gamma_k a^k X_i^k\).\ee
Note that \eqref{debut} is also true with $a$ replaced by $ aX_i$ (for $X_i\in [\hal, 1]$). 
Choosing $C$ a constant large enough (depending only on $p$) and using the generalized H\"older's inequality, we may write that  for each $1\le n\le p-1$, in view of \eqref{debut},
\begin{multline} \label{find0}
\log \Esp_{\PNbeta} \( \exp\( \frac{\gamma_n  a^n}{C}\)\indic_{\mathcal{G}_\ell} \) \le \sum_{i=0}^p |c_{i,n}| \left|\log \Esp_{\PNbeta}\( \exp\( \sum_{k=1}^{p-1} \gamma_k  a^k X_i^k\)\indic_{\mathcal{G}_\ell}\) \right|\\ \lesssim_\beta ( \beta+1)  N\ell^\d \max_{|r|\le a}\mathcal R_r,\end{multline}
where $C$ depends on $p$. The same result holds starting from  the opposite of \eqref{gammaaX}, that is we can also obtain 
\be \label{find'}
\log \Esp_{\PNbeta} \( \exp\( -\frac{\gamma_n  a^n}{C}\)\indic_{\mathcal{G}_\ell} \) \lesssim_\beta   ( \beta+1)  N\ell^\d \max_{|r|\le a}\mathcal R_r.\ee
Using again  $\Esp(e^L) \Esp(e^{-L} )\ge 1$, we obtain a two-sided bound 
\begin{equation} \label{find}
\left|\log \Esp_{\PNbeta} \( \exp\( \pm\frac{\gamma_n  a^n}{C}\)\indic_{\mathcal{G}_\ell} \) \right|   \lesssim_\beta  ( \beta+1)  N\ell^\d \max_{|r|\le a}\mathcal R_r,\end{equation}

Using H\"older's inequality again we deduce  that if  $|t|/a$ is small enough (in particular $<1$),
\begin{align}
\notag & \log \Esp_{\PNbeta} \( \exp\( \sum_{k=1}^{p-1} \gamma_k t^k\)\indic_{\mathcal{G}_\ell} \) \le 
\sum_{k=1}^{p-1} \left|\log \Esp_{\PNbeta} \( \exp\(  \mathrm{sgn}(t^k) \frac{\gamma_k a^k}{C}\indic_{\mathcal{G}_\ell} \)\)\right|^{ C \frac{|t|^k}{a^k}} \\ 
\label{hold}& \le C \frac{|t|}{a}\sum_{k=1}^{p-1} \left|\log \Esp_{\PNbeta} \( \exp\( \mathrm{sgn}(t^k)\frac{\gamma_k a^k}{C}\indic_{\mathcal{G}_\ell} \)\)\right|\lesssim_\beta \frac{|t|}{a}
(\beta +1) N\ell^\d \max_{|r|\le a}\mathcal R_r
 .\end{align}
%\cm{need to discuss $ t \ge 0$ , $t\le 0$}

%{\bf Step 2: conclusion.}\\
Inserting \eqref{defgk} and \eqref{hold} into  \eqref{exp2forme} and using  the definition of $a$, we obtain \eqref{hold0}. Inserting \eqref{hold0} into \eqref{exp2forme}, we obtain \eqref{pconc} after another application of H\"older's inequality.
\end{proof}

\subsection{Proof of Theorem \ref{CLT}}
The H\"older trick of the previous subsection gives us improved control on $T_2$, which will allow us to obtain the CLT.

\begin{lem}\label{improved T2}
Under the same assumptions,  suppose there is some $p \geq 2$ such that 
\begin{equation}\label{vanishing condition}
\(\ell N^{\frac{1}{\d}}\)^{\frac{\s}{2}}\(\max_{|r|\leq a}\mathcal{R}_r\)^{1-\frac{1}{p}}\rightarrow 0\quad \text{as } N \rightarrow \infty.
\end{equation}

If  $t=-\frac{\tau\sqrt2}{\sqrt{\beta}}N^{-1+\frac{\s}{2\d}}\ell^{-\frac{\s}{2}}$ where $\tau$ is fixed, then, 
 as $N\to \infty$, we have
%\begin{equation*}
%\log\Esp_{\PNbeta}\( e^{T_2}\indic_{\mathcal{G}_\ell}\)-\sum_{k=1}^{p-1}t^k\frac{N}{k!}B_k(\beta, \mu_0,\psi)=o(1)
%\end{equation*}
%In particular, 
\begin{equation}\label{limT2}
\log \Esp_{\PNbeta}\(e^{T_2}\indic_{\mathcal{G}_\ell}\)=-\frac{\tau\sqrt2}{\sqrt{\beta}}N^{-1+\frac{\s}{2\d}}\ell^{-\frac{\s}{2}}
NB_1(\beta,\muv, \psi)\indic_{\s \geq 0}+o(1),
\end{equation}
where $o(1)$ may depend on $\beta$.

If $t= -\frac{\tau \sqrt2}{\beta}N^{\frac{\s}{2\d}-1}\ell^{-\frac{\s}{2}}$
where $\tau $ is fixed, then, as $N \to \infty$, we have
\begin{equation}\label{limT22}
\log \Esp_{\PNbeta}\(e^{T_2}\indic_{\mathcal{G}_\ell}\)= -\frac{\tau \sqrt2}{\beta}N^{\frac{\s}{2\d}-1}\ell^{-\frac{\s}{2}}NB_1(\beta,\muv, \psi)\indic_{\s \geq 0}+o(1),
\end{equation}
where $o(1)$ is uniform as $\beta \to \infty$.

Moreover, \eqref{vanishing condition} holds for $p$ large enough if 
\be \label{small ell}
\ell \ll N^{-\frac{\s}{\d(\s+2)}}
\end{equation}
and  $\s < \s_0$, where $\s_0$ is approximately 
\begin{equation}\label{def: s0}
\begin{cases}
0.03973& \text{in }\d=1,\\
0.06059 & \text{in }\d=2.
\end{cases}
\end{equation}
\end{lem}
Notice that \eqref{small ell} holds automatically in the case $\s<0$ and in the case $\s = 0$ as soon as $\ell=o(1)$.
\begin{proof}
The first item is an immediate consequence of \eqref{pconc}: with the choice of $t$, if \eqref{vanishing condition} holds we have
\begin{equation*}
\left|tN\ell^\s \(\max_{|r|\leq a}\mathcal{R}_r\)^{1-\frac{1}{p}}\right| \lesssim \(\ell N^{\frac{1}{\d}}\)^{\frac{\s}{2}}\(\max_{|r|\leq a}\mathcal{R}_r\)^{1-\frac{1}{p}}\rightarrow 0.
\end{equation*}
 Moreover,
\begin{equation*}
\left|t^p N \ell^{(1-p)\d+p\s}\right| \lesssim N^{-p+\frac{p\s}{2\d}}\ell^{-\frac{\s p}{2}}N\ell^{(1-p)\d+p\s}=\(\ell N^{\frac{1}{\d}}\)^{\frac{p}{2}(\s-\d)+(1-\frac{p}{2})\d}\rightarrow 0
\end{equation*}
since $\s<\d$ and $p \ge 2$, establishing the desired $o(1)$ limiting behavior for $\log \Esp_{\PNbeta}\(e^{T_2}\indic_{\mathcal{G}}\)-\sum_{k=1}^{p-1}t^k\frac{N}{k!}B_k(\beta, \mu_0,\psi)$.
 Finally, using \eqref{borneB} and \eqref{tscale} we obtain 
\begin{equation*}
\left|t^k\frac{N}{k!}B_k(\beta, \mu_0,\psi)\right|\lesssim N N^{\frac{\s k}{2\d}-k}\ell^{-\frac{\s k}{2}}\ell^{(1-k)\d+\s k}=\( \ell N^{\frac{1}{\d}}\)^{(1-k)\d+\frac{\s k}{2}}.
\end{equation*}
For any $k \geq 2$, the exponent is negative as $\d-\d k+\frac{\s k}{2} \leq (k-1)(\s-\d)<0$. For $k=1$, the exponent is only negative for $\s<0$. In the $\s \geq 0$ regime, we possibly have some limit
$
\lim_{N \rightarrow \infty}(N^{-\frac{1}{\d}}\ell)^{-\frac{\s}{2}}B_1(\beta, \muv, \psi)
$. This establishes \eqref{limT2}.
The proof of \eqref{limT22} is analogous and left to the reader.

We now examine the condition under which \eqref{vanishing condition} holds.
From Proposition \ref{Kcomp}, for all $|t|\ell^{\s-\d}\le c$  small enough, we have
 %\begin{multline*}
 %\log \K_{N,\beta,\infty}^{\mathcal{G}_\ell}\(\R^\d, \mu_t,\zeta_V\circ \Phi_t^{-1}\)- \log \K_{N,\beta,\infty}\(\R^\d,\mu_0,\zeta_V\)+N(\Ent(\mu_t)-\Ent(\mu_0) ) \\=
%\mathcal{Z} (\beta, \mu_t)-\mathcal{Z}(\beta, \mu_0)+ O(\beta  N\ell^\d \mathcal R_t)
 %\end{multline*}
% where 

$$\mathcal R_t\le\max\( ( |t|\ell^{\s-\d})^{\frac\d{2\d-\s}}\(\mathscr{E}\((N^{\frac{1}{\d}} \ell)^{\frac{1}{\kappa +1}}\)  \)^{\frac{\d-\s}{2\d-\s}} ,|t|\ep^{\s-\d}\)
$$
  with 
 \begin{equation*}
\kappa:=
\begin{cases}
1-\frac{4}{4+\d-\s_+}\ & \text{if }\d \leq 5 \text{ or }\s <\d-1 \\
1-\frac{\d-1}{(\d-1)+(\d-\s)} & \text{otherwise}
\end{cases}
\end{equation*}
and
\begin{equation*}\mathscr{E}(R):= 
R^{-\kappa}\log R 
\end{equation*}
The definition \eqref{defa} together with $a \le \ell^{\d-\s}$ thus implies that 
\be \label{maxR}
\max_{|r|\le a} \mathcal R_r\le\max\(\(\mathscr{E}\((N^{\frac{1}{\d}} \ell)^{\frac{1}{\kappa +1}}\)  \)^{\frac{\d-\s}{2\d-\s}} ,a\ep^{\s-\d}\)
.\ee
Moreover, abbreviating $\mathscr{E}\((N^{\frac{1}{\d}} \ell)^{\frac{1}{\kappa +1}}\)$ into $\mathscr{E}$, 
in view of \eqref{defa}, we have
\be \label{defa2}
a\lesssim \max( \mathscr{E}^{\frac{\d-\s}{p(2\d-\s)}}, (a\ep^{\s-\d})^{\frac1p} )  \ell^{\d-\s} ,
\ee 
which together with \eqref{maxR} yields
\be \max_{|r|\le a} \mathcal R_r\lesssim  
\max\(\(\mathscr{E}\((N^{\frac{1}{\d}} \ell)^{\frac{1}{\kappa +1}}\)  \)^{\frac{\d-\s}{2\d-\s}} ,
(\frac{\ell}{\ep})^{\frac{p}{p-1}}, \( \mathscr{E}((N^{\frac{1}{\d}} \ell)^{\frac{1}{\kappa +1}}\) ^{\frac{\d-\s}{p(2\d-\s)}} (\frac{\ell}{\ep})^{\d-\s}\).\ee
Taking for instance $p=2$, 
in order to guarantee \eqref{vanishing condition}, it is thus sufficient to guarantee that 
%\cm{we now need also $\ell $ small because of the second term in the max}
\begin{equation}\label{vancon}
\frac{\s}{2}-\frac14\frac{\kappa}{\kappa+1}\(\frac{\d-\s}{2\d-\s}\)<0, \qquad \(\ell N^{\frac{1}{\d}}\)^{\frac{\s}{2}}\ell=o(1) \text{ for some }p \geq 1.
\end{equation}
This condition will only be possibly satisfied for small enough $\s$, which requires $\d=1,2$. The second condition is satisfied as soon as  $\ell \ll N^{-\frac{\s}{\d(\s+2)}}$,
which  yields \eqref{small ell}. 
For the first condition in \eqref{vancon}, we examine each dimension separately.
% \cm{can you please review now that there is $1/4$ in front?} 
In $\d=2$ we only have $\s>0$, and the condition becomes 
\begin{equation*}
\frac{\s}{2}-\frac{1}{4}\(\frac{2-\s}{8-2\s}\)\(\frac{2-\s}{4-\s}\)<0 \iff 4\s^3-33\s^2+68\s-4<0,
\end{equation*}
which is true for all $\s<\s_0:\approx 0.06059$. In $\d=1$ we always have \eqref{vanishing condition} for $\s<0$, so we only need to look at $\s\geq 0$; there, we need to check
\begin{equation*}
\frac{\s}{2}-\frac{1}{4}\(\frac{1-\s}{6-2\s}\)\(\frac{1-\s}{2-\s}\)<0\iff 4\s^3-21\s^2+26\s-1<0,
\end{equation*}
which is true for all $\s <\s_0:\approx0.03973$. 
\end{proof}
We can now prove Theorem \ref{CLT}.
\begin{proof}[Proof of Theorem \ref{CLT}]
Letting 
$t=-\frac{\tau \sqrt2}{\sqrt{\beta}}N^{\frac{\s}{2\d}-1}\ell^{-\frac{\s}{2}}$ and assembling the results of  Lemma \ref{lem2.1}, \eqref{bt0}, Lemma \ref{lem: t1} and \eqref{limT2}, under the assumptions \eqref{425} and 
\eqref{vanishing condition} we obtain
\begin{multline*}
\log  \Esp_{\PNbeta}\( \exp\(\tau \sqrt{2\beta}(N^{\frac{1}{\d}}\ell)^{-\frac{\s}{2}}\Fluct_{\muv}(\varphi)\)\indic_{\mathcal{G}_\ell}\)=\\
-\tau\frac{\sqrt2}{\sqrt{\beta}}\(\ell N^{-\frac{1}{\d}}\)^{-\frac{\s}{2}} \(\frac{1}{\cds}\int_{\Sigma}(-\Delta)^\alpha \varphi^\Sigma (\log \muv)+B_1(\beta, \muv, \psi)\) 
+\frac{\tau^2}{2}\ell^{-\s} \frac{c_{\d,\frac{\d-\s}{2}}}{\cds}\left\|\varphi^\Sigma \right\|_{\dot{H}^{\frac{\d-\s}{2}}}^2+o_N(1).
\end{multline*} Furthermore, since 
\begin{multline*}
\Esp_{\PNbeta}\( \exp\(\tau \sqrt{2\beta}(N^{\frac{1}{\d}}\ell)^{-\frac{\s}{2}}\Fluct_{\muv}(\varphi)\)\indic_{\mathcal{G}_\ell}\)=\Esp_{\PNbeta}\( \exp\(\tau \sqrt{2\beta}(N^{\frac{1}{\d}}\ell)^{-\frac{\s}{2}}\Fluct_{\muv}(\varphi)\indic_{\mathcal{G}_\ell}\)\)\\ +\PNbeta(\mathcal{G}_\ell^c)
\end{multline*}
and $\PNbeta(\mathcal{G}_\ell^c) \sim C_1e^{-C_2\beta \ell^\d N}\rightarrow 0$, we have obtained that the Laplace transform of 
\begin{equation*}
\sqrt{2\beta}(N^{\frac{1}{\d}}\ell)^{-\frac{\s}{2}}\Fluct_{\muv}(\varphi)\indic_{\mathcal{G}_\ell}
\end{equation*} 
converges in distribution to that of a Gaussian. Inserting into the above \eqref{limitvariance}, we thus  have that 
 \begin{equation*}
\sqrt{2\beta}\frac{\Fluct_{\muv}(\varphi)\indic_{\mathcal{G}_\ell}}{\(\ell N^{\frac{1}{\d}}\)^{\frac{\s}{2}}}+\frac{\sqrt2}{\sqrt{\beta}}\(\ell N^{-\frac{1}{\d}}\)^{-\frac{\s}{2}}\(\frac{1}{\cds}\int_{\Sigma}(-\Delta)^\alpha \varphi^\Sigma (\log \muv)+B_1(\beta, \muv, \psi)\)\end{equation*}
converges in distribution to a centered Gaussian of variance equal to 
$$\begin{cases}\|\varphi_0^\Sigma \|_{\dot{H}^{\frac{\d-\s}{2}}}^2& \text{if} \ \ell=1\\
\|\varphi_0 \|_{\dot{H}^{\frac{\d-\s}{2}}}^2& \text{if} \ \ell \to 0.\end{cases}$$
In view of \eqref{limitmean}, \eqref{expliB1} 
%\begin{multline*}
%\sqrt{\frac{\beta}{2}}\frac{\Fluct_{\muv}(\varphi)}{\(\ell N^{\frac{1}{\d}}\)^{\frac{\s}{2}}}+\frac{1}{\sqrt{2\beta}}\(\ell N^{-\frac{1}{\d}}\)^{-\frac{\s}{2}}\(\frac{1}{\cds}\int_{\Sigma}(-\Delta)^\alpha \varphi^\Sigma (\log \muv)+B_1(\beta, \muv, \psi)\)\\ \implies \mathcal{N}\(0, \|\varphi_0 \|_{\dot{H}^{\frac{\d-\s}{2}}}^2\)\end{multline*}
%as well since 
and $\PNbeta(\mathcal{G}_\ell^c)\rightarrow 0$,  this establishes the first claim of Theorem \ref{CLT}.

If we take $t= -\frac{\tau \sqrt2}{\beta}N^{\frac{\s}{2\d}-1}\ell^{-\frac{\s}{2}}$, we instead use \eqref{limT22}, and obtain the result in the same way.
\end{proof}
%This is currently presenting some problems, but notice that we only need to upgrade the control on the $p=1$ terms. Indeed, we have from Corollary \ref{lemsubdi}
%\begin{align*}
%N^{-\frac{\s}{\d}}|t|^p|\Ani_p(\XN,\muv,\psi)| &\lesssim N^{-\frac{\s}{\d}}\(N^{-1+\frac{\s}{2\d}}\ell^{-\frac{\s}{2}}\)^p N^{1+\frac{\s}{\d}}\ell^{p\s-(p-1)\d}\\
%&\sim N^{-(p-1)}\ell^{(p-1)\d}N^{\frac{\s p}{2\d}}\ell^{\frac{p\s}{2}}\\
%&\sim \(\ell N^{\frac{1}{\d}}\)^{\frac{p\s}{2}-(p-1)\d}
%\end{align*}
%since we will take $t\sim N^{-1+\frac{\s}{2\d}}\ell^{-\frac{\s}{2}}$. For $p \geq 2$, we have 
%\begin{equation*}
%\frac{p\s}{2}-(p-1)\d\leq (p-1)(\s-\d)<0 \implies N^{-\frac{\s}{\d}}|t|^p|\Ani_p(\XN,\muv,\psi)|\lesssim\(\ell N^{\frac{1}{\d}}\)^{\frac{p\s}{2}-(p-1)\d}\rightarrow 0
%\end{equation*}
%because $\ell N^{\frac{1}{\d}}\rightarrow \infty$. Similarly, using \eqref{cdpsi},
%\begin{align*}
%|t|^p\Fluct_{\muv}\(|D\psi|^p\) &\lesssim \(N^{-1+\frac{\s}{2\d}}\)^p\ell^{-\frac{p \s}{2}}N \int |D\psi|^p \\&\lesssim N^{-(p-1)}N^{\frac{p\s}{2\d}}\ell^{-\frac{p\s}{2}}\ell^{p\s-(p-1)\d} \\
%&\sim \(\ell N^{\frac{1}{\d}}\)^{\frac{p\s}{2}-(p-1)\d}
%\end{align*}
%which again tends to zero. 
%
%For the $p=1$ terms, $\Fluct_{\muv}(\div \psi)$ should be able to be dealt with in the same way as in \cite{P24}, so the only real work is in updating the control on $\Ani_1$. 
%
%In the $\s<0$ case ($\d=1$) we actually already have the CLT without any additional work since for $p=1$, we have
%\begin{equation*}
%\(\ell N^{\frac{1}{\d}}\)^{\frac{\s}{2}}\rightarrow 0
%\end{equation*}
%if $\s<0$. 
%

%%%%%%%%%%%%%%%

\section{Local Energies and Screening}\label{sec: ll prelim}

We turn now to the proof of Theorem \ref{Local Law}. The argument here is inspired by the bootstrap on scales carried out in \cite{L17}, \cite{AS21} and \cite{P24}, and relies heavily on the electric energy formulation introduced in \cite{PS17}. The main notions for this approach were introduced in Section~\ref{sec: fluct prelim}, where the next-order energy $\FN(\XN,\muv)$ was introduced. The typical order of this quantity is well understood, from \cite[Corollary 5.23]{S24}, see \eqref{expmomentcontrol} and the consequence  in \eqref{macrolaw2}
and serves as the base case for our bootstrap on scales. 
%We restate that result here for ease of reference. 
%\begin{lem}[Macroscopic Energy Control; Serfaty '24]\label{macro nrg}
%Suppose that $V$ is such that $\muv$ exists and is compactly supported. Assume also that $\muv$ has a bounded density. Then, for all $\beta>0$  
%\begin{multline*}
%\left|\log \Esp_{\PNbeta}\(\exp\frac{\beta N^{-\frac{\s}{\d}}}{2}\(\FN(\XN,\muv)+\(\frac{N}{2\d}\log N\)_{\indic_{\s=0}}+N\sum_{i=1}^N \zeta_V(x_i)\)\)\right| \\
%\leq C N+C_\zeta N,
%\end{multline*}
%where $C>0$ depends only on $\beta$, $\s$, $\d$ and $\|\muv\|_{L^\infty}$, and $C_\zeta$ depends only on $V$.
%\end{lem}
This section introduces notation and terminology for examining the system at the blown-up scale and for the localization to length scales $\ell \ll 1$. 
\subsection{Blowup and Subadditive Approximation}\label{subsec: ll neum}
First, it is convenient to change coordinates so that the typical interparticle distance is order $1$. In the blown-up scale the length scale $\rho_\beta N^{-1/\d}\le \ell\le 1$ will be replaced by $\rho_\beta\le L\le N^{\frac1\d}$.

Setting $\XN'=N^{\frac{1}{\d}}\XN$ and $\muv'(x)=\muv (xN^{-\frac{1}{\d}})$, we can compute  that
\begin{equation}\label{rescaleFformula}
\FN(\XN,\muv)=N^{\frac{\s}{\d}}\F(\XN',\muv')-\(\frac{N}{2\d}\log N\) \indic_{\s=0}
\end{equation}
where, for any nonnegative density $\mu$ with $\mu(\R^\d)=N$, we let  $\F$ be defined by 
\begin{equation}\label{next order}
\F(\XN, \mu):=\frac{1}{2}\iint_{\triangle^c}\g(x-y) d\left(\sum_{i=1}^N \delta_{x_i}-\mu\right)(x)d \left(\sum_{i=1}^N \delta_{x_i}-\mu\right)(y).
\end{equation}
Note that \eqref{rescaleFformula} and \eqref{energylowerbound} yield that 
\be \label{energylbbu}
\F(\XN,\mu) \ge - C N
\ee
where $C>0$ depends only on $\d,\s,\|\mu\|_{L^\infty}$.

%We will work with a generic nonnegative density $\mu$ such that $
%We will also denote by $u_N$ the potential (analogous to  $h_N$  \eqref{defhN}) at the blown-up scale 
%\be \label{defiu}
%u_N:=\g*\(\sum_{i=1}^N \delta_{x_i}- \muv\).\ee

A similar blow up computation yields
\begin{equation*}
N\zeta_V(x_i)=N^{\frac{\s}{\d}}\zeta_V'(x_i')
\end{equation*}
where
\begin{equation*}
\zeta_V'(x)=h^{\muv'}+N^{1-\frac{\s}{\d}}V\left(\frac{x}{N^{\frac{1}{\d}}}\right)-N^{1-\frac{\s}{\d}}c_V.
\end{equation*}
We can then expand as before
\begin{equation*}
\ZNbeta=\exp \left(-\beta N^{2-\frac{\s}{\d}}\I(\muv)\right)\K_{\beta}(\muv', \zeta_V')
\end{equation*}
where we define more generally
\begin{equation}
\label{defK1}
 \K_{\beta}(\mu,\zeta):=  N^{-N}\int_{(\R^\d)^N} 
\exp\left(- \beta \(\F(X_N,\mu) + \sum_{i=1}^N \zeta(x_i) \) \right)
\, dX_N.\end{equation}

%At times we will need to consider local energies of fields defined by other potentials as well. Let $\XN\subset (\R^\d)^N$ and suppose that $w$ is any electrostatic potential solving 
%\begin{equation}\label{Riesz viable potential}
%-\div(|y|^\gamma \nabla w)=\cds\(\sum_{i=1}^N \delta_{x_i}-\mu\delta_{\R^\d}\)
%\end{equation}
%in any $\Lambda \subset \R^{\d+1}$. Then we will correspondingly consider the \textit{local energy associated to }$w$\textit{ in }$\Lambda$, of that potential by 
%\begin{equation}\label{Riesz genpot local energy}
% \frac{1}{2\cds}\int_{\Lambda}|y|^\gamma|\nabla w_{\rr}|^2
%\end{equation}
%with $\rr$ is defined at blown up scale analogously to  \eqref{minimal distance} as 
%\be \rr_i:=\frac14\min(\min_{j\neq i}|x_i-x_j|, 1).\ee Notice that we have dropped the renormalization terms from Lemma \ref{riesz electric}, as in light of Proposition \ref{prop:MElb} the other terms are at most of same order. This quantity is more computationally easy to work with, and more relevant for the desired applications.

In the subsequent analysis, we will need subadditive and superadditive minimal approximations of our local energy that are purely local quantities, unlike the above energies which depend on the global configuration. We proceed to define the subadditive approximation now, again in analogy with the Coulomb gas \cite[Section 2]{AS21} and the $1$-d log gas \cite[Section 2]{P24}. The superadditive approximation is easier to present following a discussion of the screening procedure, and so we postpone that discussion to later in this section.

\subsubsection{Subadditive approximation}\label{sec:subadditive}
Let us start with our subadditive approximation, which is a Neumann energy, following the steps of \cite[Chap.~7]{S24}. It is now better to work with a generic nonnegative density $\mu$ in a generic domain $U$. Let $U\subset \R^\d$ be a domain with piecewise $C^1$ boundary and such that  $\int_U\mu=N$ is an integer. 
When $U$ is unbounded, we will need an additional decay assumption in the case $\s\le 0$: there exists $m>0$ and a set $\check{U}$ such that $\mu \ge m>0$ in $\check{U}$, such that 
\be\label{asstail} \frac{1}{\mu ((\check{U})^c)}\iint_{(\check{U})^c \times (\check{U})^c} \g(x-y)d\mu(x)d\mu(y)\ge - CN,\ee
which is easily satisfied by the blown-up equilibrium measure in the generic case we are studying.

We need a new version of the minimal distance to make the energy subadditive: we set
\begin{equation}\label{defrrc}
\rrh_i := \frac{1}{4}\min \( \min_{x_j\in U, j\neq i} |x_i-x_j|, \dist(x_i, \pa U), 1\).
\end{equation} 
This will shrink the radii of the balls when they approach $\partial U$, ensuring that all $B(x_i, \rrh_i)$ remain included in $U$ if $x_i \in U$.
Let now $\Lambda$ be a set of the form $U\times [-H,H]$ for some $H\in (0, +\infty]$ (interior case), or of the form $(U^c\times [-H,H])^c$ (exterior case).

%In order to have an energy which is always bounded from below and to ensure the convergence of the Neumann partition functions, we need for technical reasons to a%dd some energy to points that approach the boundary. 
%To that effect we   define
%\be \label{defGG}\mathsf{h}(x_i):= \(\g\(\tfrac14\dist(x_i, \pa U)\) - \g\(\tfrac14\)\)_+.\ee
If $\mu (U)=N$, for a configuration $\XN$ of points in $U \subset \R^\d$ and $ \Lambda \subset \R^{\d+1}$ of the form above, we let $u$ solve
\be \label{defiu}\left\{
\begin{array}{ll}
-\div (\yg \nab u) = \cds\displaystyle\(\sum_{i=1}^N \delta_{x_i}-\mu \drd\)& \text{in} \ \Lambda \\
 \frac{\pa u}{\pa n}=0 &\ \text{on} \ \partial \Lambda \\ [2mm]
 \nab u \to 0 &\ \text{at } \infty. \end{array}\right.
\end{equation}
If $U$ and $\Lambda$ are bounded, the condition at $\infty$ is dropped. The associated \textit{subadditive approximation} is the Neumann energy
\begin{equation}\label{minneum}
\F(\XN,\mu,\Lambda):=\frac{1}{2\cds}\( \int_{\Lambda}|y|^\gamma|\nab u_{\rrh}|^2 -\cds\sum_{i=1}^N \g\(\rrh_i\)\)-\sum_{i=1}^N \int_U \f_{\rrh_i}(x-x_i)~d\mu(x),
\end{equation} where $\f_\eta$ is defined in \eqref{extended  truncation}.
Note that $\F(\XN, \mu, \R^{\d+1})= \F(\XN, \mu)$, the quantity defined in \eqref{next order}.

Notice also that for any $\Lambda_1 \subset \Lambda_2$,
we have
\be \label{locali2}
\F\(\XN,\mu,\Lambda_2\) \leq \F(\XN,\mu,\Lambda_1)+\F\(\XN,\mu,\Lambda_2\setminus \Lambda_1\).
\ee
The proof is as in  \cite[Corollary 7.6]{S24} and relies on the following projection lemma,  which tells us that gradients minimize energy, cf. \cite{PS17}.
\begin{lem}[Projection lemma]\label{projlem}
  Assume that $U$ is an open  subset of $\R^\d$ with piecewise $C^1$ boundary, and let $\XN \subset U\times \{0\} \subset \Lambda \subset \R^{\d+1}$, where $\Lambda$ is an open subset of $\R^{\d+1}$ with piecewise $C^1$ boundary. Assume $E$ is a vector-field satisfying a relation of the form
  \begin{equation}\label{eqe}
   \left\{\begin{array}{ll}  
  -\div (\yg E)= \cds\( \sum_{i=1}^N \delta_{x_i} -\mu\drd\) &\quad \text{in} \ \Lambda \\
  E \cdot  n=0 &\quad \text{on} \ \partial \Lambda,\end{array}\right.\end{equation}
and $u$ 
  solves 
$$  \left\{\begin{array}{ll}  
  -\div(\yg\nab u)= \cds\( \sum_{i=1}^N \delta_{x_i}- \mu\drd\)& \quad \text{in}  \ \Lambda \\
  \frac{\pa u}{\pa n}=0 & \text{on} \ \partial \Lambda,\end{array}\right.$$ 
  and $ u(\frac{\pa u}{\pa n}- E \cdot n) \to 0$ as $|x|\to \infty, x\in \Lambda$ if $\Lambda$ is unbounded.
  Then
  $$\int_{\Lambda}\yg |\nab u_{\rrh}|^2 \le \int_{\Lambda} \yg  |E_{\rrh} |^2,$$ 
  where $E_{\rrh}:= E-\sum_{i=1}^N \nab \f_{\rrh_i} (x-x_i)$.
  \end{lem}
Extending $\nab u_{\rrh}$ by a zero vector field and using the projection lemma  as in \cite[Corollary 7.7]{S24},
 we have that
\be\label{comparaisontoutesp}  \F(\XN,\mu,\Lambda) \ge \F(\XN, \mu\indic_{\Lambda}, \R^\d)=\F(\XN, \mu\indic_{\Lambda}),\ee
which implies in view of \eqref{energylbbu} the lower bound
\be\label{lblocale}
 \F(\XN, \mu, \Lambda) \ge - C N
\ee
with $C$ depending only on $\d,\s,\|\mu\|_{L^\infty}$.
% We omit the proof of \eqref{comparaisontoutesp} which consists in extending $\nab u_{\rrh}$ by a $0$ vector field and using the projection lemma \ref{projlem} is as in \cite[Corollary 7.7]{S24}.

\subsubsection{Local versions}
We next turn to local versions of these energies.
First we define a new minimal distance relative to $\partial \Omega$ where $\Omega \subset \R^\d$: 
\begin{equation}\label{defrrc3}
\rrt_i := \frac{1}{4}\begin{cases} 
\min \( \min_{x_j\in \Omega, j\neq i} |x_i-x_j|, \dist(x_i, \pa U\cap \Omega), 1\) & \text{if} \ \dist(x_i, \partial \Omega) \ge \hal  \\ \min(1, \dist(x_i, \partial U \cap \Omega)) & \text{if } \dist(x_i, \pa \Omega\backslash \partial U ) \le \frac14\\
t \min \( \min_{x_j\in \Omega, j\neq i} |x_i-x_j|, \dist(x_i, \pa U\cap \Omega\backslash \pa U ), 1\)&\\
+(1-t) 
\min(1, \dist(x_i, \partial U \cap \Omega))& \text{if } \dist(x_i, \pa  \Omega\backslash \pa U)= \frac{1+t}{4} 
, t\in [0,1].\end{cases}
\end{equation}
We note that this minimal distance coincides with \eqref{defrrc} when taking $\Omega=\R^\d$.
Here, the  balls are enlarged to their largest possible values for points that approach the boundary of $\Omega$ (except for the part included in $\partial U$). This way, balls can potentially overlap the boundary and not be disjoint. This also ensures that for $x_i\in \Omega$, $\rrt_i$ does not depend on the points of the configuration that lie outside of $\Omega$. 
In addition the definition is made so that the radii are continuous with respect to the position of the points.

Given a set $\tilde \Omega \subset \R^{\d+1}$ of the form $\Omega\times [-h,h]$ or  its complement, we then let 
\be \label{Glocal2} 
\F^{\tilde \Omega}(\XN,\mu,\Lambda) :=\frac1{2\cds}\( \int_{\tilde \Omega \cap \Lambda} \yg|\nab u_{\rrt}|^2-\cds\sum_{i\in I_{\tilde \Omega}} \g\(\rrt_i\)\)-\sum_{i\in I_{\tilde \Omega}} \int_U \f_{\rrt_i}(x-x_i)\, d\mu(x) .\ee
% where we have set $\tilde{U}=U \times \R$. This time,
This definition provides the following  important superadditivity property (see \cite[Section 4.5, Lemma 7.8]{S24})
\be\label{superadd}
 \F(\XN,\mu,\Lambda) \geq \F^{\tilde \Omega}(\XN,\mu,\Lambda)+\F^{(\tilde \Omega)^c}(\XN,\mu,\Lambda).
\ee

We have the following control, as in  \cite[Proposition 4.28, Lemma 7.8]{S24}: if $\tilde \Omega \cap \Lambda$ is convex,  there exists $C_0>0$ depending only on $\d,\s$ and $\|\mu\|_{L^\infty} $ such that 
 \be \label{arrivC0} \int_{\tilde \Omega} \yg |\nab h_{\rrt}|^2 \le 4\cds\(  \F^{\tilde \Omega}(\XN, \mu, \Lambda) + C_0 \#I_{\tilde \Omega} \).
 \ee
 
 Moreover, as in \cite[Remark 7.9]{S24}, in the case that $U$ is a hyperrectangle, we can show that 
 \be \label{arrivC02}
 \int_U \yg  |\nab h_{\rrh}|^2\le C (\F(X_N, \mu, U)+C), \qquad \sum_{i=1}^N \g(\rrh_i) \le C ( \F(X_N, \mu, U)+ C).\ee
  
  %Instead we will work with 
%  We will also denote
%\begin{equation}\label{minneum tilde}
%\tilde{\F}(\XN,\mu,\Lambda):=\frac{1}{2\cds}\int_{\Lambda}|y|^\gamma|\nab u_{\rrt}|^2.
%\end{equation}
\subsubsection{Local partition function}
Associated to $U\subset \R^\d$ we also define a \textit{local partition function} respect to a height $H$ in the extended dimension. For $\zeta\ge 0$ such that $\int e^{-\zeta(x)}dx$ is convergent, $\zeta$ vanishing in the support of $\mu$,  and $u$ associated to $U$ via \eqref{defiu}, let 
\begin{equation}
\label{defK}
 \K_{\beta,H}(U, \mu,\zeta):=  N^{-N}\int_{U^N} 
\exp \(-\beta \(\F \(\XN,\mu, U \times [-H,H]\)
 + \sum_{i=1}^N \zeta(x_i)  \)\)
\, dX_N.\end{equation}
We associate a measure to this partition function by 
\begin{equation}\label{defQ}
d\Q_{\beta, H}(U,\mu,\zeta):=\frac{1}{N^N\K_{\beta, H}(U,\mu,\zeta)}\exp\(-\beta \left( \F\(\XN,\mu, U \times [-H,H]\)
 + \sum_{i=1}^N \zeta(x_i)  \)\) dX_N,
\end{equation}
and we also introduce external partition functions
\begin{equation}
\label{defKext}
 \K^{\mathrm{ext}}_{\beta,H}(U^c, \mu,\zeta):=  (N')^{-N'}\int_{(U^c)^{N'}} 
\exp\left(- \beta \( \F\(X_{N'},\mu, (U \times [-H,H])^c\)+ \sum_{i=1}^{N'} \zeta(x_i)  \)\right)
\, dX_{N'},\end{equation}
where $N'=\mu(U^c)$, assumed to be integer.

%We also introduce $\tilde{\K}$, where $\F$ is replaced by $\tilde{\F}$ in the above definitions.
% \subsection{Background Results}\label{subsec: prelim}

%
%
%
%
%
%
%
%On the good events that we will consider, $u$ in \eqref{truepot}) will inner/outer screenable in the sense \eqref{Riesz screenability}). Thus, it is a competitor in the definition of $\G_{a,h}^{\mathrm{inn/ext}}$, and we have the following.
%\begin{lem}\label{lemrestri}
%Let $u $ be the solution of \eqref{defv} used in the definition of \eqref{Glocal2}. If $u$ is screenable in the sense \eqref{Riesz screenability}), then 
%\be \F^{\Omega_t}(\XN, U) \ge \G_{a,h}^{\mathrm{inn}}(\XN|_{\Omega}, \Omega_t, U) .\ee
%and
%\be \F^{\Omega_t \setminus \Omega_{t'}'}(\XN, U) \ge \G_{a,h}^{\mathrm{ext}}(\XN|_{\Omega \setminus \Omega'}, \Omega_t \setminus \Omega_{t'}', U) .\ee
%\end{lem}
%
Coupled with \eqref{locali2}, we have the following superadditivity of partition functions.
\begin{lem}\label{lem: supadd part}
Let $U$ be as above, and suppose $U$ is partitioned into $p$ disjoint sets $Q_i$, $1 \leq i \leq p$ with $\mu(Q_i)=N_i \in \Z$. Let $H, h_1, \dots, h_p \in (0,+\infty]$ and suppose $h_i \leq H$ for all $i$. Then, 
\begin{equation}\label{eq: supadd part}
\K_{\beta,H}(U, \mu, \zeta) \geq \frac{N!N^{-N}}{N_1!\cdots N_p! N_1^{-N_1}\cdots N_p^{-N_p}}\prod_{i=1}^p \K_{\beta,h_i}(Q_i, \mu, \zeta).
\end{equation}
We also have, if $\mu(\R^\d)=N$ and $\mu(U)=\mn$,
\begin{equation}\label{eq: supadd partext}
\K_{\beta,H}(\R^\d, \mu, \zeta) \geq \frac{N!N^{-N}}{\mn!\mn^{-\mn} (N-\mn)! (N-\mn)^{N-\mn}}\K_{\beta,H}(U, \mu, \zeta)\K_{\beta,H}^{\mathrm{ext}}(U^c, \mu, \zeta).
\end{equation}
\end{lem} 
\begin{proof}
The argument is exactly as in \cite[Lemma 7.3]{S24}, with the added observation in \eqref{eq: supadd part} that we can decrease $H$ by appending a zero electric field for $|y|>h_i$ in the extended dimension to the vector field defining $\F(\XN\vert_{Q_i}, \mu, Q_i \times [-h_i,h_i])$ and applying Lemma \ref{projlem} above.
\end{proof}

\subsection{Preliminary free energy controls}
To obtain  free energy controls, we use, as in \cite{S24},  the following rewriting.
 \begin{lem}
   Let $U$ be an open subset of $\R^\d$  with bounded and piecewise $C^1$ boundary  and $\mu$ a bounded  nonnegative density such that $\mu(U)=N$ is an integer. Let $h\in (0,+\infty]$, and 
   % and $\int_{\R^\d}|\g(x-y)|d \mu(y) <\infty$. 
   let $G_U$ solve 
   \begin{equation}\label{NGf}\left\{\begin{array}{ll}
   -\div_x (\yg \nab G_U (x,x_0))=\cds\(\delta_{x_0} (x) - \frac{1}{\mu(U)} \mu (x)\drd\) & \text{in} \ U\times [-h,h] \\
   \frac{\partial G_U}{\partial n}= 0  & \text{on} \ \partial \(U\times[-h,h]\)\\ [1mm]
   \nab G_U \to 0 & \text{at} \ \infty.\end{array}\right.\ee  
   Let
   $$H_U(x,x')= G_U(x,x')-\g(x-x'),$$ where $\g$ is naturally extended to $\R^{\d+1}$.
   Then for any configuration  $X_{N}$ of points in $U$, we have 
   \begin{multline}\label{737} \F(X_{ N},\mu, U\times[-H,H]) = \hal \iint_{\R^\d\backslash \triangle } \g(x-y) d\( \sum_{i=1}^{N} \delta_{x_i} - \mu\indic_U\)(x) d\( \sum_{i=1}^{ N} \delta_{x_i} - \mu\indic_U\)(y) 
   \\
   + \hal \iint_{U\times U} H_U(x,y) d\( \sum_{i=1}^{ N} \delta_{x_i} - \mu\)(x) d\( \sum_{i=1}^{ N} \delta_{x_i} - \mu\)(y) .\end{multline}
   \end{lem}
   The proof is analogous to \cite[Lemma 7.9]{S24} except working in extended space with natural extension of $\g$ and $G_U$.

We can now obtain the a priori bound on the Neumann free energies.
\begin{prop}[Neumann free energy bound]\label{pro718}
 Let $U$ be an open subset of $\R^\d$  with bounded and piecewise $C^1$ boundary  and $\mu$ a bounded  nonnegative density in $U$ such that $\mu(U)=N$ is an integer. Let $H \in (0,+\infty]$.
 %If $\s\le 0$ and $U$ is unbounded, assume in addition that  \eqref{assumplbs}, \eqref{assgmm} and \eqref{assumpbeta} hold. \cm{add later}
Under the assumption \eqref{asstail}, we have
\be 
\label{bornesfiU} \log  \K_{\beta,H}(U,\mu,\zeta) +\Ent(\mu) \ge - C(1+ \beta)  N\ee
where $C>0$ depends only on $\d, \s, \|\mu\|_{L^\infty}$ and the constants in the assumptions.

Under the assumption that $U$ is bounded, we  have
 \be \label{bornesfiU2} \log  \K_{\beta,H}(U,\mu,0)  \le C \beta N+ N \log \frac{|U|}{N}\ee
%and 
%\begin{equation}\label{NRGbd}
%\frac{1}{N^N}\int_{U^N}\F(\XN,\mu,U\times [-h,h])~d\mu^{\otimes N}\(\XN\)\lesssim  (1+\beta) N.
%\end{equation}
\end{prop}

\begin{proof}Let us start with the lower bound.
The difference with the proof of 
\cite[Prop 7.10]{S24} is that the reference measure in the definition of $\K_{\beta,H}$ is $e^{-\zeta(x)}dx$ and not $\mu$.
By definition and using Jensen's inequality, we may write 
\begin{multline}\log \K_{\beta,H}(U,\mu,\zeta) = \log  \(N^{-N}\int_{U^N} 
\exp\left(- \beta \F(X_N,\mu,U) + \sum_{i=1}^N \zeta(x_i) -\sum_{i=1}^N \log \mu(x_i)  \right)
\, d\mu^{\otimes N} (\XN)\)\\
\ge  \frac{1}{N^N} \int_{U^N}\( -\beta \F(X_N,\mu,U) + \sum_{i=1}^N \zeta(x_i) -\sum_{i=1}^N \log \mu(x_i)  \) d\mu^{\otimes N}(\XN)\end{multline}
Inserting the rewriting \eqref{737}, expanding all the sums and observing that the terms involving $H_U$ cancel after integration against $d\mu^{\otimes N}$, we are led to 
$$\log \K_{\beta,H}(U,\mu,\zeta)\ge \int_U \zeta d\mu -\int_U \mu \log \mu+ \frac{\beta}{2N}\iint \g(x,y) d\mu(x)d\mu(y).$$ 
If $\s>0$ then $\g \ge 0$ and we deduce the lower bound. If $\s\le 0$, we argue by splitting the region into cells of size $R$ and optimizing over $R$. It is an adaptation of the proof of Propositions~7.10 and 5.14 in \cite{S24}, left to the reader, and using \eqref{asstail}.

For  the upper bound, it suffices to insert the  lower bound \eqref{lblocale} for $\F$  into \eqref{defK}.
%Finally, \eqref{NRGbd} is a corollary of the above proof of the lower bound.
\end{proof}
 
\subsection{Riesz Screening and Superadditive Approximation}
The main technical tool that we use in the proof of Theorem \ref{Local Law} is a screening procedure for Riesz gases, adapted from the procedure for Coulomb gases described in \cite[Section 7.2]{S24}. This kind of procedure is based on ideas from \cite{ACO09} and \cite{SS12}, and adapted to  Riesz gases in \cite{PS17}, then also used extensively in \cite{LS15}. The difficulty in the adaptation is in having to work in the extended space.  The version below is optimized from \cite{PS17} as was done for the Coulomb gas in \cite{AS21}, introducing a new approach in dealing with the extended dimension for the so-called ``outer screening''  (which was up to \cite{AS21} called inner screening).

A detailed description of the ideas and motivations involved in the screening can be found in \cite[Section 7.2]{S24}, some of which we summarize here. Consider a hyperrectangle $\Omega$ in which the background measure $\mu$ is bounded from below, the localization superadditivity as in \eqref{superadd}  yields
\begin{equation}\label{superadd2}
 \F(\XN,\mu, \R^\d) \geq \F^{\Omega\times [-h,h]} (\XN,\mu,\R^\d)+ \F^{(\Omega \times [-h,h])^c}(\XN,\mu, \R^\d).
\end{equation}
A matching upper bound is of course not true, but if the energy on $\Omega$ is (reasonably) well controlled we can screen the configuration $\XN$, which means modify it near $\partial \Omega$ and produce new configurations $Y_{\mn}$ and $Y_{N-\mn}$ with corresponding electric fields that have a zero Neumann  boundary condition and energy smaller than  the original ones, up to small errors, i.e.
$$ \F\(Y_{\mn}, \tilde{\mu}, \Omega \times [-h,h]\)\le  \F^{\Omega\times [-h,h]} (\XN,\mu,\R^\d)+ \text{ screening errors}
$$
and 
$$ \F\(Y_{N-\mn}, \tilde{\mu}, (\Omega \times [-h,h])^c \)\le   \F^{(\Omega \times [-h,h])^c}(\XN,\mu, \R^\d)+ \text{ screening errors}
.$$ Here the configurations $Y_{\mn}$ and $Y_{N-\mn}$ coincide with $X_N$ except in a boundary layer near $\partial \Omega$, the same for $\tilde \mu$ with $\mu$.
By subadditivity of the Neumann energy \eqref{locali2}, gluing together $ Y_{\mn}$ and $Y_{N-\mn}$ into a new configuration $Y_N$ on $\R^\d$, 
we then have 
\begin{align*}
 \F(Y_N,\mu,\R^\d)& \le  \F\(Y_{\mn}, \tilde{\mu}, \Omega \times [-h,h]\)+\F\(Y_{N-\mn}, \tilde{\mu}, (\Omega \times [-h,h])^c \)\\
 &\le
 \F^{(\Omega \times [-h,h])^c}(\XN,\mu, \R^\d)
+ \F^{\Omega\times [-h,h]} (\XN,\mu,\R^\d)+ \text{ screening errors}.
 \end{align*}
The configuration  $Y_N$ being considered as an approximation of $\XN$, this constitutes a sort of converse (with error) to the superadditivity of \eqref{superadd2}, providing an estimate of additivity error. 
 We will be able to do this for most point configurations and corresponding electric fields, this will allow the above relations to be integrated over configurations and turned  into a free energy almost additivity result. The crucial point here is that the screening errors hence addivity errors are negligible with respect to the volume of the box $\Omega$, even when $\Omega$ is small, which is what will allow to obtain the local laws down to the microscale of Theorem~\ref{Local Law}. An important fact is that screening can only be performed in regions where the density $\mu$ is bounded below by some constant $m>0$.
\subsubsection{Riesz screenability}
Before delving into a more detailed description of the procedure, we first give a definition of screenable electric fields.
We let $\mathcal Q_R$ be the set of closed hyperrectangles of the form $Q_R\times [-h,h]$ in $\R^{\d+1}$ with the  sidelengths of $Q_R$  in $[\hal R, 2R]$ and which are such that $\int_{Q_R}\mu$ is an integer.

We will screen fairly generic electric fields satisfying relations of the form \eqref{inner screening w}, respecting Neumann boundary data constraints if there are any, and for such vector fields we define their truncations as in \eqref{truncated potential}, with the truncation radii $\rrh$, 
which we notice coincides with $\rrt$ for points at distance larger than 1 from all considered boundaries (hence the truncated fields coincide as well once at distance $\ge 1$ from the boundary).
The main difference with the Coulomb case is in the need to control the energy on horizontal slices parallel to $\R^\d$ in the extended space.
In what follows $m>0$ is a positive constant: screening can only be performed in regions where  the density $\mu$ is bounded away from $0$.

\begin{defi}\label{def: Riesz screenability} Let $\mu$ be a nonnegative bounded density.
Assume $\Lambda$ is either $\R^{\d+1}$ or a  the cartesian product of a  hyperrectangle with an interval $[-H,H]$,
% for some $R \ge \max(1,\beta^{-\frac1\d})$ \cm{revoir: $\max(1, \beta^{-\frac{1}{\d-\s}})$??} 
  or the complement of such a set. 
Let  $h \le R/2$ and let
$\tilde\Omega= (Q_{R}\times [-h,h])\cap \Lambda $ (inner case), resp. $\tilde\Omega= \Lambda\backslash (Q_R\times [-h,h])$ (outer case) where $Q_R$ is a hyperrectangle of sidelengths in $[R,2R]$ with sides parallel to those of $\Lambda$, and such that 
 $\mu(\tilde\Omega)=\mn$, an integer.
Let $\ell$ and $\tilde \ell$ be such that  $R\ge \tilde \l \ge \l \ge C$, where $C$ is some constant dependent on $m$. In the inner case, assume that $\mu \ge m>0$  in $((Q_R\backslash Q_{R-2\tilde \ell})\times \{0\})\cap\Lambda$ and in the  outer case,   assume that that 
 $\mu\ge m>0$ in $ ((Q_{R+2\tilde \ell} \backslash Q_R)\times \{0\})\cap \Lambda$.
In the outer case, also assume  that 
 the faces of  $\partial  Q_{R}   $  are at distance $\ge 2 \tilde \ell$ from their respective  parallel faces of $\partial \Lambda$.
 
 In the inner screening, let $X_n$ be a configuration of points in $\tilde \Omega$ and let $w$ solve 
\begin{equation}\label{inner screening w}
\begin{cases}
-\div(|y|^\gamma \nabla w) =\cds \left(\sum_{i=1}^n \delta_{x_i}-\mu \delta_{\R^d}\right) \quad& \text{in } \tilde\Omega\\
\frac{\partial w}{\partial n} =0 & \text{on} \ \partial \Lambda \cap \tilde\Omega.\end{cases}
\end{equation}
In the outer screening, let $X_n$ be a configuration of points in $\tilde\Omega$ and let $w$ solve 
\begin{equation}\label{outer screening w}
\begin{cases}
-\div(|y|^\gamma \nabla w)=\cds \left(\sum_{i=1}^n \delta_{x_i}-\mu \delta_{\R^\d}\right) \quad &\text{in }\tilde \Omega\\
\frac{\partial w}{\partial n} =0 & \text{on} \ \partial \Lambda \cap \tilde\Omega.\end{cases}
\end{equation}

In the inner screening, denote
\begin{align}
&S(X_n,w,h)=\int_{ ((Q_{R-\tilde{\ell}}\setminus Q_{R-2\tilde{\ell}})\times [-h,h])\cap \Lambda}|y|^\gamma|\nabla w_{\rrh}|^2  \\
& S'(X_n,w,h)= \sup_x \int_{(((Q_{R-\tilde{\ell}}\setminus Q_{R-2\tilde{\ell}})\cap \square_\ell(x))\times [-h,h] )\cap \Lambda}|y|^\gamma|\nabla w_{\rrh}|^2 \\
& \label{Riesz upper energy inner screening}e(X_n,w,h)=\int_{Q_{R} \times \{-h,h\}}|y|^\gamma|\nabla w|^2.
\end{align}

In the outer screening, denote
\begin{align}
&S(X_n,w,h)=\int_{((Q_{R+2\tilde{\ell}}\setminus Q_{R+\tilde{\ell}}) \times [-h,h])\cap \Lambda}|y|^\gamma|\nabla w_{\rrh}|^2  \\
& S'(X_n,w,h)= \sup_x \int_{(((Q_{R+2\tilde{\ell}}\setminus Q_{R+\tilde{\ell}})\cap \square_\ell(x))\times [-h,h])\cap \Lambda}|y|^\gamma|\nabla w_{\rrh}|^2 \\
& \label{Riesz upper energy outer screening}e(X_n,w,h)=\int_{Q_{R}\times \{ (-(R+h),(R+h)\} 
%\cup \partial \(Q_R\times ([-(R+h),(R+h)]\setminus [-h,h])) \cap \Lambda\)
}|y|^\gamma|\nabla w|^2.
\end{align}
 %with $\square_T$ as in Appendix \ref{sec: pfscn}.\cm{I doubt that this is the correct domain of integration. Also replace $\carr_T$ by $Q_R$} 
% We also define
% \begin{equation}   \label{defM}
%  M_0:=\frac{1}{|\New|}\int_{\Old \times \{\pm h\}}\yg \nabla w\cdot \hat{n} 
%\end{equation}
%where $\Old$ is the set on which the configuration remains unchanged in the screening procedure, and $\New$ is $Q_L \setminus \Old$ in the inner screening and $Q_L^c \setminus \Old$ in the outer screening (see Appendix \ref{sec: pfscn}). For the outer screening, the domain of integration in \eqref{defM} is replaced by 
%\begin{equation*}
%M_0=\frac{1}{|\New|}\int_{\square_{T}\times \{-(L+h),(L+h)\} \cup \partial \(\square_{T}\times ([-(L+h),(L+h)]\setminus [-h,h])\)}|y|^\gamma \nabla w\cdot \hat{n}.
%\end{equation*}
%\cm{this makes no sense. You can't define $\Old $ and $\New$ until you have screened, and you can't screen until you've defined screenability. Use a quantity like $\int \yg |\nab w|$ integrated over a larger set. Why not express the condition in terms of $e$?

We say that a configuration $X_n$ and potential $w$ are screenable at height $h$ if  
\begin{equation}\label{Riesz screenability1}
\max\left(
\frac{1}{R^{\d-2} \tilde \ell^2} h^{\gamma} e(X_n,w,h)
, \frac{h^{1+\gamma}}{\ell^{\d+1}}\frac{S(X_n,w,h)}{\tilde{\ell}}\right)\leq \mathsf{c}
\end{equation}
or 
\begin{equation}\label{Riesz screenability2}
\max\left(
\frac{1}{R^{\d-2} \tilde \ell^2} h^{\gamma} e(X_n,w,h)
, \frac{h^{1+\gamma}}{\ell^{\d+1}}S'(X_n,w,h)\right)\leq \mathsf{c}
\end{equation}
%\cm{first term in the estimate was wrong, it has changed and this needs to be checked and propagated everywhere}
where $\mathsf{c}$ is a constant dependent only on upper and lower bounds for $\mu$ (and 
defined in  \eqref{definition of little c}).\end{defi}

With a notion of screenability in hand, we can define the minimal energy approximation that we will need for the screening statement. It is analogous to a Dirichlet energy.

 \begin{defi}[Best screenable potential and energy] \label{defibestpot}
With the same notation as above, we let 
 \begin{multline}
 \label{innernrj}\G_{a,h}^{\mathrm{inn/ext}}(X_n, \tilde \Omega) =\min  \Big\{  \frac{1}{2\cds}  \(\int_{\tilde \Omega} \yg |\nab w_{\rrh} |^2
 - \cds \sum_{i=1}^n \g(\rrh_i) \) - \sum_{i=1}^n \int_{\tilde \Omega }\f_{\rrh_i} (x-x_i) d\mu(x)\drd
  ,\\
  w \, \mathrm{inner/outer  \ screenable\ satisfying\ a\  relation\ of\  the\ form }\\
  \begin{cases}
-\div(|y|^\gamma \nabla w) =\cds \left(\sum_{i=1}^n \delta_{x_i}-\mu \delta_{\R^d}+\sum_j \delta_{x_j}^{(\eta_j)}\right) \quad& \text{in } \tilde\Omega\\
\frac{\partial w}{\partial n} =0 & \text{on} \ \partial \Lambda \cap \tilde\Omega\end{cases}\\
\mathrm{with\ } x_j\notin \tilde \Omega, \eta_j\le \frac14 \min(1,\dist(x_j, \partial \Lambda))\\
  \mathrm{and\ satisfying\ } \eqref{Riesz screenability1}\ \mathrm{or } \ \eqref{Riesz screenability2} \mathrm{ \ at \: level  \ } h \mathrm{ \ and \ } e(X_n,w, h) \leq a \Big\} ,\end{multline}
  (with the $\min$ understood as $+\infty$ if the set is empty),
% and 
 %\begin{multline}
 %\label{outernrj} \G_{a,h}^{\mathrm{ext}}(X_n, \Lambda) =
 %\min  \Big\{  \frac{1}{2\cds}  \(\int_{\tilde \Omega} \yg |\nab w_{\rrh} |^2
 %- \cds \sum_{i=1}^n \g(\rrh_i) \) - \sum_{i=1}^n \int_{\tilde \Omega }\f_{\rrh_i} (x-x_i) d\mu(x)\drd
  %,\\
 % w \,\mathrm{outer \ screenable \ satisfying\ }\eqref{Riesz screenability} \mathrm{\ at \: level \ }h  \mathrm{\ and \ } e(X_n,w, h)\leq a\Big\} , \end{multline} 
 with 
  $e(X_n,w, h)$ as in \eqref{Riesz upper energy inner screening} and \eqref{Riesz upper energy outer screening}, respectively. By the direct method in the calculus of variations, one may check that the minima are achieved.
  % by some $w$ which vanishes on $\partial \Omega\backslash \partial U$.
 Note that $\G$ depends on $\Lambda$ and $\mu$ but for the sake of lightness we do not retain it in the notation.
 
 We also define 
 \begin{align}\label{bestS}
&  \bar S_{a,h}(X_n)= \inf \{ S(X_n, w,h), w \ \text{achieving the min in $\G_{a,h}^{\mathrm{inn}}(X_n, \tilde \Omega) $, resp. $\G_{a,h}^{\mathrm{ext}} (X_n, \tilde \Omega) $}\} 
 \\ \label{bestS2}  & \bar S'_{a,h}(X_n)= \inf \{ S'(X_n, w,h), w \ \text{achieving the min in $\G_{a,h}^{\mathrm{inn}}(X_n, \tilde \Omega)$, resp. $\G_{a,h}^{\mathrm{ext}} (X_n, \tilde \Omega) $}\} .
 \end{align}
 
\end{defi}
%The reason for calling this a Dirichlet energy is that if the Dirichlet potential for $\Omega$ is inner/outer screenable and has the appropriate decay away from the axis, then we can write
%\begin{multline*}
 %\G_{a,h}^{\mathrm{inn}}(X_n, \Omega \times [-L,L]):=\frac{1}{2\cds}\int_{\Omega \times [-L,L]}\yg|\nabla w_{\rr}|^2, \\
   %\G_{a,h}^{\mathrm{ext}}(X_n, (\Omega \times -L,L])^c):=\frac{1}{2\cds}\int_{(\Omega \times [-L,L])^c}\yg|\nabla w_{\rr}|^2\end{multline*}
%where $w$ solves the Dirichlet problem
%\begin{equation*}
%\begin{cases}
%-\div(|y|^\gamma \nabla w)=\cds \left(\sum_{i=1}^{\mn} \delta_{x_i}-\mu \delta_{\R}\right)& \text{in }\Omega \times [-L,L], \(\text{resp. }(\Omega \times [-L,L])^c\) \\
%w =0 & \text{on }\partial (\Omega \times [-L,L]).
%\end{cases}
%\end{equation*}
When $u$ is inner/outer screenable at level $h$ with an $a$-bound on $e$, it is a competitor in the definition of $\G^{\mathrm{inn/out}}$, thus we have
\be \label{locali3}
\F^{\tilde\Omega} (\XN, \mu, \Lambda) \ge \G^{\mathrm{inn/out}}(\XN|_{\tilde \Omega},\tilde \Omega).
\ee

\subsubsection{Riesz screening}
Before we state the screening procedure formally let us give a heuristic description (see Figure \ref{fig1}). For a screenable field and configuration, we select by a mean-value argument a ``good boundary'' enclosing a set $\mathcal{O}$ (like old), in which we keep the configuration and field unchanged. We let $\New $ (like new) be the complement layer to $\Old$ and  place new points in $\New$ in a way that neutralizes the background measure, and define a new field in $\New\times [-h,h]$ with a zero Neumann boundary condition on $\partial \Old \times [-h,h]$. A novel component of the screening in the Riesz case is that we need to complete the field away from  the subspace $\R^\d$; this is done in $\Omega \times ([-R,R]\setminus [-h,h])$ by matching the current field at level $h$ and setting a Neumann zero boundary condition elsewhere.

The inner screening is analogous, except we are now working with the field in $\(\Omega \times [-R,R]\)^c$; see diagram (B) in Figure 1 below, where $E$ stands for the electric field $\nab w$.
\begin{figure}
\begin{subfigure}[t]{\textwidth}
\begin{flushleft}
\begin{tikzpicture}
%axis
\draw[->,ultra thick] (-5,0)--(5,0) node[right]{$x$};
%outline of Omega \times [-L,L]
\draw[ultra thick] (-3,-3) rectangle (3,3);
%outline of D_0
\draw[thick] (-2.5,-1.5) rectangle (2.5,1.5);

%description of D_0
\node[align=left] at (-1,1) {\footnotesize{keep field fixed here}};
\node[align=left] at (0.25,-0.5) {\tiny{keep configuration fixed}};
\draw [->, thick] (0.3,-0.35) to  (0.3,-0.05);

%dotted boundaries of D_partial
\draw[dashed] (-3, -1.5) -- (-2.5,-1.5); 
\draw[dashed] (2.5, -1.5) -- (3,-1.5); 
\draw[dashed] (-3, 1.5) -- (-2.5,1.5); 
\draw[dashed] (2.5, 1.5) -- (3,1.5); 

%new points

\draw (-2.6,0) node[cross=4pt, cyan, very thick] {};
\draw (-2.75,0) node[cross=4pt, cyan,very thick] {};
\draw (-2.9,0) node[cross=4pt, cyan, very thick] {};

\draw (2.6,0) node[cross=4pt, cyan, very thick] {};
\draw (2.75,0) node[cross=4pt, cyan, very thick] {};
\draw (2.9,0) node[cross=4pt, cyan, very thick] {};

%description of D_partial

\draw[->] (3.5,1) to [out=180, in=90] (2.75,0);
\node[align=left] at (4.6,1) {\tiny{place new points}};

\draw[->] (-3.5,-1) to [out=0, in=270] (-2.75,0);
\node[align=left] at (-4.6,-1) {\tiny{place new points}};

\draw[->] (-4,-2.5) to [out=60, in =210] (-2.5,-1);
\draw[->] (-4,-2.5) to (-2.75,-1.5);
\draw[->] (-4,-2.5) to [out=0, in=240] (-2,-1.5);
\node[align=left, text width=2cm] at (-4,-3) {\tiny{match boundary conditions}};

\draw[->] (4,2.5) to [out=315, in =15] (2.5,1);
\draw[->] (4,2.5) to (2.75,1.5);
\draw[->] (4,2.5) to [out=180, in=60] (2,1.5);
\node[align=left, text width=2cm] at (4.4,3) {\tiny{match boundary conditions}};

\node[align=left, rotate=90] at (-3.25, 1.5) {$E \cdot \vec{n}=0$};
\node[align=left, rotate=270] at (3.25, -1.5) {$E \cdot \vec{n}=0$};

%description of D_1
\node[align=left] at (0,3.25) {$E\cdot \vec{n}=0$};
\node[align=left] at (0,-3.25) {$E \cdot \vec{n}=0$};

%exterior field
\node[align=left] at (-4.5,3) {$E \equiv 0$};

%other annotations
 
% Calligraphic brace
\draw [decorate, decoration = {calligraphic brace,
        raise=5pt,
        amplitude=5pt}] (3,3) --  (-3,3);
        
       \node[align=left] at (0, 2.5) {$\Omega \times \{R\}$};
       
       \draw [decorate, decoration = {calligraphic brace,
        raise=5pt,
        amplitude=5pt}] (-2.5,-1.5) --  (2.5,-1.5);
        
       \node[align=left] at (0, -1) {$\mathcal{O} \times \{-h\}$};

\end{tikzpicture}
\caption{inner screening}
\end{flushleft}
\end{subfigure}
\begin{subfigure}{\textwidth}
\begin{flushright}
\begin{tikzpicture}
%axis
\draw[->,ultra thick] (-5,0)--(5,0) node[right]{$x$};
%outline of Omega \times [-L,L]
\draw[ultra thick] (-3,-3) rectangle (3,3);
%outline of D_0
\draw[thick] (-2.5,-1.5) rectangle (2.5,1.5);

%description of D_0
\node[align=left, text width=2cm] at (-4,1) {\tiny{keep field fixed here}};
\node[align=left] at (4.25,-0.5) {\tiny{keep configuration fixed}};
\draw [->, thick] (4.3,-0.35) to  (4.3,-0.05);

%dotted boundaries of D_partial
\draw[dashed] (-3, -0.75) -- (-2.5,-0.75); 
\draw[dashed] (2.5, -0.75) -- (3,-0.75); 
\draw[dashed] (-3, 0.75) -- (-2.5,0.75); 
\draw[dashed] (2.5, 0.75) -- (3,0.75); 

%new points

\draw (-2.6,0) node[cross=2.5pt, green, thick] {};
\draw (-2.75,0) node[cross=2.5pt, green, thick] {};
\draw (-2.9,0) node[cross=2.5pt, green, thick] {};

\draw (2.6,0) node[cross=2.5pt, green, thick] {};
\draw (2.75,0) node[cross=2.5pt, green, thick] {};
\draw (2.9,0) node[cross=2.5pt, green, thick] {};

%description of D_partial

\draw[->] (3.5,1) to [out=180, in=90] (2.75,0);
\node[align=left] at (4.6,1) {\tiny{place new points}};

\draw[->] (-3.5,-1) to [out=0, in=270] (-2.75,0);
\node[align=left] at (-4.6,-1) {\tiny{place new points}};

\draw[->] (-2,-2) to  (-3,-2);
\draw[->] (-2,-2) to [out=180, in=270] (-2.75,-0.75);
\draw[->] (-2,-2) to (-2.5,-3);
\node[align=left, text width=2cm] at (-1,-2.5) {\tiny{match boundary conditions}};

\draw[->] (2,2) to (3,2);
\draw[->] (2,2) to [out=0, in=90] (2.75,0.75);
\draw[->] (2,2) to  (2.5,3);
\node[align=left, text width=2cm] at (1,2.5) {\tiny{match boundary conditions}};

\node[align=left, rotate=270] at (-2.25, .9) {\tiny{$E \cdot \vec{n}=0$}};
\node[align=left, rotate=90] at (2.25, -.9) {\tiny{$E \cdot \vec{n}=0$}};

%description of D_1
\node[align=left] at (0,1.25) {\tiny{$E\cdot \vec{n}=0$}};
\node[align=left] at (0,-1.25) {\tiny{$E \cdot \vec{n}=0$}};

%interior field
\node[align=left] at (0,0.5) {$E \equiv 0$};

%other annotations
 
% Calligraphic brace
\draw [decorate, decoration = {calligraphic brace,
        raise=5pt,
        amplitude=5pt}] (-3,3) --  (3,3);
        
       \node[align=left] at (0, 3.5) {$\mathcal{O}^c \times \{R+h\}$};
       
       \draw [decorate, decoration = {calligraphic brace,
        raise=3pt,
        amplitude=3pt}] (-2.5,1.5) --  (2.5,1.5);
        
       \node[align=left] at (0, 1.8) {\scriptsize{$\Omega \times \{R\}$}};

\end{tikzpicture}
\caption{outer screening}
\end{flushright}
\end{subfigure}
\caption{The screening procedure}
\label{fig1}
\end{figure}
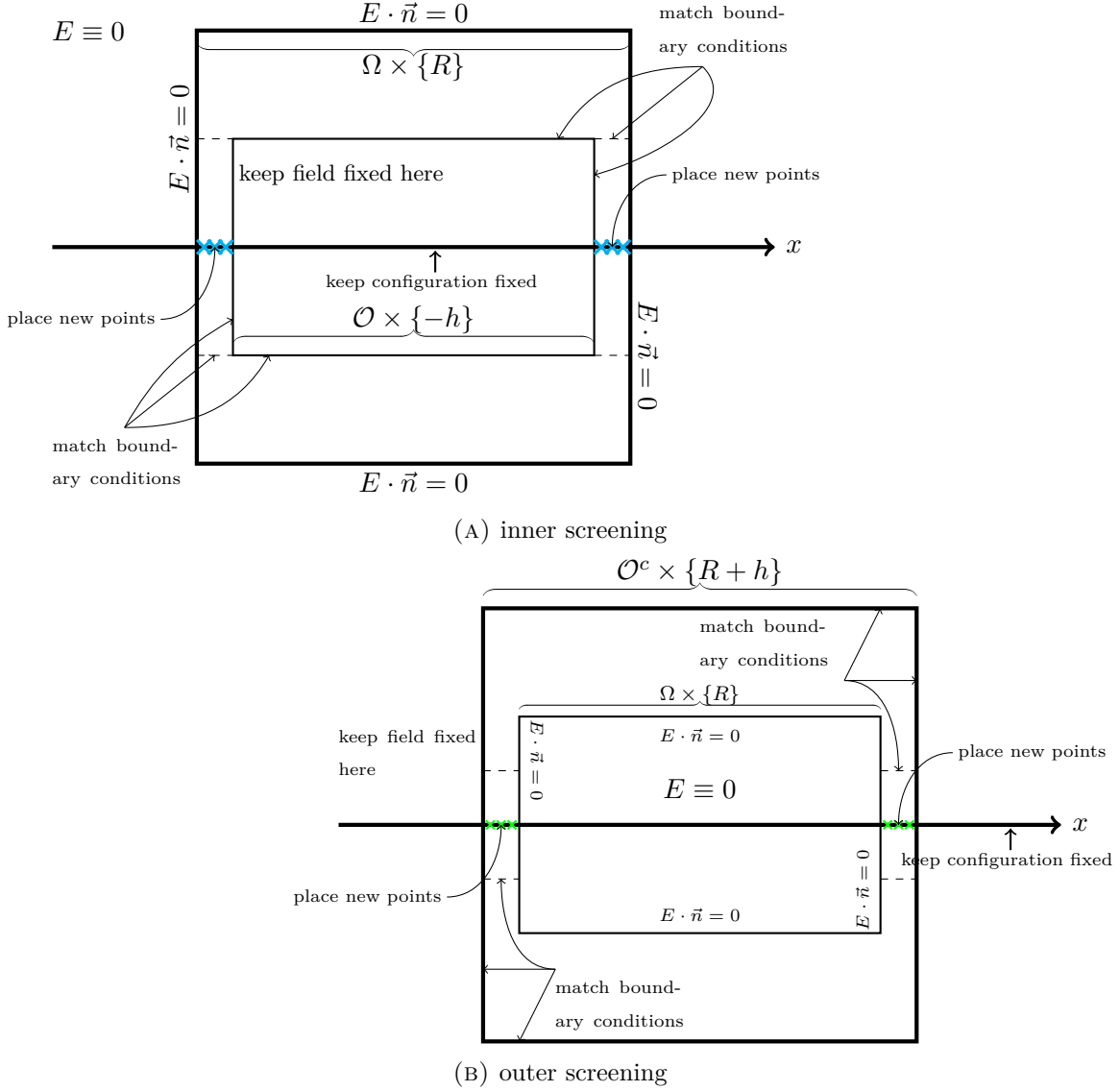

We now state the screening procedure formally.
\begin{prop}[Riesz Screening]\label{Riesz screening result}
Let us take the same assumptions as in Definition \ref{def: Riesz screenability}. Then, there exists a $C>5$ dependent only on $\d,\s, m$ and $\|\mu\|_{L^\infty}$ such that the following holds. Suppose that $X_n$ and  $w$ satisfy  \eqref{inner screening w} (respectively, \eqref{outer screening w}) and are screenable at height $h$ in the sense of Definition \ref{def: Riesz screenability}.  Then, there exists a  set $\Old\subset \R^\d$ such that $$(Q_{R-2\tilde \ell}\times \{0\})\cap \Lambda \subset \Old \times \{0\} \subset (Q_{R-\tilde \ell}\times \{0\}) \cap \Lambda \ \text{ (inner screening)}, $$ resp.  $$ (Q_{R+2\tilde \ell}^c\times \{0\}) \cap \Lambda \subset \Old\times \{0\} \subset (Q_{R+\tilde \ell}^c \times \{0\})\cap \Lambda \ \text{ (outer screening)},$$ a subset $I_{\partial} \subset \{1,\dots,n\}$, and a nonnegative density $\tilde{\mu}$ supported in $\New \subset \R^\d$ such that $\New\times \{0\}=Q_R\times \{0\} \cap \Lambda \setminus (\Old\times \{0\})$ in the inner screening and  $\New\times \{0\}=((Q_R^c\times \{0\})\cap \Lambda) \setminus ( \Old\times \{0\})$ in the outer screening, such that the following holds:
\begin{enumerate}
    \item  $n_\Old$ being the number of points of $X_n$ such that $B(x_i, \rrh_i)$ intersects $\Old$, we have
    \be \label{bornimp}
\tilde \mu(\New)= \mn-\N,\qquad |\mu(\New)-\tilde \mu(\New)|\le C\(R^{\d-1}+ \frac{S(X_n,w,h)}{\tilde \ell}\)
\ee
\be
\label{mmut2}  \|\mu -\tilde \mu\|_{L^\infty(\New)} \le \frac{m}{2},\qquad 
\int_{\New} (\tilde \mu-\mu)^2 \le C \frac{S(X_n,w,h)}{\l\tilde \ell}+ C \frac{1}{\ell} R^{\d-1}
\ee
    \item $\#I_\partial \le C \frac{S(X_n,w,h)}{\tilde{\ell}}$.
   % \cm{$M_0$ is not defined in what follows. we need to either define it or bound it directly by $e$ if possible?}
    \item The Neumann approximation is comparable to the original energy, i.e.
    \begin{align}\label{Riesz screening error}
 & \F(Y_{\mn}, \tilde{\mu}, Q_R \times [-R,R])-\(\frac{1}{2\cds}\(\int_{Q_R\times[-R,R]}|y|^\gamma|\nabla w_{\rrh}|^2-\cd \sum_{i=1}^n \g(\rrh_i)\) +\sum_{i=1}^n  \int \f_{\rrh_i}(x-x_i) d\mu(x)\)
\\ 
\notag& \le C \frac{h S(X_n,w)}{\tilde{\ell}}+\sum_{(i,j)\in J}\g(x_i-z_j)
+|n-\mn| +\(\frac{R^{2}}{\tilde{\ell}}+R\)e(X_n,w,h) \\
\notag & \hspace{3mm}+  \(\tilde{\ell}+h^{-\gamma}\)R^{\d-1}+\F(Z_{\mn-n_\Old}, \tilde{\mu}, \New\times[-h,h])
\end{align}
where
\begin{equation*}
J=\{(i,j) \in I_\partial \times \{1,\dots,\mn-n_\Old\}:|x_i-z_j|\leq \rrh_i\}
\end{equation*}
in the inner case, and  
\begin{align}\label{Riesz outer screening error}
    &\F(Y_{\mn}, \tilde{\mu}, (Q_R \times [-R,R])^c)\\
   \notag  & \qquad-\(\frac{1}{2\cds}\(\int_{(Q_R \times [-R,R])^c}|y|^\gamma|\nabla w_{\rrh}|^2 -\cd\sum_{i=1}^n \g(\rrh_i)\)+\sum_{i=1}^n  \int \f_{\rrh_i}(x-x_i) d\mu(x)\) \\
&\notag \le C  \frac{h S(X_n,w)}{\tilde{\ell}}+\sum_{(i,j)\in J}\g(x_i-z_j) 
+|n-\mn| +\frac{R}{\min(h,\tilde{\ell})}\(\frac{R^{2}}{\tilde{\ell}}+R\)e(X_n,w,h) \\
&\hspace{3mm}\notag+   \(\tilde{\ell}+h^{-\gamma}\)R^{\d-1}+\F(Z_{\mn-n_\Old}, \tilde{\mu}, \New\times[-h,h])
\end{align}
in the outer case. 
\end{enumerate}
Moreover, if  \eqref{Riesz screenability1} is used, we can take $ \Old$ to be equal to $Q_t\cap U$ for some $t \in [R-2 \tilde \ell, R-\tilde \ell]$ (resp. $U\backslash Q_t $ for some $t \in [R+\tilde \ell, R+ 2\tilde \ell]$) and we can make $t$ a measurable function of $X_n$, while if \eqref{Riesz screenability2} is used, the boundary of $\Old $ is in general piecewise affine and made of  facets parallel to the faces of $Q_R$ of sidelengths bounded above and below by constants times $\ell$, all included in some $Q_{t+\ell}\backslash Q_t$ for $t \in [R- 2\tilde \ell , R-\tilde \ell -\ell] $, resp.  $t\in [R+\tilde \ell, R+2\tilde \ell-\ell]$.

%\cm{it seems strange that the second equation has a factor $\frac{R}{\min(h,\tilde{\ell})}$ which is absent in the first}
%For inner screening, the same result holds in \eqref{Riesz screening error}) with $\F_{\mathrm{inn}}^\Omega(Y_{\mn}, \mu, U)-\F^{\mathrm{inn}}(w, \Omega$ replaced by $ \F_{\mathrm{ext}}^{\Omega^c}(Y_{\mn}, \mu, U)-\F^{\mathrm{ext}}(w, \Omega)$.
%\begin{align}\label{Riesz inner screening error}
%\nonumber \F_{\mathrm{ext}}^{\Omega^c}(Y_{\mn}, \mu, U)-&\F^{\mathrm{ext}}(w, \Omega)\lesssim_{m,\Lambda} \frac{\ell S(X_n, w)}{\tilde{\ell}}+\sum_{i,j}\g(x_i-z_j)+\tilde{\ell}L^{\d-1} \\
%&+\F_{\mathrm{inn}}^{\New_\eta}(Z_{\mn-n_\Old}, \tilde{\mu}, U) +|n-\mn|+L^{\d+1}h^{\gamma}M_0^2 +\frac{L^{1+1}}{h}e(X_n, w).
%\end{align}
\end{prop}
The method of proof essentially adapts and optimizes the approach of \cite[Section 6]{PS17} in an analogous way to the optimization of \cite[Appendix C]{AS21} for Coulomb gases. A thorough discussion can be found in \cite[Chapter 7]{S24}. The approach to outer screening is a novel adaptation for the Riesz gas. We will present the proof in Section \ref{sec: pfscn}.

\begin{remark}\label{height change}
It will be important to screen fields defined at heights $R'>R$ in the proof of almost additivity in Section \ref{sec: partition}. We can start with a field $\nabla w$ defined in $Q_R \times [-R',R']$ for any $R'>R$ and apply the screening procedure to its restriction to $Q_R \times [-R,R]$ to obtain a screened field in $Q_R \times [-R,R]$. Notice then that by appending an electric field that is uniformly zero for $|y|>R$ and using Lemma \ref{projlem} above, we can replace $\F(Y_{\mn}, \tilde{\mu}, Q_R\times [-R,R])$ in \eqref{Riesz screening error} with $\F(Y_{\mn}, \tilde{\mu}, Q_R\times [-R',R'])$ for any $R'>R$.  This yields
\begin{multline*}
\F(Y_{\mn}, \tilde{\mu}, Q_R \times [-R',R'])-\frac{1}{2\cds}\int_{Q_R\times[-R',R']}|y|^\gamma|\nabla w_{\rrh}|^2 \\
\leq \F(Y_{\mn}, \tilde{\mu}, Q_R \times [-R,R])-\frac{1}{2\cds}\int_{Q_R\times[-R,R]}|y|^\gamma|\nabla w_{\rrh}|^2,\end{multline*}
from which it follows that the screening result above allows us to screen in $Q_R \times [-R',R']$ for any $R'>R$ with the same errors as in \eqref{Riesz screening error}-\eqref{Riesz outer screening error}.
%\cm{still not sure what you mean by the second part of this remark. can you spell it out?}
\end{remark}

\section{Main Bootstrap}\label{sec: mb}
In this section we prove the main probabilistic control on local energies, following a bootstrap on scales. The argument is a generalization of \cite{P24} to the higher dimensional Riesz case, which was in turn a generalization of \cite{AS21} to the 1d-log case. See also \cite[Chapters 7-8]{S24} for a thorough description of the method.

%We first restate Theorem \ref{Local Law} from the introduction for ease of reference at blown-up scale.
%\setcounter{theo}{0}
%\begin{theo}\label{Local Law MB}
%Let $\Omega \subset \bulk'$, the blowup of the bulk $\bulk$, be a hypercube with sidelengths in $[L, 2L]$ for some $L >\omega$ with $\omega$ as in the proof of Proposition \ref{decay control} below. Let $u$ be the Riesz potential generated by a blown up configuration of points $\XN \subset \R^\d$, i.e.
%\begin{equation}\label{ll Rieszpot}
%u:=\g \ast \(\sum_{i=1}^N \delta_{x_i}-\mu\delta_{\R^\d}\)
%\end{equation}
%where $\mu$ is a measure satisfying $\int \mu= N$ and of bounded density. Then, there exists a scale independent positive constant $\C$ and a good event $\Gc_L$ such that 
%\begin{equation}
%\int_{\Omega \times [-L,L]}\yg |\nabla u_{\rr}|^2 \leq \C L^\d
%\end{equation}
%on $\Gc_\Omega$, with 
%\begin{equation}
%\PNbeta(\Gc_\Omega^c)\leq C_1e^{-C_2 \beta  L^{\d}}
%\end{equation}
%for some constants $C_1$ and $C_2$ dependent only on $\mu$.
%\end{theo}

We are going to prove local laws in blown-up scale. In order to prove Theorem \ref{Local Law}, we need to prove them for $\Lambda=\R^{\d+1}$ and $\mu$ equal to the blown-up of the equilibrium measure $\muv$. For the proof of the almost additivity of the energy in the next section, we will also need to have proven the local laws in cubes. This is why we continue to work with a generic density $\mu$ and  set $\Lambda$ in the setup of Section \ref{sec:subadditive}, and a general probability law $\mathbb{Q}_{\beta , H}$ as in \eqref{defQ}. In the case of a cube, the local laws will be valid up to the boundary. In the case of the whole space, the local laws are valid only in the bulk, i.e.~on cubes separated from the set where $\mu$ is small by a fixed positive distance in original scale.
For that reason, we let, in the case where $\Lambda=\R^{\d+1}$,
\be \bulk':=\{x\in \R^\d, \dist(x, \{\mu=0\}) >\ep N^{1/\d}\},\ee and assume that $\mu \ge m>0$ on $\bulk'$. Note that by  assumption \eqref{eq reg} we know that $\muv$ is bounded below at distance $\ep$ from $\partial \Sigma$, so $\muv'$ satisfies this condition.
In the case where $\Lambda$ is a hypercube of projection onto $\R^\d$ equal to $U$, we just let $\bulk'=U$.
All the constants in the local laws will depend on $m$, hence on  $\ep$.

We wish to prove that there exists constants $ C_1, C_2$ independent of the scale and of $\beta$, $\C_\beta$ independent of $\beta $ when $\beta \ge 1$,  such that,  $u$ being defined in \eqref{defiu}, for all $L>\rho_\beta$ and any cube $\carr_L \subset \bulk'$, there exists an event $\Gc_L$ such that
\be \label{locallawMB}
\forall \XN \in \Gc_L, \quad \ \F^{\carr_L\times [-L,L]}(\XN, \mu, \Lambda) +2 C_0 \#I_{\carr_L} 
 \leq \C_\beta  L^\d
\ee where $C_0$ is the constant in \eqref{arrivC0},  with the event $\Gc_L$ satisfying $\mathbb{Q}_{\beta, H}(\Gc_L^c)\leq C_1e^{-C_2 \beta  L^{\d}}$.

As in \cite{L17}, \cite{AS21} and \cite[Theorem 1]{P24}, this is achieved by an induction on the scale. 
To do so, we consider a cube $\carr_{2\ep}(z)\subset \Sigma'$ (later we will drop the $z$ from the notation). Then $\carr_\ep(z)\subset \bulk'$. 
We then consider $2^{-k_*} \ep=L$, and  assume that  there exists 
 an event $\mc{G}_{2L}$ with $\mathbb{Q}_{\beta, H}(\mathcal G_{2L}^c) \le C_1 e^{-C_2 \beta (2L)^\d}$ such that for every $1\le k \le k_*$
 \be \label{locallawMBind}
\forall \XN \in \Gc_{2L}, \quad \ \F^{\carr_{2^kL}\times [-2^k L,2^kL]}(\XN, \mu, \Lambda) +2 C_0 \#I_{\carr_{2^kL}} 
 \leq \C_\beta (2^kL)^\d.
\ee
Note that for $k \ge k_*$ this is also automatically satisfied thanks to \eqref{macrolaw2} (for a constant depending on $\ep$).
 
%\cm{are the $C_1$ and $C_2$ not going to change?} 
In view of \eqref{arrivC0}, this implies in particular 
\begin{multline} \label{induchyp}
\forall \XN \in \Gc_{2L}, \forall 1\le k, \quad  \int_{\square_{2^kL} \times [-2^kL,2^kL ]}\yg |\nabla u_{\rr}|^2 \leq 4\cds \mathcal{C}_\beta\(2^kL\)^\d\\
\quad \text{and} \quad C_0 \#I_{\carr_{2^kL}} \le \C_\beta  \(2^kL\)^\d.
\end{multline}
We then wish to prove that there exists an event $\Gc_L$ such that $\mathbb{Q}_{\beta, H}(\mathcal G_{L}^c) \le C_1 e^{-C_2 \beta L^\d}$ and such that for all $\XN \in \Gc_L$, \eqref{locallawMBind} holds for $k=0$ as long as $L\ge \rho_\beta$, with the same constants $C_1,C_2, \C$. This easily suffices to imply  that  \eqref{locallawMB} holds  for any cube $\carr_L \subset \bulk'$ as long as $L\ge \rho_\beta$.
Note that in order to prove that \eqref{locallawMBind} holds for $k=0$, it suffices to show it over a hyperrectangle $Q_L$ such that $\carr_L\subset Q_L \subset \carr_{\frac{3}{2}L}$, which will allow us to choose $Q_L$ such that $\mu(Q_L)$ is an integer.

\begin{prop}\label{Main Bootstrap} Let the setup be as in Section \ref{sec:subadditive} with $\Lambda$ equal to $\R^{\d+1}$ or a hypercube of height $H$,  $\mathbb{Q}_{\beta, H}(U, \mu, \zeta)$ as in \eqref{defQ} and $u$ as in \eqref{defiu}.
Suppose that there exists an event $\Gc_{2L}$ such that \eqref{induchyp} holds. Then, there is  a scale $\rho_\beta >0$ (depending only on $\beta$), $C>0$, with $4\le \rho_\beta \lesssim_\beta 1$  such that if $L \ge \rho_\beta$, the following holds. 
%Assume there is an event $\Gc_L \subset \Gc_{2L}$ such that 
\begin{equation}\forall \XN \in \Gc_L, \quad \ \F^{\carr_L\times [-L,L]}(\XN, \mu, \Lambda) + 2C_0 \#I_{\carr_L} \le \C_\beta  L^\d
\end{equation} with  
\begin{equation}
\mathbb{Q}_{\beta, h}(U, \mu, \zeta)(\Gc_{2L} \setminus \Gc_L)\leq  e^{-C\beta L^{\d}}
\end{equation}
with $ C>0$ depending only on $\d,\s,m,\ep, \|\mu\|_{L^\infty}$ and $\C_\beta\ge 1$ also possibly depending on $\beta$ when $\beta \le 1$.
\end{prop}
Once this proposition is proved, Theorem \ref{Local Law} follows 
 by a bootstrap on the scale starting from \eqref{macrolaw1} and by using \eqref{arrivC0} in the same way as the proof of \cite[Theorem 3.1]{P24} from \cite[Proposition 3.2]{P24}.
 %\cm{can't we just summarize how it's done? it's short}
 One starts at the macroscopic scale in $\bulk'$, where one has the local law outside of an exponentially small event by \eqref{macrolaw1}, and applies Proposition \ref{Main Bootstrap} iteratively down to scale $L$. At each application at scale $2^kL$, we lose an event of probability no more than $ e^{-C \beta 2^kL^{\d}}$, so the local law at
scale $L$ holds off of an event of size at most
 \begin{equation*}
 \sum_{k=0}^\infty  e^{-C\beta 2^kL^{\d}} \leq C_1 e^{-C_2\beta L^\d},
 \end{equation*} for some constants $C_1$ and $C_2$ independent of scale $L$.

 %\cm{we also need to invoke some union bound argument and a covering of the domain with cubes of size $L$ to ensure that this is true irrespective of the cube center}
 
The key technical tool that we use in the proof of Proposition \ref{Main Bootstrap} is the screening procedure Proposition \ref{Riesz screening result}, which allows us to localize our rather nonlocal next order energy and exhibit an almost additivity on scales. In order to successfully employ this procedure, we will need to guarantee that the errors generated when we screen are sufficiently small.

\subsection{Control of Screening Errors}
The first result we will need is on the rate of decay of  the electric field  away from the subspace $\R^\d$, which allows to control the $e$ terms (as in \eqref{Riesz upper energy inner screening}) in the screening estimates, a question which is  absent in the Coulomb case. This is done by viewing the electric field as a fluctuation and using as an input the local law at scales $2^kL$ \eqref{induchyp}.
In the sequel we let $\s_+=\max(\s, 0)$.

\begin{prop}[Decay estimate - control of the $e$ term]\label{decay control} Let $\Lambda $ be $\R^{\d+1}$ or a hyperrectangle as above. Let $L$ be such that $L=2^{-k_*}\ep$, with $\carr_{ \ep} \subset \bulk'$ in the case $\Lambda = \R^{\d+1}$. 
 \begin{enumerate}
 \item Let $\mathcal{G}_{2L}$ be the event that \eqref{induchyp} holds. Then, for any $\epsilon>0$ and $K \geq 1$ large enough so that
 \begin{equation*}
 \frac{L}{K} \leq \epsilon^{\frac{2}{4+\s-\d}}N^{\frac{1}{\d}},
 \end{equation*}  
 there exists $4\le \rho_\beta\lesssim_\beta 1$ and a constant $C>0$ dependent only on $K$ and $\epsilon$ such that if $L>C (\beta^{-1}\log \frac1\epsilon )^{\frac1\d}$ and $h=L/K$, there is an event $\mathcal{G} \subset \mathcal{G}_{2L}$ such that
 \begin{equation}\label{local control of top energy}
\int_{\carr_{2L} \times \{\pm h\}}|y|^\gamma|\nabla u|^2 \lesssim_\beta \epsilon L^{\d-1}\quad \mathrm{on}  \ \mathcal G,
\end{equation}with 
\begin{equation}\label{local event}
\mathbb{Q}_{\beta, H}(U, \mu, \zeta)(\mathcal{G}_{2L} \setminus \mathcal{G})\leq e^{-C \beta L^\d}.
\end{equation}
\item 
Let $\mathcal{G}_{h}$ 
be the event that \eqref{induchyp} holds for $2L=\rho_\beta$. Then,  there exists a constant $C>0$ such that given $M \ge 1$,    if $L>h\ge \rho_\beta$, there is  an event $\mathcal{G} \subset \mathcal{G}_{h}$ such that
 \begin{equation}\label{local control of top energy2}
\int_{\carr_{2L} \times \{\pm h\}}|y|^\gamma|\nabla u|^2 \lesssim_\beta M L^{\d} h^{\s_+-\d-1}\quad \mathrm{on} \ \mathcal G,
\end{equation}
with 
\begin{equation}\label{local event2}
\mathbb{Q}_{\beta, H}(U, \mu, \zeta)(\mathcal{G}_{h} \setminus \mathcal{G})\leq   M^{-\frac{\d}{2}}L^\d h^{\frac{\d( \d-\s_+-2)}{2}} e^{-CM}.
\end{equation}
\end{enumerate}
%\begin{remark}
%As will be important for outer screening, \eqref{local control of top energy} also holds when we replace the domain of integration by $\carr_{2L} \times \{\pm (L+h)\}$; the proof is the same as what follows below.  \cm{review where used}
%\cm{not clear we need that remark}
%\end{remark}
\end{prop}
\begin{proof}Let us first consider the case $\Lambda = \R^{\d+1}$. In that case, $u$ given by \eqref{defiu} is, up to an additive constant, equal to 
$\g*\(\sum_{j=1}^N \delta_{x_j'}-\mu\)$. By symmetry in $y$, it suffices to prove the result at height $+h$.
%The proof is analogous to \cite[Proposition 3.3]{P24}, and  we will abbreviate some computational details when the analogy is clear. 
For notational ease we will let $E$ be a shorthand for $\nabla u$. The main idea rests on the observation that the components of $E_i$, $1 \leq i \leq \d+1$, at a point $(\vec a, h)=(a_1,a_2,\dots,a_\d,h)$ with $h>0$, are fluctuations of smooth linear statistics. Namely, denoting $x_j'$ for the points of the configuration (at the blown-up scale) and spelling out $\nab \g$, we have 
\begin{equation}\label{spellout}
E_i(\vec a,h)=\int_{\R^\d} \kappa_{(\vec a,h)}^k(x)d\Bigg(\sum_{j=1}^N \delta_{x_j'}-\mu\Bigg)(x)=\Fluct_{\mu}\(\kappa_{(\vec a,h)}^k\(N^{\frac{1}{\d}}x\)\)
\end{equation}
where
\begin{equation}\label{LL test functions}
\kappa_{(\vec a,h)}^i(x)=\begin{cases}
\frac{-(a-x)_i}{\(|x-a|^2+h^2\)^{\frac{\s+2}{2}}} & \text{if }1 \leq i\leq \d \\
\frac{- h}{\(|x-a|^2+h^2\)^{\frac{\s+2}{2}}} & \text{if }i=\d+1.
\end{cases}
\end{equation}
With this observation, the main idea is as follows:
\begin{itemize}
\item Place $P$ points on $\carr_{2L} \times \{ h\}$; estimate $E(z)=E(z_p)+(E(z)-E(z_p))$
\item Estimate $E(z_p)$ in probability using Theorem \ref{FirstFluct} and the  local laws.
\item Use elliptic regularity and the local laws to control $\nab_{\R^\d} E$ and thus  $E(z)-E(z_p)$.
\end{itemize} 
\textbf{Step 1: Setup.}
We split $\carr_{2L}$ into $P$ equally sized subrectangles $I_p$ of sidelength comparable to $L P^{-1/\d}$, and let $z_p=(\vec a_p, h)$ denote the center of the subcube $I_p$. On each subcube $I_p$, we estimate
\begin{align}\notag
\int_{I_p \times \{h\} }|E|^2&\leq 2\sum_{i=1}^{\d+1}\(\int_{I_p \times \{ h\} }|E_i(z)-E_i(z_p)|^2+\int_{I_p  \times \{ h\} }|E_i(z_p)|^2\)\\
\notag &\lesssim \|\nabla_{\R^\d} E\|_{L^\infty(I_p \times \{ h\} )}^2\int_{I_p \times \{h\} }|z-z_p|^2+\frac{L^\d}{P}\sum_{i=1}^{\d+1}|E_i(z_p)|^2 \\
\label{intIp}&\lesssim \frac{ L^{\d+2}}{P^{1+\frac{2}{\d}}} \|\nabla_{\R^\d} E\|_{L^\infty(I_p \times \{ h\} )}^2+\frac{ L^\d}{P}|E(z_p)|^2.
\end{align}
where $\nabla_{\R^\d}$ denotes gradient in $\R^\d$. We now estimate each term separately.

{\bf Step 2. Control of  $|E(z_p)|^2$}. This is done via estimates on fluctuations of linear statistics associated to the functions $\kappa_{(\vec a,h)}^i$. The main difference between our approach here and that of \cite[Proposition 3.3]{P24} is that we use Theorem \ref{FirstFluct} for a rescaled test function directly, instead of running transport estimates for the functions $\kappa^i$. Even though these test functions are not literally rescaled versions of a compactly supported test function, we can treat them as such after a dyadic splitting. 
%
% cutoff function that is identically one on a ball of macroscopic $(N^{1/\d})$ order inside of the blowup of the bulk $\bulk'$, vanishes outside of a ball of radius twice that still contained inside of $\bulk'$, and between $0$ and $1$ everywhere else. 

Without loss of generality, let $\vec a=\vec 0$; we also focus on $i=\d+1$, since the computation for $i\leq \d$ is analogous and produces the same result. Let $x$ be the space  variable at the non-blown-up scale. Notice that 
\begin{equation*}
h^{1+\s}\kappa^{\d+1}_{\vec 0, h}\(N^{\frac{1}{\d}}x\)=\frac{- h^{2+\s}}{\(|N^{\frac{1}{\d}}x|^2+h^2\)^{\frac{\s+2}{2}}}=-\frac{1}{\(| \frac{N^{\frac1\d}x}{h} |^2+1\)^{\frac{\s+2}{2}}}:=\varphi_0\(\frac{x N^{1/\d}}{h}\),
\end{equation*}
where $\varphi_0$ is a smooth scale-independent function. Let $\chi_k$ be a partition of unity associated to dyadic annuli $B(0,2^{k+2})\backslash B(0,  2^{k-2})$ (as in Section \ref{sec: fluct prelim}). We may choose $k_*$ such that $2^{k_*}\ell$ is bounded below by $\ep>0$ and $B(0, 2^{k_*}\ell) $ does not intersect any other connected component of $\supp \mu$ (if there is more than one) than that of $0$.

We may then write $\varphi_0= \sum_{k=0}^{k_*}  (\chi_k \varphi_0)+ \varphi_1$, with $\chi_k\varphi_0$ a function supported in a dyadic annulus and $\|\varphi_1\|_{L^\infty}$ is  bounded by $ O((h N^{-1/\d})^{\s+2})$. Computing directly shows as well that 
\begin{equation*}
\left\|\nab\(\varphi_1\(\frac{\cdot N^{1/\d}}{h}\)\)\right\|_{L^\infty}=O\(\(hN^{-1/\d}\)^{\s+2}\).
\end{equation*}
The function $(2^k)^{\s+2}\chi_k \varphi_0 $ satisfies
\eqref{estxiintro} at scale $2^k\ell$  for some constant $\M>0$. 
 We may thus apply Theorem \ref{FirstFluct} to it. 
 %after noting that it satisfies \eqref{assconncomp} because it is supported in a unique connected component of $\supp \mu$. \cm{we need to justify, is this true?}
For the proof of item (1) of the proposition, we 
use the local laws at scales $2^kL$ to have a control at scale $2^k h N^{-1/\d}= 2^k\frac{L}{K} N^{-1/\d}$,  we thus obtain the bound stated in  Theorem \ref{FirstFluct} : for $\tau_k (2^k h)^{\s-\d}$ small enough, 
\be
\left|\log\Esp_{\PNbeta}\left[\exp\( \tau_k \frac{\beta}{1+\beta}\Fluct_{\mu}\((2^k)^{\s+2}(\chi_k \varphi_0)\( \frac{\cdot N^{1/\d}}{h} \)\)\indic_{\mathcal{G}_{2L}}\)\right]\right|\lesssim_\beta(|\tau_k|+|\tau_k|^2) (2^k h)^\s\ee
where the constant depends on $K$. 
Applying to $\tau_k = \tilde \tau_k (2^k h)^{\d-\s} $, we find that if $\tilde \tau_k$ is small enough, 
\begin{multline*}
\left|\log\Esp_{\PNbeta}\left[\exp\( \frac{\beta}{1+\beta}\tilde \tau_k (2^k h)^{\d-\s} (2^k)^{\s+2}\Fluct_{\mu}\((\chi_k \varphi_0)\( \frac{\cdot N^{1/\d}}{h} \)\)\indic_{\mathcal{G}'_{2L}}\)\right]\right|\\
\lesssim_\beta (\tilde \tau_k (2^k h)^\d+\tilde \tau_k^2 (2^kh)^{2\d-\s} ) .\end{multline*}

For any $\lambda$, Markov's inequality yields 
\begin{multline}\label{fluct Chern}
\PNbeta\(\left\{\Fluct_{\mu}\((\chi_k \varphi_0)\( \frac{\cdot N^{1/\d}}{h} \)\)\geq \lambda\right\} \cap \mathcal{G}_{2L}\)\\
%&\leq \exp\(-\tau ah+C\(|\tau|+|\tau|^3h^{-\s-\d}+\tau^2h^{-\s}\)\) \\
\leq \exp\(C_\beta( \tilde \tau_k(2^k h)^{\d} +\tilde \tau_k^2 (2^k h)^{2\d-\s})- \tilde \tau_k\frac{\beta}{1+\beta}(2^k h)^{\d-\s}(2^k)^{\s+2}\lambda\)
\end{multline}
hence, choosing $ \tilde \tau_k= \epsilon\frac{\beta}{C_\beta}(2^k h)^{\frac{\s-\d}{2}}$ and $ \lambda =2\epsilon \frac{1+ \beta}{\beta} (2^k h)^{\frac{\d+\s}{2}} (2^k)^{-\s-2}$
 we find that as soon as $\epsilon $ is small enough and $h$ is large enough (depending on $K$ and the other constants),  in $\mathcal{G}_{2L}$, it holds that
\be \Fluct_{\mu}\((\chi_k\varphi_0)\(\frac{\cdot N^{1/\d}}{h}\) \)\le 2\epsilon \frac{1+ \beta}{\beta}(2^k h)^{\frac{\d+\s}{2}}(2^k)^{-\s-2} \ee
except on an event of probability $\le \exp(-\beta\epsilon^2(2^{k}h)^\d)$. 
Using Proposition \ref{pro:controlfluct} at macroscopic scale, we also have 
\be  \Fluct_{\mu}\(\varphi_1\(\frac{\cdot N^{1/\d}}{h}\)\) \lesssim_\beta \(hN^{-1/\d}\)^{\s+2}N^{\frac{1}{2}+\frac{\s}{2\d}} .\ee
Taking a union bound on all these events, summing and using that $\s>\d-2$, we obtain that except for an event of probability $\le  \exp(- C'\beta\epsilon^2 h^\d)$, we have 
\begin{multline} \Fluct_{\mu}\(\kappa^i_{\vec{0}, h} \)\lesssim_\beta     h^{-1-\s }\( \epsilon h^{\frac{\d+\s}{2}}  \sum_{k=0}^{k_*} (2^k)^{\frac{\d-\s}{2}-2}+ \(hN^{-1/\d}\)^{\s+2}N^{\frac{1}{2}+\frac{\s}{2\d}} \)\\
\lesssim_\beta   \epsilon   h^{-1+\frac{\d-\s}{2}}+hN^{\frac{\d-\s-4}{2\d}}. \end{multline}
The second term can be absorbed into the first, since $h\le \epsilon^{\frac{2}{4+\s-\d}} N^{1/\d}$  and $ \d-\s-4<0$, hence
\begin{equation*}
hN^{\frac{\d-\s-4}{2\d}}\le \epsilon h^{1+\frac{\d-\s-4}{2}}=\epsilon h^{-1+\frac{\d-\s}{2}}.
\end{equation*}

Hence, we conclude that 
\be \label{boundEzp} |E(z_p)|^2\lesssim_\beta \epsilon^2   h^{\d-\s-2},\ee except on an event of probability $\le \exp(-C\beta\epsilon^2 h^\d)$.
For the proof of item (2) of the proposition where $h$ and $L$ are no longer comparable, we run the same argument, using that local laws hold on scale $h$ directly, except that after \eqref{fluct Chern}, we  choose 
$$\begin{cases}
 \tilde \tau_k= \frac1{\sqrt{C_\beta}}(2^k)^\ep \sqrt M (2^kh)^{-\d+\frac\s2}, \quad\lambda=\frac{2}{\sqrt{C_\beta}} \frac{1+\beta}{\beta}  (2^k)^\ep \sqrt M (2^k h)^{\frac\s2}(2^k)^{-\s-2}& \text{if} \ \s< 0\\
\tilde \tau_k= \frac1{C_\beta}(2^k)^\ep M(2^kh)^{-\d}, \quad  \quad\quad\lambda= 2  \frac{1+\beta}{\beta}  (2^k)^\ep(2^k h)^{\s}(2^k)^{-\s-2} & \text{if} \ \s\ge 0.
\end{cases}
$$ This yields that in $\mathcal G_{h}$, except with probability $\exp(-(2^k)^\ep M)$, we have
$$\Fluct_\mu \((\chi_k\varphi_0)\(\frac{\cdot N^{1/\d}}{h}\) \)\lesssim_\beta  \begin{cases}
\sqrt{M}  (2^k h)^{\frac{\s}{2}}(2^k)^{-\s-2+\ep}& \text{if} \ \s< 0\\ (2^k h)^{\s}(2^k)^{-\s-2+\ep} & \text{if} \ \s\ge 0.\end{cases}
 $$
Taking a union bound over these events and summing over $k$ in the same manner as above, we conclude that, choosing $\ep>0$ small enough,
\be \label{boundEzp2}|E(z_p)|\lesssim_\beta   \begin{cases}
h^{-1-\frac{\s}{2}}\sqrt{M}  & \text{if} \ \s< 0\\
h^{-1} & \text{if} \ \s\ge 0\end{cases}\lesssim_\beta \sqrt M h^{-1+(-\hal\s)_+}.\ee
where $(\cdot)_+$ denotes the positive part.

{\bf Step 3. Bound on $|\nab_{\R^\d} E|$.}  Let us start with the case of item (1).  We observe that, if $K \ge 2$, $2h \le L$, hence in 
$\carr_L \times [h/2, 2h]$, if $\carr_L \subset \bulk'$,  the bound 
$$\int_{\carr_L \times [h/2, 2h]}\yg |E|^2 \lesssim_\beta h^\d$$
is verified from \eqref{locallawMB} in $\mathcal G_{2L}$, with a constant depending on $K$.
In addition, $\nab_{\R^\d} E$ satisfies 
$$-\div (\yg \nab_{\R^\d} E) = 0 \quad \text{in} \ \carr_L \times [h/2, 2h].$$
Elliptic regularity estimates then yield that $x$ being the center of $\carr_L$,
\be\label{bdnabR2E} |\nab_{\R^\d} E(x, h)|^2 \lesssim_\beta  \frac{1}{h^{2+\d+1+\gamma}} \int_{\carr_h \times [h/2, 2h]}\yg |E|^2 \lesssim_\beta \frac{h^\d}{h^{3+\d+\gamma}}\lesssim_\beta h^{-4  +\d-\s}\ee
using $\d-1+\gamma=\s$.
For the proof of item (2), we use that the local laws hold down to scale $h$ hence all the estimates above hold directly on $\carr_h \times [-h/2,2h]$ and obtain the same result.

{\bf Step 4: Conclusion in the case of the whole space.} For item (1),
inserting \eqref{bdnabR2E} and \eqref{boundEzp}  into \eqref{intIp}, and recalling that $\d+\gamma=\s+1$, taking a union bound over the bad events, we obtain that except with probability $\le P e^{-C\beta L^\d}$, we have, using $\gamma+\d-\s=1$,
\begin{equation*}
\int_{\carr_L\times \{h\}} \yg |E|^2 \lesssim_\beta  P h^\gamma \( \frac{L^{\d+2}}{P^{1+\frac2\d}}h^{\d-\s-4}+\frac{L^\d}{P} \epsilon^2  h^{\d-\s-2}\)
\lesssim_\beta \frac{L^{\d-1}}{P^{\frac2\d}}+\epsilon^2  L^{\d-1} .
\end{equation*}
Choosing $P=  \epsilon^{-\frac\d2}$, we obtain the desired result since $\log P$ can be absorbed into $O(\beta L^\d)$ when $L$ is larger than a constant times $(\beta^{-1}\log \frac1\epsilon )^{\frac1\d}$.

For item (2), we obtain instead, inserting \eqref{bdnabR2E} and \eqref{boundEzp2}  into \eqref{intIp},
\begin{align}\notag
\int_{\carr_L\times \{h\}} \yg |E|^2 \lesssim_\beta  P h^\gamma \( \frac{L^{\d+2}}{P^{1+\frac2\d}}h^{\d-\s-4}+ M   \frac{L^\d}{P} h^{\max(-2\s,0)-2}\)\\
\lesssim_\beta \(\frac{L^{\d+2}}{P^{\frac2\d}h^3}+M  L^{\d} h^{\s-\d-1+(-\s)_+}\) .
\end{align}
Taking $P= M^{-\frac{\d}{2}}L^\d h^{\frac{\d( \d-\s_+-2)}{2}}$  to equate the last two terms yields the result.
We have thus concluded the proof in the case where $\Lambda = \R^{\d+1}$. 

{\bf Step 5. The case where $\Lambda$ is a hypercube.} 
First, let us justify that Theorem \ref{FirstFluct} holds as well for the Neumann energy setup as follows: let $U$ be a hyperrectangle in $\R^\d$ of sidelengths in $[\ell/2,2\ell]$ at the original scale, $\mu$ a density bounded below in $U$ with $N\mu(U)=\mn$ integer, $\F_N$ defined  as in \eqref{minneum}, except at the original scale, and the associated Gibbs measure $\mathbb{Q}_{N,\beta, H}(U,\mu)$
and partition function $\K_{N,\beta, H}(U,\mu) $ (we can assume that $\zeta=0$) as in \eqref{defQ} and \eqref{defK} but in original scale.
Let $\varphi$ be a test function in $U$ satisfying $\frac{\partial \varphi}{\partial n}=0$. While our analysis in Sections \ref{sec: tport}-\ref{sec: fluct} is based on an analysis via transport before applying the splitting formula Lemma \ref{Riesz Splitting Formula}, we can conduct a similar analysis post splitting as in \cite[Sections 3-4]{LS18}, finding that the expansion of next-order partition functions is governed by terms $T_1$ and $T_2$. More precisely, we can write analogously to Lemma \ref{lem2.1} that 
\begin{multline*}
\E_{\mathbb{Q}_{N,\beta, H}(U,\mu)} \left[ \exp\(-\beta t N^{1-\frac\s\d}\Fluct_{\mu}(\varphi)\)\indic_{\mathcal G}\right]
\\=\frac{1}{\K_{N,\beta, H}(U,\mu)} \int_{\mathcal G} \exp\(- \beta N^{-\frac\s\d} ( t N\Fluct_{\mu}(\varphi)+ \F_N(\XN, \mu, U\times [-H,H])) \) dX_{N}
\end{multline*}
We can use \eqref{737} to observe, by ``completing the square'' that if $\nu $ solves 
\be\label{intGU}
\int_U G_U(x,y) d\nu(y)=\varphi(x)\ee
then
\begin{multline}\label{rewrLaplaceneum}
\E_{\mathbb{Q}_{N,\beta, H}(U,\mu)} \left[ \exp\(-\beta t N^{1-\frac\s\d}\Fluct_{\mu}(\varphi)\)\indic_{\mathcal G}\right]
\\= \exp\(-\frac{\beta}{2} N^{2-\frac\s\d} t^2 \iint_{U^2}G_U(x,y) d\nu(x)d \nu(y)\)\frac{\K_{N,\beta,H}(U, \mu+t\nu) }{\K_{N,\beta, H}(U,\mu)}.\end{multline}
To solve \eqref{intGU}, we may reflect and periodize $\varphi$ across the faces of the hyperrectangle $U$, in such a way that $\varphi$ remains continuous and $\nab \varphi$ as well (thanks to the assumption $\nab \varphi \cdot \vec{n}=0$ on $\pa U$). Call $\varphi^{\mathrm{per}}$ the reflected and periodized function, then we can check that $\nu= \frac1{\cds}(-\Delta)^{\alpha}(\varphi^{\mathrm{per}})$ computed over $\R^\d$ solves \eqref{intGU} in $U$, where the fractional Laplacian of a periodic function is defined via the Fourier series representation as in \cite{RS15}. 
%\cm{does this make sense?? can we even take the fractional laplacian of a periodic function -- maybe in the Fourier sense?}
We can then analyze \eqref{rewrLaplaceneum} as in the case of the full space by using the transport map
\begin{equation*}
\begin{cases}
\div(\psi \mu)=\nu & \text{in }U\\
\psi \cdot \vec n =0 & \text{on }\partial U.
\end{cases}
\end{equation*}
for which estimates as a function of $\varphi$ are easy to obtain, in fact more easily than in the full space case treated in Section \ref{sec: tport}.
Starting from \eqref{rewrLaplaceneum}, and using the Neumann transport calculus as in \cite[Section 9.2]{S24}, we may then obtain the analogue of Theorem \ref{FirstFluct} for $\varphi$. Spelling out $\nab G_U$ we obtain an expression for the electric field in the Neumann case analogous to \eqref{spellout}, which expresses it as the fluctuation of a function in the class of the $\varphi$'s just analyzed and we can complete the proof in the same way.

\end{proof}

This estimate allows us to better understand the screening errors in Proposition \ref{Riesz screening result}, and show that the initial and screened fields are comparable.

\begin{coro}\label{tailored screening} Let $\Lambda $ be $\R^{\d+1}$ or a hyperrectangle as above. Let $\carr_L(z)$ be as above, i.e. such that $\carr_{2^{k_*}L}= \carr_\ep\subset \bulk'$ in the case $\Lambda = \R^{\d+1}$.
Let $\ \Omega=Q_L\cap \Lambda$ with $Q_L$ a hyperrectangle such that $\carr_L\cap \Lambda\subset Q_L \cap \Lambda \subset \carr_{\frac{3}{2}L}\cap \Lambda$, and $\int_\Omega \mu$ is an integer.
Let $X_n=\XN\vert_\Omega$ and let $Y_{\mn}$ be as in Proposition \ref{Riesz screening result}, and recall the definitions of $\G^{\mathrm{inn}}_{a,h}$ and $\G^{\mathrm{ext}}_{a,h}$ from \eqref{innernrj}. 
We next make the choice
\begin{equation}\label{def a}
a=\begin{cases}
C_\beta\epsilon L^{\d-1} & \text{ in case }(1) \text{ below}\\
C_\beta ML^\d h^{\s_+-\d-1} & \text{ in case }(2) \text{ below},
\end{cases}
\end{equation} for the $C_\beta$ implicitly appearing in \eqref{local control of top energy}, resp. \eqref{local control of top energy2}.
 \begin{enumerate}
 \item Let $\mathcal{G}_{2L}$ be the event in \eqref{induchyp}. Let $\epsilon>0$. There exists $K, K_1\ge 1$ such that for  $h=L/K$, $h \ge \rho_\beta $, $\tilde{\ell}=\frac{L}{K_1}$, and for  $\mathcal{G} \subset \mathcal{G}_{2L}$ as in part (1) of Proposition \ref{decay control}, we have
for all configurations $X_N$  in $\mathcal{G}$, 
  \begin{multline}\label{tailored inner screening error - ll}
\F(Y_{\mn}, \tilde{\mu}, \Omega \times [-L,L])-\G^{\mathrm{inn}}_{a,h}(X_n, \Omega \times [-L,L])\lesssim \\\sum_{(i,j)\in J}\g(x_i-z_j)+\F\(Z_{\mn-n_\Old}, \tilde{\mu}, \New \times [-h,h]\) +|n-\mn|+C_\beta \epsilon L^\d,
\end{multline}
and
\begin{multline}\label{tailored outer screening error - ll}
\F(Y_{\mn}, \tilde{\mu}, (\Omega \times [-L,L])^c)-\G^{\mathrm{ext}}_{a,h}(X_{N-n}, (\Omega \times [-L,L])^c) \lesssim \\ \sum_{(i,j)\in J}\g(x_i-z_j)+\F(Z_{\mn-n_\Old}, \tilde{\mu}, \New \times [-L,L]) +|n-\mn|+C_\beta\epsilon L^\d.
\end{multline}
\item 
Let $\mathcal{G}_{h}$ be as in item (2) of Proposition \ref{decay control}.
%be the event that \eqref{induchyp} holds at scale $\rho_\beta$. 
Let $M \ge 1$,  $L>h\ge \rho_\beta$, and $\mathcal{G} \subset \mathcal{G}_{h}$ as in item (2) of Proposition \ref{decay control}. Suppose that
\begin{equation}\label{height control}
C_\beta \frac{ML^2 h^{\s+\s_+-2\d}}{\tilde{\ell^2}} \ \mathrm{sufficiently } \, \mathrm{small} , \qquad C_\beta h L^{\frac{1-\d}{\d-\gamma}}\ \mathrm{sufficiently } \, \mathrm{large}\,  \text{ if }\, \s \geq \d-1.
\end{equation}
Then, for all configurations in $\Gc$,
 \begin{multline}\label{tailored inner screening error - gen}
\F(Y_{\mn}, \tilde{\mu}, \Omega \times [-L,L])-\G^{\mathrm{inn}}_{a,h}(X_n, \Omega \times [-L,L])\lesssim \sum_{(i,j)\in J}\g(x_i-z_j)+|n-\mn|+\\
\F(Z_{\mn-n_\Old}, \tilde{\mu}, \New \times [-h,h]) +\(\frac{h^2}{\tilde{\ell}}+\tilde{\ell}+h^{-\gamma}\)L^{\d-1}+C_\beta\( \frac{L^{2}}{\tilde{\ell}}+L\)M L^{\d} h^{\s_+-\d-1}
\end{multline}
and
\begin{multline}\label{tailored outer screening error - gen}
\F(Y_{\mn}, \tilde{\mu}, (\Omega \times [-L,L])^c)-\G^{\mathrm{ext}}_{a,h}(X_{N-n}, (\Omega \times [-L,L])^c) \lesssim \sum_{(i,j)\in J}\g(x_i-z_j)+|n-\mn|+\\ \F(Z_{\mn-n_\Old}, \tilde{\mu}, \New \times [-h,h]) +\(\frac{h^2}{\tilde{\ell}}+\tilde{\ell}+h^{-\gamma}\)L^{\d-1}+C_\beta \frac{L}{\min(h,\tilde{\ell})} \( \frac{L^{2}}{\tilde{\ell}}+L\) M L^{\d} h^{\s_+-\d-1}.
\end{multline}
\end{enumerate}
\end{coro}
\begin{proof}
We will apply Proposition \ref{Riesz screening result}, and control the error terms using Proposition \ref{decay control}. Let us focus on the proof for the inner screening, since the outer case is analogous.  Let $w$ be any potential associated to $\Omega$ as in \eqref{inner screening w} with 
\begin{align*}
\frac{1}{2\cds}\int_{\Omega \times [-L,L]}\yg |\nabla w_{\rrh}|^2 &\leq \frac{1}{2\cds}\int_{\Omega \times [-L,L]}\yg |\nabla u_{\rrh}|^2 \\
   \int_{\Omega \times \{\pm h\}}|\nabla w|^2 & \leq a.
\end{align*}
At least one such potential exists, namely the true potential $u$ given in \eqref{defiu}. 
%\cm{is the discussion useful? why not take directly $w=u$?}

\textbf{Step 1: Screenability.}
Let us comment on screenability in the general case (2) first. For $\s \leq \d-1$ we verify screenability on $S'(X_n,w,h)$, obtaining, for $h$ large enough
\begin{equation*}
\frac{h^{1+\gamma}}{h^{\d+1}}S'(X_n,w,h) \leq \frac{\mathcal{C}_\beta h^{\d+1+\gamma}}{h^{\d+1}} \leq h^\gamma \leq \mathsf{c}
\end{equation*}
on $\mathcal{G}_h$, since $\s < \d-1$ implies that $ \gamma < 0$ and $h \gg 1$. If $\s \geq \d-1$, we instead need to verify screenability with $S(X_n,w,h)$. By a covering argument we have
\begin{equation}\label{eq: coverS}
S(X_n,w,h) \leq \(\frac{L}{h}\)^{\d-1}\mathcal{C}_\beta h^\d \leq \mathcal{C}_\beta hL^{\d-1}
\end{equation} 
and obtain
\begin{equation*}
\frac{h^{1+\gamma}}{h^{\d+2}}\mathcal{C}_\beta hL^{\d-1}=\mathcal{C}_\beta \frac{L^{\d-1}}{h^{\d-\gamma}}
\end{equation*}
which can be made $\leq \mathsf{c}$ so long as $h \gg L^{\frac{\d-1}{\d-\gamma}}$, yielding the additional condition on $h$. Coupling these estimates with Proposition \ref{decay control} yields the second item in \eqref{height control}.

%\cm{Furthermore, using Jensen's inequality we find 
%\begin{equation}\label{Jensen bd}
%M_0^2=\(\frac{h^{-\gamma}}{|\New|}\int_{\Old \times \{h\}}|y|^\gamma \nabla w \cdot \hat{n}\)^2
%\leq \(\frac{h^{-\gamma}|\Old|}{(2\tilde{l})^\d}\int_{\Old \times \{h\}}|y|^\gamma \nabla w \cdot \hat{n} ~\frac{dx}{|\Old|}\)^2 \lesssim \frac{L^\d}{\tilde{\ell}^{2\d}h^{\gamma}}e(X_n,w,h)
%\end{equation}
%which coupled with Proposition \ref{decay control} yields 
%\begin{equation*}
%M_0^2 \lesssim  \frac{L^\d}{\tilde{\ell}^{2\d}h^{\gamma}}e(X_n,w,h)
 %\lesssim \frac{L^\d}{\tilde{\ell}^{2\d}h^{\gamma}}\(\frac{M^2}{h}\(\frac{L}{h}\)^\d+\frac{L^{\d+2}}{K^{\frac{2}{\d}}h^3}\)
%\end{equation*}
%hence the condition \eqref{height control}.}

For item (1), we specialize to $h=\frac{L}{K}$ and the announced $a$. We also set $\ell=\frac{L}{K}$ and $\tilde{\ell}=\frac{L}{K_1}$, with $K>K_1$ to be determined. Using the bound \eqref{local control of top energy}, we see that the first item of \eqref{Riesz screenability1}--\eqref{Riesz screenability2} becomes
\begin{equation*}
C_\beta\frac{\epsilon L^{\d-1}h^\gamma}{L^{\d-2}\tilde{\ell}^2}=C_\beta \frac{\epsilon K_1^2 L^{1+\gamma}}{K^\gamma L^2}\leq \mathsf{c}
\end{equation*}
which holds on the event $\mathcal{G}$ of Proposition \ref{decay control} for appropriate choice of constants since $\gamma<1$.

For item $(2)$, we instead find using \eqref{local control of top energy2} that the first item in \eqref{Riesz screenability1}--\eqref{Riesz screenability2} becomes 
\begin{equation*}
C_\beta \frac{h^\gamma ML^\d h^{\s_+-\d-1}}{L^{\d-2}\tilde{\ell}^2}=C_\beta \frac{ML^2 h^{\s+\s_+-2\d}}{\tilde{\ell^2}}\leq \mathsf{c}
\end{equation*}
using $\gamma=\s+1-\d$, which is the first item in \eqref{height control}.
 %as long as $h$ is larger than a constant, which may depend on $\beta $ when $\beta \le 1$. We may then redefine $\rho_\beta \lesssim_\beta$ if necessary so that this holds for $h \ge \rho_\beta$. 

\textbf{Step 2: Screening.}
We now apply and control the errors in Proposition~\ref{Riesz screening result}. Let us start with the general case (2). First, we have
\begin{equation*}
\frac{h}{\tilde{\ell}}S(X_n,w,h) \leq \mathcal{C}\frac{h^2}{\tilde{\ell}}L^{\d-1}
\end{equation*}
using the covering argument from the screenability computation in Step 1. 
%Next, using the same estimate as in \eqref{Jensen bd} we find 
%\begin{equation*}
%\tilde{\ell}L^{\d}{h^{\gamma}}M_0^2 \lesssim \tilde{\ell} L^{\d}h^{\gamma}\frac{L^\d}{\tilde{\ell}^{2\d}h^\gamma}e(X_n,w,h) \lesssim \frac{L^{2\d}}{\tilde{\ell}^{2\d-1}}e(X_n,w,h)
%\end{equation*}
 Substituting in the bound from Proposition \ref{decay control} yields 
  \begin{multline*}
\F(Y_{\mn}, \tilde{\mu}, \Omega \times [-L,L])-\(\frac{1}{2\cds}\(\int_{Q_R\times[-R,R]}|y|^\gamma|\nabla w_{\rrh}|^2-\cd \sum_{i=1}^n \g(\rrh_i)\) +\sum_{i=1}^n  \int \f_{\rrh_i}(x-x_i) d\mu(x)\)\\
\lesssim \sum_{(i,j)\in J}\g(x_i-z_j)+
\F(Z_{\mn-n_\Old}, \tilde{\mu}, \New \times [-h,h]) +|n-\mn|\\
+\(\frac{h^2}{\tilde{\ell}}+\tilde{\ell}+h^{-\gamma}\)L^{\d-1}+ C_\beta \( \frac{L^{2}}{\tilde{\ell}}+L\)M L^{\d} h^{\s_+-\d-1}.
\end{multline*}
Since the right hand side is uniform for any $w$ associated to $\Omega$ as in \eqref{inner screening w} with the requisite decay, we obtain the result by inserting the definition of $\G_{a,h}^{\mathrm{inn}}(X_n, \tilde \Omega) $.

In case (1), the screening errors in the previous computation then take the form
\begin{multline*}
\mathcal{C}_\beta \(\frac{h^2}{\tilde{\ell}}+\tilde{\ell}+h^{-\gamma}\)L^{\d-1}+\( \frac{L^{2}}{\tilde{\ell}}+L\)e(X_n,w,h) \\
\leq \mathcal{C}_\beta \(\frac{K_1}{K^2}+\frac{1}{K_1}\) L^\d+K^\gamma L^{\d-1-\gamma}+C_\beta\epsilon\(K_1^{2}+1\)L^\d
\end{multline*}
where we have inserted the control on $e(X_n,w,h)$ with the local laws parameters from Proposition \ref{decay control}. By choosing $K$ and $K_1$ appropriately and redefining $\epsilon$ from Proposition \ref{decay control}, we have that this is $\lesssim_\beta \epsilon L^\d$ as desired, since $|\gamma|<1$. A similar computation holds for the error terms in the outer screening.
\end{proof}

\subsection{Analysis of Exponential Moments}
To get control on the level of exponential moments, we need a sufficient volume of configurations for which we can screen. This is given by the following.

\begin{prop}\label{volume of configurations}
Keep the same assumptions and cases as in Corollary \ref{tailored screening}. Let $\mathcal{G}$ denote the good event of Corollary \ref{tailored screening}, and let $a$ be as in \eqref{def a}.
% and $M,h$ as in Corollary \ref{tailored screening}.\cm{the $M$ and $h$ are not chosen there -- what do you mean?}  
Let $n \leq N$, and define 
\begin{align}\label{integration restriction of good event}
\mathcal{G}_\Omega&=\{X_n \in \Omega^n: X_n=\XN \vert_\Omega \text{ for some }\XN \in \mathcal{G}\} \\
\mathcal{G}_{\Omega^c}&=\{X_{N-n} \in (\Omega^c)^{N-n}: X_{N-n}=\XN \vert_{\Omega^c} \text{ for some }\XN \in \mathcal{G}\}
\end{align}
Then, we have 
\begin{equation*}
n^{-n}\int_{\mathcal{G}_\Omega}\exp\left\{-\beta \G^{\mathrm{inn}}_{a,h}(X_n, \mu,\Omega \times [-L,L])\right\}~d\rho^{\otimes n}(X_n)\leq \K_{\beta,L}(\Omega,\mu,\zeta)\exp(\beta \varepsilon_e+\varepsilon_v)
\end{equation*} 
and 
\begin{align*}
(N-n)^{n-N}\int_{\mathcal{G}_{\Omega^c}}\exp\left\{-\beta \G^{\mathrm{ext}}_{a,h}(X_{N-n}, \mu,(\Omega \times [-L,L])^c)\right\}&~d\rho^{\otimes N-n}(X_{N-n}) \\
&\leq \K_{\beta,L}^{\mathrm{ext}}(\Omega,\mu,\zeta)\exp(\beta \varepsilon_e+\varepsilon_v),
\end{align*}
where $\rho$ denotes the confinement measure $d\rho(x)=e^{-\zeta (x)}~dx$, $\varepsilon_e$ denotes the energy error
\begin{equation*}
\varepsilon_e=\begin{cases}
|\mn-n|+C_\beta\epsilon L^\d+\tilde{\ell}^\d+\frac{\mathcal{C}L^\d}{\tilde{\ell}}  & \text{in case } (1) \\
|\mn-n|+\tilde{\ell}^\d+\frac{ \mathcal{C} hL^{\d-1}}{\tilde{\ell}}+\(\frac{h^2}{\tilde{\ell}}+\tilde{\ell}+h^{-\gamma}\)L^{\d-1}+C_\beta\( \frac{L^{2}}{\tilde{\ell}}+L\)M L^{\d} h^{\s_+-\d-1}
& \text{in case }(2)
\end{cases}
\end{equation*}
in the inner screening, and 
\begin{equation*}
\varepsilon_e=\begin{cases}
|\mn-n|+C_\beta\epsilon L^\d+\frac{\mathcal{C}L^\d}{\tilde{\ell}}  & \text{in case }(1) \\
|\mn-n|+\frac{ \mathcal{C} hL^{\d-1}}{\tilde{\ell}}+\(\frac{h^2}{\tilde{\ell}}+\tilde{\ell}+h^{-\gamma}\)L^{\d-1}+\\C_\beta\frac{L}{\min(h,\tilde{\ell})} \( \frac{L^{2}}{\tilde{\ell}}+L\)ML^\d h^{\s_+-\d-1}&\text{in case }(2)
\end{cases}
\end{equation*}
in the outer screening, and $\varepsilon_v$ denotes the volume error, with estimate
\begin{multline}\label{volerreurs}
\varepsilon_v\leq  C\frac{L^{\d-1}}{\ell^{\d-1}}\log \ell+ 
%C\mathcal C_\beta \frac{h L^{\d-1}}{\ell \tilde\ell}  +
C\tilde \ell L^{\d-1} 
%\\+ n-\mn + \alpha-\alpha'
+(\mn-n-\alpha)\log \frac{\alpha}{\alpha'}+\left(\mn-n-\alpha-\frac{1}{2}\right)\log \left(1+\frac{n-\mn}{\alpha}\right)-\frac{1}{2}\log \frac{n}{\mn},
\end{multline}
where $\alpha,\alpha'$  satisfy
\be \label{cara}
\tilde \ell L^{\d-1}\lesssim \alpha, \alpha' \lesssim \tilde \ell L^{\d-1}
%, \quad |\alpha-\alpha'|\lesssim L^{\d-1}h^{1+\gamma}+\frac{L^\d}{\tilde{\ell}}
.\ee
%\cm{I think the second estimate is 
%$$|\frac{\alpha}{\alpha'}-1|\le C \tilde \ell^{-1}+ C \frac{h}{\tilde \ell^2}$$ }
\end{prop}

\begin{proof}

We focus on the integral over $\mathcal{G}_\Omega$ (inner screening), since the argument for the integral over $\mathcal{G}_{\Omega^c}$ is an analogous application of Corollary \ref{tailored screening}. Much of the proof is as in \cite[Proposition 3.5]{P24} and \cite[Proposition 4.2]{AS21}, with the screening and combinatorial accounting updated as in \cite[Proposition 8.2]{S24}.

Let us first consider the situation in which \eqref{Riesz screenability1} holds, in which  case, for any configuration $X_n$, the screening contour $\Gamma$ is the boundary of a hyperrectangle $Q_t=\Old$ where we will denote $t:= t(X_n)$. We will also denote $n_{\Old }(X_n)$ the number of points of $X_n$ in $\Old$. We will  abbreviate $S(X_n, w,h)$ by $s$ and recall that for $X_n \in \mathcal G_\Omega$, we have as  in  \eqref{eq: coverS}
\be\label{boundsS} s\le \mathcal C_\beta h L^{\d-1}.\ee

 Let us also denote by $\Phi(X_n)$ the set of  all configurations $Y_{\mn}$ produced by screening the configuration $X_n$, i.e.~all those $Y_\mn$'s coinciding with $X_n$ in $\Old=\Old(X_n)$ and with an arbitrary configuration $Z_{\mn-n_\Old(X_n)}$ in $\New=\New (X_n)$, with all possible relabellings. We wish to integrate the screening estimates over all possible $Y_\mn$'s but we have to do it with care as a given configuration could be produced multiple times by the screening.  We will also condition on $n_\Old(X_n)$ being equal to a given integer $k\le \mn$.

The parameter $0 \le \eta\le \ell$ being given, we split the possibly values of $t$ into intervals $t\in I_p:=[L-\tilde \ell - (p+1)\eta, L-\tilde \ell -p\eta)$, $p$ integer. The integer $p$ ranges from $0$ to $ C\tilde \ell / \eta$. 
%We note that $p$ and $k$ are constrained, indeed in view of \eqref{bornimp}, we must have  on the one hand $
%\tilde \mu(\New(X_n))=\mn - n_\Old(X_n)$
%and on the other hand
 %\begin{multline}\label{816}\mu(Q_L\backslash Q_{L-(p+1)\eta})= \mu(\New(X_n)) +O(\eta L^{\d-1})\\
 %= \tilde \mu(\New(X_n))+ O( L^{\d-1}+ \frac{s}{\tilde \ell}) + O(\eta L^{\d-1})
 %=\tilde \mu(\New(X_n)) \(1+O\(\frac{1+\eta}{\tilde \ell} +\frac{s}{\tilde\ell^2 L^{\d-1}} \)\) \\
 %=(\mn-n_\Old(X_n)) \(1+O\(\frac{1+\eta}{\tilde \ell} +\frac{s}{\tilde\ell^2 L^{\d-1}} \)\)
%\end{multline} 
%where we used the upper and lower bounds on $\mu$ and the fact that  $Q_L\backslash Q_{L-\tilde \ell}\subset \New\subset Q_L\backslash Q_{L-2\tilde \ell}$.
%\cm{Revisit the above equation; we may not have recourse to it since we are integrating against $\rho$ below.}
We denote by $\mathcal A$ the set of admissible integer $p$'s and $k$'s, and note that it is included in $[0, \lfloor \frac{C\tilde \ell}{\eta}\rfloor]\times [0,\mn]$.
%, and satisfies the constraint that 
%\be \label{constraintpk}
%\mu(Q_L\backslash Q_{L-(p+1)\eta})= (\mn-k) \(1+O\(\frac{1+\eta}{\tilde \ell} +\frac{s}{\tilde\ell^2 L^{\d-1}} \)\).
%\ee

%\cm{revisit this too...}

We next let $ \nu = (t(X_n), n_\Old(X_n))\# \rho^{\otimes n}$, i.e.~the ``law'' on $\mathbb{N}\times \mathbb{N}$ equal to the ``law'' of $(t(X_n),
n_\Old(X_n)) $ for  $X_n \sim \rho^{\otimes n}.$
We 
have a natural disintegration 
\be \label{desintegmu} d\rho^{\otimes n}(X_n) = \sum_{(p,k)\in \mathcal A} d(\rho^{\otimes n})_{p,k} (X_n) \nu(p,k).\ee
 Here, the conditional measure $(\rho^{\otimes n})_{p,k}$ is supported on $\{X_n, t(X_n)\in I_p, n_{\Old}(X_n)=k\}$, and given by 
\be \label{2}(\rho^{\otimes n})_{p,k}(A)= \frac{\rho^{\otimes n} (A \cap \{ t(X_n) \in I_p, n_{\Old}(X_n)=k\})}{\rho^{\otimes n} ( \{ t(X_n) \in I_p, n_{\Old}(X_n)=k\})}= 
 \frac{\rho^{\otimes n} (A \cap \{ t(X_n) \in I_p, n_{\Old}(X_n)=k\})}{\nu(p,k)}\ee

 Then, for given $p,k$, we define the multiplicity measure 
\be\label{defmesurem}d \mathsf{m}(Y_{\mn},p,k)= \int_{X_n|Y_{\mn}\in \Phi(X_n)} d\rho^{\otimes (\mn-k)}((Y_{\mn})_{S_k^c} )  d(\rho^{\otimes n})_{p,k}(X_n) ,\ee
where for each $Y_{\mn}\in \Phi(X_n)$, $S_k$ is a set of indices of cardinality $k$ such that the configuration $(Y_{\mn})_{S_k}$ coincides with $\{ X_n\} \cap Q_{t(X_n)}$ (here $(Y_\mn)_{S}$ denotes the restriction of the $\mn$-tuple to the indices in $S$).
By Fubini's theorem, we can write 
\begin{multline}
\int_{\Omega^\mn } \exp\(-\beta \F(Y_\mn,\mu,\Omega\times [-L,L]) \) d\mathsf{m}(Y_\mn,p,k) 
\\ \ge
 \int_{\mathcal G_\Omega } \(\int_{\Phi(X_n)} \exp\(-\beta \F(Y_{\mn}, \mu,\Omega\times [-L,L]) \)d\rho^{\otimes( \mn-k)}((Y_{\mn})_{S_k^c}) \)d(\rho^{\otimes n})_{p,k}(X_n).
\end{multline}
Grouping the error estimate from Corollary \ref{tailored screening} into
\begin{equation*}
\varepsilon_e=C\begin{cases}
|\mn-n|+C_\beta \epsilon L^\d  & \text{in case } (1) \\
|\mn-n|+\(\frac{h^2}{\tilde{\ell}}+\tilde{\ell}+h^{-\gamma}\)L^{\d-1}+C_\beta\( \frac{L^{2}}{\tilde{\ell}}+L\)M L^{\d} h^{\s_+-\d-1}
 & \text{in case }  (2)
\end{cases}
\end{equation*}
with $a$ as in \eqref{def a}, we may insert into the right-hand side the relation
\begin{multline*}\F(Y_\mn, \mu,\Omega\times [-L,L]) \\ \le  \G_{a,h}^{\mathrm{inn}}(X_n, \Omega\times [-L,L]) + 
 C\Bigg(   \F(Z_{\mn-\N} ,\tilde\mu, \New\times [-h,h])   + \varepsilon_e
+
\sum_{(i,j) \in J} \g(x_i-z_j) \Bigg).
\end{multline*}
Thus, we obtain
\begin{multline}\label{33}
\int_{\mathcal G_\Omega} \exp\(-\beta \F(Y_\mn,\mu,\Omega\times [-L,L]) \)d\mathsf{m}(Y_\mn,p,k) 
\\ \ge \int_{\mathcal G_\Omega} \(\int_{\Phi(X_n)} e^{-\beta \G_{a,h}^{\mathrm{inn}}(X_n,\Omega\times [-L,L]) -C\beta\( \F(Z_{\mn -n_\Old},\tilde\mu, \New\times [-h,h])+\varepsilon_e
+
\sum_{(i,j) \in J} \g(x_i-z_j) \)} d\rho^{\otimes( \mn-k)}((Y_{\mn})_{S_k^c})  \)  \\ 
\qquad \times d(\rho^{\otimes n})_{p,k}(X_n)
\\ =\binom{\mn}{k} \int_{\mathcal G_\Omega}\(\int e^{- C\beta  \F(Z_{\mn -n_\Old}, \tilde \mu, \New\times [-h,h]) -C \beta \sum_{(i,j) \in J} \g(x_i-z_j)  } d\rho^{\otimes( \mn-k)}|_{Q_L\backslash Q_{t(X_n)}} (Z_{\mn-k})  \)   \\
\times \exp\(-\beta \G_{a,h}^{\mathrm{inn}}(X_n,\mu, \Omega \times [-L,L])  - C \beta\varepsilon_e\)d(\rho^{\otimes n})_{p,k}(X_n).
\end{multline}

To estimate the inner integral, we  make use of  Jensen's inequality coupled with a tilt as in the proof of Proposition~\ref{pro718}. The argument is exactly as in \cite[Proposition 4.2]{AS21} or \cite[Lemma 8.5]{S24}. We can rewrite the interior integral as
\begin{align*}
I=\int_{\New^{\mn-\N}}\exp \biggl(&-\beta C\left(\sum_{(i,j)\in J}\g(x_i-z_j)+\F(Z_{\mn-\N},\tilde{\mu}, \New \times [-h,h])\right) \\
&-\sum_{i=1}^{\mn-\N}\log \tilde{\mu}(z_i)\biggr)~d\tilde{\mu}\vert_{\New}^{\otimes {\mn-\N}}(Z_{\mn-\N})
\end{align*}
where we have used that $\rho$ agrees with the Lebesgue measure inside of the bulk. Applying Jensen's inequality, we then have 
\begin{equation*}
I\geq \tilde{\mu}(\New)^{\mn-\N}\exp (A+B+C)
\end{equation*}
with 
\begin{align*}
A&=\tilde{\mu}(\New)^{\N-\mn}\int_{\New^{\mn-\N}}-\beta C\sum_{(i,j)\in J}\g(x_i-z_j)~d\tilde{\mu}\vert_{\New}^{\otimes {\mn-\N}}(Z_{\mn-\N}) \\
B&= \tilde{\mu}(\New)^{\N-\mn}\int_{\New^{\mn-\N}}-\beta C\F(Z_{\mn-\N},\tilde{\mu}, \New \times [-h,h])~d\tilde{\mu}\vert_{\New}^{\otimes {\mn-\N}}(Z_{\mn-\N}) \\
C&= -\tilde{\mu}(\New)^{\N-\mn}\int_{\New^{\mn-\N}}\sum_{i=1}^{\mn-\N}\log \tilde{\mu}(z_i)~d\tilde{\mu}\vert_{\New}^{\otimes {\mn-\N}}(Z_{\mn-\N}).
\end{align*}
Recall that $\tilde{\mu}(\New)=\mn-\N$. $B$ is dealt with immediately by Proposition \ref{pro718}, and $A$ and $C$ are dealt with as in \cite[Proposition 4.2]{AS21} and \cite[Proposition 3.5]{P24}, namely
\begin{equation*}
A \gtrsim -\beta \#I_\partial \gtrsim \frac{-\beta \mathcal{C}_\beta hL^{\d-1}}{\tilde{\ell}}, 
\end{equation*}
where we have used $S(X_n,w,h) \leq \mathcal{C}_\beta hL^{\d-1}$ (see \eqref{eq: coverS}) to control $\#I_\partial \lesssim \frac{S(X_n,w,h)}{\tilde{\ell}}$, and 
$$C=- \int_{\New^{\mn-n_\Old}} \tilde \mu \log \tilde \mu \gtrsim -C\tilde{\ell}L^{\d-1}$$
from the boundedness of $\tilde{\mu}$.
%$$ C =- \int_{\New^{\mn-n_\Old}} \tilde \mu \log \tilde \mu= \cm{edit} \int_{\New} \mu-\tilde \mu+O\( \int_{\New} \frac{|\mu-\tilde \mu|^2}{\tilde \mu}\) = \mu(\New)-\tilde \mu(\New) + O\( \frac{s}{\ell \tilde \ell}+\frac{L^{\d-1}}{\ell}\)$$
%\cm{revisit..we will get a worse estimate since we are integrating against $\rho$. I think it's just $C\tilde \ell L^{\d-1}$, which we can absorb into other errors}
%where we have used \eqref{mmut2} and a Taylor expansion.
Using \eqref{boundsS}, absorbing $\frac{ \mathcal{C}_\beta  hL^{\d-1}}{\tilde{\ell}}$ into the definition of $\varepsilon_e$ yields 
\begin{multline}\label{66}
%\frac{1}{ \binom{n}{k} \mu(Q_R\backslash Q_{R-(p+1)\eta})^{n-k} } 
\int_{\Omega^\mn}  \exp\(-\beta \F(Y_\mn,\mu, \Omega) \)
d\mathsf{m} (Y_\mn,p,k) 
 \ge \exp\(- C \beta \varepsilon_e- C\tilde \ell  L^{\d-1}\)\\
 %- C\mathcal C_\beta \frac{h L^{\d-1}}{\ell \tilde\ell}cut last term?}     \\
\times \binom{\mn}{k}(\mn-k)^{\mn-k}  \int_{\mathcal G_\Omega}   
%\exp\(\mu(\New)-\tilde \mu(\New) \cm{and this one}\)
\exp\(-\beta \G_{a,h}^{\mathrm{inn}}(X_n,\mu, \Omega\times [-L,L])\)   d(\rho^{\otimes n})_{p,k}(X_n).
\end{multline}

%\begin{equation*}
%\int_{\mathcal{G}_\Omega}\exp(-\beta \G^{\mathrm{inn}}_{a,h}(X_n,\Omega\times [-L,L]))~d\rho^{\otimes n}(X_n)\leq \mn^{\mn}\K_{\beta,L}(\Omega,\mu,\zeta)\exp \left(\beta C\varepsilon_e+\varepsilon_v\right)
%\end{equation*}
%with 
%\begin{equation*}
%\varepsilon_e=\begin{cases}
%|\mn-n|+C_\beta\epsilon L^\d+\frac{\mathcal{C}L^\d}{\tilde{\ell}}  & \text{in case } (1) \\
%|\mn-n|+\frac{ \mathcal{C} hL^{\d-1}}{\tilde{\ell}}+\(\frac{h^2}{\tilde{\ell}}+\tilde{\ell}+h^{-\gamma}\)L^{\d-1}+C_\beta\( \frac{L^{2}}{\tilde{\ell}}+L\)M L^{\d} h^{\s_+-\d-1}
%& \text{in case }(2).
%\end{cases}
%\end{equation*}
%where we have inserted the definition of $a$ from \eqref{def a} and 
%\begin{equation*}
%\varepsilon_v=C\tilde{\ell}^\d+\log \frac{C\tilde{\ell}}{\eta}+\log \left(\frac{n!(\mn-\N)!|\New|^{n-\N}}{\mn!(n-\N)!(\mn-\N)^{\mn-\N}}\right).
%\end{equation*}

%As in \eqref{816},
%in view of \eqref{bornimp}, we have  $\tilde \mu(\New)=  \mn-k$, on the other hand 
%we have $\mu(\New)= \mu(Q_L \backslash Q_{L-(p+1)\eta}) +O(\eta L^{\d-1})$. \cm{revisit}
% and on the other hand 
%$$\tilde \mu(\New)-\mu(\New)= \tilde \mu(\New)  O \(\frac{1}{\tilde \ell}+\frac{s}{\tilde \ell^2 R^{\d-1}}\)= (\mn -k) O \(\frac{1}{\tilde \ell}+\frac{s}{\tilde \ell^2 R^{\d-1}}\).$$ 
%Inserting this into \eqref{66}, 
Passing factors onto the left-hand side
and summing over $p$ and $k$ after multiplying by $\nu(p,k)$, in view of \eqref{desintegmu} we obtain
\begin{multline}
\int_{\mathcal G_\Omega} \exp\(-\beta \G_{a,h}^{\mathrm{inn}}(X_n,\mu, \Omega\times [-L,L])\)   d\rho^{\otimes n}(X_n)
\\
\le 
e^{C\beta \varepsilon_e+C\varepsilon_v} \sum_{(p,k)\in \mathcal A}
% \nu(p,k) \frac{e^{ (\mn-k) - \mu(Q_L \backslash Q_{L-(p+1)\eta}) }}{\binom{\mn}{k}(\mn-k)^{\mn-k}  }
\frac{\nu(p,k)}{\binom{\mn}{k}(\mn-k)^{\mn-k}  }\int_{\Omega^\mn}  e^{-\beta \F(Y_\mn,\mu, \Omega\times [-L,L]) }
d\mathsf{m} (Y_\mn,p,k) ,
\end{multline}
%\cm{cut numerator?}
where we have set $\varepsilon_v=\tilde \ell L^{\d-1}.$
%$$\varepsilon_v:= C\mathcal C_\beta \frac{h L^{\d-1}}{\ell \tilde\ell}  +(\tilde \ell+\eta) L^{\d-1}.$$
%\cm{cut first term and $\eta$?}
We may rewrite this as 
\begin{multline}\label{825}
\int_{\mathcal G_\Omega} \exp\(-\beta \G_{a,h}^{\mathrm{inn}}(X_n, \mu, \Omega\times [-L,L])\)   d\rho^{\otimes n}(X_n)
\\ \le e^{C\beta \varepsilon_e+C\varepsilon_v} \max_{(k,p)\in \mathcal A} \(  
%e^{ (\mn-k) - \mu(Q_L \backslash Q_{L-(p+1)\eta}) }
\frac{\binom{n}{k}\rho(Q_L\backslash Q_{L-(p+1)\eta})^{n-k}} {\binom{\mn}{k}(\mn-k)^{\mn-k}  }\)  \\ \times \sum_{(p,k)\in \mathcal A}  \frac{\nu(p,k)}{ \binom{n}{k}\rho(Q_L\backslash Q_{L-(p+1)\eta})^{n-k} }\int_{\Omega^\mn}  \exp\(-\beta \F(Y_\mn,\mu, \Omega\times [-L,L]) \)
d\mathsf{m} (Y_\mn,p,k) .
\end{multline}
%\cm{cut exponential?}
%Now we recall that by definition of the disintegration, for any measurable function $f$ we have
%$$ \sum_{p,k} d(\mu^n)_{p,k} (X_n) \nu(p,k) f(X_n,p,k)=\int f(X_n, p(X_n), n_\Old (X_n) )d\mu^n(X_n) .$$
But  by definition of $\mathsf{m}$ \eqref{defmesurem}, we have 
\begin{align*}& 
\sum_{(p,k)\in \mathcal A}  \frac{\nu(p,k)}{ \binom{n}{k}\rho(Q_L\backslash Q_{L-(p+1)\eta})^{n-k} }\int_{\Omega^\mn}  \exp\(-\beta \F(Y_\mn,\mu,\Omega\times [-L,L]) \)
d\mathsf{m} (Y_\mn,p,k)  \\
& =\sum_{(p,k)\in \mathcal A}\int_{\Omega^\mn }  \int_{X_n|Y_{\mn}\in \Phi(X_n)} d\rho^{\otimes (\mn-k)}((Y_{\mn})_{S_k^c} )  d(\rho^{\otimes n})_{p,k}(X_n)\nu(p,k)
\frac{e^{-\beta \F(Y_\mn,\mu,\Omega\times [-L,L]) }}{\binom{n}{k}\rho(Q_L\backslash Q_{L-(p+1)\eta})^{n-k} } . \end{align*}
Inserting \eqref{2}, letting $S_{k}'$ denote a set of indices of cardinality $n-k$ such that the points of $(X_{n})_{S_{k}'}$ coincide with $ \{X_n\} \cap \Old (X_n)$, we obtain
\begin{align*}& 
\sum_{(p,k)\in \mathcal A}  \frac{\nu(p,k)}{ \binom{n}{k}\rho(Q_L\backslash Q_{L-(p+1)\eta})^{n-k} }\int_{\Omega^\mn}  e^{-\beta \F(Y_\mn,\mu,\Omega\times [-L,L]) }
d\mathsf{m} (Y_\mn,p,k)  \\
&=
\sum_{(p,k)\in \mathcal A}\int_{\Omega^\mn }  \int_{X_n|Y_{\mn}\in \Phi(X_n)} d\rho^{\otimes (\mn-k)}((Y_{\mn})_{S_k^c} )  d\rho^{\otimes n}(X_n)\indic_{n_\Old (X_n)= k, t(X_n) \in I_p}
\frac{e^{-\beta \F(Y_\mn,\mu,\Omega\times [-L,L]) }}{\binom{n}{k}\rho(Q_L\backslash Q_{L-(p+1)\eta})^{n-k} } \\
& =\sum_{(p,k)\in \mathcal A} \int_{\Omega^\mn }  \int_{X_n|Y_{\mn}\in \Phi(X_n)} d\rho^{\otimes (\mn-k)}((Y_{\mn})_{S_k^c} )  d\rho^{\otimes k}((X_n)_{S'_{k}})d\rho^{\otimes (n-k)}((X_n)_{(S'_{k})^c})\indic_{n_\Old (X_n)= k, t(X_n) \in I_p}
\\
&\qquad\qquad\times \frac{e^{-\beta \F(Y_\mn,\mu,\Omega\times [-L,L]) }}{\binom{n}{k} \rho(Q_L\backslash Q_{L-(p+1)\eta})^{n-k} } \\
&= \sum_{(p,k)\in \mathcal A} \int_{\Omega^\mn }  \int_{X_n|Y_{\mn}\in \Phi(X_n)} d\rho^{\otimes \mn}(Y_{\mn})  d\rho^{\otimes (n-k)}((X_n)_{(S'_{k})^c})\indic_{n_\Old (X_n)= k, t(X_n) \in I_p}
 \frac{e^{-\beta \F(Y_\mn,\mu,\Omega\times [-L,L]) }}{\binom{n}{k}\rho(Q_L\backslash Q_{L-(p+1)\eta})^{n-k} } \\
& = \int_{\Omega^{\mn}} e^{-\beta \F(Y_\mn,\mu,\Omega\times [-L,L]) }d\rho^{\otimes \mn}(Y_{\mn})  \sum_{(p,k)\in \mathcal A}  \int_{X_n|Y_{\mn}\in \Phi(X_n)}   d\rho^{\otimes (n-k)}((X_n)_{(S'_{k})^c}) \frac{\indic_{n_\Old (X_n)= k, t(X_n) \in I_p}
}{\binom{n}{k}\rho(Q_L\backslash Q_{L-(p+1)\eta})^{n-k} } \\
 &\le  \int_{\Omega^{\mn}} e^{-\beta \F(Y_\mn,\mu,\Omega\times [-L,L]) }d\rho^{\otimes \mn}(Y_{\mn})\sum_{(p,k)\in \mathcal A} 1.
\end{align*}
 Inserting that $\#\mathcal A\le C\frac{\mn \tilde \ell}{\eta}$,  we conclude that 
\begin{multline}\label{838fin}
n^{-n}\int_{\Omega^n} \exp\(-\beta \G_{a,h}^{\mathrm{inn}}(X_n, \mu, \Omega)\)d\rho^{\otimes n}(X_n)
\le 
C \frac{\mn \tilde \ell}{\eta} e^{ C (\beta \varepsilon_e+C\varepsilon_v) }\\
\times \max_{(k,p)\in \mathcal A} \( 
 %e^{\mn-k- \mu(Q_L\backslash Q_{L-(p+1)\eta}) }
\frac{\mn^{\mn}}{n^n} \frac{\binom{n}{k}\rho(Q_L\backslash Q_{L-(p+1)\eta})^{n-k}} {\binom{\mn}{k}(\mn-k)^{\mn-k}  }\) 
\frac{1}{\mn^{\mn}}\int_{\mathcal G_\Omega } e^{-\beta \F(Y_\mn,\mu,\Omega\times [-L,L]) } d\rho^{\otimes \mn}(Y_{\mn} ) .\end{multline}
We now let  $(k_0,p_0)\in \mathcal A$ be an integer couple where the max is achieved.
%, which must satisfy \eqref{constraintpk}. \cm{review}
Denoting for shortcut  $\alpha=\mn-k_0$ and $\alpha'= 
\rho(Q_L\backslash Q_{L-(p_0+1)\eta})$, the relation \eqref{cara} is then satisfied, and  we next
  study the factor 
$$\frac{\mn^{\mn}}{n^n}\frac{\binom{n}{k_0} 
%e^{\mn-k_0-\alpha' }
\rho(Q_L\backslash Q_{L-(p_0+1)\eta})^{n-k_0}} {\binom{\mn}{k_0}(\mn-k_0)^{\mn-k_0}  }=
\frac{\mn^{\mn}}{n^n} \frac{\binom{n}{k_0} 
%e^{\alpha -\alpha'}
  (\alpha')^{n-k_0} }{\binom{\mn}{k_0}(\alpha)^{\mn-k_0} }.$$
  By Stirling's formula as in \cite[Proposition 4.2]{AS21}, substituting $k_0=\mn-\alpha$, 
  we have
\begin{align*}
& \log
\( \frac{\mn^{\mn}}{n^n} \frac{\binom{n}{k_0} 
%e^{\alpha -\alpha'}
  (\alpha')^{n-k_0} }{\binom{\mn}{k_0}(\alpha)^{\mn-k_0} }
\)
%\\ &=
%\alpha-\alpha'- \mn+n +
%(n-k_0) \log \alpha - (n-k_0) \log (n-k_0) -\hal \log \frac{\mn(n-k_0)}{(\mn-k_0) n}-O(1)\\
\\& =
%\alpha-\alpha' +n-\mn
% -(n-\mn+\alpha) \log \frac\alpha{\alpha'} +
 (n-\mn+\alpha) \log \frac{ \alpha' }{(n-\mn+\alpha)}-\hal \log \frac{\mn (n-\mn+\alpha)}{n\alpha}+O(1)\\
& =
%\alpha-\alpha'+n-\mn 
-(n-\mn+\alpha) \log \frac\alpha{\alpha'}  -(n-\mn+\alpha+\hal) \log (1+\frac{n-\mn}{\alpha}) -\hal \log \frac{\mn}{n} +O(1).
\end{align*}Choosing $\eta=\ell^{-1}$  and inserting into \eqref{838fin} we obtain the result (the outer case is completely analogous).

We next turn to the second case, where  \eqref{Riesz screenability2} is used in the screening. Then, for each configuration $X_n$, $\Gamma$,  the boundary of $\Old$ is a piecewise affine boundary with facets parallel to the faces of $Q_L$, included in some $Q_{t+\ell}\backslash Q_\ell$.  We need to split the phase-space into regions where we know this boundary location up to a $O(\eta)$ error. There are $O(\frac{\tilde \ell}{\ell})$ strips of the form $Q_{t+\ell}\backslash Q_\ell$ in which $\Gamma$ can lie. For each of them  there are $O(\frac{L^{\d-1}}{\ell^{\d-1}})$ facets forming $\Gamma$, and $O(\ell/\eta)$ choices for each so that we know their location up to an error $\eta$. This allows to partition the phase-space into  $C_1\frac{\tilde \ell}{\ell}\(C_2 \frac{\ell}{\eta}  \)^ {C_3 \frac{L^{\d-1}}{\ell^{\d-1}}}$ sets  such that each $\Gamma$ in one of these sets 
 is  at distance $\le \eta$ of a reference boundary. Up to using translations and increasing  the factor $C_2$ in the multiplicity, we may ensure that each $\Gamma$ in entirely included in one such set. This allows to run the same proof as in the first case, except with the multiplicity $\frac{\mn\tilde \ell}{\eta}$  at the end replaced by 
 $\mn \frac{\tilde \ell}{\ell}\(C \frac{ \ell}{\eta}  \)^ { \frac{L^{\d-1}}{\ell^{\d-1}}}$,
which adds a term $O\(  \frac{L^{\d-1}}{\ell^{\d-1}}\log \frac{\ell}{\eta}\)$ to the volume errors.
The choice $\eta=\ell^{1-\d}$ yields the announced estimate.

\end{proof}
 
\begin{rem} \label{bdd vol error}
When summing the contributions over $\Omega$ where $n$ points fall and $U\backslash \Omega $ where $N-n$ points fall, the errors of~\eqref{volerreurs} compensate and add up to a well bounded error. More precisely, 
if  $\alpha, \alpha'$, respectively $\gamma, \gamma'$ satisfy ~\eqref{cara}, then for every $n$ we have
\begin{multline}\label{648}
%\alpha-\alpha'
-(n-\mn+\alpha) \log \frac{\alpha}{\alpha'} - (n-\mn+\alpha+\hal) \log \(1+\frac{n-\mn}{\alpha}\)  - \hal \log \frac{n}{\mn} 
\\ 
%\gamma-\gamma'
-(\mn-n +\gamma) \log \frac{\gamma}{\gamma'} - (\mn-n+\gamma+\hal) \log \(1+\frac{\mn-n}{\gamma}\)  - \hal \log \frac{N-n}{N-\mn} \\ \le  C\tilde{\ell}L^{\d-1}.
%C\( \frac{R^{\d-1}}{\tilde \ell} + \frac{s^2}{\tilde \ell^3 R^{\d-1}}\).
\end{multline}
\end{rem}
\begin{proof}
The proof is exactly as in \cite[Remark 8.4]{S24} and \cite[Remark 3.8]{P24}.
\end{proof}

\subsection{Conclusion}
We now have all of the tools to complete the proof of Proposition \ref{Main Bootstrap}. We specialize to Case 1 of the previous results, focusing on the local laws bootstrap. Coupling Proposition \ref{volume of configurations} with Remark \ref{bdd vol error} and a subadditivity argument yields first the following control on the level of exponential moments.
\begin{prop}\label{energy control for good points}
Keep the same assumptions as in Corollary \ref{tailored screening} and above results, and let $\mathcal{G}$ denote the good event from Corollary \ref{tailored screening}. Denote by $\mathcal{G}^{n}$ the configurations in $\mathcal{G}$ who have $n$ points in $\Omega$. Let $u$ be as in \eqref{defiu}. Then, if $\C_\beta$ is chosen large enough depending on $\d,\s,\|\mu\|_{L^\infty},\ep$ (and $\beta$ for $\beta \le 1$), 
\begin{multline*}
\mathbb{E}_{\mathbb{Q}_{\beta,\frac{L}{K}} (U,\mu,\zeta)}\left(\exp \left(\frac{\beta}{2}\(\F^{\Omega \times [-L,L]} (\XN, \mu, \Lambda) + 2C_0 \# I_{\Omega}\)\right)\indic_{\mathcal{G}^{n}}\right) \\
\le C
\frac{\K_{\frac{\beta}{2},\frac{L}{K}}\left(\Omega, \mu, \zeta\right)}{\K_{\beta,\frac{L}{K}}(\Omega, \mu, \zeta)}\exp \left(\beta \frac{\mathcal{C}_\beta}{8}L^\d+C\beta|\mn-n|\right)
\end{multline*}
for some constant $C$ dependent only on the upper and lower bounds of $\mu$. 
\end{prop}
\begin{proof}

We recall that $\Omega = Q_L \cap \Lambda$ and we denote $\Omega_L=\Omega \times [-L,L]$ and $\omc_{L}= Q_{L-2\tilde \ell}\cap \Lambda$ where $Q_{L-2\tilde \ell}$ is the hyperectangle with same center as $Q_L$.

The proof relies on the following superadditivity computation based on \eqref{superadd}, using also \eqref{arrivC0} to get some nonnegativity and \eqref{locali3} (here we drop the $\mu$ dependence in the notation for $\F$).
\begin{align*}
\frac{\beta}{2}\F^{\omc_L}(\XN, \Lambda)-\beta \F(\XN,\Lambda)&\leq \frac{\beta}{2}\F^{\omc_L}(\XN, \Lambda)-\beta \F^{\Omega_L}(\XN, \Lambda)-\beta\F^{\Lambda\backslash \Omega_L}(\XN, \Lambda) \\
&\leq \frac{\beta}{2}\F^{\omc_L}(\XN, \Lambda)-\frac{\beta}{2} \F^{\Omega_L}(\XN, \Lambda)-\frac{\beta}{2} \F^{\Omega_L}(\XN, \Lambda)-\beta\F^{\Lambda\backslash \Omega_L}(\XN, \Lambda) \\
&\leq -\frac{\beta}{2}\F^{\Omega_{L} \setminus \omc_{L}}(\XN, \Lambda)
-\frac{\beta}{2} \F^{\Omega_L}(\XN, \Lambda)-\beta\F^{\Lambda\backslash \Omega_L}(\XN, \Lambda) \\
&\leq -\frac{\beta}{2} \F^{\Omega_L}(\XN, \Lambda)-\beta\F^{\Lambda\backslash \Omega_L}(\XN, \Lambda)   +\frac{\beta}{2}C_0n \\
&\leq -\frac{\beta}{2} \G^{\mathrm{inn}}_{a, \frac{L}{K}}(\XN\vert_\Omega, \Omega_L)-\beta \G^{\mathrm{ext}}_{a, \frac{L}{K}}(\XN \vert_{ \Omega^c},\Lambda \backslash\Omega_L)+\frac{\beta}{2}C_0n
\end{align*}
with $a=C_\beta \epsilon L^{\d-1}$ as in \eqref{def a}. We can then compute using Proposition \ref{volume of configurations} exactly as in \cite[Proposition 3.7]{P24}
\begin{align*}
&\Esp_{\mathbb{Q}_{\beta, h} ( U,\mu,\zeta)} \left(\exp \left(\frac{\beta}{2}\F^{\omc_L}(\cdot, \Lambda)\right)\indic_{\mathcal{G}_{n}}\right)\\
&\leq \frac{1}{N^N\K_{\beta, \frac{L}{K}}(U,\mu,\zeta)}{N \choose n}\int_{\Omega^n \cap \mathcal{G}^\Omega}\exp \left(-\frac{\beta}{2}\G^{\mathrm{inn}}_{a, \frac{L}{K}}(X_n,\Omega_L)+\frac{\beta}{2}C_0n\right)~d\rho^{\otimes n}(X_n) \\
&\times \int_{(\Omega^c)^{N-n}\cap \mathcal{G}^{\Omega^c}}\exp \left(-\beta \G^{\mathrm{ext}}_{a, \frac{L}{K}}(X_{N-n},\Omega_L^c)\right)~d\rho^{\otimes N-n}(X_{N-n}) \\
&\leq \frac{{N \choose n}}{N^N \K_{\beta,\frac{L}{K}}(U,\mu,\zeta)}n^n(N-n)^{N-n}\K_{\frac{\beta}{2}, L/K}\left(\Omega,\mu, 0\right)\K_{\beta,L/K}(\Omega^c,\mu,0)\exp \left(\beta \varepsilon_e+\varepsilon_v+\frac{\beta}{2}C_0n\right)
\end{align*}
%\cm{is it indeed first $h$ then $L/K$?}
where $\varepsilon_e$ and $\varepsilon_v$ denote the energy and volume errors 
\begin{equation*}
\varepsilon_e=C(C_\beta \epsilon L^\d+\tilde{\ell}^\d+|\mn-n|) \qquad
\varepsilon_v=
%C\frac{L^{\d-1}}{\ell^{\d-1}}\log \ell+ 
%C\mathcal C_\beta \frac{h L^{\d-1}}{\ell \tilde\ell}  +
C\tilde \ell L^{\d-1} 
\end{equation*}
where we have dropped the $\frac{L^\d}{\tilde{\ell}}$ and $C\frac{L^{\d-1}}{\ell^{\d-1}}\log \ell$ errors because they are controlled by the remaining terms.
The superadditivity of the Neumann partition functions \eqref{eq: supadd partext} then yields
\begin{multline*}
\Esp_{\mathbb{Q}_{\beta, \frac{L}{K}} ( U,\mu,\zeta)} \left(\exp \left(\frac{\beta}{2}\F^{\omc_L}(\cdot, \Lambda)\right)\indic_{\mathcal{G}_{n}}\right)\\
\leq \frac{\mn! (N-\mn)!n^n(N-n)^{N-n}}{n!(N-n)!\mn^\mn(N-\mn)^{N-\mn}}\frac{\K_{\frac{\beta}{2},\frac{L}K}\left(\Omega,\mu,\zeta\right)}{\K_{\beta,\frac{L}K}(\Omega,\mu,\zeta)}\exp \left(\beta \varepsilon_e+\varepsilon_v + \frac\beta 2 C_0 n\right).
\end{multline*}
Stirling's formula  yields 
\begin{equation*}
\Esp_{\mathbb{Q}_{\beta, h} ( U,\mu,\zeta)} \left(\exp \left(\frac{\beta}{2}\F^{\omc_L}(\cdot, \Lambda)\right)\indic_{\mathcal{G}^{n}}\right)\lesssim \frac{\K_{\frac{\beta}{2},\frac{L}{K}}\left(\Omega,\mu,\zeta\right)}{\K_{\beta,\frac{L}{K}}(\Omega,\mu,\zeta)}\exp \left(\beta \varepsilon_e+\varepsilon_v + \frac\beta 2 C_0 n\right).
\end{equation*}
%\cm{need to keep $L/K$ for denominator}
We then write $C_0n \le C_0 \mn + C_0 |n-\mn|$, and assume (without loss of generality) that $\C_\beta  \ge 8\|\mu\|_{L^\infty} C_0$, which ensures that $C_0 \mn \le \frac{\C}{8}L^\d$. Choosing $\epsilon$ small enough in Proposition \ref{decay control} and $K_1$ defining $\tilde{\ell}=\frac{L}{K_1}$ large enough in Corollary \ref{tailored screening} makes $\ep_e$ and $\ep_v$ small enough and we obtain the result. 
\end{proof}

\subsubsection{Proof of Proposition \ref{Main Bootstrap}}
We establish that
\begin{equation}\label{pre local law}
\Esp_{\mathbb{Q}_{\beta, h} ( U,\mu,\zeta)} \( \exp \(\frac{\beta}{2}\(\F^{\Omega_L}(\cdot, \Lambda) +2C_0 \#I_{\Omega} \) \)\indic_{\mathcal{G}}\) \leq \exp\(\C_\beta \frac{\beta}{4} L^\d\).
\end{equation}
To prove this, we sum the control in Proposition \ref{energy control for good points} over all possible $n$. First, we can restrict the number of $n$ needed to consider by using discrepancy estimates.  Before doing so, we note that we can always reflect the configuration and the potential as in Step 5 of the proof of Proposition \ref{decay control}, in such a way that there are no boundary issues and we may 
use  Proposition~\ref{pro:controlfluct} coupled with Theorem \ref{Local Law} at the length scale $2L$  to obtain  that either $|n-\mn|\lesssim L^{\d-1}$ or 
\begin{equation*}
|\mn -n| \leq K \max \(\mathcal{C}_\beta^{1/2}L^{\frac{\d+\s}{2}}, \mathcal{C}_\beta^{1/3}L^{\frac{2\d+\s}{3}}\) \leq K\mathcal{C}_\beta^{1/2}L^{\frac{2\d+\s}{3}},
\end{equation*} for some $K>0$, 
since $\s<\d$.
%
%
%
%
%if $|n-\mn| \gtrsim L^\d$ then the discrepancy estimate Proposition \ref{pro:controlfluct} coupled with Theorem \ref{Local Law} at the length scale $2L$ yields
%\begin{equation*}
%L^{2\d}\lesssim |n-\mn|^2 \lesssim  L^{\d+\s}
%\end{equation*}
%which is a contradiction for $L$ large enough since $\s<\d$, and so $|n-\mn|\lesssim L^\d$. Then, either $D \lesssim L^{\d-1}$ or  
%\begin{equation*}
%\frac{|n-\mn|^3}{L^{\s+\d}} \lesssim \mathcal{C}L^\d  \implies |n-\mn|\le K \mathcal{C}^{1/3}L^{\frac{2\d+\s}{3}},
%\end{equation*} for some constant $K>0$.
Thus, %\cm{see again for $L/K$}
\begin{multline*}
\sum_{|\mn-n|\leq K\mathcal{C}_\beta^{1/2}L^{\frac{2\d+\s}{3}}}\Esp_{\PNbeta}\left(\exp \(\frac{\beta}{2}
\(\F^{\Omega_L}(\cdot, \Lambda) +2C_0 \#I_{\Omega} \)\)
\indic_{\mathcal{G}^{n}}\right) \\ \lesssim 
\sum_{|\mn-n|\leq K\mathcal{C}_\beta^{1/2}L^{\frac{2\d+\s}{3}}}\frac{\K_{\frac{\beta}{2},\frac{L}{K}}\left(\Omega,\mu,\zeta\right)}{\K_{\beta,\frac{L}{K}}(\Omega,\mu,\zeta)}\exp \left(\beta \left(\frac{\mathcal{C}_\beta}{8}L^\d+C_0 \mn +C C_0 K\mathcal{C}_\beta^{1/2}L^{\frac{2\d+\s}{3}}\right)\right).
\end{multline*}%\cm{insert $\beta$ everywhere}
Using Proposition \ref{pro718}, we can control the ratio of partition functions uniformly by $e^{C(1+\beta)L^\d}$.  The remaining error terms are bounded at strictly smaller order. Making $\C_\beta$ larger (but still scale independent and $\beta$-independent for $\beta \ge 1$) if necessary, in particular $\C_\beta \ge 1$, we find \eqref{pre local law}. A Chernoff bound immediately implies
\begin{multline}
\PNbeta\(\left\{\F^{\Omega_L}(\cdot, \Lambda)+2C_0 \#I_{\Omega}\geq \mathcal{C}_\beta L^\d \right\} \cap \mathcal{G}\) 
 \\ \leq e^{-\frac{\beta}{2}\mathcal{C}_\beta L^\d}\Esp_{\PNbeta}\(\exp \(\frac{\beta}{2}
 \(\F^{\Omega_L}(\cdot, \Lambda) +2C_0 \#I_{\Omega} \)
\indic_{\mathcal{G}} \)\)\leq e^{-\frac{\beta}{4}\mathcal{C}_\beta L^\d}\le e^{-\frac{\beta }{4}L^\d},
\end{multline}
with $\Gc_L=\Gc$ as in \eqref{local event}. Coupling this with the bound in \eqref{local event} and absorbing one constant into the other establishes Proposition \ref{Main Bootstrap}, and as discussed after the statement of that proposition Theorem \ref{Local Law}. 

\section{Expansion of the Next Order Partition Functions}\label{sec: partition}
The goal of this section is to leverage the local laws of Theorem \ref{Local Law} to establish a quantitative approximation of the local Neumann partition functions \eqref{defK} in the form \eqref{2eway0}. This approximation is crucial in the approach in Section \ref{sec: fluct} to improve the error estimates in Theorem \ref{FirstFluct} and obtain the CLT, Theorem \ref{CLT}. 
\subsection{Almost Additivity of the Next-Order Partition Functions}
We start with an almost additivity that generalizes the result for the Coulomb gas from \cite[Proposition 8.10]{S24} to nonCoulomb Riesz gases. The proof is similar in spirit, using the updated terminology and local laws for the Riesz gas alongside of the Riesz screening procedure Proposition \ref{Riesz screening result}. %This proposition is stated at blown-up scale. 
We recall the definitions of partition functions in \eqref{defK} and \eqref{defKext}.

\begin{prop}\label{prop: aa}
Assume $\d \geq 1$ and $\s \in (\d-2, \d)$. Let $U \subset \R^\d$ be a bounded, open set with piecewise $C^1$ boundary, and suppose as well that $\mu$ is bounded above and below in $U$. Assume that Theorem \ref{Local Law} holds for $\mu$ down to a minimal order one lengthscale $\rho_\beta$.

Assume that  $U\subset \bulk'$ can be written as a disjoint union of $p$ hyperrectangles $Q_i$ with sidelengths in $[L,2L]$ with $L \geq \rho_\beta$ such that  $\mu(Q_i)=N_i$ integer. Let $h_1, \dots, h_p \in [L,R]$. Then, for any event $\mathcal{G}$ there is some constant $C>0$, depending only on $\d$ and $\mu$ such that 
\begin{equation}\label{part ext}
\left|\log \K_{\beta,\infty}^{\mathcal{G}}(\R^\d, \mu, \zeta)-\(\log \K_{\beta,R}^{\mathrm{ext},\mathcal{G}_U}(U^c, \mu, \zeta)+\sum_{i=1}^p \log \K_{\beta,h_i}^{\mathcal{G}_{Q_i}}(Q_i, \mu, \zeta)\)\right| \lesssim_\beta R^\d \mathscr{E}( L)
\end{equation}
where $\mathcal{G}_\Omega$ denotes the set of configurations in $\Omega$ that are restrictions of a configuration in $\mathcal{G}$ and the error rate $\mathscr{E}(L)=o(1)$ is given by 
\begin{equation}\label{error rate}
\mathscr{E}(L):= \begin{cases}
L^{-1+\frac{4}{4+\d-\s_+}}\log L & \text{if }\d \leq 5 \text{ or }\s <\d-1 \\
L^{-1+\frac{\d-1}{(\d-1)+(\d-\s)}}\log L & \text{otherwise.}
\end{cases}
\end{equation}
Additionally, we have 
\begin{equation}\label{part noext}
\left|\log \K_{\beta,R}(U,\mu,\zeta)-\sum_{i=1}^p \log \K_{\beta,h_i}(Q_i,\mu,\zeta)\right|\lesssim_\beta R^\d \mathscr{E}(L).
\end{equation} 
%\cm{I think we need something like Proposition \ref{decay control} for the Neumann potential to say this. I modified the different arguments that used it previously, so perhaps cut it?}
\end{prop}
%\begin{remark}
%While the error in \eqref{part ext} at first glance appears to live at two different scales and thus differ from \cite[Proposition 8.10]{S24}, which gives an error of the form $pL^\d \mathcal{R}(\d,\s,L)$, notice that these are actually the same since if $V=\bigsqcup_{i=1}^p Q_i$ then $p=O\(\frac{R^\d}{L^\d}\)$. We choose to write $R^\d$ because of the error term that appears from outer screening in the proof below lives at scale $R$.
%\end{remark}
As we will see below, a key ingredient in optimizing the error rate $\mathcal{R}$ is that we have local laws valid down to the microscale, which will allow us to make use of Remark \ref{height change} to screen at smaller heights than we were able to in the proof of Theorem \ref{Local Law}.

\begin{proof}
We give the proof for $\mathcal{G}=\R^{\d N}$, it is a straightforward modification to extend to general $\mathcal{G}$. As in the proof of \cite[Proposition 8.10]{S24}, it suffices to prove upper bounds since the corresponding lower bounds follow from the superadditivity result \eqref{eq: supadd partext} and Stirling's formula; shrinking the heights from $h$ to $h_i$ is a consequence of the same argument as in the proof of \eqref{eq: supadd part}. Namely, we can decrease $h$ by appending a zero electric field for $|y|>h_i$ in the extended dimension to the electric field defining $\F(\XN\vert_{Q_i}, \mu, Q_i \times [-h_i,h_i])$ and applying Lemma \ref{projlem}.   %\cm{is it really going in the right direction?}

%For the upper bound, we start by separating $Q_1$; we will then iterate the approach on $\log \K_{\beta, R}^{\mathrm{ext}}\(\R^\d \setminus \bigcup_{i=1}^j Q_i, \mu, \zeta\)$ as in the proof of \cite[Proposition 8.10]{S24}. 
The key input is that we now have local laws down to the minimal length scale in $\bulk'$. 

\textbf{Step 1: restricting to a good event.}
Let $N$ denote the number of points in $U$, $n_i$ be the number of points in $Q_i$, and $\mn_i=\mu(Q_i)$. We also set 
\begin{equation*}
\hat{Q}_i:=\{x \in Q_i: \dist(x,\partial Q_i) \leq r\}
\end{equation*}
for $r$ to be determined, and 
\begin{multline}\label{good event}
\Gc=\bigg\{\XN \in (\R^\d)^N: |n_i-\mn_i|\leq \epsilon ~\forall i, \: \sup_x \int_{\(\hat{Q}_1 \cap \square_r(x)\) \times [-r,r]}|y|^\gamma |\nabla u_{\rr}|^2 \leq \mathcal{C}_\beta r^\d, \\
e(X_n, u, r) \leq C_\beta  M L^{\d} r^{(\s)_+-\d-1}
 \bigg\}
\end{multline}
with $\epsilon, \tilde{\ell}, M, K$ and $r$ to be determined and $u$ the  potential \eqref{defiu} for $\Lambda = \R^{\d+1}$.

The supremum condition is satisfied by $O\(\frac{L^{\d-1}}{r^{\d-1}}\)$ applications of the local law Theorem~\ref{Local Law}, each whose complement has probability bounded by $ \exp \(-C\beta r^\d\)$. The control on $e(X_n, w, r)$ follows from Proposition \ref{decay control} on a further restricted event for an appropriate $C_\beta$; applying Proposition \ref{decay control}, and using from Theorem \ref{Local Law} that $\PNbeta(\Gc_r^c) \leq \(\frac{L}{r}\)^\d e^{-\beta r^\d}$, we find 
\begin{equation*}
\PNbeta(\mathcal{G}^c) \leq \(\frac{L}{r}\)^{\d}e^{-C r^\d}+ 
M^{-\frac{\d}{2}}L^\d h^{\frac{\d( \d-\s_+-2)}{2}} e^{-CM}+\sum_{i=1}^p\PNbeta(|n_i-\mn_i|>\epsilon)
\end{equation*}
redefining $C$. 

Let's next analyze the event with bad discrepancy. Applying  the discrepancy estimate  \eqref{disc2buriesz}--\eqref{disc1buriesz}, covering  $\Omega_\delta \setminus \Omega$, where $\Omega=Q_i$ by $\frac{L^{\d-1}\delta}{\mathfrak{R}^\d}$ balls of size $\mathfrak{R}>\rho_\beta$, and using Theorem~\ref{Local Law} at scale $\mathfrak{R}$  one finds for each $i$ that
\begin{equation}\label{eq: aa disc}
|n_i-\mn_i|\lesssim L^{\d-1}\delta+\sqrt{L^{2\d-1+\gamma}\frac{\mathfrak{R}}{\delta}}=\epsilon
\end{equation}
off of an event of size 
\begin{equation*}
\PNbeta(|n_i-\mn_i|>\epsilon)\lesssim \frac{L^{\d-1}\delta}{\mathfrak{R}^\d}e^{-C\beta \mathfrak{R}^\d}.
\end{equation*}
Hence we find
\begin{equation}\label{bad event aa}
\PNbeta(\mathcal{G}^c) \leq \(\frac{L}{r}\)^{\d}e^{-C r^\d}+ L^\d r^{\frac{\d(\s_+-(\d-2)}{2}}e^{-CM}+Cp\frac{L^{\d-1}\delta}{\mathfrak{R}^\d}e^{-C \mathfrak{R}^\d}
\end{equation}
again redefining $C$. The third term in \eqref{bad event aa} will be smaller than $\frac{1}{4}$ as soon as we take $\mathfrak{R}^\d \gtrsim \log L$ for large enough constant; comparing the terms in the discrepancy bound \eqref{eq: aa disc} then and optimizing leads to taking $\delta=L^{\frac{1+\gamma}{3}}$, and 
\begin{equation}\label{eq: def ep}
|n_i-\mn_i|\lesssim L^{\d-1+\frac{1+\gamma}{3}}\log^{\frac{1}{2\d}}L:=\epsilon.
\end{equation}
We will optimize the remaining parameters in the above in such a way that the probability in \eqref{bad event aa} is no more than $\frac{1}{2}$.

\textbf{Step 2: superadditivity and screening.}
First, we write
\begin{equation*}
N^{-N}\int_{\Gc^c}\exp\(-\beta \F\(\XN,\mu,\R^{\d+1}\)\)~d\rho^{\otimes N}=\K_{\beta,\infty}\(\R^\d, \mu, \zeta\)\PNbeta(\mathcal{G}^c)\leq\frac{1}{2}\K_{\beta,\infty}\( \R^\d, \mu, \zeta\)
\end{equation*}
where we are again using $\rho(x)~dx$ as notation for the measure $e^{-\beta \zeta(x)}~dx$, under the assumption that we have tuned the parameters in \eqref{bad event aa} to guarantee $\PNbeta(\mathcal{G}^c)\leq \frac{1}{2}$.

Then, via the superadditivity \eqref{superadd} and \eqref{locali3}, 
\begin{multline*}
\frac{N^N}{2}\K_{\beta,\infty}\(\R^\d, \mu, \zeta \)\leq \int_{\Gc}\exp\(-\beta \F\(\XN,\mu,\R^{\d+1}\)\)~d\rho^{\otimes N} \\
\leq \sum_{\substack{|n_i-\mn_i|\leq \epsilon, \\ \sum n_i=N}}{N \choose n_1 \: n_2 \dots n_p}\prod_{i=1}^p\int_{\Gc}\exp\(-\beta \G_{b,r}^{\mathrm{inn}}(\cdot, Q_i\times [-h_i,h_i] )\)~d\rho^{\otimes n}(X_n)\\
\times \int_{\Gc}\exp\(-\beta \G_{b,r}^{\mathrm{ext}}\(\cdot, \(U\times [-R,R]\)^c\)\)~d\rho^{\otimes N-n}(X_{N-n}).
\end{multline*}
where $b$ denotes the expression $C_\beta M L^{\d} r^{\s_+-\d-1}$ controlling $e(X_n,u,r)$
 %At this first step we of course have $L=R$ and $Q_1=U$, but we have kept the notation $Q_i \times [-L,L]$ for the inner screening and $\(U \times [-R,R]\)^c$ for the outer screening to make the later iteration $p \geq 2$ as clear as possible. 
%\cm{actually not clear- shouldn't there be $p$ terms in the product above?}  
%For the later steps when $p \geq 2$, 
and we have replaced $Q_{i} \times [-L, L]$ in the equation above with any $h_{i} \in [L,R]$ using Remark \ref{height change}.

We next replace the energies $\G$ in the exponents with the corresponding Neumann energies \eqref{minneum} using the screening on the level of partition functions, Proposition \ref{volume of configurations}. Applying Remark~\ref{bdd vol error}, we obtain 
\begin{multline}\label{firstpass}
\K_{\beta,\infty}\(\R^\d, \mu, \zeta\)\leq 2\prod_{i=1}^p\K_{\beta,h_i}(Q_i,\mu, \zeta)\K_{\beta,R}^{\mathrm{ext}}(U^c, \mu, \zeta)\sum_{\substack{|n_i-\mn_i|\leq \epsilon, \\ \sum n_i=N}}{N \choose \mn_1, \: \mn_2, \dots, \mn_p}N^{-N}\prod_{i=1}^p \mn_i^{\mn_i} \\
\times \exp\(Cp\(\epsilon+ \tilde{\ell}^\d+ \(\frac{r^2}{\tilde{\ell}}+\tilde{\ell}+r^{-\gamma}\)L^{\d-1}+C_\beta \frac{L}{\min(r,\tilde{\ell})}\(\frac{L^{2}}{\tilde{\ell}}+L\)
M L^{\d} r^{\s_+-\d-1}\)\)
\end{multline}
%\cm{again product of $p$ terms here?}
using $r <L$, provided \eqref{height control} holds and, if $\s \geq \d-1$, $r \gg L^{\frac{\d-1}{\d-\gamma}}$. An application of Stirling's formula yields
\begin{equation*}
{N \choose \mn_1 \: \mn_2 \dots \mn_p}N^{-N}\prod_{i=1}^p \mn_i^{\mn_i} \leq C\sqrt{\frac{N}{\prod_{i=1}^p\mn_i}}\leq C
\end{equation*}
and so the summands are controlled by $C\epsilon p$.

%The summands can be controlled by $C\epsilon$ using Stirling's formula.  %\cm{is that still ok when there are $p$ terms?}

%Note that, a priori, when $p \geq 2$ one may want to choose different parameters for $M,r$ and $P$ below at the different scales $L$ and $R$, modifying \eqref{good event} as well. To keep the notation as simple as possible we absorb them for now since in this case we will choose the same parameters for the inner and outer screening.

\textbf{Step 3: error optimization.} 
We now need to optimize the choice of $r$, with competing terms that prefer $r$ smaller or larger. Observe that we want to take $M\ge 1$ as small as possible while keeping the estimate in \eqref{bad event aa} small. 

%If $\s>0$ we may as well take $M=Cr^{\s}$ since we will need to take $r$ below at some power of $L$; in the $\s \leq 0$ case this is not so good, so there we will take $M=C\log L$ for $C$ large enough.\cm{outdated} 
Notice that \eqref{bad event aa}  
%\begin{equation*}
%\PNbeta(\mathcal{G}^c)\leq C_1\(\frac{L}{r}\)^{\d-1}e^{-C_2r^\d}+\(\frac{L}{r}\)^\d e^{-C_2M^2}+\frac{1}{4}
%\end{equation*}\cm{review second term}
can be made smaller than $\frac{1}{2}$ for any $r$ at order $L^\nu$ and $M\gtrsim \log L$ . So, we may as well optimize 
\begin{equation}\label{opt 1}
L^{\d-1+\frac{1+\gamma}{3}}\log^{\frac{1}{2\d}}L+\tilde{\ell}^\d +\(\frac{r^2}{\tilde{\ell}}+\tilde{\ell}+r^{-\gamma}\)L^{\d-1}+C_\beta\frac{L}{\min(r,\tilde{\ell})}\(\frac{L^{2}}{\tilde{\ell}}\)M L^{\d} r^{\s_+-\d-1}
\end{equation} in \eqref{firstpass} using \eqref{eq: def ep} subject to 
\begin{equation}\label{opt 2}
C_\beta \frac{L^2}{\tilde{\ell}^2}Mr^{\s+\s_+-2\d}\leq \mathsf{c}, \qquad r \gg L^{\frac{\d-1}{\d-\gamma}} \text{ if }\s\geq \d-1,
\end{equation}
%\cm{shouldn't we insert $\ep$ from \eqref{eq: def ep} as well?}
where we have dropped $L$ in the parenthesis in \eqref{opt 1} because $\tilde{\ell}<L$ implies the $\frac{L^{2\d}}{\tilde{\ell}^{2\d-1}}$ term is dominant. Comparing the optimization problems for $\tilde{\ell}<r$ and $r<\tilde{\ell}$, one sees in either case that the optimum choice is as the parameters approach $r$. So, we select $\ell=\tilde{\ell}=r$, and choose $M=C\log L$ for $C$ large enough.
% and proceed to discuss the $\s>0$ and $\s\leq 0$ cases.
%Here, set $M=Cr^{\s}$  for $C$ large enough. 
Making the simplification $\tilde{\ell}=r$, we see that we need to optimize (absorbing some terms)
\begin{equation}\label{eq: opt 2}
L^{\d-1+\frac{1+\gamma}{3}}\log^{\frac{1}{2\d}}L+rL^{\d-1}+L^{\d+3}r^{\s_+-\d-3}M
\end{equation}
subject to 
\begin{equation}\label{eq: modcon}
C_\beta ML^2 r^{\s+\s_+-2\d-2}\leq \mathsf{c}, \qquad r \gg L^{\frac{\d-1}{\d-\gamma}} \text{ if }\s\geq \d-1,
\end{equation}
where we have absorbed the error term $\ell^{1-\d}r^{-\gamma}L^{\d-1}=r^{-\s}L^{\d-1}$ into $rL^{\d-1}$ since $\s>-1$.
Balancing the second and third terms in \eqref{eq: opt 2} requires 
\begin{equation}\label{def }
r^{\d+4-\s_+}=ML^4.
\end{equation}
The first term in \eqref{eq: opt 2} is always smaller with this choice of $r$, since $\frac{1+\gamma}{3}=\frac{\s+2-\d}{3}$ and
\begin{equation*}
\frac{4}{\d+4-\s_+}>\frac{\s+2-\d}{3}\iff 12>(\s+2-\d)(\d-\s_++4)
\end{equation*}
which follows from $\s \in (\d-2,\d)$. Notice also that this choice of $r$ satisfies \eqref{eq: modcon}. Indeed, for the first condition we observe that, with $\gamma=\s+1-\d$ and this choice of $r$,
\begin{equation*}
C_\beta ML^2 r^{\s+\s_+-2\d-2}=\frac{C_\beta ML^2}{r^{\d+4-\s_+}r^{\d-2-\s}}=\frac{C_\beta r^{\s-(\d-2)}}{L^2} \ll 1 
%\begin{equation*}
%C_\beta\frac{ML^{2\d}r^{\s_+-\d-1}}{r^{2\d+\gamma}}=C_\beta \frac{ML^{2\d+2}}{L^2 r^{2\d+2-\s_+}r^\s}=C_\beta \frac{r^\d}{L^2r^\s} \leq r^{\d-2-\s} \ll 1
\end{equation*}
since $r<L$ and $\s>\d-2$. For the second condition, since $M=C\log L$ it is sufficient to check that $\frac{4}{4+(\d-\s)} \geq\frac{\d-1}{\d-\gamma}$ (since $\s\geq \d-1 \implies \s \geq 0 \implies \s_+=\s)$. Indeed, since $\gamma=\s+1-\d$ this is true if and only if 
\begin{equation*}
\frac{4}{4+(\d-\s)} \geq \frac{\d-1}{(\d-1)+(\d-\s)}.
\end{equation*}
which is true for $\d \leq 5$ because for $a>0$, the quantity $\frac{x}{x+a}$ is increasing in $x$. For $\d>5$, we instead need to take then $r=CL^{\frac{\d-1}{\d-\gamma}}\log L$.

%\subsubsection*{$\s \leq 0$}
%This case only appears for $\d=1$. Here, set $M=C\log L$  for $C$ large enough and again make the simplification $\tilde{\ell}=r$. To obtain a rough bound, we see we need to optimize (absorbing some terms)
%\begin{equation*}
%r+\frac{L^4 \log L}{r^4}, \qquad \frac{L^2 \log L}{r^{4}} \leq \mathsf{c}
%\end{equation*}
%where we have dropped the $r \gg L^{\frac{\d-1}{\d-\gamma}}$ condition since $\s \leq \d-1$. Balancing the first two terms requires $r \gtrsim L^{\frac{4}{5}}$, which also satisfies the third constraint.

We conclude that 
\begin{equation}\label{aa bd}
\K_{\beta,\infty}\(\R^\d,\mu,\zeta\)\leq C\epsilon p \prod_{i=1}^p\K_{\beta,h_i}(Q_i, \mu, \zeta)\K_{\beta,R}^{\mathrm{ext}}(U^c,\mu, \zeta)\exp \( C_\beta L^\d\mathscr{E}(L)\)
\end{equation}
where $\mathscr{E}(L)$ is as in \eqref{error rate}.
Taking logarithms yields the result.

%\textbf{Step 4: conclusion.}
%The same approach can now be iterated on $\log \K_{\beta, R}^{\mathrm{ext}}\(\R^\d \setminus \bigcup_{i=1}^j Q_i, \mu, \zeta\)$ as in the proof of \cite[Proposition 8.10]{S24}. 
The proof of \eqref{part noext} is the same using the local laws for the measure associated to $\K_{\beta, R}(U,\mu,\zeta)$
\end{proof}

The goal is now to use this to derive a precise asymptotic expansion of local partition functions. First, we prove the analogue of \cite[Proposition 6.2]{AS21} at constant density. It is largely a corollary of the previous proposition.
\begin{lem}\label{lem: const dens}
Assume $\d \geq 1$ and $\s \in (\d-2,\d)$. Then, there exists a function $\mf:(0,\infty)\rightarrow \R$ and a constant $C>0$ depending only on $\d$ such that  if $L^\d \in \Z$ we have
\begin{equation}
\left|\frac{\log \K_{\beta,L}(\carr_L,1,0)}{\beta|\carr_L|}+\mf(\beta)\right|\lesssim_\beta \mathscr{E}(L),
\end{equation}
with $\mathscr{E}$ as in \eqref{error rate}.
\end{lem}
%\cm{connect $\mf$ to the rate function of Leble-S, in the introduction}
\begin{proof}
The proof is exactly as in \cite[Proposition 6.2]{AS21} using Proposition \ref{prop: aa}, with our new error term. 
One first uses superadditivity \eqref{eq: supadd part} of the partition function to derive
\begin{equation*}
\frac{1}{\beta}\log \K_{\beta, 2L} \(\square_{2L},1, 0\) \geq O\(\frac{\log N}{\beta}\)+\frac{2^\d}{\beta}\log \K_{\beta,L}\(\square_L,1, 0\)
\end{equation*}
and so, setting $\phi(L):=\frac{\log \K_{\beta, L}(\square_L,1,0)}{\beta L^\d}$ one has
\begin{equation*}
\phi(2L)\geq \phi(L)+O\(\frac{\log L}{\beta L^\d}\).
\end{equation*}
Iterating the previous estimate, we find
\begin{equation*}
\phi(\infty) \geq \phi(L)+O\(\sum_{k=1}^\infty \frac{\log L}{\beta\(2^k L\)^\d}\)
\end{equation*}
where $\phi(\infty)$ is defined by $\limsup_{L \rightarrow \infty}\phi(L)$, which simplifying yields 
\begin{equation*}
\phi(L)\leq \phi(\infty)+O\(\frac{\log L}{\beta L^\d}\)
\end{equation*}
as in \cite{AS21}.
On the other hand, using Proposition \ref{prop: aa} we can write
\begin{equation*}
\phi(2L) \leq \phi(L)+C_\beta\mathscr{E}(L)
\end{equation*}
which we can sum to find $\phi(\infty) \leq \phi(L)+C_\beta\mathscr{E}(L)$. Defining $\mf(\beta):=-\phi(\infty)$ yields the result.
\end{proof}
This Lemma allows us to compute the next-order partition function for other constant densities as well; indeed, by scaling and setting $Q_L'=m^{\frac{1}{\d}}Q_L$, we find 
\begin{align}\label{part scaling}
\nonumber\frac{\log \K_{\beta,\infty}(Q_L,m,0)}{\beta|Q_L|}&=m^{1+\frac{\s}{\d}}\frac{\log \K_{\beta m^{\frac{\s}{\d}}, \infty}(Q_L',1,0)}{\beta m^{\frac{\s}{\d}}|Q_L'|}-\frac{m}{\beta}\log m+\frac{1}{2\d}(m\log m) \indic_{\s=0} \\
&=-m^{1+\frac{\s}{\d}}\mf\(\beta m^{\frac{\s}{\d}}\)-\frac{m}{\beta}\log m+\frac{m}{2\d}(\log m )\indic_{\s=0}+O\(\mathscr{E}(L)\)
\end{align}

In what follows, we will need the  assumption  \eqref{425}.
%As discussed in the introduction, this is a {\it no-phase transition} assumption at the effective temperature  $\beta \muv(x)^{\frac{\s}{\d}}$. 

\subsection{Comparison of Partition Functions by Transport}
We would now like to leverage the formula \eqref{part scaling} to obtain expansions for inhomogeneous densities $\mu$. In this section we return to normal coordinates, since the expansions we will derive will be leveraged to obtain results about fluctuations of linear statistics stated in usual coordinates. Just as in \cite[Chapter 9]{S24}, \eqref{defK} has the counterpart at the usual scale \begin{equation}\label{defK bd}
\K_{N,\beta,\ell}(U,\mu,\zeta):=\int_{U^N}\exp\(-\beta N^{-\frac{\s}{\d}}\(\F_N(\XN,\mu,U\times [-\ell,\ell])+N\sum_{i=1}^N\zeta(x_i)\)\)~d\XN
\end{equation}
where $\F_N(\XN,\mu,U\times[-\ell,\ell])$ and $N\zeta$ are analogous to $\F\(\XN', \mu', UN^{\frac{1}{\d}}\times \left[-\ell N^{\frac{1}{\d}}, \ell N^{\frac{1}{\d}}\right]\)$ and $\zeta'$ as in Section \ref{subsec: ll neum}. 
% \cm{well, they are not really equal since there is a $N^{\frac\s\d}$ rescaling factor - or just say they are the analogues}  
A superscript $\mathcal{G}$ in \eqref{defK bd} will indicate a restriction of the integral to some good event $\mathcal{G}\subset U^N$. As in \eqref{defK fluct}, when $U=\R^\d$ and $h=\infty$ we abbreviate \eqref{defK bd} by $\K_{N,\beta}$. We also define a Gibbs measure associated to this next-order partition function by 
\begin{equation}\label{defQ bd}
d\Q_{N,\beta,\ell}(U,\mu,\zeta):=\frac{1}{\K_{N,\beta,\ell}(U,\mu,\zeta)}\exp\(-\beta N^{-\frac{\s}{\d}}\(\F_N(\XN,\mu,U\times [-\ell,\ell])+N\sum_{i=1}^N\zeta(x_i)\)\)~d\XN.
\end{equation}

If $\mu$ is sufficiently smooth then it cannot vary too much in a small cube of sidelengths  $\ell \ll 1$. In particular then, we should be able to compare $\K_{N,\beta,\ell}(Q_\ell,\mu, \zeta_V)$ to that previously obtained for constant densities, more precisely $\K_{N,\beta,\ell}\(Q_\ell,\fint_{Q_\ell}\mu, \zeta\)$. We will make this comparison explicit by a transport as in \cite[Chapter 9]{S24}. First, let us recall some information about transporting Neumann energies in cubes from \cite[Section 9.2]{S24}. Let $\psi_t: Q_\ell \to Q_\ell$ be a Lipschitz vector field depending continuously on a  parameter $t\in [0,1]$ and satisfying $\psi_t \cdot \nu=0$ on $\partial Q_\ell$. Let us define the flow $\Phi_t: Q_\ell \to Q_\ell$ to be the solution to 
\be \label{defflow}
\left\{ \begin{aligned}
&\frac{d \Phi_t}{dt} (x)= \psi_t (\Phi_t(x))\\
& \Phi_0(x)=x.\end{aligned}\right. \ee
This flow is well-defined for $t\in [0,1]$ by standard ODE theory. Moreover, it is standard to check that if $\mu$ is a probability density then the push-forward 
\be \label{defmutflow}\mu_t := \Phi_t \# \mu\ee solves 
\be\label{continuityeq} \partial_t \mu_t + \div (\psi_t \mu_t)= 0\quad \text{in} \ Q_\ell .\ee
In \cite{S24}, the following lemma is shown, it is valid in the same way in $Q_\ell$.
 \begin{lem}\label{1221}Let $\mu_0$ and $\mu_1$ be two probability densities on $\R^\d$ and let  $\mu_s= (1-s) \mu_0 + s\mu_1$ be their linear interpolant. 
 Assume that for  $s \in [0,1]$,  $\psi_s$ is a Lipschitz vector field on $\R^\d$, depending continuously on $s$, and satisfying
 \be \label{divpsimu}
- \div (\psi_s\mu_s)= \mu_1-\mu_0.\ee
 Then defining $\Phi_s$ as in \eqref{defflow}, we have that $\mu_s = \Phi_s \# \mu_0$.
 \end{lem}    

\begin{lem}\label{lem: vardens}
Assume $\ell N^{\frac{1}{\d}}>\rho_\beta$, and let $Q_\ell$ be a hyperrectangle of sidelengths in $(\ell, 2\ell)$. Let $\mu_0,\mu_1$ be two Lipschitz  densities bounded  above and below by positive constants in $Q_\ell$, a hyperrectangle of sidelengths in $[\ell, 2\ell]$ with $N\mu_0(Q_\ell)=N\mu_1(Q_\ell)=\mn$ an integer. Let $h>0$ and let $\mathcal{G}$ denote an event on which the local laws hold for $\Q_{N,\beta,h}(Q_\ell, \mu_s,0)$. Then 
\begin{multline}\label{compdeskcube}
|\log \K_{N,\beta,h}^{\mathcal{G}}(Q_\ell, \mu_1,0)-\log \K_{N,\beta,h}(Q_\ell, \mu_0,0)+N(\Ent_{Q_\ell}(\mu_1)-\Ent_{Q_\ell}(\mu_0))|\\
\lesssim_\beta  (1+\beta)  N\ell^\d \(\ell^2  ( |\mu_0|_{C^1}+  |\mu_1-\mu_0|_{C^1} )  |\mu_1-\mu_0|_{C^1} 
+ \ell  |\mu_1-\mu_0|_{C^1}\),\end{multline}
where $C$ depends only on $\d$ and a lower bound for $\mu_0$ and $\mu_1$, and where $\Ent_{Q_\ell}(\mu)$ is the entropy restricted to $Q_\ell$, i.e.~$\int_{Q_\ell}\mu \log \mu$. 
\end{lem}
\begin{proof} 
Since we are working with the Neumann energy in a cube, we need to find a transport that preserves the cube and solves \eqref{divpsimu}. For that  we let $\varphi$ solve 
\be \label{solxi}
\left\{
\begin{array}{ll}
-\Delta \varphi= \mu_1-\mu_0& \text{in} \ Q_\ell\\ [2mm]
\frac{\pa \varphi}{\pa n}=0 & \text{on} \ \pa Q_\ell .\end{array}\right.
\ee
By elliptic regularity and scaling we have  
$$|\varphi|_{C^1}\le C \ell^2|\mu_1-\mu_0|_{C^1}, \quad |\varphi|_{C^2}\le C \ell|\mu_1-\mu_0|_{C^1}.$$
Define, for $0\le s \le 1$, the linear interpolant
$ \mu_s= (1-s) \mu_0+ s\mu_1$.
Setting $$\psi_s:= \frac{\nab \varphi}{\mu_s},$$ 
we thus have $-\div (\psi_s\mu_s)= \mu_1-\mu_0$ i.e.~\eqref{divpsimu} is satisfied, thus 
$ \mu_1=\Phi_1\#\mu_0.$
Moreover, by simple estimates 
\begin{multline} \label{418}|\psi_s|_{C^1}\le C \(\left\| \frac{1}{\mu_s}\right\|^2_{L^\infty} |\mu_s|_{C^1}    |\varphi|_{C^1}+   \left\|\frac{1}{\mu_s}\right\|_{L^\infty}  |\varphi|_{C^2}\) \\
\le C \( \ell^2  \(  |\mu_0|_{C^1} + |\mu_1-\mu_0|_{C^1}\)  |\mu_1-\mu_0|_{C^1} 
+ \ell  |\mu_1-\mu_0|_{C^1}\),
\end{multline}
where $C$ depends only on $\d$ and a lower bound for $\mu_0$ and $\mu_1$. Computing as in \eqref{09}, we have 
\begin{multline*}
\log \K_{N,\beta,h}^{\mathcal{G}}(Q_\ell, \mu_1, 0)-\log \K_{N,\beta,h}(Q_\ell, \mu_0, 0)+N \left(\Ent_{Q_\ell}(\mu_1)- \Ent_{Q_\ell}(\mu_0)\right) \\
=\log \Esp_{\Q_{N,\beta,h}\(Q_\ell, \mu_0,0\)} \exp\biggl(-\beta N^{-\frac{\s}{\d}}\(\F_N\(\Phi_1(\XN), \mu_1, Q_\ell\times [-h,h]\) -\F_N\(\XN,\mu_0, Q_\ell \times [-h,h]\)\) \\
+\Fluct_{\mu_0}(\log \det D\phi_1)\biggr)\indic_{\mathcal{G}}.
\end{multline*}
The first term in the exponent is given by (cf. \eqref{dtF})
\begin{equation*}
-\beta N^{-\frac{\s}{\d}}\int_0^1 \Ani_1(\Phi_s(\XN), \mu_s, \psi \circ \Phi_s^{-1},Q_\ell) ds
\end{equation*}
where $\Ani_1 (\XN, \mu, \psi, U)$ is the Neumann variant of \eqref{dtF} defined in \cite[Section 9.2]{S24},  for which the same control as in 
 Proposition \ref{pro:commutator} holds. \cm{I had to change here} 
 Using then the local laws on $\mathcal{G}$ coupled with Proposition~\ref{pf loiloc}, this is controlled by
\begin{equation*}
\beta N^{-\frac{\s}{\d}}\int_0^1 |\Ani_1(\Phi_s(\XN), \mu_s, \psi \circ \Phi_s^{-1},Q_\ell)|~ds\lesssim_\beta \beta N\ell^\d \|\psi\|_{C^1}
\end{equation*} 
The second term is similarly well-controlled; \eqref{Linf} and the local laws yield
\begin{equation*}
\left|\Fluct_\mu(\log \det D\phi_1)\right|\lesssim_\beta N\ell^\d\|\psi\|_{C^1}.
\end{equation*}
%Integrating \eqref{derivlogkteqx2} we have 
%\begin{multline}\label{419}
%\log \K_{N,\beta,h}^{\mathcal{G}}(Q_\ell,\mu_1, 0)- \log \K_{N,\beta,h}(Q_\ell,\mu_0, 0)+N(\Ent_{Q_\ell}(\mu_1)-\Ent_{Q_\ell}(\mu_0))
%\\
%= \int_0^1 \Esp_{\Q_{N,\beta}(Q, \mu_s,0)} \( -\beta N^{-\frac\s\d} \Ani_1(\XN, \mu_s, \psi_s)
%+\Fluct( \div\psi_s)\)\indic_{\mathcal{G}}ds.
%\end{multline}
%\cm{need to modify with the $\mathcal{G}$}.
%Using the rough bound for fluctuations \eqref{Linf}), Proposition \ref{pro:commutator} to bound $\Ani_1$ and the local laws for the true potential \eqref{ll Rieszpot}) generated by a configuration in $Q_\ell$ distributed according to $\Q_{N,\beta,h}(Q_\ell, \mu_s,0)$ we deduce that 
% \begin{multline}
% \left|\log \K_{N,\beta,h}^{\mathcal{G}}(Q_\ell,\mu_1, 0)- \log \K_{N,\beta,h}(Q_\ell,\mu_0, 0)+N(\Ent_{Q_\ell}(\mu_1)-\Ent_{Q_\ell}(\mu_0))\right|
%\\ \le \beta \chi(\beta) N\ell^\d \int_0^1 |\psi_s|_{C^1}ds.\end{multline}
With \eqref{418} we deduce the result.
\end{proof}

We now are able to write an expansion of the next-order partition functions as in \eqref{2eway0}.

\subsection{Relative Expansion of Next-Order Partition Functions}
The idea for obtaining \eqref{2eway0} is based on tools from \cite{S22}, with the added difficulty that the transport \eqref{Rtransport} is nonlocal. We will use the almost additivity Proposition \ref{prop: aa} to split the consideration between $U=Q_R$ for some $R>\ell$ to be determined and $Q_R^c$, applied to $Q_i$ with sidelengths in $[r,2r]$. We will compare $\K_{N,\beta, r}(Q_i,\mu,\zeta)$ to that of $\K_{N,\beta,r}(Q_i, \fint \mu, \zeta)$ in Lemma \ref{lem: const dens} using Lemma \ref{lem: vardens}, and optimize over the choice of $r$ to improve our error bounds. Our analysis of $\K_{N,\beta,R}^{\mathrm{ext}}(Q_R,\mu,\zeta)$ will rely on the decay of \eqref{Rtransport} away from $Q_\ell$.
 \begin{prop}[Comparison of Partition Functions]\label{Kcomp}
 Assume $\s \in (\d-2,\d)$ and that $\ell N^{\frac{1}{\d}}>\rho_\beta$. Assume $\supp \varphi  \subset \carr_\ell\subset \carr_{2\ep}\subset \Sigma$.  Let $\psi$ be given by Proposition \ref{transport}, and $\mu_t:=(\mathrm{Id}+t\psi)\#\muv$. Let $\mathcal{G}_\ell$ be the event of Theorem~\ref{FirstFluct}, and set
\begin{equation}\label{def kappa}
\kappa:=
%1-\frac{2\d+2}{3\d+2-\s_+} .
\begin{cases}
1-\frac{4}{4+\d-\s_+}\ & \text{if }\d \leq 5 \text{ or }\s <\d-1 \\
1-\frac{\d-1}{(\d-1)+(\d-\s)} & \text{otherwise.}
\end{cases}
\end{equation}
 Then, for all $|t|\ell^{\s-\d}$ is small enough, we have
 \begin{multline*}
 \log \K_{N,\beta,\infty}^{\mathcal{G}_\ell}\(\R^\d, \mu_t,\zeta_V\circ \Phi_t^{-1}\)- \log \K_{N,\beta,\infty}\(\R^\d,\mu_0,\zeta_V\)+N(\Ent(\mu_t)-\Ent(\mu_0) ) \\=
\mathcal{Z} (\beta, \mu_t)-\mathcal{Z}(\beta, \mu_0)+ O_\beta((1+\beta)  N\ell^\d \mathcal R_t)
 \end{multline*}
 with 
 \be \mathcal{Z}(\beta, \mu)= -N\beta \int \mu^{1+\frac\s\d} \mf(\beta \mu^{\frac\s\d}) + \frac{N\beta}{2\d}\indic_{\s=0}\int\mu \log \mu\ee
and 
\begin{equation}\label{final error rate}
\mathcal{R}_t= \max\( ( |t|\ell^{\s-\d})^{\frac\d{2\d-\s}}\(\mathscr{E}\((N^{\frac{1}{\d}} \ell)^{\frac{1}{\kappa +1}}\)  \)^{\frac{\d-\s}{2\d-\s}} ,|t|\ep^{\s-\d}\)
\end{equation}where $\mathscr{E}$ is as in \eqref{error rate} and $\ep $ is as in \eqref{bulk}.
 \end{prop}
Together with \eqref{maxmu} this will provide \eqref{2eway0}.
 \begin{proof} 
 Before we start, we replace $\log \K_{N,\beta,\infty}\(\R^\d,\mu_0,\zeta_V\)$ with $\log \K_{N,\beta,\infty}^{\mathcal{G}_\ell}\(\R^\d,\mu_0,\zeta_V\)$; since the difference between the two is just $\log \PNbeta(\mathcal{G}_\ell)=\log (1+O(e^{-C\beta N \ell^\d}))=O(e^{-C\beta N \ell^\d})$ so  exponentially small, we will be able to absorb it into the error terms we generate below.
 
If $\ell$ is large enough, we just use simple additivity at some scale $R$. Otherwise, it is better to include $Q_\ell$ into some larger cube $Q_R$ with $R$ to be determined. 

\textbf{Step 1: Cutoff via Almost Additivity.}
Since $\supp \varphi \subset \bulk$, we may include $\supp \varphi$ in a hyperrectangle $Q_R$ at distance $\ge \ep>0$ from $\partial \Sigma$, and such that $\muv(Q_R) $ is an integer, for some $R\le \ep$.

First, we apply almost additivity (Proposition \ref{prop: aa})  on $U_1=Q_R$ and $U_2=\Phi_t(Q_R)$. Note that if $t\ell^{\s-\d}$ is sufficiently small, $|\Phi_t-\id |\le t\|\psi\|_{L^\infty}$ is small as well, and thus both $U_1$ and $U_2$ are  at distance $\ge \ep$ from $\partial \Sigma$, hence sets  where $\muv$ and $\mu_t$ are bounded below and $\zeta_V$, $\zeta_V\circ \Phi_t$ vanish, and local laws hold.

We will partition $U_1$ and $U_2$  into sets $Q_i$ with sidelengths in $[r,2r]$ with $r$ to be optimized; this yields   
\begin{multline}\label{eq: firstexp}
\log \K_{N,\beta,\infty}^{\mathcal{G}_\ell}\(\R^\d,\mu_t,\zeta_V\circ \Phi_t^{-1}\)- \log \K_{N,\beta,\infty}^{\mathcal{G}_\ell}\(\R^\d,\mu_0,\zeta_V\)=\\
\sum_{i=1}^p \( \log \K_{N,\beta,r}^{\mathcal{G}_{\Phi_t(Q_i)}}\(\Phi_t(Q_i), \mu_t, 0\)- \log \K_{N,\beta,r}^{\mathcal{G}_{Q_i}}\(Q_i, \mu_0, 0\)\) +
NR^\d O_\beta\(\mathscr{E}( rN^{\frac{1}{\d}})\)\\
+\log \K_{N,\beta,R}^{\mathrm{ext}, \mathcal{G}_{\Phi_t(Q_R^c)}}\(\Phi_t(Q_R^c),\mu_t,\zeta_V \circ \Phi_t^{-1}\)-\log \K_{N,\beta,R}^{\mathrm{ext},\mathcal{G}_{Q_R^c}}\(Q_r^c,\mu_0, \zeta_V\)
\end{multline}
where $\mathscr{E}$ is as in \eqref{error rate} and $\mathcal{G}_\Omega$ denotes the set of configurations in $\Omega$ that are restrictions of a configuration in $\mathcal{G}_\ell$.

Since $\mu_t= (\id+ t\psi) \# \mu_0$, a  computation based on \eqref{explicitmu} and \eqref{tscale} yields that  if $t \ell^{\s-\d}$ is small enough, 
\be \label{maxmu}
\|\mu_t\|_{C^1} \lesssim  |t| \|\psi\|_{C^2}\lesssim \ell^{\d-\s}\frac{1}{\ell^{\d-\s+1}}=\ell^{-1}.\ee 

\textbf{Step 2: Comparison of Local Partition Functions.}
We next analyze 
\begin{equation*}
\sum_{i=1}^p\log \K_{N,\beta,r}^{\mathcal{G}_{Q_i}}\(Q_i, \mu, 0 \)
\end{equation*}
with $\mu$ either $\mu_0$ or $\mu_t$ above. The idea is reduce to Lemma \ref{lem: const dens} using Lemma \ref{lem: vardens}. Using the scaling identity from \eqref{part scaling} we find  for each $i$ that 
\begin{multline*}
\log \K_{N,\beta,r}^{\mathcal{G}_{Q_i}}\(Q_i, \mu, 0 \) +N\Ent_{Q_i}(\mu)\\
=-\beta N |Q_i|\(\overline{\mu}^{1+\frac{\s}{\d}}\mf\(\beta \overline{\mu_i}^{\frac{\s}{\d}}\)-\frac{1}{2\d}\overline{ \mu}\log \overline{\mu} \indic_{\s=0}\)+\(\frac{\beta}{2\d}\mn \log N\)\indic_{\s=0}\\
+(1+\beta)Nr^\d O_\beta\(r\|\mu\|_{C^1}+r^2\|\mu\|_{C^1}^2\indic_{\s=0}+\mathscr{E}\(r N^{\frac{1}{\d}}\)\)
\end{multline*}
with $\mn=N\int_{Q_U}\mu$,  and $\overline{\mu}_i=\fint_{Q_i}\mu$. Running the same computation as in the proof of \cite[Lemma 6.1]{S22} to replace $\mu$ with $\overline{\mu}$ using \eqref{425}, and summing over $i$, we obtain
\begin{multline}
\sum_{i=1}^p\log \K_{N,\beta,r}^{\mathcal{G}_{Q_i}}\(Q_i, \mu, 0 \)+N\Ent_{Q_i}(\mu)=-\beta N\int_{U}\mu^{1+\frac{\s}{\d}}\mf\(\beta \mu^{\frac{\s}{\d}}\)-\frac{\beta}{2\d}N\indic_{\s=0}\int_U \mu\log \mu\\+\(\frac{\beta}{2\d}\mn \log N\)\indic_{\s=0}
+
(1+\beta)NR^\d O_\beta\(r\|\mu\|_{C^1}+r^2\|\mu\|_{C^1}^2\indic_{\s=0}+\mathscr{E}\(r N^{\frac{1}{\d}}\)\)
\end{multline}
with $\mathscr{E}$ as in \eqref{error rate}. 
Requiring $r\|\mu\|_{C^1}\lesssim 1$, equivalently $r< \ell$ in view of \eqref{maxmu},   to absorb some errors, we need to optimize 
$\mathscr{E}(r N^{\frac1\d})+r\|\mu\|_{C^1}$
over choices of $r<\ell$. Letting $-\kappa$ with $\kappa>0$ denote the exponent on $L$ in \eqref{error rate}, we see that the right choice of $r$ is 
\begin{equation*}
r=N^{-\frac{\kappa}{\d(\kappa+1)}}\|\mu\|_{C^1}^{-\frac{1}{\kappa+1}}.
\end{equation*} With this choice, we have $r\|\mu\|_{C^1}\lesssim 1$ thanks to the fact that $\ell \ge \rho_\beta N^{-1/\d}$ and $\|\mu\|_{C^1}\lesssim \ell^{-1}$. Thus, the above optimization is the correct one, moreover $r \lesssim \ell$, and up to multiplying $r$ by a constant, we have $ r <\ell \le R$, hence the partitioning of $Q_R$ into size $r$ hyperrectangles makes sense (at least provided $R/\ell$ is large enough).
If $R/r$ is large enough, we use \cite[Lemma 5.13]{S24} to  partitioning $Q_R$ into hyperrectangles of size comparable to $r$ and with quantized $\mu$-mass, i.e.~$N\mu(Q_i)$ integer. If  $R/r$ is not, then we do not partition further $Q_R$.
In all cases this yields 
\begin{multline}\label{eq: innerexp}
\sum_{i=1}^p\log \K_{N,\beta,r}^{\mathcal{G}_{Q_i}}\(Q_i, \mu, 0 \)+N\Ent_{Q_i}(\mu)=-\beta N\int_{U}\mu^{1+\frac{\s}{\d}}\mf\(\beta \mu^{\frac{\s}{\d}}\)-\frac{\beta}{2\d}N \indic_{\s=0}\int_U \mu\log \mu\\ +\(\frac{\beta}{2\d}\mn \log N\)\indic_{\s=0}+O_\beta \( (1+\beta)NR^\d\mathscr{E}\( (N^{\frac{1}{\d}} \|\mu\|_{C^1}^{-1})^{\frac{1}{\kappa +1}}\)\).
\end{multline}

\textbf{Step 3: External estimate.}
We turn to estimating
\begin{equation*}
\log \K_{N,\beta,R}^{\mathrm{ext}, \mathcal{G}_{\Phi_t(Q_R^c)}}\(\Phi_t(Q_R^c),\mu_t,\zeta_V \circ \Phi_t^{-1}\)-\log \K_{N,\beta,R}^{\mathrm{ext},\mathcal{G}_{Q_R^c}}\(Q_r^c,\mu_0, \zeta_V\)
\end{equation*}
in \eqref{eq: firstexp}. By the same computation as in \eqref{09} we obtain 
\begin{multline*}
\log \K_{N,\beta,R}^{\mathrm{ext}, \mathcal{G}_{\Phi_t(Q_R^c)}}(\Phi_t(Q_R^c), \Phi_t\#\muv, \zeta_V \circ \Phi_t^{-1})-\log \K_{N,\beta,R}^{\mathrm{ext},\mathcal{G}_{Q_R^c}}(Q_R^c, \muv, \zeta_V)\\
+N \left(\Ent_{\Phi_t(Q_R^c)}(\Phi_t\#\muv)- \Ent_{Q_R^c}(\muv)\right) \\
=\log \Esp_{\Q_{N,\beta,R}^{\mathrm{ext}, \mathcal{G}_{Q_R}}\(Q_R^c, \muv,\zeta_V\)} \exp\biggl(-\beta N^{-\frac{\s}{\d}}\(\F_N\(\Phi_t(\XN), \Phi_t\# ({\muv}_{|Q_R^c}), \(\Phi_t(Q_R)\times [-R,R]\)^c\)\) \\
-\F_N\(\XN,{\muv}_{|Q_R^c}, (Q_R\times [-R,R])^c\)+\Fluct_{\muv}(\log \det D\phi_t)\biggr)\end{multline*}
The first term in the exponent is equal by \eqref{dtF} to
\begin{equation*}
\int_0^t \Ani_1(\Phi_s(\XN), \Phi_s\#(\mu_{|Q_R^c}), \psi \circ \Phi_s^{-1})~ds
\end{equation*}
which we can control, using the local laws and modifying the proof of Corollary \ref{lemsubdi}, by 
\begin{multline*}
\int_0^t \Ani_1(\Phi_s(\XN), \Phi_s\#(\mu_{|Q_R^c}), \psi \circ \Phi_s^{-1})~ds\lesssim_\beta |t|N\|\varphi_0\|_{C^4}\ell^\s\sum_{k\ge \log_2 \frac{R}\ell } \frac{1}{(2^{\d-\s})^k}\\
\lesssim |t|N\|\varphi_0\|_{C^4}\ell^\s\(\frac{\ell}R\)^{\d-s}\lesssim_\beta   |t|N\|\varphi_0\|_{C^4}\ell^\d R^{\s-\d}
\end{multline*} The second term is similarly well-controlled;  \eqref{Linf} and \eqref{tscale} yield also
\begin{equation*}
\left|\Fluct_{\muv}(\log \det D\phi_t)\right|\lesssim_\beta |t| \|\varphi_0\|_{C^4}\sum_{k \geq \log_2\frac{R}{\ell}} \frac{\ell^\d N\(2^k\ell\)^\d}{\(2^k \ell\)^{2\d-\s}}\lesssim_\beta  |t|N\|\varphi_0\|_{C^4}\ell^\d R^{\s-\d}.
\end{equation*}
We thus have the estimate
\begin{equation}\label{733}
\log \K_{N,\beta,R}^{\mathrm{ext}, \mathcal{G}_{\Phi_t(Q_R^c)}}\(\Phi_t(Q_R^c),\mu_t,\zeta_V \circ \Phi_t^{-1}\)-\log \K_{N,\beta,R}^{\mathrm{ext},\mathcal{G}_{Q_R^c}}\(Q_r^c,\mu_0, \zeta_V\)\lesssim_\beta |t|N\|\varphi_0\|_{C^4}\ell^\d R^{\s-\d}.
\end{equation}
%This estimate is not very good when $\ell$ is close to the macroscale $\ell=1$. Instead we may use crude a priori bounds provided by Proposition~\ref{pro718}: \eqref{bornesfiU} (after rescaling) provides a lower bound, while the upper bound is immediate via  \eqref{lblocale} (after rescaling), and obtain instead
%\be\label{734}
%\log \K_{N,\beta,R}^{\mathrm{ext}, \mathcal{G}_{\Phi_t(Q_R^c)}}\(\Phi_t(Q_R^c),\mu_t,\zeta_V \circ \Phi_t^{-1}\)-\log \K_{N,\beta,R}^{\mathrm{ext},\mathcal{G}_{Q_R^c}}\(Q_r^c,\mu_0, \zeta_V\)\lesssim (1+\beta) N \mu_V(Q_R^c).
%\ee
\textbf{Step 4: Conclusion.} We apply  
\eqref{eq: innerexp} to $\mu_t$ and to $\muv$,
 use that $\|\mu_t\|_{C^1}\lesssim \ell^{-1}$,  insert into \eqref{eq: firstexp} and  add the bounds of the previous step  to obtain
\begin{multline*}
\log \K_{N,\beta,\infty}^{\mathcal{G}_\ell}\(\R^\d, \mu_t,\zeta_V\circ \Phi_t^{-1}\)- \log \K_{N,\beta,\infty}\(\R^\d,\mu_0,\zeta_V\) +N\(\Ent(\mu_t)-\Ent(\muv)\)=\\
 -N\beta \int_{Q_R} \mu_t^{1+\frac\s\d} \mf(\beta \mu_t^{\frac\s\d}) + \frac{N\beta}{2\d}\indic_{\s=0}\int_{Q_R}\mu_t \log \mu_t 
 +N\beta \int_{Q_R} \muv^{1+\frac\s\d} \mf(\beta \muv^{\frac\s\d}) - \frac{N\beta}{2\d}\indic_{\s=0}\int_{Q_R}\muv \log \muv\\
 +O_\beta\(|t|N\|\varphi_0\|_{C^4}\ell^\d R^{\s-\d}\) +O_\beta\( (1+\beta)NR^\d
\mathscr{E}\((N^{\frac{1}{\d}} \ell)^{\frac{1}{\kappa +1}}\)\).
\end{multline*}
We can insert the remaining integral over $Q_R^c$ into the above as well as in the computation in \eqref{borneB}. In particular, one has
\begin{multline*}
\bigg|-N\beta \int_{Q_R^c} \mu_t^{1+\frac\s\d} \mf(\beta \mu_t^{\frac\s\d}) + \frac{N\beta}{2\d}\indic_{\s=0}\int_{Q_R^c}\mu_t \log \mu_t \\
 +N\beta \int_{Q_R^c} \mu^{1+\frac\s\d} \mf(\beta \mu^{\frac\s\d}) - \frac{N\beta}{2\d}\indic_{\s=0}\int_{Q_R^c}\muv \log \muv\biggr|\lesssim_\beta N\beta \int_{Q_R^c}|t||D\psi| \lesssim_\beta |t| N \beta \ell^{\d}R^{\s-\d}
\end{multline*}
using \eqref{tscale}. Incorporating this into the error terms we find 
\begin{multline*}
\log \K_{N,\beta,\infty}^{\mathcal{G}_\ell}\(\R^\d, \mu_t,\zeta_V\circ \Phi_t^{-1}\)- \log \K_{N,\beta,\infty}\(\R^\d,\mu_0,\zeta_V\)  +N\(\Ent(\mu_t)-\Ent(\muv)\)\\
=\mathcal{Z}(\beta,\mu_t)-\mathcal{Z}(\beta,\mu_0) 
+O_\beta\(|t|N\|\varphi_0\|_{C^4}\ell^\d R^{\s-\d}\) +O_\beta\( (1+\beta)NR^\d
\mathscr{E}\( (N^{\frac{1}{\d}} \ell)^{\frac{1}{\kappa +1}}\)\).
\end{multline*}
We now have to optimize the sum of the errors over $R\le \ep$. Equating the two terms leads to the choice $$R= \min \( (|t|\ell^\d)^{\frac{1}{2\d-\s}} \( \mathscr{E}\( (N^{\frac{1}{\d}} \ell)^{\frac{1}{\kappa +1}}\)\)^{-\frac{1}{2\d-\s}},\ep\)$$ and an error rate   $\max(( |t|\ell^{\s-\d})^{\frac\d{2\d-\s}} \mathscr{E}^{\frac{\d-\s}{2\d-\s}}, |t|\ep^{\s-\d}).$

%The above optimization is valid so long as $R\lesssim 1 $, i.e.~$\ell \lesssim \(\mathscr E\)^{\frac{1}{2\d-\s}}$. Otherwise, 
%we take $R=\ell$ and do not split $Q_R$ into smaller cells, instead of \eqref{eq: firstexp} we obtain 
%\begin{multline}\label{eq: firstexp2}
%\log \K_{N,\beta,\infty}^{\mathcal{G}_\ell}\(\R^\d,\mu_t,\zeta_V\circ \Phi_t^{-1}\)- \log \K_{N,\beta,\infty}^{\mathcal{G}_\ell}\(\R^\d,\mu_0,\zeta_V\)=\\
% \( \log \K_{N,\beta,r}^{\mathcal{G}_{\Phi_t(Q_\ell)}}\(\Phi_t(Q_\ell), \mu_t, 0\)- \log \K_{N,\beta,r}^{\mathcal{G}_{Q_i}}\(Q_i, \mu_0, 0\)\) +
%NR^\d O\(\mathscr{E}( rN^{\frac{1}{\d}})\)\\
%+\log \K_{N,\beta,R}^{\mathrm{ext}, \mathcal{G}_{\Phi_t(Q_R^c)}}\(\Phi_t(Q_R^c),\mu_t,\zeta_V \circ \Phi_t^{-1}\)-\log \K_{N,\beta,R}^{\mathrm{ext},\mathcal{G}_{Q_R^c}}\(Q_r^c,\mu_0, \zeta_V\)
%\end{multline}

%----
%instead of using the bound from Step $3$, we apply Proposition \ref{pro718} \cm{revoir car c'etait cense etre pour des domaines bornes} which yields  
%\begin{multline*}
%\log \K_{N,\beta,\infty}^{\mathcal{G}_\ell}\(\R^\d, \mu_t,\zeta_V\circ \Phi_t^{-1}\)- \log \K_{N,\beta,\infty}\(\R^\d,\mu_0,\zeta_V\) +N\(\Ent(\mu_t)-\Ent(\muv)\)\\=
%\mathcal{Z} (\beta, \mu_t,Q_R)-\mathcal{Z}(\beta, \mu_0,Q_R)
%+NR^\d\mathrm{Error}\(\d,\s,N\)+\beta N \muv(Q_R^c)
%\end{multline*}
%and we need to optimize $NR^\d\mathrm{Error}\(N\)+\beta N \muv(Q_R^c) \sim NR^\d\Error+N\(\dist(Q_R, \partial \Sigma)\)^{\d+1-\alpha}$. The result follows by choosing $R \sim 1$ and $\dist(Q_R, \partial \Sigma)$ small enough so that $NR^\d\mathrm{Error}\(\d,\s,N\)$ is dominant.

\end{proof}

\section{Fractional harmonic extension}\label{sec: flap}
%\subsection{Fractional Harmonic Extension}
In this section we state and prove existence and regularity of the fractional harmonic extension, i.e.~the solution $\varphi^\Sigma$ to
\begin{equation}
\begin{cases}
\varphi^{\Sigma}=\varphi & \text{in }\Sigma \\
(-\Delta)^\alpha \varphi^\Sigma=0 & \text{in } \Sigma^c.
\end{cases}
\end{equation}
for $\Sigma \subset \R^\d$, as well as provide estimates for it.  We could not find a succinct statement of the existence result for solutions of fractional elliptic problems in unbounded domains in the literature, although it is certainly known. For completeness we offer its proof here. Interior and boundary regularity for solutions of equations like \eqref{def aharmext} are well-studied, and indeed the optimal regularity is understood to be $C^\alpha$, see \cite{RoS14}, \cite[Theorem 1.1]{RoW24} and references therein. The regularity of $\frac{\varphi^{\Sigma}-\varphi}{\dist(x,\Sigma)^\alpha}$ and its scaling properties then becomes an interesting question, which we rely crucially on here.

\begin{lem}\label{def aharmext}
Let $U$ be a neighborhood of $\Sigma$.
Let  $0<\alpha<1$ and let $\varphi\in \dot H^\alpha(\R^\d)$, where $\dot{H}^\alpha$ is the Sobolev space defined via \eqref{soboFT}. Suppose that $\partial \Sigma$ is a $C^{1,1}$ boundary. Then, there exists a unique function $\varphi^{\Sigma} \in \dot{H}^\alpha(\R^\d)$ solving \eqref{aharmext}. If $(-\Delta)^\alpha \varphi \in L^\infty(\R^\d)$ we also have the boundary estimate
\begin{equation}\label{aharm reg}
\left\|\frac{\varphi^{\Sigma}-\varphi}{\dist(x,\Sigma)^\alpha}\right\|_{C^{\sigma}(\overline{U \setminus \Sigma})} \lesssim \|(-\Delta)^\alpha \varphi\|_{L^\infty(U \setminus \Sigma)}
\end{equation}
for all $\sigma \in (0,\alpha)$. Furthermore, if 
 $\supp \varphi \subset \carr_\ell $, for some cube $\carr_\ell$ of size $\ell$ included  in $\bulk$, then
\begin{equation}\label{D2}
\left\|\frac{\varphi^{\Sigma}-\varphi}{\dist(x,\Sigma)^\alpha}\right\|_{C^{\sigma}(\overline{U \setminus \Sigma})} \lesssim \|\varphi\|_{L^\infty}\ell^\d
\end{equation} 
for all $\sigma \in (0,\alpha)$. 

If in addition we assume that $(-\Delta)^\alpha \varphi\in C^{k}(\R^\d)$ and $\partial \Sigma$ is $C^{k+1}$,  we also have 
\begin{equation}\label{aharm higher reg}
\left\|\frac{\varphi^{\Sigma}-\varphi}{\dist(x,\Sigma)^\alpha}\right\|_{C^{\sigma}(\overline{U \setminus \Sigma})} \lesssim \|(-\Delta)^\alpha \varphi\|_{C^{(\sigma-\alpha)}(U \setminus \Sigma)}
\end{equation}
for all $\sigma \in (\alpha,k)$ and $\sigma, \sigma\pm \alpha \notin \mathbb{N}$. \eqref{D2} holds as well in the mesoscopic interior case under the assumption that $\partial \Sigma$ is $C^{k+1}$;
%\cm{still this assumption?} 
no additional regularity on the test function $\varphi$ is needed.
%\begin{equation}\label{D3}
%\|f\|_{C^{\sigma}(U \setminus \Sigma)} \lesssim \|\varphi_0\|_{L^\infty}\ell^{d}.
%\end{equation}
\end{lem}
The regularity assumptions on $(-\Delta)^\alpha \varphi$ and $\partial \Sigma$ are so that we can apply the local regularity results \cite[Theorem 1.2]{RoS16} and \cite[Theorem 1.4]{ARo20}.
%We could not find a succinct statement of this result for solutions of fractional elliptic problems in unbounded domains in the literature, although it is certainly known. For completeness we offer its proof in Appendix \ref{sec: flap}. Interior and boundary regularity for solutions of equations like \eqref{def aharmext}) are well-studied (see ) and indeed the optimal regularity is understood to be $C^\alpha$, see \cite{RoS14}, \cite[Theorem 1.1]{RoW24} and references therein. The regularity of $\frac{\varphi^{\Sigma}-\varphi}{\dist(x,\Sigma)^\alpha}$ then becomes an interesting question, which we rely crucially on here.

\begin{proof}[Proof of Lemma \ref{def aharmext}]
First, let us prove existence of $\varphi^\Sigma$ in $\dot{H}^\alpha$. The proof requires us to solve a fractional Dirichlet problem on the unbounded domain $\Sigma^c$. Existence is standard, and follows for instance the variational technique used on bounded domains in \cite{R16} (and \cite{FKV15} for more general equations). For completeness and to clarify that it applies to unbounded domains in our situation, we sketch a proof here.

Let 
\begin{equation}\label{H^alpha energy}
\Lambda(u):=\frac{1}{2}\int_{\R^{2\d}}(u(x)-u(z))^2\frac{c_{\d,\alpha}}{|z-x|^{\d+2\alpha}}~dxdz,
\end{equation}
where $c_{\d,\alpha}$ is the constant associated to the fractional Laplace kernel in (\ref{def fraclap}). Notice that $\Lambda(u)$ is an equivalent definition of $[u]_{\dot H^\alpha}^2$, where  $[u]_{\dot H^\alpha}^2$ is as in (\ref{soboFT}) (see \cite[Chapter 2]{DPV12} for a discussion of the kernel definition of fractional Sobolev spaces). Since $\d\geq 2>2\alpha$, $\dot{H}^\alpha$ continuously imbeds into $L^q$ for all $q \in \left[p, \frac{2\d}{\d-2\alpha}\right]$ \cite[Theorem 6.5]{DPV12}. $\Lambda(u)$ is in fact defined then for all $u \in \dot{H}^\alpha$ and is equivalent to $\|u\|_{\dot{H}^\alpha}^2$.

Let us minimize $\Lambda(u)$ over the set $\dot{H}^\alpha_\varphi=\{u \in \dot H^\alpha:u=\varphi \text{ on }\Sigma\}$; first, $\Lambda(u)$ is bounded below on $\dot H^\alpha_\varphi$, and there is at least one element of $\dot H^\alpha_\varphi$ for which $\Lambda(u)$ is finite (namely $\varphi$). Thus, $\inf_{\dot H^\alpha_\varphi}\Lambda(u)$ exists and is finite. Now, consider a minimizing sequence $u_n$ in $\dot H^\alpha_\varphi$, with $\Lambda(u_n)=\|u_n\|_{\dot H^\alpha}^2 \rightarrow \inf_{\dot H^\alpha_\varphi}\Lambda(u)$. $\dot H^\alpha$ is complete, so there is some $u \in \dot H^\alpha$ such that $u_n \rightarrow u$ and 
\begin{equation*}
\Lambda(u)=\|u\|_{\dot H^\alpha}^2=\lim_{n \rightarrow \infty}\|u_n\|_{\dot H^\alpha}^2=\lim_{n \rightarrow \infty}\Lambda(u_n)=\inf_{\dot H^\alpha_\varphi}\Lambda(u).
\end{equation*}
Furthermore, $u \in \dot H^\alpha_\varphi$: for any $n$ and $\varphi \in C_c^\infty(\Sigma)$ we have
\begin{equation*}
\left|\int(\varphi-u)\varphi\right|=\left|\int(u_n-u)\varphi \right|\leq \|u_n-u\|_{L^2}\|\varphi\|_{L^2} \leq \|u_n-u\|_{H^\alpha}\|\varphi\|_{L^2} \rightarrow 0
\end{equation*}
and so $u=\varphi$ a.e.~on $\Sigma$. It remains to see that $u$ solves \eqref{aharmext}. Let $\varphi \in C^\infty_c(\Sigma^c)$ be an arbitrary test function. Then, $u+\ep \varphi \in \dot H^\alpha_\varphi$ and since $u$ is a minimizer of $\Lambda$ we have
\begin{align*}
0=\frac{d}{d\ep}\big\vert_{\ep=0}\Lambda(u+\ep \varphi)&=\frac{d}{d\ep}\big\vert_{\ep=0}\frac{1}{2}\int_{\R^{2\d}}(((u(x)-u(z))+\ep (\varphi(x)-\varphi(z)))^2\frac{c_{\d,\alpha}}{|z-x|^{\d+2\alpha}}~dxdz \\
&=\int_{\R^{2\d}}(u(x)-u(z))(\varphi(x)-\varphi(z))\frac{c_{\d,\alpha}}{|z-x|^{\d+2\alpha}}~dxdz.
\end{align*}
So,
\begin{equation*}
\int_{\R^{2\d}}(u(x)-u(z))\varphi(x)\frac{c_{\d,\alpha}}{|z-x|^{\d+2\alpha}}=\int_{\R^{2\d}}(u(x)-u(z))\varphi(z)\frac{c_{\d,\alpha}}{|z-x|^{\d+2\alpha}}
\end{equation*}
and thus by symmetry
\begin{align*}
\int_{\R^{2\d}}(u(x)-u(z))\varphi(z)\frac{c_{\d,\alpha}}{|z-x|^{\d+2\alpha}}&=\int_{\R^\d}\varphi(z)\left(\int_{\R^\d}(u(x)-u(z))\frac{c_{\d,\alpha}}{|z-x|^{\d+2\alpha}}~dx\right)dz\\
&=\int \varphi(z) (-\Delta)^\alpha u=0.
\end{align*}
Since $\varphi \in C_c^\infty(\Sigma^c)$ was arbitrary,  we conclude that $(-\Delta)^\alpha u=0$ a.e.~in $\Sigma^c$ as desired, and we set $\varphi^\Sigma =u$.

%Now, since $d>2\alpha$, the embedding referenced above (\cite[Theorem 6.5]{DPV12}) guarantees $u \in L^2$, and thus $u \rightarrow 0$ as $|x|\rightarrow +\infty$. Thus, we can cutoff $u$ and apply the results of \cite[Chapter 5]{R16} to obtain $u \in L^\infty (\Sigma^c)$ with $\|u\|_{L^\infty} \lesssim \|\varphi\|_{L^\infty}$. We are then able to apply the regularity results \cite[Theorem 1.1,1.2]{ROS17} to the function 

We now apply the regularity result of \cite[Theorem 1.2]{RoS16} to
\begin{equation*}
\tilde{u}=\varphi^\Sigma-\varphi
\end{equation*}
solving
\begin{equation*}
\begin{cases} 
\tilde{u}=0 & \text{in }\Sigma\\
(-\Delta)^\alpha \tilde{u}=-(-\Delta)^\alpha \varphi & \text{in }\Sigma^c
\end{cases}
\end{equation*}
since $\Sigma$ has a $C^{1,1}$ boundary and $(-\Delta)^\alpha \varphi \in L^\infty$ to conclude the boundary estimate (\ref{aharm reg}).

Now, let us examine the case where $\varphi=\varphi_0\(\frac{\cdot-z}{\ell}\)$. Let
\begin{equation}\label{def f}
\frac{\varphi^\Sigma-\varphi}{\dist(x,\partial \Sigma)^\alpha}:=f,
\end{equation}
which we have just seen is $C^\sigma$ regular, with $C^\sigma$ norm controlled by $\left\|(-\Delta)^\alpha\varphi\right\|_{L^\infty(U \setminus \Sigma)}$. The bound \eqref{D2} then follows from a careful analysis of \eqref{aharm reg}. The second term on the right hand side can be carefully analyzed from the definition of the fractional Laplacian;
\begin{multline*}
(-\Delta)^\alpha \varphi(x)=\text{P.V.}\int (\varphi(x)-\varphi(y))  \frac{\cds}{|x-y|^{\d+2\alpha}}~dy 
=-\int \varphi(y)\frac{\cds}{|x-y|^{\d+2\alpha}}~dy\lesssim \|\varphi\|_{L^\infty}\ell^\d.
\end{multline*}
The second equation follows from $x \notin \supp \varphi$, and the bounds in the integral come from the fact that $|x-y|$ is bounded from below at order $1$ and $\varphi$ is only supported on a set of size $\ell$. This is \eqref{D2}.

For higher regularity, we use \cite[Theorem 1.4]{ARo20}. There, if $\sigma \in (\alpha,k)$ (the $k$ is because $\partial \Sigma$ is $C^{k+1}$ regular by assumption) is such that $\sigma, \sigma\pm \alpha \notin \mathbb{N}$ then, with $f$ as in \eqref{def f}
\begin{equation*}
\|f\|_{C^{\sigma}(\overline{U \setminus \Sigma})} \lesssim \|(-\Delta)^\alpha \varphi\|_{C^{\sigma-\alpha}(\overline{U \setminus \Sigma})}+\|\varphi^\Sigma-\varphi\|_{L^\infty(\R^\d)}  \lesssim \|(-\Delta)^\alpha \varphi\|_{C^{\sigma-\alpha}(\overline{ \Sigma^c})}
\end{equation*}
where we have used \cite[Theorem 1.2]{RoS16} to control the $L^\infty$ term. We have $(-\Delta)^\alpha \varphi \in C^{\sigma-\alpha}(\Sigma^c)$ since we assume $\sigma<k$ and $\varphi \in C^{k+\alpha}(\R^\d)$; this yields \eqref{aharm higher reg}. 

Specializing now to the mesoscopic case,  it is slightly easier to differentiate 
\begin{equation*}
(-\Delta)^\alpha \varphi=\text{P.V.} \int \frac{\varphi(x)-\varphi(y)}{|x-y|^{\d+2\alpha}}~dy=\int_{\Sigma} \frac{-\varphi(y)}{|x-y|^{\d+2\alpha}}~dy
\end{equation*}
since $\varphi(x)=0$ for $x \in \Sigma^c$. A computation shows that for  $m\leq k$ we have, for $x \in \Sigma^c$, using $\dist(\supp \varphi, \partial \Sigma) \ge \ep>0$,
\begin{equation}\label{fraclap der}
\left|\nab^{\otimes m} (-\Delta)^\alpha \varphi(x)\right|\lesssim \int_{\Sigma} \frac{|\varphi(y)|}{|x-y|^{\d+2\alpha+m}}~dy \lesssim \|\varphi\|_{L^\infty} \ell^\d.
\end{equation}
H\"older interpolation then yields the result \eqref{D2}.
\end{proof}

We will also need more precise information about the decay of $\varphi^\Sigma$ at infinity; this is given by the following.

\begin{lem}\label{lem: decayxisigma}
Let $\varphi$ be as in the assumptions of Proposition \ref{tscale}. For any $m \le k$, if $x\in U^c$,  we have
\begin{equation}\label{decayxisigma}
 \left|\nab^{\otimes m} \varphi^\Sigma(x)\right|\lesssim \frac{\left\|(-\Delta)^\alpha \varphi\right\|_{L^\infty}}{|x-z|^{\s+m+2}}.
 \end{equation}

 Additionally, if   $\supp \varphi\subset \bulk$, then we have for any $m\le k$ that, for $x \in U^c$, 
\begin{equation}\label{decayxisigma-scale}
 \left|\nab^{\otimes m} \varphi^\Sigma(x)\right|\lesssim \frac{\ell^{\d+2}\|(-\Delta)^\alpha \varphi\|_{L^\infty}+\ell^\d \|\varphi\|_{L^\infty}}{|x|^{\s+m+2}} .
 \end{equation}
\end{lem}
\begin{proof} Without loss of generality, let us assume that $z=0$ (thus $0\in \Sigma$).
Let us start with the case $\ell=1$. Following the proof of Lemma \ref{lem1} with $v=\varphi^\Sigma-\varphi$, we may write  
 \begin{equation*}
 \dist(x,\partial \Sigma)^\alpha\((-\Delta)^\alpha \varphi^\Sigma(x)-(-\Delta)^\alpha \varphi(x)\)=\overline{c_\alpha}c_\circ+o\(\dist(x,\partial \Sigma)^\alpha\)
 \end{equation*}
in $\Sigma$, where 
 \begin{equation}\label{ccirc}
 |c_\circ|\leq \left\|\frac{\varphi^\Sigma-\varphi}{\dist(x,\partial \Sigma)^\alpha}\right\|_{L^\infty(U \setminus \Sigma)}\leq \|(-\Delta)^\alpha \varphi\|_{L^\infty(U \setminus \Sigma)};
 \end{equation}
 the second inequality follows from (\ref{aharm reg}). As a result, we may expand 
 \begin{equation}\label{RO expansion}
 (-\Delta)^\alpha \varphi^\Sigma(x)=(-\Delta)^\alpha \varphi(x)+\frac{\overline{c_\alpha}c_\circ}{\dist(x,\partial \Sigma)^\alpha}+o(1)
 \end{equation}
 in $\Sigma$. Now, the key is that, as discussed in \eqref{eqars2}, we can recover $\varphi^\Sigma=\g \ast (-\Delta)^\alpha \varphi^\Sigma$, i.e.
 \begin{equation*}
 \varphi^\Sigma(x)=\int \frac{(-\Delta)^\alpha \varphi^\Sigma(y)}{|x-y|^{\s}}~dy,
 \end{equation*}
 %\cm{from this form and integrating \eqref{ptwise fraclap} over $\Sigma$, we get that $|D^\gamma %\varphi^\Sigma|\lesssim \ell^\s$ in $U^c$}
 where $(-\Delta)^\alpha \varphi^\Sigma(y)$ is only supported in $\Sigma$. Taylor expanding the denominator to second order, we find for $x\in \Sigma^c$, 
 \begin{multline}\label{Texp}
 \varphi^\Sigma(x)=\int (-\Delta)^\alpha \varphi^\Sigma (y)\( \frac{1}{|x|^\s}-\s \sum_i \frac{x_iy_i}{|x|^{\s+2}}+O\(\frac{|y|^2}{|x|^{\s+2}}\)\)~dy \\
 =\frac{1}{|x|^\s}\int (-\Delta)^\alpha \varphi^\Sigma(y)~dy-\s\sum_i \frac{x_i}{|x|^{\s+2}} \int y_i(-\Delta)^\alpha \varphi^\Sigma(y)~dy+O\(\frac{1}{|x|^{\s+2}}\int_{\Sigma} |(-\Delta)^\alpha \varphi^\Sigma(y)||y|^2~dy\).
 \end{multline}
Since the fractional Laplacian is a mean zero operator, the first term cancels. Furthermore, for all $i$ we have 
\begin{equation*}
\int y_i(-\Delta)^\alpha \varphi^\Sigma(y)~dy=\int (-\Delta)^\alpha y_i \varphi^\Sigma=0
\end{equation*}
via fractional integration by parts, where we have used the odd symmetry of $y_i$ to conclude that $(-\Delta)^\alpha y_i=0$. Integrating  \eqref{RO expansion} over $\Sigma$ and using \eqref{ccirc} to control $c_\circ$, it follows that
 \begin{equation*}
 \left| \varphi^\Sigma(x)\right|\lesssim\frac{1}{|x|^{\s+2}}\int_{\Sigma} |y|^2|(-\Delta)^\alpha \varphi^\Sigma(y)|~dy \lesssim \frac{\left\|(-\Delta)^\alpha \varphi\right\|_{L^\infty}}{|x|^{\s+2}}.
  \end{equation*}
  %\cm{no, this is not what integrating \eqref{ptwise fraclap} gives. Alternatively I don't see how $(-\Delta)^\alpha \varphi^\Sigma$ is controlled by $(-\Delta)^\alpha \varphi$}
 The argument for the derivative is exactly analogous; passing the derivative through the integral, we find for any multiindex $\gamma$ of length $|\gamma|=m$,
 \begin{equation*}
D^\gamma \varphi^\Sigma(x)=\int (-\Delta)^\alpha \varphi^\Sigma(y) D^\gamma\(\frac{1}{|x-y|^{\s}}\)~dy.
 \end{equation*}
 Taylor expanding $ D^\gamma\(\frac{1}{|x-y|^{\s}}\)$ about $x$ and using that 
 \begin{equation*}
 \left|\partial_i\partial_j D^\gamma\(\frac{1}{|x|^{\s}}\)\right|\lesssim \frac{1}{|x|^{\s+m+2}}
 \end{equation*}
 yields the result via the same argument as above.

Next, consider the mesoscopic case $\ell<1$. We may write the same expansion \eqref{RO expansion}, although we now have from (\ref{D2}) that 
\begin{equation}\label{ccirc-scale}
 |c_\circ|\leq \left\|\frac{\varphi^\Sigma-\varphi}{\dist(x,\partial \Sigma)^\alpha}\right\|_{L^\infty(U \setminus \Sigma)}\lesssim \ell^\d\|\varphi\|_{L^\infty}.
 \end{equation}
 We expand as in (\ref{Texp}); the integral of the $\frac{\overline{c_\alpha}c_\circ}{\dist(x,\partial \Sigma)^\alpha}$ term in \eqref{RO expansion} now yields 
 \begin{equation*}
 O\(\frac{1}{|x|^{\s+2}}\int_\Sigma \left|\frac{\overline{c_\alpha}c_\circ}{\dist(x,\partial \Sigma)^\alpha}\right||y|^2~dy\)=O\(\frac{\ell^\d\|\varphi\|_{L^\infty}}{|x|^{\s+2}}\)
 \end{equation*}
 using (\ref{ccirc-scale}). The contribution from $(-\Delta)^\alpha \varphi$ in \eqref{RO expansion} now yields 
 \begin{equation*}
 \frac{1}{|x|^{\s+2}}\int_\Sigma|y|^2|(-\Delta)^\alpha \varphi(y)|~dy
 \end{equation*}
 which we split into two contributions: those in $\carr_{2\ell}$ and those away. In $\carr_{2\ell}$, we  find 
 \begin{equation*}
  \frac{1}{|x|^{\s+2}}\int_{\carr_{2\ell}}|y|^2|(-\Delta)^\alpha \varphi(y)|~dy \lesssim \frac{\ell^{2+\d}\|(-\Delta)^\alpha \varphi\|_{L^\infty}}{|x|^{\s+2}} .
 \end{equation*}
 Away from $\carr_{2\ell}$ we seek to estimate the decay of $(-\Delta)^\alpha \varphi$. Since we are at distance $\ge \ell $ from $\supp \varphi$, we can write 
 \begin{equation*}
\left| (-\Delta)^\alpha \varphi(x)\right|=\left|\int \frac{-\varphi(y)}{|x-y|^{\d+2\alpha}}~dy\right| \lesssim \frac{\ell^\d\|\varphi\|_{L^\infty}}{|x-z|^{\d+2\alpha}}
 \end{equation*}
 since $|x-y|\gtrsim |x-z|$ due to $\dist(x,\supp \varphi)\gtrsim \ell$. Thus, 
  \begin{equation*}
  \frac{1}{|x|^{\s+2}}\int_{\Sigma \setminus \carr_{2\ell}}|y|^2|(-\Delta)^\alpha \varphi(y)|~dy \lesssim \frac{\ell^\d\|\varphi\|_{L^\infty}}{|x|^{\s+2}}\int_\ell^1 \frac{r^{\d-1}}{r^{\d+2\alpha-2}}~dr \lesssim \frac{\ell^{\s+2}\|\varphi\|_{L^\infty}}{|x|^{\s+2}}.
 \end{equation*}
 Combining these estimates, we have
 \begin{equation*}
 \frac{1}{|x|^{\s+2}}\int_\Sigma|y|^2|(-\Delta)^\alpha \varphi(y)|~dy \lesssim \frac{\ell^{\d+2}\|(-\Delta)^\alpha \varphi\|_{L^\infty}}{|x|^{\s+2}}+\frac{\ell^{\s+2}\|\varphi\|_{L^\infty}}{|x|^{\s+2}}.
  \end{equation*} Combining with the $c_\circ $ term contribution, we get the announced estimate for $m=0$.
  A similar argument as in the $\ell=1$ case recovers the analogous estimates for the behavior of $\nab^{\otimes m} \varphi^\Sigma$ as $|x|\rightarrow +\infty$.
  Since $\s+2>\d$ and $\ell<1$, we conclude with the result.

%We proceed by scaling (\ref{decayxisigma}); as remarked in the proof of Theorem \ref{CLT}, if $\varphi(\cdot)=\varphi_0\(\frac{\cdot-z}{\ell}\)$, then 
% \begin{equation*}
% \varphi^\Sigma(x)=\varphi_0^{\frac{\Sigma-z}{\ell}}\(\frac{x}{\ell}\).
% \end{equation*}
% It follows from (\ref{decayxisigma}) that, as $|x|\rightarrow +\infty$,
% \begin{equation*}
% \left|\varphi^\Sigma(x)\right|\lesssim \frac{\left\|(-\Delta)^\alpha \varphi_0\right\|_{L^\infty}}{\ell^{2\alpha}}\frac{1}{\ell^k}\frac{\ell^{\s+k+2}}{|x|^{\s+k+2}} \sim \frac{\left\|(-\Delta)^\alpha \varphi_0\right\|_{L^\infty}\ell^{2\s+2-\d}}{|x|^{\s+k+2}}
%  \end{equation*}
%  using $\s+2-2\alpha=2\s+2-\d.$
 \end{proof}

We next state a result that controls the fractional Laplacian in terms of the derivatives of a function.
\begin{lem}\label{lemestxi}
Assume $\varphi$ is supported in some cube $\carr_\ell$ of sidelength $\ell$ and 
\eqref{estxi} holds for $k \ge 2$. Then 
\be \|(-\Delta)^{\alpha}\varphi\|_{L^\infty}\lesssim \M \ell^{-2\alpha}.\ee
\end{lem} \begin{proof}
Using the integral definition of the fractional Laplacian \eqref{def fraclap}, 
we need to compute 
\begin{equation*}
\lim_{\eta \downarrow 0}\int_{B_\eta(0)^c} \(\varphi(x)-\varphi(x+y)\)\frac{c_{\d,\alpha}}{|y|^{\d+2\alpha}}~dy. 
\end{equation*}
First consider $x\in \carr_{2\ell}$.
For $y \in \carr_{2\ell}$, Taylor expanding,  we find that the integrand is given by
\begin{multline*}
\left|\int_{B_{4\ell}(0)\setminus B_\eta(0)}\left(-\nabla \varphi \cdot y+|y|^2O\(D^2\varphi\)\right)\frac{c_{\d,\alpha}}{|y|^{\d+2\alpha}}~dy\right| \\
=\left|\int_{B_{4\ell}(0)\setminus B_\eta(0)}\left(|y|^2O\(D^2\varphi\)\right)\frac{c_{\d,\alpha}}{|y|^{\d+2\alpha}}~dy\right| \lesssim  \|\varphi\|_{C^2} \int_0^\ell r^{1-2\alpha}~dr \lesssim \|\varphi\|_{C^2}\ell^{2-2\alpha}
\end{multline*}
by spherical symmetry of the order $y$ term. For $|y|>4\ell$ we estimate roughly, 
\begin{equation*}
\left|\int_{B_{4\ell}(0)^c}\(\varphi(x)-\varphi(x+y)\)\frac{c_{\d,\alpha}}{|y|^{\d+2\alpha}}~dy\right| \lesssim \|\varphi\|_{L^\infty}\int_{4\ell}^\infty r^{-1-2\alpha}~dr \lesssim \|\varphi\|_{L^\infty}\ell^{-2\alpha}.
\end{equation*}
Assembling these results, we find that  for $x \in \carr_{2\ell}$, 
\begin{align}\label{casintS0}
\left|(-\Delta)^\alpha \varphi(x)\right|\lesssim \ell^{2-2\alpha}\|\varphi\|_{C^2}+\ell^{\s}\left\| \varphi\right\|_{L^\infty}\lesssim \M \ell^{-2\alpha}.
\end{align}

We next turn to the case $x\in \R^\d\backslash \carr_{2\ell}$. Since $\varphi $ vanishes outside $\carr_{\ell}$, we don't need to use the principal value in the definition of the fractional Laplacian, and instead have
\begin{equation*}
(-\Delta)^\alpha \varphi(x)=-\int_{x+y\in \carr_{\ell}} \varphi(x+y)\frac{c_{\d,\alpha}}{|y|^{\d+2\alpha}}~dy.
\end{equation*} We note that for $x\notin \carr_{2\ell}$, $x+y \in \carr_\ell$ implies that $|y|\ge \ell$. The integral is thus bounded 
by $\|\varphi\|_{L^\infty} \ell^{-2\alpha}\le \M \ell^{-2\alpha}$. Combining with  \eqref{casintS0}, the conclusion follows.

\end{proof}

We can now state the scaling that we need for the behavior of the fractional Laplacian of the fractional harmonic extension in $\Sigma$.
\begin{lem}\label{Linfty} 
Let $\varphi$ be as in the assumptions of Lemma \ref{lem: decayxisigma}. Let $m\leq k$. If $\supp \varphi \subset \carr_\ell \subset \bulk$ with $\ell \leq 1$, then 
\begin{equation}\label{ptwise fraclap}
\left|\nab^{\otimes m} \(\dist(x,\partial \Sigma)^\alpha(-\Delta)^\alpha \varphi^\Sigma(x)\)\right|\lesssim  \ell^{\d+2}\|(-\Delta)^\alpha \varphi\|_{L^\infty}+ \begin{cases}
\|\varphi\|_{C^{2+m}}\ell^{2+\s-\d} & \text{if }x \in \carr_{2\ell} \\
\|\varphi\|_{L^\infty} \frac{\ell^\d}{|x-z|^{2\d-\s+m}} & \text{if }x \in \Sigma \setminus \carr_{2\ell} 
%\|\varphi_0\|_{L^\infty} \ell^\d  & \text{if }x \in \Sigma \setminus \bulk,
\end{cases}
\end{equation}
%\cm{I'm not sure of the third case}
%In the case where $\ell=1$ and $\supp \varphi$ may overlap $\partial \Sigma$, \cm{remove} we instead have
%\begin{equation}
%\left|D^\gamma\((-\Delta)^\alpha \varphi^\Sigma(x)\)\right|\lesssim \|\varphi\|_{C^{2+m}} \qquad \text{if }x \in \bulk.
%\end{equation}
%For $x \notin \bulk$, we can write $\varphi^\Sigma=\varphi+v$, where 
%\begin{equation}
%\left|D^\gamma\(\dist(x,\partial \Sigma)^\alpha(-\Delta)^\alpha v(x)\)\right|\lesssim \|(-\Delta)^\alpha \varphi\|_{C^{k-\alpha}}.
%\end{equation}
\end{lem}
\begin{proof}

Let us first consider
% $\ell<1$ and choose 
$m=0$.   Note that by Lemma \ref{def aharmext}, we have 
\begin{equation}\label{reprlema1}
\varphi-\varphi^\Sigma = \dist(\cdot , \partial \Sigma)^\alpha f, \quad \|f\|_{C^{\sigma}(\overline{U\backslash \Sigma})}\lesssim \|\varphi\|_{L^\infty}\ell^\d.
\end{equation}

Next, we argue similarly as in the proof of Lemma \ref{lemestxi}. First, take $x \in \carr_{2\ell}$. 
  Using the integral definition of the fractional Laplacian \eqref{def fraclap}, 
we need to compute 
\begin{equation*}
\lim_{\eta \downarrow 0}\int_{B_\eta(0)^c} \(\varphi^\Sigma(x)-\varphi^\Sigma(x+y)\)\frac{\cds}{|y|^{\d+2\alpha}}~dy. 
\end{equation*}
For $|y| \leq 4\ell$,  $\varphi^\Sigma=\varphi$ in the domain of integration.  Taylor expanding,  we find that the integrand is given by
\begin{multline*}
\left|\int_{B_{4\ell}(0)\setminus B_\eta(0)}\left(-\nabla \varphi \cdot y+|y|^2O\(D^2\varphi\)\right)\frac{\cds}{|y|^{\d+2\alpha}}~dy\right| \\
=\left|\int_{B_{4\ell}(0)\setminus B_\eta(0)}\left(|y|^2O\(D^2\varphi\)\right)\frac{\cds}{|y|^{\d+2\alpha}}~dy\right| \lesssim  \|\varphi\|_{C^2} \int_0^\ell r^{1-2\alpha}~dr \lesssim \|\varphi\|_{C^2}\ell^{2-2\alpha}
\end{multline*}
by spherical symmetry of the order $y$ term. For $|y|>4\ell$ we estimate roughly. First,
\begin{equation*}
\left|\int_{B_{4\ell}(0)^c}\(\varphi^\Sigma(x)-\varphi^\Sigma(x+y)\)\frac{\cds}{|y|^{\d+2\alpha}}~dy\right| \leq |\varphi(x)|\left|\int_{B_{4\ell}(0)^c}\frac{\cds}{|y|^{\d+2\alpha}}~dy\right|+\left|\int_{B_{4\ell}(0)^c}\frac{\cds\varphi^\Sigma(x+y)}{|y|^{\d+2\alpha}}~dy\right| 
\end{equation*}
using $\varphi^\Sigma(x)=\varphi(x)$ for $x \in \carr_{2\ell}$. The first term directly yields 
\begin{equation*}
|\varphi(x)|\left|\int_{B_{4\ell}(0)^c}\frac{\cds}{|y|^{\d+2\alpha}}~dy\right|\lesssim \|\varphi\|_{L^\infty}\int_\ell^\infty r^{-1-2\alpha}~dr \lesssim \|\varphi\|_{L^\infty}\ell^{-2\alpha}.
\end{equation*}
On the other hand, using \eqref{reprlema1}, for $x+y \in U$, we have
\begin{equation*}
\left|\varphi^\Sigma(x+y)\right|\lesssim \|\varphi\|_{L^\infty}\ell^\d
\end{equation*} since for $y \in B_{4\ell}(0)^c$, $\varphi(x+y)=0$. Using \eqref{decayxisigma-scale} for $x+y\notin U$, we obtain 
\begin{multline*}
\left|\int_{B_{4\ell}(0)^c}\varphi^\Sigma(x+y)\frac{\cds}{|y|^{\d+2\alpha}}~dy\right| \lesssim \|\varphi\|_{L^\infty}\ell^\d \int_\ell^\infty r^{-1-2\alpha}~dr\\
+\( \ell^\d\left\| \varphi\right\|_{L^\infty}+ \ell^{\d+2}\|(-\Delta)^\alpha \varphi\|_{L^\infty}\)  \int_1^\infty r^{-\s-3-2\alpha}dr\\
\lesssim \( \ell^{\s}\left\| \varphi\right\|_{L^\infty}+ \ell^{\d+2}\|(-\Delta)^\alpha \varphi\|_{L^\infty}\) 
\end{multline*}
where we used that $2\alpha=\d-\s$.
% \cm{we could I imagine write something better, but this suffices}. 
Assembling these results,  we find that  for $x \in \carr_{2\ell}$, 
\begin{align}\label{casintS}
\left|(-\Delta)^\alpha \varphi^\Sigma(x)\right|\lesssim \ell^{2+\s-\d}\|\varphi\|_{C^2}+\ell^{\s-\d}\left\| \varphi\right\|_{L^\infty}+ \ell^{\d+2}\|(-\Delta)^\alpha \varphi\|_{L^\infty}.
\end{align}

We next turn to the case $x\in \Sigma \backslash \carr_{2\ell}$. Since $\varphi^\Sigma (x)=0$, we don't need to use the principal value in the definition of the fractional Laplacian, and instead have
\begin{equation*}
(-\Delta)^\alpha \varphi^\Sigma(x)=-\int \varphi^\Sigma(x+y)\frac{\cds}{|y|^{\d+2\alpha}}~dy.
\end{equation*} We split this integral into three parts: $x+y \in \Sigma$, $x+y \in \Sigma^c\cap U$ and $x+y\in \R^\d\backslash U$. For $x+y \in \Sigma$, since $\varphi^\Sigma (x+y)= \varphi(x+y)$ vanishes unless $x+y \in Q_{\ell}$, we have
\begin{equation*}
-\int_{x+y \in \Sigma} \varphi^\Sigma(x+y)\frac{\cds}{|y|^{\d+2\alpha}}~dy \lesssim \|\varphi\|_{L^\infty}  \frac{\ell^\d}{|x-z|^{\d+2\alpha}}.
\end{equation*}
 For $x+y \in \Sigma^c\cap U$, notice that we have $|y|\ge \dist(x,\partial \Sigma)$ and $|y|\ge \dist(x+y, \partial \Sigma)$. Since $\varphi=0$ in $\Sigma^c$, using \eqref{reprlema1}, we have
\begin{multline*}
-\int_{x+y \in \Sigma^c\cap U} \varphi^\Sigma(x+y)\frac{\cds}{|y|^{\d+2\alpha}}~dy=-\int_{x+y \in \Sigma^c\cap U}f(x+y)\dist(x+y, \partial \Sigma)^\alpha \frac{\cds}{|y|^{\d+2\alpha}}~dy \\ \lesssim \ell^\d \|\varphi\|_{L^\infty}\int_{x+y \in \Sigma^c\cap U}\frac{|y|^\alpha}{|y|^{\d+2\alpha}}\, dy \lesssim \ell^\d \|\varphi\|_{L^\infty}\int_{\dist(x,\partial \Sigma)}^\infty r^{-1-\alpha}\lesssim \ell^\d \|\varphi\|_{L^\infty} \dist(x,\partial \Sigma)^{-\alpha}.
\end{multline*}
Finally,  if $x+y\in U^c$, this implies that $|y|$ is bounded below by some positive $\ep>0$. Using \eqref{decayxisigma-scale}, we then find
\begin{multline*}
-\int_{x+y \in U^c} \varphi^\Sigma(x+y)\frac{\cds}{|y|^{\d+2\alpha}}~dy\lesssim 
(\ell^{\d+2}\|(-\Delta)^\alpha \varphi\|_{L^\infty}+\ell^\d \|\varphi\|_{L^\infty} )
 \int_{x+y\in U^c} \frac{1}{|y|^{2\d-\s} }\frac{1}{|x+y|^{\s+2}}dy\\ \lesssim
 (\ell^{\d+2}\|(-\Delta)^\alpha \varphi\|_{L^\infty}+\ell^\d \|\varphi\|_{L^\infty} )
  \int_{\ep}^\infty r^{-\d-3} dr\lesssim (\ell^{\d+2}\|(-\Delta)^\alpha \varphi\|_{L^\infty}+\ell^\d \|\varphi\|_{L^\infty} ).
\end{multline*}
This yields
\begin{equation*}
\left|(-\Delta)^\alpha \varphi^\Sigma(x)\right|\lesssim\|\varphi\|_{L^\infty} \( \frac{\ell^\d}{|x-z|^{2\d-\s}}
+ 
\ell^\d \dist(x,\partial \Sigma)^{-\alpha} \)+(\ell^{\d+2}\|(-\Delta)^\alpha \varphi\|_{L^\infty}+\ell^\d \|\varphi\|_{L^\infty} )
\end{equation*} in this case.

In view of \eqref{casintS}, this concludes the proof of  \eqref{ptwise fraclap} for $m=0$.

Now, let us turn to derivatives, i.e.~$m\ge 1$. The same argument as above works for the estimate inside of $\carr_{2\ell}$. For the decay, we start with $x \in \hat \Sigma$; there, $\dist(x,\partial \Sigma)^\alpha$ and its derivatives are bounded so we need only consider $(-\Delta)^\alpha \varphi^\Sigma$. It is again easier to use the definition
\begin{equation*}
(-\Delta)^\alpha \varphi^\Sigma=\text{P.V.}  \int \frac{\varphi^\Sigma(x)-\varphi^\Sigma(y)}{|x-y|^{\d+2\alpha}}~dy=\int \frac{-\varphi^\Sigma(y)}{|x-y|^{\d+2\alpha}}\, dy
\end{equation*}
since $\varphi^\Sigma(x)=0$ for $x \in \bulk \setminus \carr_{2\ell}$. A computation shows that for any mutiindex $\gamma$ with $|\gamma|=m\leq k$ we have
\begin{equation*}
\left|D^\gamma (-\Delta)^\alpha \varphi^\Sigma\right|\lesssim \int \frac{|\varphi^\Sigma(y)|}{|x-y|^{\d+2\alpha+m}}\, dy 
% \sim \int_{O(\ell)}O(1)\lesssim \ell^\d
\end{equation*}
We again split this integral into three parts: $y \in \Sigma$, $y \in U \setminus \Sigma$, and $y \in U^c$. For $y \in \Sigma$, since $\varphi^\Sigma=\varphi$ there we have
\begin{equation*}
\int_{y \in \Sigma}\left| \varphi^\Sigma(y)\right|\frac{1}{|x-y|^{\d+2\alpha+m}}~dy \lesssim \|\varphi\|_{L^\infty}\int_{O(\ell)}\frac{O(1)}{|x-z|^{\d+2\alpha+m}}~dy \lesssim \frac{\ell^\d\|\varphi\|_{L^\infty}}{|x-z|^{\d+2\alpha+m}}
\end{equation*}
since the integral is only defined for $y \in \supp(\varphi)$, which forces $|x-y| \gtrsim |x-z|$. For $y \in U \setminus \Sigma$, we have using \eqref{reprlema1} that 
\begin{multline*}
\int_{y \in U \setminus \Sigma} \left|\varphi^\Sigma(y)\right|\frac{\cds}{|y|^{\d+2\alpha+m}}\,dy=\int_{y \in U \setminus \Sigma}\left|f(y)\right|\dist(y, \partial \Sigma)^\alpha \frac{1}{|y|^{\d+2\alpha+m}}\,dy \\ \lesssim \ell^\d \|\varphi\|_{L^\infty}\int_{y \in U \setminus \Sigma}\frac{|y|^\alpha}{|y|^{\d+2\alpha+m}}~dy \lesssim \ell^\d \|\varphi\|_{L^\infty}\int_{\dist(x,\partial \Sigma)}^\infty r^{-1-\alpha-m}\lesssim \ell^\d \|\varphi\|_{L^\infty} \dist(x,\partial \Sigma)^{-\alpha-m}
\end{multline*}
since $y$ needs to be order $\dist(x,\partial \Sigma)^\alpha$ for $x+y \in \Sigma^c$. Since we are inside of the bulk, $\dist(x,\partial \Sigma)\sim 1$ and thus the integral contribution $y \in \Sigma$ dominates. Finally,  if $y\in U^c$, this implies that $|y|$ is bounded below by some positive $\ep>0$. Using \eqref{decayxisigma-scale}, we then find
\begin{multline*}
-\int_{y \in U^c}| \varphi^\Sigma(y)|\frac{\cds}{|y|^{\d+2\alpha+m}}~dy\lesssim( \ell^{\d+2}\|(-\Delta)^\alpha \varphi\|_{L^\infty}+\ell^\d \|\varphi\|_{L^\infty})  \int_{x+y\in U^c} \frac{1}{|y|^{2\d-\s+m} }\frac{1}{|y|^{\s+2}}dy\\ \lesssim (\ell^{\d+2}\|(-\Delta)^\alpha \varphi\|_{L^\infty}+\ell^\d \|\varphi\|_{L^\infty})\int_{\ep}^\infty r^{-\d-3-m} dr\lesssim \ell^{\d+2}\|(-\Delta)^\alpha \varphi\|_{L^\infty}+\ell^\d \|\varphi\|_{L^\infty},
\end{multline*}
and again the contribution from $y \in \Sigma$ dominates, yielding the required decay bound.

Finally, the desired control in $\Sigma \setminus \bulk$ follows from \eqref{assumpw} and \eqref{D2}.

\end{proof}

Finally, we establish the formula for the limit variance in the mesoscopic case.
\begin{lem}\label{variancemeso}
Assume that $\supp \varphi \subset \bulk$ and  $\varphi= \varphi_0(\frac{\cdot -z}{\ell}) $ for some fixed function $\varphi_0$.
As $\ell \to 0$ we have  
\be\label{limitvariance}
\ell^{-\s}\|\varphi^\Sigma\|_{\dot H^{\frac{\d-\s}{2}}}^2\to  \|\varphi_0\|_{\dot H^{\frac{\d-\s}{2}}}^2.\ee
and 
\be\label{limitmean}
 \ell^{-\s}\int_\Sigma(-\Delta)^\alpha \varphi^\Sigma(\log \muv) \to 0.\ee
\end{lem}
\begin{proof}

Let us start with \eqref{limitvariance}; we will use the equivalent characterization of the homogeneous Sobolev norm from the proof of Lemma \ref{def aharmext}, namely
\begin{equation}\label{seminorm}
\|u\|_{\dot{H}^{\frac{\d-\s}{2}}}^2=\frac{1}{2}\int_{\R^{2\d}}(u(x)-u(y))^2\frac{c_{\d,\alpha}}{|y-x|^{2\d-\s}}~dxdz
\end{equation}
using $2\alpha=\d-\s$. Notice that for $\varphi=\varphi_0\(\frac{\cdot-z}{\ell}\)$,
\begin{equation*}
(-\Delta)^\alpha \varphi=\ell^{-2\alpha}\(\(-\Delta\)^\alpha \varphi_0\)\(\frac{\cdot -z}{\ell}\)
\end{equation*}
and hence $\varphi^\Sigma(\cdot)=\varphi_0^{\frac{\Sigma-z}{\ell}}\(\frac{\cdot-z}{\ell}\).$  We will separate 
\begin{equation*}
\ell^{-\s}\|\varphi^\Sigma\|_{\dot H^{\frac{\d-\s}{2}}}=I_1+I_2
\end{equation*}
where
\begin{align*}
I_1&:=\frac{c_{\d,\alpha}\ell^{-\s}}{2}\int_{\Sigma \times \Sigma}\frac{\(\varphi_0^{\frac{\Sigma-z}{\ell}}\(\frac{x-z}{\ell}\)-\varphi_0^{\frac{\Sigma-z}{\ell}}\(\frac{y-z}{\ell}\)\)^2}{|x-y|^{2\d-\s}}~dxdy ,\\
I_2&:=\frac{c_{\d,\alpha}\ell^{-\s}}{2}\int_{(\Sigma \times \Sigma)^c}\frac{\(\varphi^\Sigma(x)-\varphi^\Sigma(y)\)^2}{|x-y|^{2\d-\s}}~dxdy.
\end{align*} Notice that via a change of variables we may write
\begin{multline*}
I_1=\frac{c_{\d,\alpha}\ell^{-\s}}{2}\int_{\Sigma \times \Sigma}\frac{\(\varphi_0^{\frac{\Sigma-z}{\ell}}\(\frac{x-z}{\ell}\)-\varphi_0^{\frac{\Sigma-z}{\ell}}\(\frac{y-z}{\ell}\)\)^2}{|x-y|^{2\d-\s}}~dxdy\\
=\frac{c_{\d,\alpha}}{2}\int_{\frac{\Sigma-z}{\ell}\times \frac{\Sigma-z}{\ell}}\frac{\(\varphi_0^{\frac{\Sigma-z}{\ell}}\(u\)-\varphi_0^{\frac{\Sigma-z}{\ell}}\(w\)\)^2}{|u-w|^{2\d-\s}}~dudw=\frac{c_{\d,\alpha}}{2}\int_{\frac{\Sigma-z}{\ell}\times \frac{\Sigma-z}{\ell}}\frac{\(\varphi_0\(u\)-\varphi_0\(w\)\)^2}{|u-w|^{2\d-\s}}~dudw\end{multline*}
and so the difference $\ell^{-\s}\|\varphi^\Sigma\|_{\dot H^{\frac{\d-\s}{2}}}^2-  \|\varphi_0\|_{\dot H^{\frac{\d-\s}{2}}}^2$ is merely
\begin{equation}\label{intest}
\ell^{-\s}\|\varphi^\Sigma\|_{\dot H^{\frac{\d-\s}{2}}}^2-  \|\varphi_0\|_{\dot H^{\frac{\d-\s}{2}}}^2=I_2-\frac{c_{\d,\alpha}}{2}\int_{\(\frac{\Sigma-z}{\ell}\times \frac{\Sigma-z}{\ell}\)^c}\frac{\(\varphi_0\(u\)-\varphi_0\(w\)\)^2}{|u-w|^{2\d-\s}}~dudw.
\end{equation}
Notice that the same change of variables implies 
\begin{equation*}
\ell^{-\s}\|\varphi^\Sigma\|_{\dot H^{\frac{\d-\s}{2}}}^2=\left\|\varphi_0^{\frac{\Sigma-z}{\ell}}\right\|_{\dot{H}^{\frac{\d-\s}{2}}}^2.
\end{equation*}
Since $\varphi_0^{\frac{\Sigma-z}{\ell}}$ is defined as the function coinciding with $\varphi_0$ on $\frac{\Sigma-z}{\ell}$ and minimizing $\dot{H}^{\frac{\d-\s}{2}}$ norm, it follows that 
\begin{equation*}
\ell^{-\s}\|\varphi^\Sigma\|_{\dot H^{\frac{\d-\s}{2}}}^2-  \|\varphi_0\|_{\dot H^{\frac{\d-\s}{2}}}^2=\left\|\varphi_0^{\frac{\Sigma-z}{\ell}}\right\|_{\dot{H}^{\frac{\d-\s}{2}}}^2-  \|\varphi_0\|_{\dot H^{\frac{\d-\s}{2}}}^2\leq 0.
\end{equation*}
It is thus sufficient in \eqref{intest} to show that 
\begin{equation*}
\left|\frac{c_{\d,\alpha}}{2}\int_{\(\frac{\Sigma-z}{\ell}\times \frac{\Sigma-z}{\ell}\)^c}\frac{\(\varphi_0\(u\)-\varphi_0\(w\)\)^2}{|u-w|^{2\d-\s}}~dudw\right|\rightarrow 0
\end{equation*}
since $I_2 \geq 0$. Notice that since $2 \, \supp \varphi_0 \subset \frac{\Sigma-z}{\ell}$, we may rewrite it via symmetry as 
\begin{multline*}
\left|-c_{\d,\alpha}\int_{(\supp\varphi_0)\times \(\frac{\Sigma-z}{\ell}\)^c}\frac{\varphi_0(u)^2}{|u-w|^{2\d-\s}}~dudw\right|\lesssim \|\varphi_0\|_{L^\infty}^2 \int_{\supp \varphi_0} \int_{\dist\(\supp \varphi_0, \partial \(\frac{\Sigma-z}{\ell}\)\)}^\infty \frac{r^{\d-1}}{r^{2\d-\s}}~drdu \\
\lesssim \|\varphi_0\|_{L^\infty}^2 \dist\(\supp \varphi_0, \partial \(\frac{\Sigma-z}{\ell}\)\)^{\s-\d} \lesssim \|\varphi_0\|_{L^\infty}^2\ell^{\d-\s} \rightarrow 0
\end{multline*}
and we conclude the result.

Now for \eqref{limitmean} it is easier to use the equivalent characterization of the mean from the proof of \eqref{defmean}, namely
\begin{equation*}
\ell^{-\s}\int_\Sigma (-\Delta)^\alpha \varphi^\Sigma (\log \muv)=\cds \ell^{-\s} \int_{\R^\d}(\div \psi)\muv=\cds \ell^{-\s} \int_\Sigma (\div \psi)\muv
\end{equation*}
where $\psi$ is the transport map defined by \eqref{Rtransport}. Integrating by parts and using that $\psi \muv \equiv 0$ on $\partial \Sigma$,
\begin{equation*}
\cds \ell^{-\s}\int_\Sigma (\div \psi)\muv=-\cds \ell^{-\s}\int_\Sigma \psi \cdot \nabla \muv.
\end{equation*}
In the bulk, we use that $\nabla \muv$ is bounded below and \eqref{tscale} to bound
\begin{multline*}
\left|-\cds \ell^{-\s}\int_{\bulk} \psi \cdot \nabla \muv\right|\lesssim \ell^{-\s}\int_{\carr_{2\ell}} \frac{\|\varphi_0\|_{C^3}}{\ell^{\d-\s-1}}+\ell^{-\s}\int_{\bulk \setminus \carr_{2\ell}}\frac{\ell^\d\|\varphi_0\|_{C^3}}{|x-z|^{2\d-\s-1}}\\
\lesssim \ell \|\varphi_0\|_{C^3}+\ell^{\d-\s}\|\varphi_0\|_{C^3}\int_\ell^1\frac{r^{\d-1}}{r^{2\d-\s-1}}~dr \lesssim \ell \|\varphi_0\|_{C^3}\rightarrow 0.
\end{multline*}
Outside of $\bulk$, we use that $\nabla \muv$ decays like $\dist(x,\partial \Sigma)^{-\alpha}$ coupled with \eqref{tscale} for $|x-z|$ at order one to see that 
\begin{equation*}
\left|-\cds \ell^{-\s}\int_{\Sigma \setminus \bulk} \psi \cdot \nabla \muv\right|\lesssim \ell^{-\s}\int_{\Sigma \setminus \bulk}\|\varphi_0\|_{C^3}\ell^\d \dist(x,\partial \Sigma)^{-\alpha} \lesssim \|\varphi_0\|_{C^3}\ell^{\d-\s}\rightarrow 0.
\end{equation*}
since $\dist(x,\partial \Sigma)^{-\alpha}$ is an integrable singularity. Coupling the estimates in $\bulk$ and $\Sigma \setminus \bulk$ yields \eqref{limitmean}. 

%\cm{To estimate
%\begin{equation*} 
%\ell^{-\s}\int_\Sigma|(-\Delta)^\alpha \varphi^\Sigma|
%\end{equation*}
%we use Lemma \ref{Linfty}. In $\bulk$, $\dist(x,\partial \Sigma)^\alpha$ is smooth and bounded below, so we can use \eqref{ptwise fraclap} to bound
%\begin{equation*}
%\ell^{-\s}\int_{\bulk}|(-\Delta)^\alpha \varphi^\Sigma| \lesssim \ell^{-\s}\( \|\varphi_0\|_{C^2}\int_{Q_{2\ell}}\ell^{\s-\d}+\|\varphi_0\|_{L^\infty}\int_{\bulk\setminus Q_{2\ell}}\frac{\ell^\d}{|x-z|^{2\d-\s}}\) \lesssim \|\varphi_0\|_{C^2}
%\end{equation*}
%In $\Sigma \setminus \bulk$, we use that \eqref{ptwise fraclap} yields $|(-\Delta)^\alpha \varphi^\Sigma|\lesssim \frac{\|\varphi_0\|_{L^\infty}\ell^\d}{\dist(x,\partial \Sigma)^\alpha}$ and the integrability of $\dist(x,\partial \Sigma)^{-\alpha}$ to conclude.}

\end{proof}

\section{Proof of the Screening Result - Proposition \ref{Riesz screening result}}\label{sec: pfscn}

 We focus on the proof of outer screening; we will discuss the main idea of inner screening and what computational changes are necessary at the end of this section. From this point on, we denote $E_{\rr}=\nabla w_{\rr}$ to emphasize that we are considering the electric field. 
\setcounter{subsection}{0}
\subsection{The Setup}
The first part of the proof consists of using the energy bounds to find a good boundary outside of which to construct the screened configuration. We have two separate cases, depending on which screenability condition is satisfied.

If  \eqref{Riesz screenability1} is verified, then using a mean value argument as in \cite[Section 6.2, Step 1]{PS17} and \cite[Section C.1]{AS21}, we can find a $T \in [R-2\tilde{\ell}+2,R-\tilde{\ell}-2]$ such that 
\begin{align}\label{boundM}
    \int_{(Q_{T+2}\setminus Q_{T-2})\times[-h,h]}|y|^\gamma|E_{\rrh}|^2 &\leq \frac{S(X_n,w,h)}{\tilde{\ell}} :=M\\
  \label{boundMbord}  \int_{\partial Q_T \times [-h,h]}|y|^\gamma|E_{\rrh}|^2 &\lesssim M 
\end{align}%\cm{are we sure it's $\times [-h,h]$? }
where $Q_T \subset Q_R$ and $Q_T \in \mathcal{Q}_T$ (i.e. in particular $\mu(Q_T)$ is integer). We then take $\Gamma:=\partial Q_T$, which in one dimension is simply two points on the axis.

Otherwise, \eqref{Riesz screenability2} holds. Then, via a mean-value argument analogous to \cite[Appendix A]{S24}, we can find some $t \in [R-2\tilde{\ell}, R-\tilde{\ell}-\ell]$ such that 
\begin{equation*}
\int_{(Q_{t+\ell} \setminus Q_t) \times [-h,h]}\yg |E_{\rrh}|^2 \leq C\frac{S(X_n,w,h)\ell}{\tilde{\ell}}.
\end{equation*}
We then apply a mean-value argument in the strip $Q_{t+\ell}\setminus Q_t$ and find a piecewise affine boundary $\Gamma$ in $Q_{t+\ell}\setminus Q_t$ with faces parallel to those of $Q_R$ and sidelengths of order $\ell$ such that 
\begin{equation}\label{intGamma}
\int_{\Gamma \times[-h,h]}\yg |E_{\rrh}|^2 \leq C\frac{S(X_n,w,h)}{\tilde{\ell}}, \qquad \sup_x \int_{\(\Gamma \cap \square_{\ell}(x)\) \times [-h,h]}\yg |E_{\rrh}|^2 \leq CS'(X_n,w,h)
\end{equation}
and 
\begin{equation*}
\int_{\Gamma_1 \times [-h,h]}\yg |E_{\rrh}|^2 \leq C\frac{S(X_n,w,h)}{\tilde{\ell}}
\end{equation*}
where $\Gamma_1$ denotes the $1$-neighborhood of $\Gamma$. In both cases we let $M=\frac{S(X_n,w,h)}{\tilde{\ell}}$ and in the latter case $M_\ell=CS'(X_n,w,h)$.

$\Gamma$ encloses a set $\Old$, in which we will keep $X_n$ and the associated electric field $E$ unchanged. The modifications to the electric field will take place in the set $\New:=Q_L \setminus \Old$. 
 We also let $M_0^+$ and $M_0^-$ be averaged electric fluxes on the top and bottom of our hyperrectangular region:
\begin{align}
    M_0^+:=\frac{h^{-\gamma}}{|\New|}\int_{\Old \times \{h\}}|y|^\gamma E_{\rr}\cdot \vec{n}, \label{definition of M} \\
    M_0^-:=\frac{h^{-\gamma}}{|\New|}\int_{\Old \times \{-h\}}|y|^\gamma E_{\rr}\cdot \vec{n} ,\label{definition of M - neg}
\end{align}
and set 
\be M_0= h^\gamma M_0^++h^\gamma M_0^-.\ee

With this region in tow, we partition our space (as in \cite[Section 6.2, Step 2]{PS17}) into the following subregions, and solve elliptic problems in each:
\begin{enumerate}
    \item $D_0:= \Old \times [-h, h]$
    \item $D_\partial:=\New \times[-h, h]$
    \item $D_1:=(Q_R \times [-\max(R,H), \max(R,H)])\setminus (D_0 \cup D_\partial).$
    %\cm{if $\Lambda$ is a hyperrectangle of height $H$, we need to max at $H$}
\end{enumerate}
where $H$ is the height of $\Lambda \subset \R^{\d+1}$ in the case where we are proving local laws in $\Lambda$. We partition $Q_R\setminus \Old$ into regions $H_k$ with piecewise affine boundary and sidelengths at scale $\ell$, in $\left[\frac{\ell}{C}, \ell C\right]$. Let $\tilde{H}_k$ denote the cells $H_k \times [-h, h]$.

Observe that the delineation of our points into old and new sets might intersect some of the ``smeared'' points; these smeared regions will have to be modified appropriately. We let $I_\partial$ denote the set of charges that are smeared by the boundary $\Gamma$, i.e.
\begin{equation*}
    I_\partial:=\{i: B(x_i, \rr_i) \cap \Gamma \ne \emptyset \}.
\end{equation*}
Set
\begin{equation}\label{defnkapp}
    n_k:=\cds \int_{\tilde{H}_k} \sum_{i \in I_\partial}\delta_{x_i}^{(\rr_i)}
\end{equation}
to be the amount of smear in a region $\tilde{H}_k$. We let $\N$ denote the number of smeared charges and the number of charges we want wholly unchanged in $\Old$, i.e.
\begin{equation*}
    \N:=\# I_\partial + \# (\{i:x_i \in \Old\} \setminus I_\partial).
\end{equation*}
The goal will be to place $\mn-\N$ sampled points in $\New:=Q_R\setminus \Old$, where $\mn= \mu(Q_R)$. For each $k$, choose constants $m_k$ such that 
\begin{equation}\label{defmk}
    \cds m_k |H_k|=\int_{\partial D_0 \cap \partial \tilde{H}_k}|y|^\gamma E_{\rr}\cdot \vec{n}+M_0|H_k|-n_k.\end{equation}
If $m_k$ is small enough, namely $|m_k|\leq \frac{1}{2}m$ (where $m$ is a lower bound for $\mu$), then we can guarantee $\int_{H_k}\mu+m_k|H_K| \in \mathbb{N}$.

Define 
\begin{equation*}
    \tilde{\mu}:=\mu+\sum_{k}m_k\indic_{H_k}.
\end{equation*}
On the other hand, we have immediately from the divergence theorem and \eqref{outer screening w}
\begin{equation*}
    \frac{1}{\cds}\int_{\partial D_0}|y|^\gamma E_{\rr}\cdot \vec{n}=\int_\Old d\mu-\N+\frac{1}{\cds}\sum_{k}n_k.
\end{equation*}
Hence, by definition \eqref{defmk},
\begin{align*}
    \tilde{\mu}(\New)&=\mu(\New)+\sum_k m_k|H_k| \\
    &=\mn-\mu(\Old)+\frac{1}{\cds}\int_{\partial D_0 \setminus (\Old \times \{-h,h\})}|y|^\gamma E_{\rr}\cdot\vec{n}\\ &+\frac{M_0}{\cds}\sum_k|H_k|-\frac{1}{\cds}\sum_k n_k \\
    &=\mn-\mu(\Old)+\frac{1}{\cds}\int_{\partial D_0}|y|^\gamma E_{\rr}\cdot\vec{n}-\frac{1}{\cds}\sum_k n_k \\
    &=\mn-n_\Old.
\end{align*}
With all of these quantities defined, we are in a position to construct a new screened field outside of $D_0$.

\subsection{Defining the Electric Field}
We define the screened electric field in each of the different subregions.

First we have $E_1$, which completes the smeared charges, defined by
\begin{equation*}
    E_1:=\sum_k \indic_{\tilde{H}_k}\nabla h_{1,k},
\end{equation*}
where $h_{1,k}$ solves 
\begin{equation*}
    \begin{cases}
    -\div(|y|^\gamma \nabla h_{1,k})=\cds \sum_{i \in I_\partial}\delta_{x_i}^{(\rr_i)} & \text{in } \tilde{H}_k \\
   \frac{ \partial h_{1,k}}{\partial n} =0 & \text{on }\partial \tilde{H}_k \setminus \partial D_0 \\
   \frac{ \partial h_{1,k}}{\partial n}=-\frac{n_k}{\int_{F_k}|y|^\gamma} & \text{on } F_k,
    \end{cases}
\end{equation*}
where $F_k$ is the face of $\partial \tilde{H}_k$ touching $\partial D_0$, if it exists. This is solvable by \eqref{defnkapp}.

$E_2$ balances the top region, $D_1$ and is defined by 
\begin{equation*}
    E_2:=\sum_k \indic_{\tilde{H}_k}\nabla h_{2,k},
\end{equation*}
where $h_{2,k}$ solves
\begin{equation*}
    \begin{cases}
    -\div(|y|^\gamma h_{2,k})=\cds m_k\indic_{ H_k} \delta_{\R^\d} & \text{in } \tilde{H}_k \\
   \frac{ \partial h_{2,k}}{\partial n}=-M_0^+ & \text{on }H_k \times \{h \} \\
   \frac{ \partial h_{2,k}}{\partial n} =-M_0^- & \text{on }H_k \times \{-h\} \\
  \frac{  \partial h_{2,k}}{\partial n}=g_k & \text{on the rest of }\partial \tilde{H}_k,
    \end{cases}
\end{equation*}
where $g_k\equiv 0$ if $H_k$ doesn't touch $\Gamma$, and $g_k=-E_{\rr}\cdot \vec{n}+\frac{n_k}{\int_{F_k}|y|^\gamma}$ otherwise, with $\vec{n}$ throughout the outward normal from $D_0$. This is solvable by \eqref{defmk}.

$E_3$ gives us the sampled configuration $Z_{\mn-n_\Old}$ in $\New$ and is defined by 
\begin{equation*}
    E_3=\nabla h_3\indic_{D_\partial},
\end{equation*}
where $h_3$ solves the Neumann problem
\begin{equation*}
    \begin{cases}
    -\div(|y|^\gamma \nabla h_3)=\cds \left(\sum_{j=1}^{\mn-\N}\delta_{z_j}-\tilde{\mu}\delta_{\R^\d}\right) &\text{in }\New \times [-h,h] \\
       \frac{ \partial h_3}{\partial n}=0 & \text{on }\partial \(\New \times [-h,h]\).
    \end{cases}
\end{equation*}

Finally, $E_4$ gives us the screened electric field in $D_1$ and is defined by 
\begin{equation*}
    E_4:=\nabla h_4,
\end{equation*}
where $h_4$ solves
\begin{equation*}
    \begin{cases}
    -\div(|y|^\gamma \nabla h_4)=0 &\text{in }D_1 \\
    \frac{\partial h_4 }{\partial n} =-\phi & \text{on }\partial D_1,
    \end{cases}
\end{equation*}
where $$\phi:=\indic_{\partial D_1 \cap \partial D_0}E\cdot \vec{n}-\indic_{\partial D_1 \cap \partial D_\partial \cap \{y>0\}}M_0^+-\indic_{\partial D_1 \cap \partial D_\partial \cap \{y<0\}}M_0^-.$$

Now, set $\Escr_{\rrh}:=(E_1+E_2+E_3)\indic_{D_\partial}+E_4\indic_{D_1}+E_{\rrh}\indic_{D_0}$ and add back in the truncations, by setting
\begin{equation*}
    \Escr:=\Escr_{\rrh}+\sum_{i=1}^{\mn}\nabla \f_{\overline{\mathsf r}_i}(x-y_i),
\end{equation*}
where $y_i$ correspond to the points  (in $\R^\d$) of the new configuration $Y_{\mn}=(\{X_n\} \cap \Old) \cup Z_{\mn-\N}$, and $\overline{\mathsf r}$ are the (possibly changed) minimal distances, blown up versions of \eqref{defrrc} for the new configuration $Y_{\mn}$. Due to the Neumann condition, no divergence is created across boundaries when we set $\Escr$ to vanish outside of our region. By definition, we have
\begin{equation*}\left\{
\begin{array}{ll}
    -\text{div}(|y|^\gamma \Escr)=\cds \Big(\displaystyle\sum_{i \in Y_{\mn}}\delta_{y_i}-\mu\delta_{\R^\d} \Big) &\text{in } Q_R\times [-R,R] \\
    \Escr \cdot \vec{n}=0 &\text{on }\partial (Q_R \times[-R,R]).\end{array}\right.
\end{equation*}

\subsection{Estimating Constants}
Instead of estimating $M_0^+$ and $M_0^-$ using Cauchy-Schwarz immediately as in \cite{PS17}, we instead carry these constants through our calculations. This will allow us to be as precise as possible in our estimates of $M_0$. As we discussed above, the screening process requires that $|m_k|\le\frac{m}{2}$. It will be convenient in the proof of the local laws to have a bound $\left\|\frac{\mu-\tilde{\mu}}{\tilde{\mu}}\right\|_{L^\infty(\New)} \leq C<1$; in order to obtain this, we will actually need a bit more than $|m_k|\le \frac{m}{2}$. So, we seek $|m_k| \le \frac{m}{3}$. %\cm{this is not very clear, why do we even need a strict bound?}

First observe, using the bound \eqref{boundM} and the discrepancy estimates \eqref{discest} applied on balls $B_\alpha$ of radius $1$ (at blown-up scale) near $\Gamma$, we have
\be\label{prelim nk}
n_{k,\alpha} \lesssim  \|\mu\|_{L^\infty} + \( \int_{H_k \cap B_\alpha} \yg |E_{\rrh}|^2 \)^{\hal}
\ee
and summing over $\alpha$  (choosing the $O(\ell^{\d-1})$ balls to form a finite covering of the desired region), we find 
$$%\begin{cases}
 n_k \lesssim  \|\mu\|_{L^\infty} \ell^{\d-1} + \ell^{\frac{\d-1}{2}}  (\min (M, M_\ell))^\hal .
 %\\
%\sum_k n_k \lesssim \|\mu\|_{L^\infty} R^{\d-1} + \ell^{\frac{\d-1}{2}}M^\hal ,\end{cases}
$$
We also have
\begin{equation*}
n_k^2=\(\sum_\alpha n_{k,\alpha}\)^2 \lesssim \ell^{\d-1}\sum_\alpha n_{k,\alpha}^2 
\lesssim \|\mu\|_{L^\infty}\ell^{2\d-2}+\ell^{\d-1}\int_{H_k \cap [Q_{T+2}\setminus Q_{T-2}] \times [-h,h]}\yg |E_{\rrh}|^2
\end{equation*}
using \eqref{prelim nk}, which yields 
\be\label{final nk}
\sum n_k^2 \lesssim R^{\d-1}\ell^{\d-1}\|\mu\|_{L^\infty}^2+M\ell^{\d-1}
\ee
since the number of $H_k$ intersecting $\Gamma$ is bounded by $O\(\frac{R^{\d-1}}{\ell^{\d-1}}\)$. We will make use of this estimate below. In the same way $\#I_\partial \lesssim M+R^{\d-1}$.
Hence, by definition \eqref{defmk}, using the above bounds, Cauchy-Schwarz and \eqref{intGamma}, we have
\begin{align*}
    |m_k| 
   & \le C\( \frac{1}{|H_k|}\int_{\partial D_0 \cap \partial \tilde{H}_k}|y|^\gamma |E_{\rrh}\cdot \vec{n}|+M_0+\frac{n_k}{|H_k|}\) \\
    &\le C\( M_0+\ell^{-\d} \int_{\partial D_0 \cap \partial \tilde{H}_k}|y|^\gamma|E_{\rrh} |+\ell^{-\d}(\min (M,M_\ell))^\hal \ell^{\frac{\d-1}{2}}+\|\mu\|_{L^\infty}\ell^{\d-1})\) \\
    &\le C \( M_0+\ell^{-\d}\sqrt{M}\sqrt{|\partial D_0 \cap \partial \tilde{H}_k|h^\gamma}+(\min (M,M_\ell))^\hal \ell^{-\frac{\d+1}{2}}+\|\mu\|_{L^\infty}\ell^{-1})\) \\
    &\le C\( M_0+\ell^{-\d}\ell^{\frac{\d-1}{2}}h^{\frac{1+\gamma}{2}}(\min (M,M_\ell))^\hal \)+\frac{C}{\ell}\|\mu\|_{L^\infty},
    % \\
  %  &\le C \( M_0+\ell^{-\frac{\d+1}{2}}h^{\frac{1+\gamma}{2}}\sqrt{M} + 1 +\ell^{-\d}(M+\ell^{\d-1}) \)+\frac{\eta}{\ell}\|\mu\|_{L^\infty},
\end{align*}  after absorbing some terms, using $\eta\le 1$ and $h> 1$. To obtain that this is  less than  $\frac{m}3$, 
squaring, rearranging and using the bounds on $M$ we reduce to  the sufficient conditions that \begin{equation}\label{definition of little c}
  \frac{C}{\ell}\|\mu\|_{L^\infty}<\frac{m}{6}, \qquad M_0^2+\frac{h^{1+\gamma}\min (M,M_\ell)}{\ell^{\d+1}} \leq \mathsf{c}
\end{equation}
for some fixed constant $\mathsf{c}$ (depending on $m$), as long as $\ell$ is large enough and $h>1$. The first condition is satisfied trivially for large enough $\ell>1$. Inserting the definition of $M$ and $M_\ell$, this corresponds to the second part of the  screenability condition \eqref{Riesz screenability1}--\eqref{Riesz screenability2}. 
%\cm{the first condition has changed, no $\eta$ anymore but a condition $\eta<1$, check where it propagates}
Furthermore, using Cauchy-Schwarz and the  definition \eqref{Riesz upper energy inner screening},  we have 
\begin{equation}\label{Jensen bd}
(M_0^+)^2=\(\frac{h^{-\gamma}}{|\New|}\int_{\Old \times \{h\}}|y|^\gamma E \cdot \vec{n}\)^2\le \frac{|\Old|}{|\New|^2} h^{-\gamma} e(X_n,w,h)\lesssim \frac{R^\d}{R^{2(\d-1)} \tilde \ell^2} h^{-\gamma} e(X_n,w,h)
\end{equation}
and the same for $M_0^-$, 
hence 
\begin{equation*}
M_0^2 \le\frac{1}{R^{\d-2} \tilde \ell^2} h^{\gamma} e(X_n,w,h).\end{equation*}
Substituting this into the above yields the screenability condition \eqref{Riesz screenability1}, resp. \eqref{Riesz screenability2}. %\cm{previous estimate hence condition was wrong -- please check this}

Notice that this condition yields a nice $L^\infty$ bound
\begin{equation*}
\left\|\frac{\mu-\tilde{\mu}}{\tilde{\mu}}\right\|_{L^\infty(H_k)}=\left\|\frac{m_k}{\mu+m_k}\right\|_{L^\infty(H_k)} \leq \frac{\frac{m}{3}}{m-\frac{m}{3}}=\frac{1}{2}.
\end{equation*}
Since $\tilde{\mu}$ is defined separately on the $H_k$, we then also have
\begin{equation*}
\left\|\frac{\mu-\tilde{\mu}}{\tilde{\mu}}\right\|_{L^\infty(\New)}\leq\frac{1}{2}.
\end{equation*}
%\cm{only on $\New_\eta$, not on $\New$}

%\cm{To get the $L^1$ bound on $\mu$ vs. $\tilde{\mu}$, notice that as in \cite[Section C.3, Step 2]{AS21}, with our choice of boundary,
%\begin{align*}
%\left|\int_{\New}\mu-\tilde{\mu}\right| \lesssim \sum_k |m_k||H_k| &\lesssim \tilde{l}^\d M_0+\int_{\Gamma \times [-h_2,h_1]}|y|^\gamma|E_{\rr}|+\sum_k n_k^2\\
% &\lesssim\tilde{l}L^{\d-1}M_0+ L^{\frac{\d-1}{2}}h^{\frac{1+\gamma}{2}}M^{1/2}+M\\
% &\lesssim \tilde{l}L^{\d-1}M_0+L^{\d-1}h^{1+\gamma}+M\\
%&\lesssim L^{\d-1}h^{1+\gamma}+\frac{S(X_n,w)}{\tilde{l}},
%\end{align*}
%since screenability implies a uniform bound on $M_0^2$ and hence a uniform upper bound on $M_0$. We've also used the estimate $\sum_k n_k^2 \lesssim M$ from above.
%
%
%
%
%Another option is to keep the $h$ on the $L$ and $M$, and get 
%\begin{equation*}
%\tilde{l}^\d M_0+h^{\frac{1+\gamma}{2}}L^{\d-1}+h^{\frac{1+\gamma}{2}}M
%\end{equation*}
%which we can control at order $\ll L^\d$.
%
%
%
%
%Or cut?}
%

%%Finally, to get the $L^2$ bound on $\mu$ vs $\tilde{\mu}$ we use Cauchy-Schwarz on the integral and obtain
%%\begin{align*}
%%\int_{\New}(\mu-\tilde{\mu})^2 \lesssim l\sum_k m_k^2 &\lesssim \frac{h}{l}\int_{\Gamma \times [-h_2,h_1]}|E_{\rr}|^2+\frac{1}{l}\sum_kn_k^2+lM_0^2 \\
%%&\lesssim \frac{hS(X_n,w)}{l\tilde{l}}+lM_0^2 \\
%%&\lesssim l,
%%\end{align*}
%%where the last line follows from the screenability condition. We now use these estimates and typical elliptic estimates to control the screened field.
%
%
\subsection{Estimating the screened field}
We first estimate $E_1$; the idea is to use an analog of \cite[Lemma A.2]{S24}. An examination of the proof shows that the $(\#I)^2\ell^{2-\d}$ term there is obtained from applying the Coulomb version of \cite[Lemma 6.4]{PS17} to 
\begin{equation*}
\begin{cases}
-\div(\yg \nabla v)=c\frac{|F_k|}{\tilde{H_k}}
   \frac{ \partial v}{\partial n} =0 & \text{on }\partial \tilde{H}_k \setminus \partial D_0 \\
   \frac{ \partial v}{\partial n}=c & \text{on } F_k,
\end{cases}
\end{equation*}
with $c=-\frac{n_k}{\int_{F_k}|y|^\gamma}$. We instead just apply \cite[Lemma 6.4]{PS17} to control this term; however, since $\tilde{H_k}$ has a height of length $h$, an examination of the proof induces an aspect ratio $\frac{h}{\ell}$. Hence,
\begin{equation}
    \int_{\tilde{H_k}}|y|^\gamma|\nabla h_{1,k, \rrh}|^2 \lesssim \frac{hn_k^2}{\int_{F_k}\yg}+\sum_{i \in H_k}\g(\rrh_i).
    \end{equation}
    %\cm{I put the $\ell^{2-\d}$ from the Coulomb case, but maybe it's $\ell^{-\s}$ here?}
    Using \eqref{eq:13} we can bound the sum over all $\g(\rr_i)$ in $D_\partial$ by $M+\#I_\partial \lesssim M+R^{\d-1}$ as seen above, and so 
    
 \begin{multline*}
    \int_{D_\partial}\yg |E_{1,\rrh}|^2 \lesssim h\sum_k \frac{n_k^2}{\int_{F_k}|y|^\gamma}+M+R^{\d-1}\lesssim \frac{1}{\ell^{\d-1}h^\gamma}\(\ell^{\d-1}R^{\d-1}+\ell^{\d-1}M\)+M+R^{\d-1} \\
    \lesssim \(1+h^{-\gamma}\)\(M+R^{\d-1}\)
\end{multline*}
using \eqref{final nk} and 
\begin{equation*}
\int_{F_k}|y|^\gamma \sim |\partial H_k|\int_{-h}^h |y|^\gamma \sim \ell^{\d-1}h^{1+\gamma}.
\end{equation*}

We next turn to $E_2$. 
%Notice that $g_k$ and $m_k$ are defined in such a way that 
%\begin{equation}\label{main relation}
%\int_{F_k}|y|^\gamma g_k+\cds m_k|H_k|-M_0|H_k|=0,
%\end{equation}
%\cm{is this correct? isn't some $n_k$ missing? and where is this relation used?}\cm{$n_k$ is in the definition of $g_k$}
%so that 
%\begin{equation*}
%\int_{\partial \tilde{H}_k}|y|^\gamma g_k=-M_0|H_k|+\int_{F_k}|y|^\gamma g_k=-\cds m_k|H_k|
%\end{equation*}\cm{check as well}
%Since $h$ is not necessarily at the same order as $\ell$, we cannot apply \cite[Lemma 6.4]{PS17} directly on each $\tilde{H}_k$. However, a close examination of the proof indicates we need to multiply the estimate there by the aspect ratio $\frac{h}{\ell}$, 
which bounding similarly with \cite[Lemma 6.4]{PS17} yields
\begin{align*}
\int_{D_\partial}|y|^\gamma|E_{2,\rrh}|^2&\lesssim \frac{h}{\ell} \ell \sum_k\left(\int_{F_k}|y|^\gamma |g_k|^2+(M_0^+)^2h^\gamma \ell^\d+(M_0^-)^2h^\gamma \ell^\d\right) \\
&\lesssim h\int_{\Gamma \times [-h,h]}|y|^\gamma |E_{\rrh}|^2+h\sum_k \frac{n_k^2}{\int_{F_k}|y|^\gamma}+h\left((M_0^+)^2h^\gamma \ell^\d+(M_0^-)^2h^\gamma \ell^\d\right)\frac{R^{\d-1}\tilde{\ell}}{\ell^\d} \\
&\lesssim hM+Mh^{-\gamma}+R^{\d-1}h^{-\gamma}+\tilde{\ell}R^{\d-1}h^{1-\gamma}M_0^2
\end{align*}
%\cm{something is wrong here because $M_0^2 = ((M_0^+)^2+(M_0^-)^2) h^{2\gamma}$, so $h^{-\gamma}$ on the last line? Also how did the second term involving $n_k$ get bounded?}
where we have borrowed from above the estimate on $h\sum_k \frac{n_k^2}{\int_{F_k}|y|^\gamma} $. Since $\gamma<1$ and $h>1$, we may rewrite this as 
 \begin{equation*}
 \int_{D_\partial}|y|^\gamma|E_{2,\rrh}|^2 \lesssim hM+R^{\d-1}h^{-\gamma}+\tilde{\ell}R^{\d-1}h^{1-\gamma}M_0^2.
 \end{equation*}
%%We can control this further by observing 
%%\begin{equation*}
%% \int_{F_k^+}|g_k|^2+\int_{F_k^-}|g_k|^2\leq \int_{F_k^+}|E_{\rr}|^2+\int_{F_k^-}|E_{\rr}|^2+n_k^2 \left(\frac{1}{|F_k^+|}+\frac{1}{|F_k^-|}\right) \lesssim \int_{\partial \square_T}|E_{\rr}|^2+M \lesssim M,
%%\end{equation*}
%%and so 
%%\begin{equation*}
%%\int_{D_\partial}|E_2|^2\lesssim lM+\tilde{l}lM_0^2.
%%\end{equation*}
%%

 For $E_3$, we obtain directly from the definition of $  \F(Z_{\mn-n_\Old}, \tilde{\mu}, \New\times [-h,h])$ (see \eqref{minneum}) and the fact that $\f_\eta$ is uniformly bounded in $L^1$ for small $\eta$ by \eqref{eq:intf}  that 
\begin{equation*}
 \frac{1}{2\cds}   \int_{\New \times [-h,h]}\yg|\nabla h_{3, \rrh}|^2 \leq  \F(Z_{\mn-n_\Old}, \tilde{\mu}, \New \times [-h,h])+\cds\sum_{j=1}^{\mn-n_\Old}\g(\rrh_j)+C(\mn-n_\Old).
\end{equation*}
%\cm{should it be on $\New_\eta$ or on $\New$?}\cm{edited definition of $E_3$ to match the Coulomb notes; $E_3$ should only be defined in $\New_\eta \times [-h,h]$ I think}

Finally, to bound the top field $E_4$, we use \cite[Lemma 6.4]{PS17} at scale $R$. This yields
\begin{align*}
    \int_{D_1}|E_4|^2 \lesssim R \int_{\partial D_1}|y|^\gamma|\phi|^2&\leq R|\New|h^{-\gamma} M_0^2+ R\int_{\square_{T+1}\times \{-h,h\}}|y|^\gamma|\nabla w|^2 \\
    &\leq R\tilde{\ell}R^{\d-1}h^{-\gamma}M_0^2+Re(X_n,w,h) \\
    &\lesssim R^{\d}\tilde{\ell}h^{-\gamma}M_0^2+Re(X_n,w,h)
\end{align*}
using the definition of $\phi$; the multiplication by $|\New|$ comes from integrating the constants $M_0^+$ and $M_0^-$ over where they are supported in the definition of $\phi$, namely $\New \times \{h\}$ and $\New \times \{-h\}$, respectively.
%\cm{again this seems wrong given your definition of $M_0$.  There should be an $h^{-\gamma}$. And where is the $|\New|$ coming from?}
It remains to put it all together and obtain the requisite screening estimate. We kept the original electric field fixed in $D_0$, so combining the above estimates allows us to write
\begin{align*}
\int_{Q_R \times [-R,R]}|\Escr_{\rrh}|^2 &\leq \int_{D_0}|y|^\gamma|\nabla w_{\rrh}|^2+Ch M+CR^{\d-1}+R^{\d-1}h^{-\gamma}+CR^{\d}\tilde{\ell}h^{-\gamma}M_0^2+CRe(X_n,w,h)\\
&+C(\mn-n_\Old)+C \(2\cds \F(Z_{\mn-n_\Old}, \tilde{\mu}, \New \times [-h,h])+\cds\sum_{i=1}^{\mn-n_\Old}\g(\rrh_j)\).
\end{align*}
where we have used $h \le R$ to absorb $\tilde{\ell}R^{\d-1}h^{1-\gamma}M_0^2$ from $E_2$ into the $\tilde{\ell}R^\d h^{-\gamma}M_0^2$ error above.
%\cm{terms from above estimates of $E_1$ and $E_2$ seem to be missing. Also, the fix to the above estimate needs to be propagated}
Using Lemma \ref{projlem} we can replace the screened electric field with the gradient defining $ \F(Y_{\mn}, \tilde{\mu}, Q_R \times [-R,R])$ , we find
with a uniform bound on $\f_\eta$ in $L^1$ for small $\eta$ that
\begin{multline*}
\F(Y_{\mn}, \tilde{\mu}, Q_R \times [-R,R])-\(\int_{Q_L\times[-R,R]}|y|^\gamma|\nabla w_{\rrh}|^2-\cd \sum_{i=1}^n \g(\rr_i)\) \\
\leq -\frac{1}{2\cds}\int_{\New \times [-R,R]}|y|^\gamma|\nabla w_{\rrh}|^2+\frac{1}{2}\sum_{\{i \in \{1,\dots,n\}:x_i \notin \Old\}}\g(\rrh_i) \\
+C\sum_{j=1}^{\mn-n_\Old}\g(\rrh_j)+ChM+CR^{\d-1}+R^{\d-1}h^{-\gamma}+C\tilde{\ell}R^{\d} h^{-\gamma}M_0^2+CR e(X_n,w,h)+\\
C \F(Z_{\mn-n_\Old}, \tilde{\mu}, \New \times [-h,h])
+C\sum_{i,j}\g(x_i-z_j)+C|n-\mn|+C(\mn-n_\Old).
\end{multline*}
%\cm{estimate to fix}
%\cm{revoir les justifs ici}
Using \eqref{Jensen bd} we can rewrite the above bound as  
\begin{multline*}
\F(Y_{\mn}, \tilde{\mu}, Q_R \times [-R,R])-\(\int_{Q_R\times[-R,R]}|y|^\gamma|\nabla w_{\rrh}|^2-\cd \sum_{i=1}^n \g(\rrh_i)\) \\
\leq -\frac{1}{2\cds}\int_{\New \times [-R,R]}|y|^\gamma|\nabla w_{\rrh}|^2+\frac{1}{2}\sum_{\{i \in \{1,\dots,n\}:x_i \notin \Old\}}\g(\rrh_i) 
+C\sum_{j=1}^{\mn-n_\Old}\g(\rrh_j)\\+ChM+CR^{\d-1}+R^{\d-1}h^{-\gamma}+C\(\frac{R^{2}}{\tilde{\ell}}+R\)e(X_n,w,h)+
C \F(Z_{\mn-n_\Old}, \tilde{\mu}, \New \times [-h,h])\\
+C\sum_{i,j}\g(x_i-z_j)+C|n-\mn|+C(\mn-n_\Old).
\end{multline*}
%\cm{again: to fix}

Next, we would like to control $\frac{1}{2}\sum_{\{i \in \{1,\dots,n\}:x_i \notin \Old\}}\g(\rrh_i) -\frac{1}{2\cds}\int_{\New \times [-R,R]}|y|^\gamma|\nabla w_{\rrh}|^2$ by the number of points not in $\Old$, but the possible blowup of $\g(\rrh_i)$ presents an issue. We adjust the truncation parameter and apply \cite[Lemma 4.13]{S24}, but need to shrink $\Old$ a tad in order to guarantee that it does not intersect $B(x_i, \frac{1}{4})$ for all $x_i \notin \Old$. To do this, we simply observe that $\square_{T-4}\subset \Old$ and write 
\begin{align*}
&\frac{1}{2}\sum_{\{i \in \{1,\dots,n\}:x_i \notin \Old\}}\g(\rrh_i) -\frac{1}{2\cds}\int_{\New \times [-R,R]}|y|^\gamma|\nabla w_{\rrh}|^2 =\frac{1}{2}\sum_{\{i \in \{1,\dots,n\}:x_i \notin \Old\}}\g(\rrh_i) \\
 &-\frac{1}{2\cds}\int_{(\square_R \setminus \square_{T-4})\times [-R,R]}|y|^\gamma|\nabla w_{\rrh}|^2+\frac{1}{2\cds}\int_{( \Old \setminus\square_{T-4})\times [-R,R]}|y|^\gamma|\nabla w_{\rrh}|^2 \\
&\leq \frac{1}{2}\sum_{\{i \in \{1,\dots,n\}:x_i \notin \square_{T-4}\}}\g(\tilde{\rr}_i) -\frac{1}{2\cds}\int_{(\square_R \setminus \square_{T-4})\times [-R,R]}|y|^\gamma|\nabla w_{\tilde{\rr}}|^2-\sum_{\{i: x_i \in \Old\setminus \square_{T-4}\}}\g(\rrh_i)\\
&+\frac{1}{2\cds}\int_{(\Old \setminus \square_{T-4} )\times [-R,R]}|y|^\gamma|\nabla w_{\rrh}|^2 +\sum_{\{i: x_i \notin \square_{T-4}\}} \int_{\square_R \setminus \square_{T-4}} (\f_{\tilde{\rr}_i}-\f_{\rrh_i})(x-x_i)~d\mu  \\
& \leq C(n-\N)+CM,
\end{align*}
where $\tilde{\rr}_i$ is defined to be $\frac{1}{4}$ for $x_i \notin \Old$ and is kept fixed otherwise. This allows us to cancel all contributions of $\f$ and $\g$ for $x_i \in \Old \setminus \square_{T-4}$, and bound the remaining contributions of $\f$ and $\g$ by $C(n-\N)$ since $\tilde{\rr}_i$ is bounded below for such $i$ and $\f_\eta$ is again controlled uniformly in $L^1$ for small $\eta$. We have bounded negative contributions of $L^2$ energy by zero, and $\frac{1}{2\cds}\int_{(\Old \setminus \square_{T-4})\times [-R,R]}|y|^\gamma|\nabla w_{\rrh}|^2 \leq CM$. 
Finally, using \eqref{arrivC02} we can bound 
\begin{align*}
\sum_{j=1}^{\mn-n_\Old}\g(\rrh_j)&\leq C\left( \F(Z_{\mn-n_\Old}, \tilde{\mu}, \New \times [-h,h])+C(|\mn-\N|)\right) \\
&\leq C\left(\F(Z_{\mn-n_\Old}, \tilde{\mu}, \New \times [-h,h])+\tilde{\ell}R^{\d-1}\right),
\end{align*}
controlling $\mn-\N=\tilde{\mu}(\New)\lesssim \tilde{\ell}R^{\d-1}$ by our $L^\infty$ control on $\tilde{\mu}$, completing the argument.
%\cm{added screening back in with $\F$ from old document, need to double check that parameters in the final estimates are the same}

%Replacing $\rr$ with $\overline{\mathsf{r}}$ using \cite[Lemma 4.13]{S24}, we can rewrite the above estimate as 
%\begin{align*}
%\F(Y_{\mn}, \tilde{\mu}, Q_L \times [-L,L])&\leq \frac{1}{2\cds}\int_{\Old \times [-L,L]}|y|^\gamma|\nabla w_{\rr}|^2 
%+C\sum_{i,j\in J}\g(x_i-z_j)+Ch M+C(\mn-n_\Old)\\
%&+Le(X_n,w,h) +C\tilde{\ell}L^{\d} h^{\gamma}M_0^2+C  \F(Z_{\mn-n_\Old}, \tilde{\mu}, \New_{\eta}\times[-h,h])
%\end{align*}
%and bounding the energy in $\Old \times[-L,L]$ by that in $Q_L \times [-L,L]$ yields the result.

\subsection{Outer Screening}
Let us first comment on the changes to the setup that are required for screening in $(Q_R \times [-R,R])^c$. First, we choose our good boundary $\Gamma$ exterior to $Q_R$. If \eqref{Riesz screenability1} is satisfied, we find $T \in [R+\tilde{\ell}+2, R+2\tilde{\ell}-2]$ such that
\begin{align}
\nonumber    M:=\int_{(Q_{T+4}\setminus Q_{T-4})\times[-h,h]^k}|y|^\gamma|E_{\rrh}|^2 &\leq \frac{S(X_n,w,h)}{\tilde{\ell}} \\
\label{newT}    \int_{\partial Q_T \times [-h,h]^k}|y|^\gamma|E_{\rrh}|^2 &\lesssim M 
\end{align}
and again take $\Gamma=\partial Q_T$.

Otherwise, the second condition  \eqref{Riesz screenability2} is satisfied. Then, via a mean-value argument analogous to \cite[Appendix A]{S24}, we can find some $t \in [ R+\tilde{\ell}+\ell, R+2\tilde{\ell}]$ such that 
\begin{equation*}
\int_{(Q_{t} \setminus Q_{t-\ell}) \times [-h,h]}\yg |E_{\rrh}|^2 \leq C\frac{S(X_n,w,h)\ell}{\tilde{\ell}}.
\end{equation*}

We then apply a mean-value argument in the strip $Q_{t} \setminus Q_{t-\ell}$ and find a piecewise affine boundary $\Gamma$ in $Q_{t} \setminus Q_{t-\ell}$ with faces parallel to those of $Q_L$ and sidelengths of order $\ell$ such that 
\begin{equation*}
\int_{\Gamma \times[-h,h]}\yg |E_{\rrh}|^2 \leq C\frac{S(X_n,w,h)}{\tilde{\ell}}, \qquad \sup_x \int_{\(\Gamma \cap \square_{\ell}(x)\) \times [-h,h]}\yg |E_{\rrh}|^2 \leq CS'(X_n,w,h)
\end{equation*}
and 
\begin{equation*}
\int_{\Gamma_1 \times [-h,h]}\yg |E_{\rrh}|^2 \leq C\frac{S(X_n,w,h)}{\tilde{\ell}}
\end{equation*}
where $\Gamma_1$ denotes the $1$-neighborhood of $\Gamma$. In both cases we let $M=C\frac{S(X_n,w,h)}{\tilde{\ell}}$ and in the latter case $M_\ell=CS'(X_n,w,h)$.

We will leave the configuration unchanged in $\Old=Q_T^c$, and only place new points in $\New=Q_R^c \setminus Q_T^c$. We now partition space so that we only change the field near $\partial (Q_R\times [-R,R])$. Namely, we define
\begin{enumerate}
    \item $D_0:= \left(Q_T \times \left[-(R+h), (R+h)\right]\right)^c$
    \item $D_\partial:=\New \times[-h, h]$
    \item $D_1:=\left(\square_T \times [-(R+h),(R+h)]\right)\setminus(\Omega \times [-R,R] \cup D_\partial).$
\end{enumerate}
We partition $\New$ into cells  $H_k$ with sidelengths at scale $\ell$, in $\left[\frac{\ell}{C}, \ell C\right]$, and let $\tilde{H}_k$ denote the rectangles $H_k \times [-h, h]$. Then, we set 
\begin{align*}
    M_0^+:=\frac{h^{-\gamma}}{|\New|}\int_{(\partial D_1 \cap \{y>0\}) \setminus \partial (D_\partial \cup (Q_R \times [-R,R]))}E\cdot \vec{n}  \\
    M_0^-:=\frac{h^{-\gamma}}{|\New|}\int_{(\partial D_1 \cap \{y<0\})  \setminus \partial (D_\partial \cup (Q_R \times [-R,R))}E\cdot \vec{n}
\end{align*}
and denote by $M_0$ the sum $h^\gamma M_0^++h^\gamma M_0^-$. The sets and quantities $\New$, $H_k$, $n_k$, $I_\partial$, $n_\Old$, $m_k$ and $\tilde{\mu}$ are then all defined analogously to the outer screening, as is the screenability condition. $E_1$, $E_2$, $E_3$ and $E_4$ are defined in exactly the same manner as in the outer screening. Setting $\Escr_{\rrh}:=(E_1+E_2+E_3)\mathbf{1}_{D_\partial}+E_4\mathbf{1}_{D_1}+E_{\rrh}\mathbf{1}_{D_0}$ and adding back the truncations we have
\begin{equation*}
    \Escr:=\Escr_{\rrh}+\sum_{i=1}^{\mn}\nabla \f_{\overline{\mathsf r}_i}(x-y_i),
\end{equation*}
where $Y_{\mn}=(\{X_n\} \cap \Old) \cup Z_{\mn-\N}$, and $\overline{\mathsf r}$ are the (possibly changed) minimal distances for the new configuration $Y_{\mn}$. Due to the Neumann condition, no divergence is created across boundaries when we set $\Escr$ to vanish outside of our region. By definition then, we have
\begin{equation*}
    \begin{cases}
    -\text{div}(|y|^\gamma \Escr)=\cds \left(\sum_{i \in Y_{\mn}}\delta_{y_i}-\mu \right) &\text{in } (Q_R \times [-R,R])^c \\
    \Escr \cdot \vec{n}=0 &\text{on }\partial (Q_R \times [-R,R]).
    \end{cases}
\end{equation*}
Since the geometry of $D_\partial$ is unchanged and the equations are the same as with outer screening, all of the estimates on $E_1$, $E_2$, $E_3$ are the same. The sidelengths of $D_1$ are not necessarily of the same order, so the estimate on $E_4$ needs to be multiplied by an aspect ratio of $\frac{R}{\min(h,\tilde{\ell})}$. 
%$D_1$ is slightly different, since it now has $\sim L \times \tilde{l}$ rectangular legs on each side; a close examination of the proof of \cite[Lemma 6.4]{PS17} shows that we need to multiply our estimate by the aspect ratio $\frac{L}{\tilde{l}}$. 
%
%This gives the new estimate
%\begin{align*}
%    \int\int_{D_1}|E_4|^2 \lesssim \frac{L^2}{\tilde{\ell}} \int_{\partial D_1}|y|^\gamma|\phi|^2&\leq \frac{L^2}{\tilde{\ell}}|\New|h^{\gamma}M_0^2+ \frac{L^2}{\tilde{\ell}}e(X_n,w)\\
%    &\leq L^{\d+1}h^{\gamma}M_0^2+\frac{L^2}{\tilde{\ell}}e(X_n,w). 
%    \end{align*}

%%The only change comes in estimating $E_4$, since we now have contributions from the vertical sides of the boundary. However, these are estimated by $LM \sim l \frac{S(X_n,w)}{\tilde{l}}$, so the error terms remain unchanged. 
Thus,
\begin{align*}
&\int_{(Q_R \times [-R,R])^c}|y|^\gamma|\Escr_{\rrh}|^2 \leq \int_{D_0}|y|^\gamma |\nabla w_{\rrh}|^2+Ch M+C\frac{R}{\min(h,\tilde{\ell})}\(R^{\d}\tilde{\ell}h^{-\gamma}M_0^2+Re(X_n,w,h)\)\\
&+CR^{\d-1}+R^{\d-1}h^{-\gamma}+C(\mn-n_\Old)+C\(2\cds  \F(Z_{\mn-n_\Old}, \tilde{\mu}, \New \times [-h,h])+\cds\sum_{i=1}^{\mn-n_{\Old}}\g(\rrh_i)\).
\end{align*}%\cm{check}

%Using Lemma \ref{projlem} we can replace the screened electric field with the gradient defining $\F$ as before, and then replace $\rr$ with $\overline{\mathsf{r}}$ using \cite[Lemma 4.13]{S24}, to find 
%\begin{align*}
%\F(Y_{\mn}, \tilde{\mu}, (Q_L \times &[-L,L])^c)\leq \frac{1}{2\cds}\int_{D_0}|y|^\gamma|\nabla w_{\rr}|^2 
%+C\sum_{i,j\in J}\g(x_i-z_j)+C\ell M+C(\mn-n_\Old)\\
%&+C\frac{L}{\min(h,\tilde{\ell})}\(L^{\d}\tilde{\ell}h^{\gamma}M_0^2+Le(X_n,w,h)\)+C  \F(Z_{\mn-n_\Old}, \tilde{\mu}, \New_{\eta}\times[-h,h])
%\end{align*}
%Bounding the energy in $D_0$ by that in $(Q_L \times [-L,L])^c$ yields the result.

We find exactly as before then that
\begin{multline*}
\F(Y_{\mn}, \tilde{\mu}, (Q_R \times [-R,R])^c)-\(\int_{(Q_R\times[-R,R])^c}|y|^\gamma|\nabla w_{\rrh}|^2-\cd \sum_{i=1}^n \g(\rrh_i)\) \\
\leq -\frac{1}{2\cds}\int_{D_\partial \cup D_1}|y|^\gamma|\nabla w_{\rrh}|^2+\frac{1}{2}\sum_{\{i \in \{1,\dots,n\}:x_i \notin \Old\}}\g(\rr_i)+C\sum_{j=1}^{\mn-n_\Old}\g(\rrh_j) +CR^{\d-1}+R^{\d-1}h^{-\gamma}\\
+Ch M +C\frac{R}{\min(h,\tilde{\ell})}\(\frac{R^{2}}{\tilde{\ell}}+R\)e(X_n,w,h)+C\F(Z_{\mn-n_\Old}, \tilde{\mu}, \New \times [-h,h])+C\sum_{i,j}\g(x_i-z_j) \\
+C|n-\mn|+C(\mn-n_\Old).
\end{multline*}%\cm{check}
Next, we would like to control $\frac{1}{2}\sum_{\{i \in \{1,\dots,n\}:x_i \notin \Old\}}\g(\rrh_i) -\frac{1}{2\cds}\int_{D_\partial \cup D_1}|y|^\gamma|\nabla w_{\rrh}|^2$ by the number of points not in $\Old$, but the possible blowup of $\g(\rr_i)$ again presents an issue. We adjust the truncation parameter and apply \cite[Lemma 4.13]{S24}, but again need to shrink $\Old$ a tad in order to guarantee that it does not intersect $B(x_i, \frac{1}{4})$ for all $x_i \notin \Old$. To do this, we simply observe that $\square_{T+4}^c\subset \Old$ and write 
\begin{align*}
&\frac{1}{2}\sum_{\{i \in \{1,\dots,n\}:x_i \notin \Old\}}\g(\rrh_i) -\frac{1}{2\cds}\int_{D_\partial \cup D_1}|y|^\gamma|\nabla w_{\rrh}|^2  \leq \frac{1}{2}\sum_{\{i \in \{1,\dots,n\}:x_i \notin \Old\}}\g(\rrh_i) -\frac{1}{2\cds}\int_{D_\partial }|y|^\gamma|\nabla w_{\rrh}|^2 \\
&=\frac{1}{2}\sum_{\{i \in \{1,\dots,n\}:x_i \notin \Old\}}\g(\rrh_i) -\frac{1}{2\cds}\int_{(\square_R^c \setminus \square_{T+4}^c)\times [-h,h]}|y|^\gamma|\nabla w_{\rrh}|^2\\
&\hspace{2cm}+\frac{1}{2\cds}\int_{( \square_{T+4})\setminus \New\times [-h,h]}|y|^\gamma|\nabla w_{\rrh}|^2 \\
&\leq \frac{1}{2}\sum_{\{i \in \{1,\dots,n\}:x_i \notin \square_{T+4}^c\}}\g(\tilde{\rr}_i) -\frac{1}{2\cds}\int_{(\square_R^c \setminus \square_{T+4}^c)\times [-h,h]}|y|^\gamma |\nabla w_{\tilde{\rr}}|^2-\sum_{\{i: x_i \in \Old\setminus \square_{T+4}^c\}}\g(\rrh_i) \\
&+\frac{1}{2\cds}\int_{(  \square_{T+4}\setminus \New)\times [-h,h]}|y|^\gamma|\nabla w_{\rrh}|^2 +\sum_{\{i: x_i \notin \square_{T+4}^c\}} \int_{\square_R^c \setminus \square_{T+4}^c} (\f_{\tilde{\rr}_i}-\f_{\rrh_i})(x-x_i)~d\mu  \\
& \leq C(n-\N)+CM,
\end{align*}
where $\tilde{\rr}_i$ is defined to be $\frac{1}{4}$ for $x_i \notin \Old$ and is kept fixed otherwise. This allows us to cancel all contributions of $\f$ and $g$ for $x_i \in \Old \setminus \square_{T+4}^c$, and bound the remaining contributions of $\f$ and $\g$ by $C(n-\N)$ since $\tilde{\rr}_i$ is bounded below for such $i$ and $\f_\eta$ is again controlled uniformly in $L^1$ for small $\eta$. We have bounded negative contributions of $L^2$ energy by zero, and $\frac{1}{2\cds}\int_{(\square_{T+4} \setminus \New)\times [-h,h]}|\nabla w_{\rr}|^2 \leq CM$. 
Finally, using Proposition \ref{prop:MElb} we can bound 
\begin{align*}
\sum_{j=1}^{\mn-n_\Old}\g(\rrh_j)&\leq C\left( \F(Z_{\mn-n_\Old}, \tilde{\mu}, \New \times [-h,h])+\mn-\N\right) \\
&\leq C\left(  \F(Z_{\mn-n_\Old}, \tilde{\mu}, \New \times [-h,h])+\tilde{\ell}R^{d-1}\right),
\end{align*}
controlling $\mn-\N=\tilde{\mu}(\New)\lesssim \tilde{\ell}R^{d-1}$ by our $L^\infty$ control on $\tilde{\mu}$, completing the argument.

\appendix

\newcommand{\de}{\partial}
\newcommand{\on}{\overline{\nabla}}
\newcommand{\fls}{(-\Delta)^s}

\newcommand{\So}{\mathcal{S}}
\newcommand{\Lo}{\mathcal{L}}

\newcommand{\Co}{{\mathcal{C}}}

\newcommand{\average}{{\mathchoice {\kern1ex\vcenter{\hrule height.4pt
width 6pt depth0pt} \kern-9.7pt} {\kern1ex\vcenter{\hrule
height.4pt width 4.3pt depth0pt} \kern-7pt} {} {} }}
\newcommand{\ave}{\average\int}
\newpage
\section{Some remarks about the fractional obstacle problem - by Xavier Ros-Oton}
Let $s\in(0,1)$, and consider the operator
\begin{equation} \label{L}
(-\Delta)^s u(x) = c_{n,s}\int_{\R^n}\big(u(x)-u(x+y)\big)\frac{dy}{|y|^{n+2s}},
\end{equation}
We want to study the \textbf{obstacle problem}
\begin{equation} \label{obst-pb}
\min\big\{(-\Delta)^s u,\, u-\psi\big\}=0\quad \textrm{in}\quad B_1\subset\R^n.
\end{equation}
The function $u$ is assumed to be bounded in $\R^n$.
In some cases, we want to consider the \emph{global} problem, which for $n>2s$ is
\begin{equation} \label{obst-pb-global}
\begin{split}
\min\big\{(-\Delta)^s u,\, u-\psi\big\}  = & \,\,0 \quad \textrm{in}\quad \R^n \\
u\longrightarrow & \,\,0\quad \textrm{at}\quad \infty.
\end{split}
\end{equation}
When $n=2s=1$ the global problem becomes
\begin{equation} \label{obst-pb-global-1}
\begin{split}
\min\big\{\sqrt{-\Delta} u,\, u-\psi\big\}  = & \,\,0 \quad \textrm{in}\quad \R \\
\frac{u(x)}{-\log|x|}\longrightarrow & \,\,\kappa\quad \textrm{at}\quad \infty,
\end{split}
\end{equation}
while when $n=1<2s$ we have
\begin{equation} \label{obst-pb-global-1<}
\begin{split}
\min\big\{(-\Delta)^s u,\, u-\psi\big\}  = & \,\,0 \quad \textrm{in}\quad \R \\
\frac{u(x)}{-|x|^{1-2s}}\longrightarrow & \,\,\kappa\quad \textrm{at}\quad \infty,
\end{split}
\end{equation}
for some constant $\kappa>0$.
Notice that in that case we need $\psi\ll -\log|x|$ or $\psi\ll-|x|^{1-2s}$ at $\infty$.

Moreover, the constant $\kappa$ is the total mass of $(-\Delta)^s u$ in $\R$.

\subsection{Known results}

The optimal regularity of solutions was first established in \cite{CSS} for $\psi\in C^{2,1}$, and these arguments were refined in \cite{CDS} in order to establish\footnote{There is actually an important detail that was omitted in \cite{CDS}, and is the fact that one needs to prove that blow-ups are \textit{convex}. This was not done in \cite{CDS}, but it follows from the results in \cite{FJ,RTW,CDV}.} the same result under minimal assumptions on the obstacle $\psi$.

\begin{theo}[\cite{CDS}]
Let $\psi\in C^{1+s+\delta}(B_1)$ for some $\delta>0$, and $u$ be any viscosity solution of \eqref{obst-pb}.
Then, $u$ is $C^{1+s}$ in $B_{1/2}$, with 
\[\|u\|_{C^{1+s}(B_{1/2})} \leq C\left(\|\psi\|_{C^{1+s+\delta}(B_1)} + \|u\|_{L^\infty(\R^n)}\right).\]
The constant $C$ depends only on $n$, $s$, and $\delta$.
\end{theo}

The regularity of free boundaries was also established for the first time in \cite{CSS} for $\psi\in C^{2,1}$, and under weaker assumptions%\footnote{The result was stated (without proof) in \cite{CDS} under the weaker assumption $\psi\in C^{1+s+\delta}$; however it is not clear how to do that, and more work would be needed. What is true is that the same proof in \cite{CSS} applies in case $\psi\in C^{1+2s+\delta}$, which is what we state here.} 
on the obstacle in \cite{CDS,RTW}.
Combining the results in \cite{CSS,CDS} with those in \cite{RS17} (see also \cite{CRS}) and in \cite{FRS}, we get the following.
Here, the function $d$ denotes the distance to the contact set $\{u=0\}$.

\begin{theo}[\cite{CSS,CDS,RS17,RTW}] \label{FB-thm}
Let $\psi$ be such that $\psi\in C^{1+2s+\delta}(B_1)$ for some $\delta>0$, and $u$ be any solution of~\eqref{obst-pb}.
Then, for every free boundary point $x_\circ\in \{u>\psi\}\cap B_{1/2}$ we have:
\begin{itemize}
\item[(i)] either 
\[(u-\psi)(x) = c_{x_\circ}d^{1+s}(x) + O(|x-x_\circ|^{1+s+\alpha}),\]
with $c_{x_\circ}>0$.

\item[(ii)] or
\[(u-\psi)(x) = O(|x-x_\circ|^{1+s+\alpha}).\]
\end{itemize}
Moreover, in case (i) the free boundary is $C^{1,\alpha}$ in a neighborhood of $x_\circ$.

Furthermore, we have 
\[\|(u-\psi)/d^{1+s}\|_{C^{\alpha}(\overline{\{u>\psi\}}\cap B_{1/2})} \leq C\left(\|(-\Delta)^s\nabla\psi\|_{L^\infty(B_1)} + \|u\|_{L^\infty(\R^n)}\right).\]
The constants $C$ and $\alpha>0$ depend only on $n$, $s$, and $\delta$.
\end{theo}

Notice that this gives a dichotomy between \textbf{regular points} (i), and \textbf{degenerate/singular points} (ii).
We have no information a priori on how many regular or singular points there could be.

Notice also that at any regular point $x_\circ$ we have
\[c_{x_\circ}=\lim_{\Omega\ni x\to x_\circ} \frac{u-\psi}{d^{1+s}},\]
while at degenerate/singular points this limit is zero.

Concerning the \emph{generic} regularity of free boundaries, the best known results are those in \cite{FR,FT,CC} and \cite{colombofigalli,KM}.

\begin{theo}[\cite{CC}]
Let $s\in(0,1)$ with $n>2s$ and $\psi\in C^5_c(B_1)$.
Let $u_t$ be the family of solutions of \eqref{obst-pb-global} with $\psi_t:=\psi+t$, for $t\in(-1,1)$.
Assume in addition $n\leq 3$.

Then, for almost every $t$, all points of the free boundary $\partial\{u_t>0\}$ are regular, in the sense of (i) above.
\end{theo}

On the other hand, we also have a nondegeneracy condition:

\begin{prop}[\cite{BFR}] 
Let  $s\in(0,1)$, $n>2s$, and  $\psi\in C^{3,\gamma}_c(\R^n)$ satisfying
\[\Delta \psi \leq 0 \quad \textrm{in}\quad \{\psi>0\}.\]
Let $u$ be the global solution of \eqref{obst-pb-global}.
Then, for any free boundary point $x_\circ$ we have 
\[\|u-\psi\|_{B_r(x_\circ)} \geq c_1r^2>0\]
for all $r\in (0,1)$.
The constant $c_1$ might depend on $u$, but not on the point $x_\circ$.
\end{prop}

%It seems likely that one can quantify the constant $c_1$ from \cite{BFR}, with a more explicit dependence on $\psi$.

Finally, the following result shows that without the sign assumption on $\Delta\psi$, such nondegeneracy may fail at \emph{every} free boundary point.

\begin{prop}[\cite{FR}]\label{weird}
Let $s\in(0,1)$ and $m\in \{1,2,3,...\}$.
Given any bounded~$C^\infty$ domain $\Omega\subset\R^n$, there exists an obstacle $\psi\in C^\infty(\R^n)$ with $\psi\to0$ at $\infty$, and a global solution to the obstacle problem  \eqref{obst-pb-global}, such that the contact set is exactly $\{u=\psi\}=\Omega$, and moreover at every free boundary point $x_\circ$ we have
\[(u-\psi)(x) = c_\circ d^{m+s}(x) + O(|x-x_\circ|^{m+1+s}),\]
where $d$ is the distance to $\Omega$.

In particular, when $m\geq2$, all free boundary points are of the type (ii) above.
\end{prop}

Notice that, thanks to the results in \cite{FR}, such type of degenerate solutions ($m\geq2$) are very rare, in the sense that generically we expect $m=1$.

\subsection{Some new results: nondegeneracy}

Our first goal is to prove some quantitative nondegeneracy results for the global obstacle problem \eqref{obst-pb-global}.
More precisely, we want to get uniform lower bounds for the constants $c_{x_\circ}$ in Theorem \ref{FB-thm}(i).

\begin{prop}\label{quantitative-regular}
Let $s\in(0,1)$, $n>2s$, and $\psi\geq0$ satisfying
\[\|\psi\|_{L^1}=1,\qquad {\rm supp}\,\psi\subset B_M,\qquad \|\psi\|_{C^{2+\gamma}(B_{\rho}(x_\circ))} \leq M, \qquad \Delta \psi \leq 0 \quad \textrm{in}\quad \{\psi>0\},\]
with $\gamma>s$.
Let $u$ be global solution to \eqref{obst-pb-global}, and assume that $x_\circ \in \partial\{u>\psi\}$ and that the free boundary can be written as a Lipschitz graph in $B_{\rho}(x_\circ)$, with Lipschitz constant bounded by $M$.

Then, $x_\circ$ is a regular point, and the constant in Theorem \ref{FB-thm}(i) satisfies 
\[c_{x_\circ} \geq \delta>0,\]
where $\delta$ depends only on $n$, $s$, $M$, and $\rho$.
\end{prop}

\begin{rem}
Notice that the main two assumptions in this result are:
\begin{itemize}
\item The global assumption $\Delta \psi\leq0$

\item The local assumption on the geometry of the free boundary
\end{itemize}
Both assumptions are somewhat necessary, in the following sense.
Without the first one we may have solutions with smooth free boundaries, for which all points are degenerate (recall Proposition \ref{weird}).
Without the second one, we may still have solutions with degenerate points (of order 2).
\end{rem}

In case $n\leq 2s$ we have the following analogue result.

Here, we denote $\Gamma(x)=-\log|x|$ if $s=\frac12$, and $\Gamma(x)=-|x|^{1-2s}$ if $s>\frac12$.

\begin{prop}\label{quantitative-regular-1}
Let $n=1$, $s\in[\frac12,1)$, $\gamma>s$, and $u$ be global solution to \eqref{obst-pb-global-1} or \eqref{obst-pb-global-1<} with $\psi$ satisfying
\[\|\psi\|_{C^{2+\gamma}([x_\circ-\rho,x_\circ+\rho]))} \leq M, \quad \psi'' \leq 0 \ \, \textrm{for}\ \,|x|<M, \quad \frac{u-\psi}{\kappa\Gamma} \geq \rho \quad \textrm{for}\quad |x|>M. \]
Assume that $x_\circ \in \partial\{u>\psi\}$ with $u>\psi$ in $(x_\circ,x_\circ+\rho)$ and $u=\psi$ in $(x_\circ-\rho,x_\circ]$.

Then, $x_\circ$ is a regular point, and the constant in Theorem \ref{FB-thm}(i) satisfies 
\[c_{x_\circ} \geq \delta>0,\]
where $\delta$ depends only on $s$, $M$, and $\rho$.
\end{prop}

To prove Propositions \ref{quantitative-regular} and \ref{quantitative-regular-1}, we will first need the following.

\begin{lem}\label{lem-BFR}
Let $s\in(0,1)$, $n\geq1$, and $\psi$ as in Proposition \ref{quantitative-regular} or \ref{quantitative-regular-1}.
Let $u$ be global solution to \eqref{obst-pb-global} or \eqref{obst-pb-global-1} or \eqref{obst-pb-global-1<}.
Then, for any free boundary point $x_\circ\in B_{\rho/2}(z)$ we have 
\begin{equation}\label{nondeg-quant}
\|u-\psi\|_{L^\infty(B_r(x_\circ))} \geq c_Mr^2>0
\end{equation}
for all $r\in (0,1)$, where  $c_M$ depends only on $n$, $s$, $\kappa$, $\rho$, $\delta$, and $M$.
\end{lem}

\begin{proof}
\noindent \textsc{Case 1.} Assume first $n>2s$.

As in \cite{BFR}, we consider $w:=(-\Delta)^s u\geq0$ in $\R^n$.
Since $ {\rm supp}\,w\subset \{u=\psi\}$, then for any $x_1\in \{u>\psi\}\cap B_{M+1}$ we have
\[-(-\Delta)^{1-s}w(x_1) = c_{n,s}\int_{\R^n}\frac{w(z)dz}{|x_1-z|^{n+2(1-s)}} \geq c\int_{B_M} w = c\int_{\R^n} |w| .\]

Now, by classical estimates for Riesz potentials (see \cite{Stein}), it follows that
\[\int_{\R^n} \big|(-\Delta)^s u\big| \geq c\|u\|_{L^{\frac{n}{n-2s}}_{\rm weak}(\R^n)} \geq c\|\psi\|_{L^{\frac{n}{n-2s}}_{\rm weak}(\R^n)} \geq c_M>0,\]
where we used that $u\geq \psi\geq0$ and that $\|\psi\|_{L^1}=1$.

Since $u$ is a global solution, we deduce
\[\Delta u = -(-\Delta)^{1-s}w \geq c_M>0 \quad \textrm{in}\quad \{u>\psi\}\cap B_{M+1}.\]

We now observe that at every free boundary point $x_\circ\in B_{\rho/2}(z)$ and for any $r\in(0,\rho/2)$ we have $u(x_\circ)=\psi(x_\circ)$ and
\[0=(-\Delta)^s u(x_\circ) \leq (-\Delta)^s \psi(x_\circ) \leq Cr^{2-2s}\|\psi\|_{C^{1,1}(B_r(x_\circ))} + C\int_{B_r^c(x_\circ)} \frac{\psi(x_\circ) - \psi(x_\circ+y)}{|y|^{n+2s}}. \]
Rearranging terms and using the assumptions on $\psi$, we get
\[c\|\psi\|_{L^1} -Cr^{2-2s}  \leq Cr^{-2s}\psi(x_\circ).\]
Choosing a small $r>0$ (depending only on $n$, $s$, $\rho$, and $M$), we deduce that 
\[\psi(x_\circ)\geq c_M>0\]
for some constant $c_M$.
Since $x_\circ$ was an arbitrary free boundary point, and by regularity of $\psi$, it follows that there exists $r_M>0$ such that 
\[{\rm dist}\big(\{\psi=0\}\cap B_{\rho/2}(z),\,\{u=\psi\}\cap B_{\rho/2}(z)\big)  \geq {r_M}>0.\]
%\[\psi>0 \quad \textrm{in}\quad \big(\{u>\psi\}\cap B_{\rho/2}(z)\big)+B_{r_M}.\]
Then, the proof of \eqref{nondeg-quant} finishes exactly as in \cite[Proof of Lemma 3.1]{BFR}.

\vspace{2mm}

\noindent \textsc{Case 2.} 
Assume now $n=1\leq 2s$. 
Then, $w:=(-\Delta)^s u\geq0$ satisfies
 ${\rm supp}\,w\subset \{u=\psi\}$, and then for any $x_1\in \{u>\psi\}$ with $|x_1|<M+1$
\[-(-\Delta)^{1-s}w(x_1) = c_{1,s}\int_{\R^n}\frac{w(z)dz}{|x_1-z|^{1+2(1-s)}} \geq c\int_{B_M} w=c\kappa>0.\]
Hence,
\[u'' = -(-\Delta)^{1-s}w \geq c\kappa> 0 \quad \textrm{in}\quad \{u>\psi\}\cap \{|x|<M+1\}.\]
Thus, around such free boundary point $x_\circ\in \{|x|<M-\rho\}$ we have $u''\geq c\kappa>0$ and $\psi''\leq 0$ in $(x_\circ-\rho,x_\circ+\rho)$, and in particular $(u-\psi)(x_\circ\pm r)\geq c\kappa>0$.
\end{proof}

We now prove the following.

\begin{proof}[Proof of Propositions \ref{quantitative-regular} and  \ref{quantitative-regular-1}]
By Lemma \ref{lem-BFR} and Theorem \ref{FB-thm} above, it follows that $x_\circ$ is a regular point, and the free boundary is a $C^{1,\alpha}$ graph in $B_{\rho/2}(x_\circ)$, for some $\alpha>0$.
Moreover, thanks to \cite[Theorem 1.2]{AR} this implies that the free boundary is a $C^{2+\gamma-s}$ graph in $B_{\rho/4}(x_\circ)$, and a quick inspection of the proof shows that its $C^{2+\gamma-s}$ norm is bounded by a constant depending only on $n$, $s$, $\gamma$, $M$, and $\rho$.
Then, by \cite[Theorem 1.4]{AR} (applied to the derivatives of $u-\psi$) we deduce that 
\[\left\|\frac{\partial_i (u-\psi)}{d^s}\right\|_{C^{1+\gamma-s}(B_{\rho/8}(x_\circ))} \leq C\]
and therefore
\[\big|u-\psi - c_{x_\circ} d^{1+s}\big| \leq C|x-x_\circ|^{2+\gamma},\]
with $C$ depending only on $n$, $s$, $M$, and $\rho$.

Combining this with the nondegeneracy condition \eqref{nondeg-quant}, we get 
\[ c_Mr^2 - c_{x_\circ} r^{1+s} \leq \left\|(u-\psi)-c_{x_\circ} d^{1+s}\right\|_{L^\infty(B_r(x_\circ))} \leq Cr^{2+\gamma}.\]
Hence, we deduce
\[c_Mr^{1-s} - Cr^{\gamma+1-s} \leq c_{x_\circ}.\]
Choosing $r>0$ such that $Cr^\gamma=\frac12 c_M$, the result follows.
\end{proof}

\subsection{Some new results: behavior of $(-\Delta)^su$ near the free boundary}

We next want to prove a new result concerning the local behavior of $(-\Delta)^su$ near the free boundary.

\begin{prop} \label{prop1}
Let $\psi$ be such that $\psi\in C^{1+2s+\delta}(B_1)$ for some $\delta>0$, and $u$ be any solution of~\eqref{obst-pb}.
Define 
\[d(x)={\rm dist}(x,\{u=\psi\})\qquad \textrm{and}\qquad d_-(x):={\rm dist}(x,\{u>\psi\})\]
Assume that  $x_\circ\in \{u>\psi\}\cap B_{1/2}$ is a regular free boundary point, i.e., we have
\[(u-\psi)(x) = c_{x_\circ}d^{1+s}(x) + O(|x-x_\circ|^{1+s+\alpha}),\]
with $c_{x_\circ}>0$, and the free boundary is $C^{1,\alpha}$ in a neighborhood of $x_\circ$.

Then, we have 
\[\big|(-\Delta)^s u(x) - \bar c_s c_\circ d_-^{1-s}(x)\big| \leq C|x-x_\circ|^{1-s+\alpha}\quad \textrm{in}\quad \{u=\psi\}\cap B_{1/2},\]
where $\kappa_s:= \bar c_s/(1-s)$ and $\bar c_s$ is given by \eqref{eq.constant0}.
The constant $C$ depends only on $n$, $s$, $\delta$, $\alpha$, $\|(-\Delta)^s\nabla \psi\|_{L^\infty}$, and the $C^{1,\alpha}$ norm of the free boundary in $B_{1/2}$.
\end{prop}

We will first need the following result, similar to \cite[Lemma 2.6]{FR2}.

\begin{lem} \label{lem1}
Let $s\in(0,1)$, and  $\Omega\subset \R^n$ be any $C^{1,\alpha}$ domain, with $0\in \partial \Omega$.
Let $d(x)={\rm dist}(x,\partial\Omega)$.
Assume that $f\in L^\infty(\Omega\cap B_1)$ and $v$ solves 
\[
\left\{
\begin{array}{rcll}
\fls v & = & f & \quad\textrm{in } B_1\cap \Omega\\
v & = & 0 & \quad\textrm{in } B_1\setminus \Omega,
\end{array}
\right.
\]
and let us define $c_\circ$ as the unique constant such that
\[
v(x)= c_\circ d^s(x) + O(|x|^{s+\alpha})\quad\text{in}\quad \Omega.
\]
Then, 
\[
\big|\fls v(x) -  \kappa_s c_\circ d^{-s}(x)\big|\leq  C|x|^{\alpha}\,d^{-s}(x) \quad\text{in}\quad B_{1/2}\setminus \Omega,
\]
where $\bar c_s$ is given by \eqref{eq.constant0} below, and $C$ depends only on $n$, $s$, $\Omega$, $\|f\|_{L^\infty}$.
\end{lem}

\begin{proof}
First, it follows from \cite[Lemma 2.6]{FR2} that for any $e\in \mathbb S^{n-1}$
\begin{equation}
\label{eq.constant0}
\fls (x_n)_+^{s} = \bar c_s (x_n)_-^{-s}\quad\text{with}\quad \bar c_s = -\frac{\Gamma(1+s)}{\Gamma(1-s)}.
\end{equation}

Then, by \cite{RS17}, we know that $v/d^s\in C^\alpha(\overline\Omega\cap B_{1/2})$.
Thus, if we define $c_\circ\coloneqq(v/d^s)(0)$ we then have
\[\big|(v/d^s)(x) - c_\circ\big|\leq C|x|^{\alpha},\]
and multiplying this by $d^s$ we get
\[
\big|v(x)- c_\circ d^s(x)\big| \leq C|x|^{s+\alpha}\quad\text{in}\quad \Omega.
\]
Notice that such expansion also implies that
\[\big|v(x) - c_\circ (x_n)_+^s\big| \leq C|x|^{s+\alpha},\]
where we assume that $\nu=e_n$ is the normal vector to $\partial\Omega$ at the origin.

Now, thanks to the previous expansion, and since $v\equiv0$ in $\Omega^c\cap B_1$, we find that for $x=-te_n \in \Omega^c$, with $t>0$,
\[(-\Delta)^s v(x) = c_\circ(-\Delta)^s (x_n)_+^s+O(t^{-s+\alpha}) = \bar c_s c_\circ t^{-s}+O(t^{-s+\alpha}),\]
where we used \eqref{eq.constant0}.
Since this can be done not only at the origin but at every boundary point $z\in \partial\Omega \cap B_{1/2}$, we deduce that for every $x=z-t\nu_z\in \Omega^c\cap B_{1/2}$
\[(-\Delta)^s u(x) = \bar c_s c_z t^{-s}+O(t^{-s+\alpha}).\]
For each $x\in\Omega^c$ we can choose $z\in\partial\Omega$ such that $|x-z|=d(x)$, and then we deduce
\[(-\Delta)^s u(x) = \bar c_s c_z d^{-s}(x)+O(d^{-s+\alpha}(x)).\]
Since $c_z=c_\circ+O(|z|^\alpha) = c_\circ+O(|x|^\alpha)$, we finally get
\[(-\Delta)^s u(x) = \bar c_s c_\circ d^{-s}(x)+O(|x|^\alpha)\,d^{-s}(x),\]
as wanted.
\end{proof}

We can now give the:

\begin{proof}[Proof of Proposition \ref{prop1}]
We want to apply Lemma \ref{lem1} to the functions $\partial_i(u-\psi)$.
More precisely, let $x\in \{u=\psi\}\cap B_{1/2}$, and let $z$ be its closest free boundary point, and denote $x=z-t_\circ\nu$ with $\nu\in \mathbb S^{n-1}$ and $t_\circ>0$.
Up to a rotation, we may assume $\nu=e_n$.
Then, it follows from Lemma \ref{lem1} (applied to $v=\partial_n(u-\psi)$) that 
\[\big|\fls v(x) - \bar c_s c_\circ t_\circ^{-s}\big|\leq  Ct_\circ^{\alpha-s}.\]
Moreover, the same holds for any point $y$ in the segment joining $x$ and $z$: if $y=z-te_n$, with $t\in[0,t_\circ]$, then
\[\big|\partial_n \fls (u-\psi)(z-te_n) - \bar c_s c_\circ t^{-s}\big|\leq  Ct^{\alpha-s}.\]
where we used the definition of $v$.
Integrating in $t$, and using that ---thanks to \eqref{obst-pb} and the fact that $(-\Delta)^s(u-\psi)$ is continuous---
\[\fls (u-\psi)(z)=0,\]
we deduce 
\[\left|\fls (u-\psi)(z-te_n) - \frac{\bar c_s}{1-s} c_\circ t^{1-s}\right|\leq  Ct^{1-s\alpha}.\]
The result follows by recalling that $t=d_-(x)$ and the fact that $x\in \{u=\psi\}\cap B_{1/2}$ was arbitrary.
\end{proof}

%%%%%%%%%%%%%%
\bibliographystyle{amsalpha}
\bibliography{Riesz.bib}{}

\begin{thebibliography}{ROTLW25}
\bibitem[CSS08b]{CSS}
\bysame, \emph{Regularity estimates for the solution and the free boundary of
  the obstacle problem for the fractional {L}aplacian}, Invent. Math.
  \textbf{171} (2008), no.~2, 425--461. \MR{2367025}
\end{thebibliography}
%%%%%%%%%%%%%%%%%%

%%%%%%%%%%%
\end{document}